%% file: BlindDeconvAugModified.tex
\documentclass[10pt]{article}
\usepackage{authblk}
\usepackage{fullpage}
\usepackage{amsmath, amsthm, amssymb, bm, calligra}
\usepackage{mathrsfs}
\usepackage{cite}
\usepackage{graphicx}
\usepackage{epstopdf}
\usepackage{multirow}
\usepackage{url}

\usepackage{subfigure} 
\usepackage{pgf}
\usepackage{tikz}
\usepackage{pgfplots}
\newlength{\fheight}
\newlength{\fwidth}
\usepackage{enumerate}
\usepackage{thmtools}
\usepackage{thm-restate}

\usepackage[hidelinks]{hyperref}

\newtheorem{theorem}{Theorem}

\newtheorem{proposition}{Proposition}

\newcommand{\vm}[1]{\boldsymbol{#1}}

\DeclareMathAlphabet{\mathcalligra}{T1}{calligra}{m}{n}

\renewcommand{\H}{{\operatorname{H}}}

\def\vec{\ensuremath\text{vec}}
\def\diag{\ensuremath\text{diag}}
\newcommand{\pr}{\ensuremath\mbox{\rm{Pr}}}

\def\bs{\ensuremath\boldsymbol}

\def\k{\kappa}

\def\hh{\vm{h}}
\def\mm{\vm{x}}

\def\mm{\vm{m}}

\def\l{\ell}

\begin{document}

\setlength\fheight{0.2\columnwidth}
\setlength\fwidth{0.42\columnwidth}

\title{From Blind deconvolution to Blind Super-Resolution through convex programming}
\author{Augustin Cosse}
\affil{Courant Institute of Mathematical Sciences and Center for Data Science, NYU, NYC.\\
D\'epartement de Math\'ematiques et Applications, Ecole Normale Sup\'erieure, Ulm, Paris,\\
Universit\'e Paris Sciences et Lettres.}
%\author[3]{Ali Ahmed}
%\affil[3]{Department of Mathematics, MIT, Cambridge, MA.}
%\author[3]{Laurent Demanet}
\date{\today}
\maketitle
\parindent 0pt

\maketitle

\begin{abstract}
This paper discusses the recovery of an unknown signal $\bs x\in \mathbb{R}^L$ through the result of its convolution with an unknown filter $\bs h \in \mathbb{R}^L$. This problem, also known as blind deconvolution, has been studied extensively by the signal processing and applied mathematics communities, leading to a diversity of proofs and algorithms based on various assumptions on the filter and its input. Sparsity of this filter, or in contrast, non vanishing of its Fourier transform are instances of such assumptions. The main result of this paper shows that blind deconvolution can be solved through nuclear norm relaxation in the case of a fully unknown channel, as soon as this channel is probed through a few $N \gtrsim \mu^2_m K^{1/2}$ input signals $\bs x_n = \bs C_n\bs m_n$, $n=1,\ldots,N,$ that are living in known $K$-dimensional subspaces $\bs C_n$ of $\mathbb{R}^L$. This result holds with high probability on the genericity of the subspaces $\bs C_n$ as soon as $L\gtrsim K^{3/2}$ and $N\gtrsim K^{1/2}$ up to log factors. Our proof system relies on the construction of a certificate of optimality for the underlying convex program. This certificate expands as a Neumann series and is shown to satisfy the conditions for the recovery of the matrix encoding the unknowns by controlling the terms in this series. We apply specific concentration bounds to the first two terms in the series, in order to reduce the sample complexities, and bound the remaining terms through a more general argument. The first term in the series is bounded through the subexponential Bernstein inequality. The second term is decomposed into two contributions, each corresponding to a fourth order gaussian chaos. The first contribution, containing the univariate fourth order monomials in the random vectors defining the subspaces $\bs C_n$, is bounded through a matrix version of the Rosenthal-Pinelis inequality. The second contribution, containing the cross terms, is bounded by using a decoupling argument for U-Statistics. An incidental consequence of the result of this paper, following from the lack of assumptions on the filter, is that nuclear norm relaxation can be extended from blind deconvolution to blind super-resolution, as soon as the unknown ideal low pass filter has a sufficiently large support compared to the ambient dimension $L$. Numerical experiments supporting the theory as well as its application to blind super-resolution are provided.
\end{abstract}
\hspace{1cm}\begin{minipage}{14.5cm}\date{\textbf{Acknowledgement.} AC was supported by the FNRS, FSMP, BAEF and Francqui Foundations. AC thanks MIT Math, Harvard IACS and The University of Chicago, for hosting him during this work. AC is grateful to Laurent Demanet and Ali Ahmed for interesting discussions as well as Joel Tropp for pointing out the matrix version of the Rosenthal-Pinelis inequality.}\end{minipage}

\section{\label{introduction}Introduction}

In \textit{Blind deconvolution}, a single unknown signal $x[n]$ is convolved with an unknown filter $h[n]$, resulting in the output signal $y[n] = h[n]\ast x[n]$. An additional additive noise $e[n]$ is also sometimes considered. The applications of this problem range from image processing, including medical and astronomical imaging as well as super-resolution~\cite{segall2003high}, to communication. The principal motivation behind this problem lies in the cost of high fidelity imaging and transmission devices. When images obtained through a lower quality device can be deblurred successfully, this is often more interesting than acquiring images through higher quality devices. Mathematically, the discrete problem reads
\begin{eqnarray}\mbox{find} \quad & \bs x, \bs h \in \mathbb{R}^L\nonumber \\
\mbox{subject to} \quad & \bs y[\ell]=(\bs h\ast \bs x)[\ell] + e[\ell],\quad \ell=1,\ldots ,L\label{blindDeconvZero}
\end{eqnarray}

Where $\bs h\ast \bs x$ is used to denote the discrete convolution. Problem~\eqref{blindDeconvZero} is ill-posed in the general case~\cite{bertero1998introduction}. As an illutration of the ill-posedness of the problem, note that both $\bs x$ and $\bs h$ are defined up to a scaling $\alpha$. A common approach to fix this uncertainty will be to assume that both $\|\bs x\|$ and $\|\bs h\|$ are unitary. Another improvement can be obtained by increasing the number of probing signals, thus replacing $\bs x$ with the input matrix $X = [\bs x_1,\ldots \bs x_N]$ and by requiring the probing signals $\bs x_i$ to "live" in lower dimensional subspaces. When multiple input signals are considered, the discrete convolution reads

\begin{equation}\label{blindDeconv0}
\bs y_n[\ell] = (\bs h\ast \bs x_n)[\ell]= \sum_{\ell' = 0}^{L-1}h[\ell']x_n[(\ell - \ell')\text{mod $L$}].
\end{equation}

Sensitivity to the noise also seems to remain a serious problem in blind deconvolution~\cite{campisi2007blind}. Several approaches have been introduced in order to first make this problem well-posed and then solve it efficiently. It seems that the oldest references to blind deconvolution go back to the 1970's~\cite{stockham1975blind, cannon1976blind}. Most approaches from engineering and signal processing assume a kernel $\bs h$ with small support which is reasonable when dealing with applications in communications where the filter is usually representing a blur. When dealing with Green operators such as in acoustics or inverse scattering, the support of the perturbation $m_1(\bs x)$ is not necessarily small anymore and those assumptions will not hold. Some of the most relevant attempts at solving blind deconvolution are listed in section~\ref{connectionsExistingWork} below. We essentially skim through those approaches. The list is non-exhaustive. For a more detailed summary we suggest to turn to~\cite{campisi2007blind} and references therein. 

\subsection{An interesting experiment}

When considering a single input, ill-posedness of the blind deconvolution problem manifests itself through an erroneous decomposition of the recovered image into a contribution to the filter and a weaker (partially recovered) image.  Such a wrong recovery is illustrated in Figs.~\ref{SingleFrameFail} and~\ref{SingleFrameFail2} below.
When considering multiple snapshots that are sufficiently distinct from each other yet sufficiently compressible in some appropriate domain; such as multiple slices of a same image volume, or multiple snapshots generated from a time sequence; the ill-posedness of the problem is mitigated and the inacurate splitting of the recovered image into a perturbation to the filter and some incomplete recovery of the snapshot itself is not permitted anymore. Indeed such error terms would accumulate in the subsequent slices and affect the other measurements of the sequence. In other words, if the first image is inaccurately recovered, then the filter is affected as well and leads to errors in the second frame as well. Such a wrong reconstruction of the second frame will then imply a modification of the filter to match the measurements which might not match the modification implied by the erroneous reconstruction of the first frame. The error is thus gradually reduced throughout the sequence of inputs. In practice, only a few ($N \gtrsim 1$) sufficiently distinct snapshots seem to be enough to prevent incorrect reconstruction.

\begin{figure}[h!]\centering
  \begin{minipage}{0.8\textwidth}
    \hspace{0.6cm}\includegraphics[trim = 0cm 0cm 0cm 7cm, clip, width=.9\textwidth]{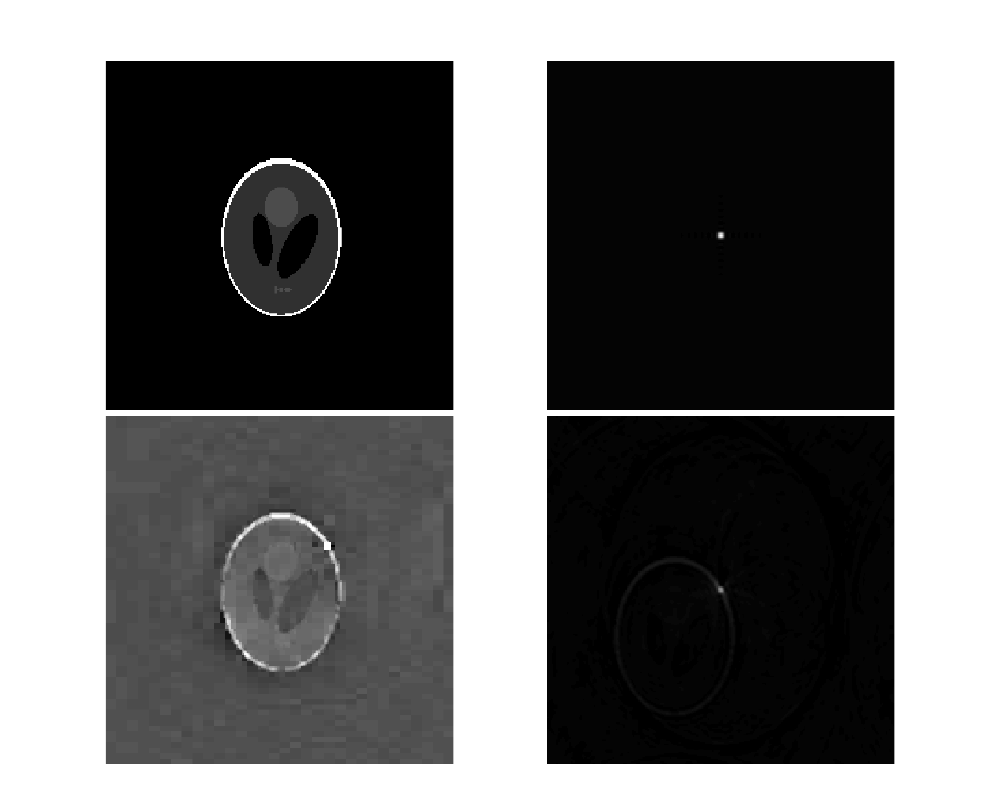}
  \end{minipage}
\caption{\label{SingleFrameFail}Blind deconvolution performed through the nuclear norm minimization program~\eqref{nuclearNormDeconvolutionLinearMap} with a single frame (i.e $n=1$). (Top left) Original image, (top right) point spread function corresponding to the ideal low pass filter, (bottom left) Recovered image, (bottom right) recovered low pass filter. Note the ghost image that appears in the recovered filter. This phenomenon is further highlighted in Fig.~\ref{SingleFrameFail2} below. When only one probing image is considered, the algorithm is not able to separate the low pass filter from the image. In the recovered image shown at the bottom, some of the finer details can also be seen to have completely disappeared. This image should be compared to Fig.~\ref{superResolution1} and~\ref{superResolution2} below where multiple slices of a same volume are used.}
\end{figure}

\begin{figure}[h!]\centering
\input{FrequencyExtensionSingleFrame2b}
\caption{\label{SingleFrameFail2}Further illustration of the overlap between the recovered low pass filter and the recovered image resulting from an insufficient number of samples. When considering a single image, nuclear norm minimization is unable to split the input image from the filter. The use of a few (sufficiently distinct) compressible images makes the problem better posed. This phenomenon is also highlighted by the phase diagrams given in Fig.~\ref{figure1} below.}
\end{figure}
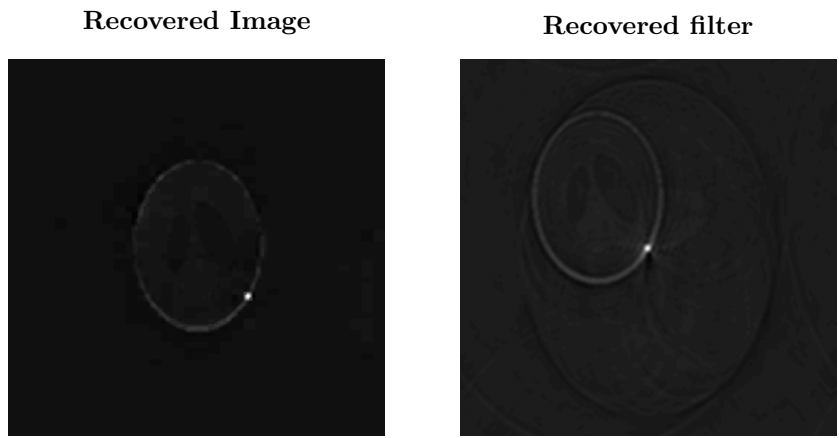

\subsection{\label{connectionsExistingWork}Connections with existing work}

Numerous approaches have been introduced to tackle the blind deconvolution problem. Without entering excessively into the details, we review the most important ones below. An important number of results rely on alternating methods that estimate each of the unknown signals sequentially rather then jointly. A common trend among those approaches is to use a statistical description of the input signal $\bs x$ which enables the use of the Bayesian framework and associated maximum a posteriori (MAP) estimators based on either joint or marginal probability distributions. 

\begin{itemize}
\item In~\cite{ayers1988iterative} Ayers and Dainty compute the Fourier transform of either of the unknown signals before dividing the output of the convolution by this Fourier transform to get an estimate of the Fourier transform of the second unknown. The method then proceeds iteratively.
\item Introductions to Bayesian methods in blind deconvolution can be found in~\cite{levin2009understanding, campisi2007blind}. In those papers, Levin et al. compare the joint MAP estimator against the MAP estimator for the filter only. They show that the former, even when combined with sparsity prior on the input signal gradient will generally favor recovery of a blurred input instead of a blurring filter thus leading to an inaccurate reconstruction. When using the MAP estimator for the filter only, however, for a gaussian prior on the input, it is possible to recover the true filter. They suggest to solve the resulting formulation through an expectation maximation algorithm. A similar approach is developed in~\cite{lagendijk1989blur} by Lagendijk et al. For both estimators, penalty terms favoring sparsity in the derivatives are considered. Other use of penalty functions include~\cite{you1999blind} (anisotropic diffusion) and~\cite{chan1998total} (TV regularization). 
\item In~\cite{reeves1992blur}, Reeves and Mersereau model the image as an autoregressive (AR) process and the blurr as a moving average (MA) process thus reducing problem~\eqref{blindDeconv0} to the recovery of $\bm H$ and $\bm A$ satisfying $\bs y = \bm H(\bm I-\bm A)^{-1}n_1 + n_2$ where $n_1$ and $n_2$ are noise models and $\bm H$ and $\bm A$ are respectively the filter and image (AR) parameters. The approach is somewhat equivalent to the Wiener filter. The original output $\bs y$ but one pixel is used to get an estimate of the filter and input image parameters and this estimate is then used to compute the misfit at the remaining pixel. The objective is the average over all pixels of this one pixel difference.
\item Another interesting recent line of work uses the bilinearity of the multiplication in Fourier space to rewrite the blind deconvolution/amplitude calibration problem as a linear problem in the inverse of the filter $h$ on the one hand and the input signals on the other. Early references along that line include~\cite{bilen2014convex} as well as~\cite{gribonval2012blind}, where the authors  provide numerical experiments in the case of inputs defined as $S$-sparse combinations of known basis vectors. The authors consider a filter with unknown phase, amplitude or both. The experiments show that the recovery is improved when the number of input signals is increased. This line of work extends the early approach from signal processing (see the aforementioned~\cite{ayers1988iterative}). Recovery guarantees along that line were recently provided in~\cite{ling2016self} as well as~\cite{cosse2017note} where the subspace assumption is removed. Despite its interest in terms of the lack of incoherence condition on the filter, the principal disadvantage of the least squares approach is its lack of stability and, as a consequence, its inability to handle functions with vanishing or very small entries. 

%The recovery guarantees given in~\cite{ling2016self} are given on the signal $(1/\bs h, \;\bs x)$ and not on $(\bs h, \bs x)$ which means that the error on $\bs h[k]$ is really $\mathcal{O}(1/(\varepsilon \|\bs h_0^{-1}\|+ 1/\bs h_0[k])$. Entries of small magnitude thus lead to large instabilities, even for relatively low levels of noise compared to the filter as a whole. 

%The motivation behind this work lies in the fact that channels with vanishing or small entries are typically those encountered in practical situation such as denoising in MRI, where the support of the filter is compact, or source imaging where the Green function is oscillating and thus exhibits a large number of vanishing entries. Finally, typical filters in communcation, will exhibit small and vanishing entries as well.
\end{itemize}

The approach developed in this paper follows the recent line of work on semidefinite programming relaxations of non linear and non convex problems in order to derive robustness and unconditional convergence garantees. Some of the most relevant results along that particular line of work include the following
\begin{itemize}
\item In~\cite{ahmed2014blind} Ahmed et al. certify recovery of both $\bs h$ and $\bs x$ up to a scaling through minimization of the nuclear norm of the matrix $\bm X$ used as a proxy for the rank one matrix $\bs h\bs x^*$. They consider a single input signal and construct their certificate of optimality using the \textit{golfing scheme}~\cite{gross2011recovering}. The main difference with the result of this paper is that their result requires both the impulse response and the input signal to "live" in lower dimensional subspaces. They certify recovery through the construction of a dual certificate whenever $L\gtrsim K_1+K_2$ up to log factors. $K_1$ and $K_2$ denote respectively the dimensions of the filter and input  subspaces. This paper is the last of a line of work by Ahmed and Romberg including~\cite{ahmed2013compressive} and~\cite{ahmed2015compressive2} where the authors study mixing and recovery of signals generated from a few basis elements and~\cite{ahmed2015compressive} where the authors consider optimal sampling and recovery of a similar ensemble. In each case the recovery of the matrix encoding the ensemble is carried out through nuclear norm minimization.

\item The idea from~\cite{ahmed2014blind} was further refined in~\cite{ahmed2016leveraging} where the authors remove the subspace assumption on the filter, while maintaining a sparsity constraint together with sufficient incoherence. In this paper,  a series of inputs $\bs x_n$ has to be recovered from convolutions of the form $\bs y = \bs w\ast \bs x_n$ where $\bs w$ is assumed to be $S$-sparse. As in this paper, the inputs are assumed to live in known $K$-dimensional subspaces. Recovery of both the inputs $\bs x_n$ and the filter $\bs w$ is guaranteed as soon as $K+S\log^2S\lesssim L/\log^4(LN)$ and $N\gtrsim \log^2(LN)$ for inputs $\bs x_n$ that are sufficiently distinct from each other. The relation of~\cite{ahmed2014blind} to the present paper is further discussed in section~\ref{sec:num}.

\item In~\cite{ling2015self}, Ling et al. study calibration problems of the form $\bs y = \bm D\bm A\bs x$ for unknown ($S$-sparse) signals $\bs x$ and diagonal calibration matrix $\bm D\in \mathbb{R}^{L\times L}$ defined as $\bm D = \mbox{diag}(\bm B\bs h)$ with $\bm B\in \mathbb{R}^{L\times M}$ such that $\bm B^*\bm B = \bm I$ and $\bs h\in \mathbb{R}^M$. They show that both $\bs x$ and $\bm D$ can be recovered through minimization of the $1$-norm of the matrix $\bm X$ used as a proxy for the rank one matrix $\bm X_0 = \bs h_0\bs x_0^T$ encoding the unknowns of the original problem. Exact recovery is certified for $\bm A\in \mathbb{R}^{L\times K}$ ($L<K$) random gaussian whenever $L\gtrsim SM$ up to $\log$ factors and for $\bm A$ random Fourier for comparable scalings. Corresponding stability results are provided. 

\item In~\cite{ling2015blind} the same authors study the recovery of input signals $\bs x_i$ and filters $\bs h_i$ when measurements $y_i$ are given by the sum of the outputs of the filters. They prove recovery as soon as $L\gtrsim R\max(K,\mu_h^2 N)$ where $R$ is the number of channels, $K$ is a bound on the support of the filters and $N$ is the size of the inputs. 

\item In~\cite{kech2016optimal}, Kech et al. study the general bilinear inverse problem with sparsity or subspace constraints on both $\bs h$ and $\bs x$ and certify injectivity of the bilinear map for $L\gtrsim S_1+S_2$ (when considering sparse vectors) or $L\gtrsim N_1 + N_2$ (when considering signals living in known subspaces). $S_1,S_2$ and $N_1,N_2$ respectively denote the sparsity and subspace dimension of each of the signals involved in the bilinear product.

\end{itemize}

Besides convexifying and linearizing approaches,  a recent trend from the statistics and optimization communities studies initialization and convergence guarantees of non convex optimization algorithms. The first results along that trend were obtained in~\cite{candes2015wirtinger}. Similar ideas have been applied to blind deconvolution by Cambareri et al.~\cite{cambareri2016non} 
as well as Li et al.~\cite{li2016rapid}. 

The proof of our main result follows the general idea developed in~\cite{candes2009exact} by Cand\`es et al., although, transposing this idea to the blind deconvolution framework.

Finally, an incidental result following from the lack of assumptions on the filter in the main result of this paper, is the extension of nuclear norm minimization from blind deconvolution to blind super-resolution. This idea is further discussed in section~\ref{sec:num}. The use of convex programming in blind super-resolution was recently discussed by Yang et al. in~\cite{yang2016super,yang2016non}. In these papers, the authors derive recovery guarantees for the atomic norm, in the case of the modulation of an unknown spike train satisfying a minimum separation condition and modulated by random waveforms generated from a random low dimensional subspace whose basis vectors satisfy some incoherence condition. 

Our paper derives recovery guarantees for a multiple inputs formulation similar to the one in~\cite{bilen2014convex} by using the nuclear norm relaxation framework introduced in~\cite{ahmed2014blind}. In the Fourier domain, problem~\eqref{blindDeconv0} can read as an affine rank minimization problem which can in turn be solved through nuclear norm minimization (see~\cite{fazel02ma, fazel2001rank,recht10gu}). Certifying recovery through the nuclear norm minimization program then relies on the construction of a so-called dual certificate. The main conclusion of this paper is that the restrictions on the impulse response $\bs h$ in~\cite{ahmed2014blind} can be lifted by considering a few $N\gtrsim \mu^2_m \sqrt{K}$ (up to log factors) input signals rather than a single one. The recovery also requires a sufficient number of measurements ($L\gtrsim K^{3/2}\mu^2_h$). 

The rest of the paper is organized as follows: Section~\ref{notations} derives the mathematical formulation of the problem and introduces the underlying notations. Section~\ref{secMainResult} states the main result of the paper.
Section~\ref{sec:strategy} outlines the main argument of the proof. The proof relies on several lemmas. Most of them are only stated in section~\ref{sec:strategy} and their proofs are detailled in section~\ref{auxiliaryLemmasProof}. Numerical experiments supporting the statement of Theorem~\ref{theorem:BlindDeconvTh} are provided in section~\eqref{sec:num} in which applications to super-resolution and medical imaging are discussed. Finally, the result of Theorem~\ref{theorem:BlindDeconvTh} is commented in section~\ref{sec:conclusion} together with a few related open questions.

%
%\subsection{Source imaging}
%
%In source imaging, including passive imaging~\cite{garnier2016passive,garnier2009passive}, the support of an unknown source is to be recovered from measurements at given receivers. Considering the wave equation,
%
%\begin{align}
%\frac{1}{c^2(\bs x)} \frac{\partial^2 u(\bs x,t)}{\partial t^2} - \Delta u(\bs x,t) = n(t,\bs x),
%\end{align}
%
%the source term $n(t,\bs x)$ has to be estimated from the the wavefield $u(t,\bs x_r)$ at some receivers positions, $\{\bs x_r\}$. A commom approach to this problem is to use the empirical correlation,
%
%\begin{align}
%C(\bs x_1,\bs x_2,\tau) = \frac{1}{T}\int_{0}^T  u(t,\bs x_1)u(t+\tau,\bs x_2)\;dt.
%\end{align}
%
%the wavefield is then expressed from the source $\bs n(t,\bs x)$ through the Green operator, as
%
% \begin{align}
%u(\bs x_1,t) = \int G(\bs x_1,\bs x,t) n(\bs x,t)d\bs x.\label{GreenOp}
%\end{align}
%
%This approach is also known as interferometry (see for example~\cite{demanetjugnon,jugnon2013interferometric}). Expression~\eqref{GreenOp} can be simplified when the medium is homogeneous ($c(\bs x) = c_0$). When the medium is complex, computing the Green operator might be very complicated and if this operator is unknown, it is not possible to express the wavefield directly from the sources. In this case passive imaging can be seen as a blind deconvolution problem where multiple unknown sources are probing a medium whose Green function is itself unknown. Moreover, in this case, it is clear that the Green function is not sparse and will vanish on some of the domain points.

\subsection{\label{notations}Problem Formulation}

As explained above, assuming both an unknown impulse response and unknown probing signals is ill-posed in the general case. For this reason  this chapter considers the problem of estimating the channel response $\bs h$ and the input signals whenever those input signals belong to $K$-dimensional generic subspaces $\mathcal{C}_n$ of $\mathbb{R}^L$. In other words we reduce the number of unknowns from $L\times N$ to $K\times N + L$ for $LN$ measurements. Intuitively, it should now be clear that by taking $L$ sufficiently large, the problem should become better posed. Each of the generic subspace $\mathcal{C}_n$ will be represented by a corresponding basis matrix $\bm C_n$ with gaussian i.i.d entries such that $\bs C_n[k,\ell] \sim \mathcal{N}(0,\frac{1}{L})$, thus with variance $\sigma^2 = 1/L$. Consequently, the input signals can thus read $\bs x_n = \bs C_n\bs m_n,\quad n=1,\ldots, N$ for some coefficients $\bs m_n \in \mathbb{R}^K$. For notational convenience, we let $\bs m$ denote the vector concatenating each of the coefficient vectors $\bs m_n$ as $\bs m = [\bs m_1,\ldots, \bs m_N]$. We will work in the Fourier domain since it turns the convolution~\eqref{blindDeconv0} into a Hadamard product. Let us introduce the DFT matrix $\bs F$ whose rows are defined as
\begin{align}\label{gr}
\bs f_\ell [k]= \frac{1}{\sqrt{L}}e^{-2\pi i (k-1) \ell/L},\quad \ell, k \in\{1,\ldots, L\}.
\end{align}
Note that the DFT matrix $\bs F$ satisfies $\bs F\bs F^* = \bs F^*\bs F = \bs I$ where $\bs F^*$ is used to define the inverse discrete Fourier tranform. We will use the notation $\hat{\bs y}$ to denote the Fourier transform of $\bs y$ and $\bs y^*$ to denote the conjugate transpose of $\bs y$. Let $[L] = \{1, \ldots, L\}$ and $[N] = \{1,\ldots, N\}$. We use $\hat{\bs c}_{\ell,n}$ to represent the vector defined from the basis $\bs C_N$ by putting the $\ell^{th}$ coefficient of the Fourier transform of the $k^{th}$ column of $\bs C_n$ at position $(n-1)K + k$ for each $1\leq k\leq K$ and zeros everywhere else. Using the canonical vector $\bs e_\ell$, $\bs c_{\ell,n}$ can thus read as $\hat{\bs c}_{\ell,n}[k] = \bs e_n\otimes (\bs f_\ell^*\bs C_n) $. Given this definition, and with the subspace decomposition $\bs x_n = \bs C_n\bs m_n$, in the Fourier domain, problem~\eqref{blindDeconv0} can be written as
\begin{eqnarray}\begin{split}	\label{deconv2}
\hat{\bs y}_n[\ell] &= \sqrt{L}\hat{\bs h}[\ell]\hat{\bs x}_n[\ell],\qquad  &\qquad(\ell,n)\in [L]\times [N],\\
 &= \sqrt{L} (\bs F\bs h)[\ell] (\bs F\bs C_n \bs m_n)[\ell],&\qquad(\ell,n)\in [L]\times [N],\\
& = \sqrt{L}\langle \bs f_\ell, \bs h\rangle\langle \hat{\bs c}_{\ell,n}, \bs m\rangle,&\qquad(\ell,n)\in [L]\times [N].
\end{split}
\end{eqnarray}
Now using the Frobenius inner product,
\begin{eqnarray}\begin{split}	\label{deconv3}
\hat{\bs y}_n[\ell] & = \sqrt{L} \langle \bs f_\ell \hat{\bs c}_{\ell,n}^*, \bs h\bs m^* \rangle, &\qquad(\ell,n)\in [L]\times [N].
\end{split}
\end{eqnarray}
The scaling $\sqrt{L}$ comes from equality between the Fourier transform of the convolution $y$ and the product of the Fourier transforms of $x$ and $h$. From now on we will include the scaling in the definition of $\hat{\bs c}_{\ell,n}$,  $\hat{\bs c}_{\ell,n} :=\hat{\bs c}_{\ell,n}\sqrt{L}$. 
We let the rank one matrices $\bm A_{\ell,n}$ encode the outer product $\bs f_\ell \hat{\bs c}_{\ell,n}^*$ so that the last line of~\eqref{deconv3} can be written compactly as $\hat{\bs y}_n[\ell] = \langle \bs A_{\ell,n},\bs h\bs m^*\rangle$. Following definition of those matrices, problem~\eqref{deconv3} can now be recast as the following affine rank minimization problem,  
\begin{subequations}\label{affineRankDeconvolution}
\begin{eqnarray}
\underset{\bm X}{\text{minimize}} & \quad \mbox{rank}(\bm X)& \\
\mbox{subject to}&\quad  \hat{\bs y}_{\ell,n} = \langle \bs A_{\ell,n},\bm X \rangle. &\qquad(\ell,n)\in [L]\times [N].
\end{eqnarray}
\end{subequations}
The matrix $\bm X$ is used as a proxy for the rank one matrix $\bs h\bs m^* $ encoding the impulse response and the coefficient vector. If the original problem is well posed, then the unique solution to problem~\eqref{affineRankDeconvolution} is given by $\bs h\bs m^*$. Affine rank minimization is hard in the general case. A common approach that has revealed efficient so far when dealing with affine rank minimization is to replace the hard minimal rank constraint by the minimization of the nuclear norm (see~\cite{fazel02ma, recht10gu}). Following this trend from convex optimization, we thus relax problem~\eqref{affineRankDeconvolution} into
\begin{subequations}\label{nuclearNormDeconvolution}
\begin{eqnarray}
\underset{\bm X}{\text{minimize}} & \quad \|\bm X\|_*& \\
\text{subject to} &\quad  \hat{y}[\ell,n] = \langle \bs A_{\ell,n},\bm X \rangle &\qquad(\ell,n)\in [L]\times [N].
\end{eqnarray}
\end{subequations}
We will sometimes write problem~\eqref{nuclearNormDeconvolution} compactly by introducing the linear map $\mathcal{A}\;:\;\mathbb{C}^{L\times KN} \mapsto \mathbb{C}^{L N}$ and defined from the matrices $\bm A_{\ell,n}$ as 
\begin{align}
\{\mathcal{A}(\bm X)\}_{(\ell,n)} \equiv\langle \bm A_{\ell,n}, \bm X\rangle = \mathcal{A}_{\ell,n}(\bm X).\qquad (\ell,n)\in [L]\times [N]. \label{DefinitionLinearMap}
\end{align}
Using definition~\eqref{DefinitionLinearMap}, formulation~\eqref{nuclearNormDeconvolution} can now read compactly as 
\begin{subequations}\label{nuclearNormDeconvolutionLinearMap}
\begin{eqnarray}
\underset{\bm X}{\text{minimize}} & \quad \|\bm X\|_*& \\
\text{subject to} &\quad  \mathcal{A}(\bm X) = \bs y.
\end{eqnarray}
\end{subequations}
The next section summarizes the strategy that will be used to certify recovery of the rank one matrix $\bm X_0 = \bs h\bs m^* $ through the convex relaxation~\eqref{nuclearNormDeconvolution}. We assume without loss of generality that $\|\bs h\| = 1$ and $\|\bs m\|=1$ which implies $\|\bm X_0\|_F = \|\bs h\bs m^*\| =1$. The next section gives the main result of the paper.

\subsection{\label{secMainResult}Main Result}

The main result of this paper shows that exact recovery through the nuclear norm relaxation~\eqref{nuclearNormDeconvolutionLinearMap} occurs with high probability (on the genericity of the subspaces) as soon as the dimension $K$ of the subspaces, the ambient dimension $L$ and the number of input signals $N$ obey $L\gtrsim K^{3/2} \mu^2_h$ and $N \gtrsim K^{1/2} \mu^2_m$ up to log factors. The coherence factors $\mu^2_h$ and $\mu^2_m$ measure the respective spreading of $\bs h$ and $\bs m$ and are defined as  
\noindent\begin{minipage}{0.5\linewidth}
\begin{equation}
\mu^2_m = \sup_n\frac{N\cdot \|\bs m_n\|^2}{\|\bs m\|^2},\label{coherencem}
\end{equation}
\end{minipage}%
\begin{minipage}{0.5\linewidth}
\begin{equation}
\mu^2_h = \sup_\ell \frac{L\cdot |\hat{h}[\ell]| }{\|\bs h\|^2}\label{coherenceh}
\end{equation}
\end{minipage}%
We can now state the main result of this paper.
\begin{theorem}\label{theorem:BlindDeconvTh}
Let $\bs C_{n}$ be random matrices of size $L\times K$ with gaussian independent and identically distributed (i.i.d) entries, i.e, $C_{n}[\ell,k]\sim\mathcal{N}(0,1/L)$ representing corresponding $K$-dimensional subspaces. Let $\bs m_n$ denote $K$-dimensional vectors representing the $L$-dimensional input signals $\bs x_n$ in the bases $\bs C_n$ so that $\bs x_n = \bs C_n\bs m_n$. The nuclear norm relaxation of~\eqref{nuclearNormDeconvolutionLinearMap} recovers the rank one matrix $\bm X_0 = \bs h\bs m^*$ where $\bs m = [\bs m_1,\ldots, \bs m_N]$, from the measurements $\bs y_n=\bs h\ast \bs x_n$, with probability at least $1 - c_1(LN)^{-\beta} - \sqrt{\left(\frac{K\mu_h^2}{L} + \frac{\mu_m^2}{N}\right)}$, where $c_1,\beta$ are known positive constants, as soon as 
\begin{equation*}
L\gtrsim K^{3/2}\mu^2_h ,\qquad N\gtrsim K^{1/2}\mu^2_m.
\end{equation*}
\end{theorem}

Numerical experiments are provided in section~\ref{sec:num} to support the claim of Theorem~\ref{theorem:BlindDeconvTh}. The proof of this theorem relies on the construction of a dual certificate and its analysis by means of a Neumann series and appropriate concentration bounds. The exponents arising in the sample complexities are due to the proof technique used to control the terms appearing in the expansion defining the certificate. Explicit concentration bounds are used for the first two terms, leading to sample complexities of $K^{1/2}$ and $K^{3/2}$. It is likely that extending those specific bounds to higher order terms in the Neumann series can further reduce those exponents.

\section{\label{sec:strategy}Proof of Theorem~\ref{theorem:BlindDeconvTh}}

This section outlines the argument that will be used to certify exact recovery of the matrix $\bs h\bs m^*$ through the nuclear norm minimization program~\eqref{nuclearNormDeconvolution}. It is organized as follows. Section~\ref{sec:optimalityBD} derives uniqueness and optimality conditions for problem~\eqref{nuclearNormDeconvolution}. Proving optimality and uniqueness of the solution $\bs h\bs m^*$ for the nuclear norm minimization program~\eqref{nuclearNormDeconvolution} is equivalent to exhibiting a dual vector $\bs Y$ defining a hyperplane separating the nuclear norm ball and the affine subspace $\mathcal{S}\equiv\left\{\bs X\;:\;\mathcal{A}(\bs X) = \bs b\right\}$. Proving Theorem~\ref{theorem:BlindDeconvTh} can thus be done by first suggesting a candidate for $\bs Y$
and then, proving that $\bs Y$ is indeed corresponding to a separating hyperplane (i.e that it satisfies the conditions of section~\ref{sec:optimalityBD}). Section~\ref{subgaussianAndSubexponential} recalls the important concepts that are needed to bound subexponential and subgaussian random variables. Such variables will appear extensively in the construction of $\bs Y$. Section~\ref{sectionCandidateDualCert} introduces a candidate for the normal vector $\bs Y$ and then explains how this particular ansatz can be proved to satisfy the optimality conditions from section~\ref{sec:optimalityBD} and thus to correspond to a valid separating hyperplane for the underlying nuclear norm and affine subspace.

\subsection{\label{sec:optimalityBD}Optimality and Uniqueness}

Proving optimality of the solution $\bm X_0 = \bs h\bs m^*$ is equivalent to exhibiting a (dual) vector $\bm Y$ that would be simultaneously normal to one of the supporting hyperplanes to the nuclear norm ball $\|\bm X\|_*\leq 1$ at $\bs h\bs m^*$ and to the affine subspace $\mathcal{A}(\bm X) = \bs y$. This is also equivalent to finding a subgradient to the nuclear norm that is normal to the affine subspace. Such a vector implies that the affine subspace is locally tangent to the nuclear norm ball (or equivalently to the nuclear norm) at $\bs h\bs m^*$ and so, that any further reduction in the nuclear norm would require to leave the affine subspace. Another way to understand the geometric meaning of such a normal vector is that it defines a separating hyperplane between the nuclear norm ball (resp. nuclear norm) and the affine subspace. We now derive the mathematical conditions defining the normal vector. The tangent space $T$ of the nuclear norm at $\bs h\bs m^*$ can be defined as 
\begin{align}
T = \left\{\bs h\bs x^* + \bs y\bs m^*, \bs x\in \mathbb{R}^{KN}, \bs y\in \mathbb{R}^L\right\}\label{tangentSpaceNN}
\end{align}
The projection $\mathcal{P}_T(\bm Y) = (\bm Y)_T$ onto this tangent space  can be defined as (see for example~\cite{candes2009exact})
\begin{align}
\mathcal{P}_T(\bm Y) = \bs h\bs h^*\bm Y + \bm Y\bs m\bs m^* - \bs h\bs h^*\bm Y\bs m\bs m^*.\label{projectorTBD}  
\end{align}
As a consequence, the corresponding projector onto the orthogonal complement  of $T$, $\mathcal{P}_T^\perp(\bm Y)$ is defined simply as $(\bm Y)_T^\perp = \mathcal{P}_T^\perp(\bm Y) = \bm Y - \mathcal{P}_T(\bm Y)$. From the definition of those two projectors, for a general matrix $\bm X = \bm U\bm \Sigma \bm V$, the subdifferential of the nuclear norm is known to be defined as~\cite{recht2010guaranteed}
\begin{align}
\partial \|\bm X\|_* =\left\{\bm Y\;|\;\mathcal{P}_T(\bm Y)=\bm U\bm V,\;\|\mathcal{P}_T^\perp(\bm Y)\|_\infty \leq 1\right\}\label{eq:subdifferentialNN}
\end{align}
Here $\|\bm X\|_\infty$ denotes the operator norm of the matrix $\bm X$. From definition~\eqref{eq:subdifferentialNN} and the discussion above, proving optimality of $\bm X_0$ thus reduces to exhibiting a vector $\bm Y$ in the range of $\mathcal{A}^*$ such that
\begin{align}
(\bm Y)_T = \bs h\bs m^*,\qquad \|(\bm Y)_{T^\perp}\|_\infty\leq  1.\label{ConditionforDualCert}
\end{align}
Finding a vector $\bm Y$ that satisfies the conditions~\eqref{ConditionforDualCert} above might be hard. In particular, the strict condition $(\bm Y)_T = \bs h\bs m^*$ might be difficult to satisfy. Some recent results~\cite{recht11si, ahmed2014blind, demanet2014stable} actually show that such an \\textit{{exact} dual vector is not always required and that one can instead certify recovery through the construction of a so-called \textit{inexact} dual certificate by a careful analysis of the properties of the linear map $\mathcal{A}$. Let $\bs Z$ be such that $\mathcal{A}(\bm Z) = 0$. One can always define two matrices $\bs H_\perp$ and $\bs M_\perp$ such that $[\bs h,\;\bs H_{\perp}]$ and 
$[\bs m,\;\bs M_\perp]$ are unitary and such that $\bs H_\perp$ and $\bs M_\perp$ satisfy $\langle \bs H_\perp \bs M_\perp, \mathcal{P}_T^\perp(\bs Z)\rangle = \|\mathcal{P}_T^\perp(\bs Z)\|_*$ (it suffices to take $\bs M_\perp$ and $\bs H_\perp$ to be the left and right unitary matrices in the singular value decomposition of $\mathcal{P}_{T}^\perp(\bm Z)$ multiplied by the sign matrix from the eigenvalues. For $\bm M = \bs h\bs m^*$, we then have~\cite{recht11si}, 
\begin{align}
\|\bs M + \bs Z\|_*&\geq \langle\bs h\bs m^* +  \bs H_\perp\bs M_\perp^* ,\bs M + \bs Z\rangle  \label{inexactDual}\\
&\geq   \|\bs M\|_* + \langle\bs h\bs m^* +  \bs H_\perp\bs M_\perp^* - \bs Y,\bs Z\rangle\nonumber\\
&\geq \|\bs M\|_* +  \langle \bs h\bs m^* - \mathcal{P}_T(\bs Y), \mathcal{P}_T(\bs Z)\rangle + \langle \bs H_\perp\bs M_\perp^T - \mathcal{P}_T^\perp(\bs Y), \mathcal{P}_T^\perp(\bs Z) \rangle \nonumber \\
&\geq \|\bs M\|_* - \|\bs h\bs m^* - \mathcal{P}_T(\bs Y)\|_F\|\mathcal{P}_T(\bs Z)\|_F + (1 -  \| \mathcal{P}_T^\perp(\bs Y)\|  )\|\mathcal{P}_T^\perp(\bs Z)\|_*\nonumber
\end{align}
In the first inequality, we use the duality between the nuclear norm and the operator norm,
\begin{align}
\|\bm A\|_* = \sup_{\|\bm X\|_{\infty}\leq 1} \langle \bs A, \bm X \rangle.
\end{align}
In the second line, we use the fact that since $\bs Z$ is in the kernel of the linear map $\mathcal{A}$, any dual vector in the range of $\mathcal{A}^*$ implies $\langle \bs Y,\bs Z\rangle  = 0$. In the last line, we use H\"older's inequality. The last line of~\eqref{inexactDual} shows that $\|\bs M + \bs Z\|_*>\|\bs M\|_*$ as soon as 
\begin{align} 
(1 - \| \mathcal{P}_T^\perp(\bs Y)\|)\|\mathcal{P}_T^\perp(\bs Z)\|_* - \|\bs h\bs m^* - \mathcal{P}_T(\bs Y)\|_F\|\mathcal{P}_T(\bs Z)\|_F >0\label{inexactStep1}\end{align}
In particular, this condition thus implies uniqueness of the solution $\bs X_0 =\bs h\bs m^*$ on top of optimality which followed from~\eqref{ConditionforDualCert}.

From the condition~\eqref{inexactStep1} also follows the notion of inexact dual certificate. The idea of an inexact dual certificate relies on relaxing the difficult constraint $\mathcal{P}_T(\bm Y) = \bs h\bs m^*$ at the expense of a slight strengthening of the operator norm constraint $\|\mathcal{P}_T^\perp(\bm Y)\|_\infty\leq 1$. Formally, this requires finding the relation between $\|\mathcal{P}_T^\perp(\bs Z)\|_*$ and $\|\mathcal{P}_T(\bs Z)\|_F$. This is where the norm of $\mathcal{A}$ will be needed. To relate the norms of $\mathcal{P}_T(\bm Z)$ and $\mathcal{P}_T^\perp(\bm Z)$ we start by using the fact that $\bm Z$ is in the nullspace of $\mathcal{A}$, so that 
\begin{align}
\left|\|\mathcal{A}(\mathcal{P}_T(\bm Z))\| - \|\mathcal{A}(\mathcal{P}^\perp_T(\bm Z))\|\right|\leq \|\mathcal{A}(\mathcal{P}_T(\bm Z) + \mathcal{P}_T^\perp(\bm Z))\| = 0
\end{align}
which implies $\|\mathcal{A}(\mathcal{P}_T(\bm Z))\| = \|\mathcal{A}(\mathcal{P}^\perp_T(\bm Z))\|$. We then have 
\begin{align}
\|\mathcal{A}(\mathcal{P}_T(\bm Z))\| \leq \|\mathcal{A}\|\|\mathcal{P}^\perp_T(\bm Z)\|.\label{eq:boudAT}
\end{align}
The proof of lemma~\ref{lemmaT} below can be used to show that
\begin{align}
\|\mathcal{P}_T\mathcal{A}^*\mathcal{A}\mathcal{P}_T - \mathcal{P}_T\|_2\leq \frac{1}{2},\quad \mbox{w.h.p.}
\end{align}
which can then be used to derive 
\begin{align}
\left|\|\mathcal{A}\mathcal{P}_T(\bm Z)\|_2^2 - \|\mathcal{P}_T(\bm Z)\|_2^2\right|\leq \frac{1}{2}\|\mathcal{P}_T(\bm Z)\|_2^2\label{eq:tmpEqn}
\end{align}
Equation~\eqref{eq:tmpEqn} in turn implies 
\begin{align}
\frac{1}{2}\|\mathcal{P}_T(\bm Z)\|_2^2\leq \|\mathcal{A}\mathcal{P}_T(\bm Z)\|_2^2.\label{tmpBoundPT}
\end{align}
Substituting expression~\eqref{tmpBoundPT} into~\eqref{eq:boudAT}, one gets,
\begin{align}
\frac{1}{\sqrt{2}}\|\mathcal{P}_T(\bm Z)\|_2\leq \|\mathcal{A}\|\|\mathcal{P}^\perp_T(\bm Z)\|_F.
\end{align}
Finally, using the equivalence of the norms yields
\begin{align}
\frac{1}{\sqrt{2}}\|\mathcal{P}_T(\bm Z)\|_F \leq \|\mathcal{A}\|\|\mathcal{P}^\perp_T(\bm Z)\|_*\label{eq:BoundNormA}
\end{align}
We can now use this last relation together with the result of~\eqref{inexactDual} to obtain a final lower bound on $\|\bm M+ \bm Z\|_*$ as,
\begin{align}
\|\bm M+\bm Z\|_*\geq \|\bm M\|_* +\left(-\sqrt{2}\|\bs h\bs m^*-\mathcal{P}_T(\bm Y)\|_F\|\mathcal{A}\| + \left(1-\|\mathcal{P}_T^\perp(\bm Y)\|\right)\right)\|\mathcal{P}_T^\perp(\bm Z)\|_*\label{eq:uniquenessTmp}
\end{align}
Certifying recovery of $\bm X_0 = \bs h\bs m^*$ thus reduces to finding a dual vector $\bm Y$ satisfying the following relation, for all $\bm Z \in \mbox{null}(\mathcal{A})$, $\gamma  =\|\mathcal{A}\|$,
\begin{align}
\left(1 - \| \mathcal{P}_T^\perp(\bs Y)\|- \sqrt{2}\gamma\|\bs h\bs m^* - \mathcal{P}_T(\bs Y)\|_F\right) \|\mathcal{P}_T^\perp(\bs Z)\|_* >0.
\end{align}
From~\eqref{eq:BoundNormA}, if $\mathcal{P}_T^\perp(\bm Z) = 0$ then $\mathcal{P}_T(\bm Z)$ vanishes as well, which implies $\bm Z = 0$ and we can thus assume that $\|\mathcal{P}_T^\perp(\bm Z)\|>0$. With this in mind and using~\eqref{eq:uniquenessTmp}, we can thus move on and assume recovery as soon as there exists a $\bm Y$ satisfying
\begin{align}
\left(1 - \| \mathcal{P}_T^\perp(\bs Y)\|- \sqrt{2}\gamma\|\bs h\bs m^* - \mathcal{P}_T(\bs Y)\|_F\right)>0. 
\end{align}
In the rest of the chapter, we will construct a $\bm Y\in \mbox{Ran}\mathcal{A}^*$ satisfying the following two conditions,
\begin{align}
\|\bs h\bs m^* - \mathcal{P}_T(\bs Y)\|_F\leq \frac{1}{\sqrt{2}\gamma c_1},\qquad \| \mathcal{P}_T^\perp(\bs Y)\|_{\infty}<1-\frac{1}{c_1}
\label{dualCertConditions}
\end{align}
for some constant $c_1$. Computing the operator norm of $\mathcal{A}$ can be done through proposition~\ref{bernstein} and gives $\sqrt{K\log(LN) +\beta\log(L)}$ with probability at least $1-L^{-\beta}$.
The next section introduces some results on sub-gaussian and sub-exponential random variables as well as the corresponding Bernstein inequality which will be used extensively to derive the recovery guarantees.

\subsection{Sub-gaussian and sub-exponential tails\label{subgaussianAndSubexponential}}

Concentration through the Bernstein inequality requires the terms in the sum to be bounded which is not possible with the the linear map $\mathcal{A}$ given the infinite tail of the gaussian basis matrices $\bs C_n$. However, the gaussian distribution is known to be bounded with respect to the probability measure and so is its moments generating function. For this reason, we will consider the following ensemble of (Orlicz) norms introduced for a general measure $\mu$ and a non decreasing, convex function $\Psi(x)$ such that $\Psi(0)=0$ (see~\cite{van1996weak}) as,
\begin{align}
\|X\|_{\Psi} = \inf\{k>0\;|\;\int \Psi\left(\frac{|X|}{k}\right)\;d\mu\;\leq 1\}.\label{eq:OrliczNorm}
\end{align}
For the probability measure and the functions $\Psi_q(x) = e^{x^q} - 1$, we have the following concentration result which will be used extensively throughout the proofs of Theorem~\ref{theorem:BlindDeconvTh} to construct the dual certificate (see~\cite{tropp2015introduction, koltchinskii10nu, koltchinskii2013remark})
\begin{proposition}[Bernstein concentration (Orlicz version)]\label{bernstein}
Let $\bs Z_1,\ldots, \bs Z_P$ be i.i.d. random matrices of size $m\times n$ with $\mathbb{E}\{\bs Z_i\} = 0$. Define 
\begin{align}
\sigma_Z = \max \left\{\left\|\frac{1}{P}\sum_{i=1}^P \mathbb{E}\left(\bs Z_i\bs Z_i^*\right)\right\|^{1/2}, \left\|\frac{1}{P}\sum_{i=1}^P \mathbb{E}\left(\bs Z_i^*\bs Z_i\right)\right\|^{1/2}\right\}\label{varianceBernstein}
\end{align}
Let $X = \|\bs Z\|$ and assume that $\|X\|_{\Psi_q}\leq U^{(q)}<\infty$ for some $q\geq 1$. Further let $M = m+n$. 
Then there exists a constant $C>0$ such that $\forall t>0$, the following bound holds with probability $1-e^{-t}$.
\begin{align}
\left\|\bs Z_1+ \ldots+ \bs Z_P\right\|\leq C P\max\left\{\sigma \sqrt{\frac{t+\log(M)}{P}}, U^{(q)}\left(\log \frac{U^{(q)}}{\sigma}\right)^{1/q}\frac{t+ \log(M)}{P}\right\}\label{eq:BernsteinSum}
\end{align} 
\end{proposition}

In the rest of this section we will show how to derive bounds on the Orlicz norm of random variables with sub-gaussian or sub-exponential tails. Those are two important classes of random variables which will appear in the proof of Theorem~\ref{theorem:BlindDeconvTh}. We start by introducing a formal characterization of those random variables. Those two classes are respectively bounded in the $\|X\|_{\Psi_2}$ and $\|X\|_{\Psi_1}$ norms for $\Psi_2(x) = e^{x^2}- 1$ and $\Psi_1(x) = e^x - 1$ as we will see. We now define the notion of sub-gaussian and sub-exponential random variables. Propositions~\ref{subGauss1} and~\ref{subExp1} can be found, for example in~\cite{wainwright2015high}.

\begin{proposition}[Equivalent characterization of Sub-Gaussian random variables~\label{subGauss1}]
For a zero mean random variable, the following properties are equivalent
\begin{enumerate}[{\upshape i)}]
\item There is a constant $\sigma$ such that $\mathbb{E}e^{\lambda (X-\mu)}\leq e^{\frac{\lambda^2\sigma^2}{2}}$ for all $\lambda\in \mathbb{R}$.
\item There is a constant $c\geq 1$ and a Gaussian variable $Z\sim \mathcal{N}(0,\tau^2)$ such that 
\begin{align}
\mathbb{P}\left(|X|\geq A\right)\leq c\mathbb{P}\left(|Z|\geq A\right),\quad\mbox{for all $A\geq 0$.} 
\end{align} 
\item For all $\lambda\in [0,1)$, $\displaystyle \mathbb{E}e^{\frac{\lambda X^2}{2\sigma^2}}\leq \frac{1}{\sqrt{1-\lambda}}.$
\end{enumerate}
\end{proposition}	

The first condition is also called \textit{Laplace transform condition}. The second one is known as \textit{sub-gaussian tail estimate}. A similar ensemble of equivalent definitions can be derived for sub-exponential random variables,

\begin{proposition}[Equivalent characterization of Sub-exponential random variables~\label{subExp1}]
For a zero mean random variable, the following properties are equivalent
\begin{enumerate}[{\upshape i)}]
\item There are non-negative numbers $(\nu,b)$ such that
\begin{align}
\mathbb{E}\{e^{\lambda (X-\mu)}\}\leq e^{\frac{\nu^2\lambda^2}{2}}\qquad \mbox{for all $|\lambda|<\frac{1}{b}$}\label{BoundMoMGenFunSubExp}
\end{align}
\item There is a positive number $c_0>0$ such that $\mathbb{E}\{e^{\lambda X}\}<\infty$ for all $|\lambda|\leq c_0$
\item There are constants $c_1,c_2>0$ such that
\begin{align}
\mathbb{P}\left(|X|\geq t\right)\leq c_1e^{-c_2t}\qquad \mbox{for all $t>0$}.
\end{align}
\end{enumerate}
\end{proposition}

From propositions~\ref{subGauss1} and~\ref{subExp1} above and in particular from (i), it should be clear that sub-gaussian random variables are sub-exponential. We now show that subexponential random variables are always bounded in the $\Psi_1$ norm whether subgaussian variables are always bounded in the $\Psi_2$ norm. For sub-exponential random variables, the bound on the expectation of the moment generating function~\eqref{BoundMoMGenFunSubExp} provides a direct bound on the $\Psi_1$-norm. Indeed note that by using definition~\eqref{BoundMoMGenFunSubExp} together with~\eqref{eq:OrliczNorm} for $\Psi(x) = \Psi_1(x)$ (subexponential), we have,
\begin{align*}
\inf \{u>0\;:\;\mathbb{E}\exp(X/u)\leq 2\} &\leq \inf \{u> b\;:\;\exp(\nu^2/2u^2)\leq 2\} \\
& \lesssim  \max(b,\frac{\nu}{\sqrt{2\log 2}}). \\
&\lesssim \max(b,\nu).\label{lastBoundPsi1}
\end{align*}

The subexponential parameters $(b,\nu)$ thus provide a bound on the Orlicz $1$-norm of \textit{subexponential} variables. For \textit{sub-gaussian} variables, we can use the following result from~\cite{vershynin2010introduction}:

\begin{proposition}[\label{EqNormPsi}Equivalence of the subexponential and subgaussian norms]
A random variable $X$ is sub-gaussian if and only if $X^2$ is sub-exponential. Moreover,
\begin{align}
\|X\|_{\Psi_2}^2\leq \|X^2\|_{\Psi_1}\leq 2\|X\|_{\Psi_2}^2.
\end{align}
\end{proposition}

From proposition~\ref{EqNormPsi}, to derive abound on the $\Psi_2$ norm of a subgaussian variable $X$, one can thus simply use the bound on the moment generating function of the corresponding sub-exponential variable $X^2$ and then take the square root of this bound. Two other important results which will be assumed throughout the paper are $\|X_1X_2\|_{\psi_1}\leq \|X_1\|_{\psi_2}\|X_2\|_{\psi_2}$ whenever $X_1$ and $X_2$ are subgaussians with $\|X_1\|_{\psi_2},\|X_2\|_{\psi_2}<\infty$, and $\|X - \mathbb{E}X\|_{\psi_q}\leq 2\|X\|_{\psi_q}$ (see for example~\cite{ahmed2016leveraging} and lemma 7 in~\cite{ahmed2015compressive2}).

An important class of variables that will be used in the proof of Theorem~\ref{theorem:BlindDeconvTh} are chi-squared variables, $Z=X^2$ with $X\sim \mathcal{N}(0,1)$. Let us first show that those variables are subexponential. The expectation of the moment generating function $e^{\lambda X}$ reads,
\begin{align}
\mathbb{E} e^{\lambda (Z-1)}&= \frac{1}{\sqrt{2\pi}}\int e^{-X^2/2\sigma^2}e^{\lambda (X^2-1)}\;dX\nonumber\\
& = \frac{1}{\sqrt{2\pi}}e^{-\lambda}\int e^{-(1/2-\lambda)X^2}\;dX\nonumber\\
& = \frac{1}{\sqrt{\pi}}\frac{e^{-\lambda}}{\sqrt{1-2\lambda}}\int e^{-X^2}\;dX\nonumber \\
& = \frac{1}{\sqrt{\pi}}\frac{e^{-\lambda}}{\sqrt{1-2\lambda}}.\label{chiSquaredRV1}
\end{align}
The last equality~\eqref{chiSquaredRV1} holds whenever $\lambda<1/2$. In order to derive a bound of the form~\eqref{BoundMoMGenFunSubExp} we first require $|\lambda|<1/4$ to bound the denominator of~\eqref{chiSquaredRV1}. For this interval, we want a $\nu$ satisfying $\mathbb{E} e^{\lambda (Z-1)}= e^{-\lambda}\; e^{-\log\sqrt{(1-2\lambda)\pi}}\leq e^{\lambda^2\nu}$. One can check that this holds for $\nu=2$. $Z$ is thus sub-exponential with parameters $(2,4)$. The proof of Theorem~\ref{theorem:BlindDeconvTh} will also repeatedly require us to bound the Orlicz norm of a sum of the form $Z = \sum_k X_k^2$ for independent $X_k\sim\mathcal{N}(0,1)$. Since the $X_k$ are independent, the expectation of the moment generating function of $Z$ reads,
\begin{align}
\mathbb{E}e^{\lambda \sum_{k=1}^n X_k^2} = \prod_{k=1}^n \mathbb{E} e^{\lambda X_k^2} \leq \prod_{k=1}^n e^{\lambda^2\nu^2}\leq e^{n\lambda^2\nu^2}, \qquad \mbox{for all $|\lambda|<1/b$.}
\end{align}
This shows that $\chi = \sum_{k=1}^n X_k^2$ is subexponential with parameters $(b, \nu \sqrt{n})$. The next section introduces a candidate for the dual vector $\bs Y$.

\subsection{\label{sectionCandidateDualCert}Ansatz}

In order to satisfy the conditions~\eqref{dualCertConditions}, whenever the map $\mathcal{A}^*\mathcal{A}$ concentrates to the identity, the most obvious choice for $\bs Y$ would be to consider the certificate $\bs Y_1$ defined as,  
\begin{align}
\bs Y_1 &= \mathcal{A}^*\mathcal{A}(\bs h\bs m^*).
\end{align}  
However, this particular construction doesn't match the sample complexity observed empirically (see section~\ref{sectionPhaseTrans} and the discussion therein). In particular, at sample complexities observed empirically, it fails to satisfy the condition $\mathcal{P}_T(\bs Y_1)\approx \bs h\bs m^*$. For this reason, we consider instead the better ansatz,
\begin{align}
\bs Y = \mathcal{A}^*(\mathcal{A}\mathcal{P}_T)(\mathcal{P}_T\mathcal{A}^*\mathcal{A}\mathcal{P}_T)^{-1}\bs h\bs m^*\label{candidateCert0},
\end{align}
where, as in~\cite{candes2009exact}, the notation $(\mathcal{P}_T\mathcal{A}^*\mathcal{A}\mathcal{P}_T)^{-1}\bs h\bs m^*$ really means the element $\bs F$ in $T$ obeying $(\mathcal{P}_T\mathcal{A}^*\mathcal{A}\mathcal{P}_T)\bs F = \bs h\bs m^*$. In particular, the candidate certificate $\bs Y$ is well defined as soon as $(\mathcal{P}_T\mathcal{A}^*\mathcal{A}\mathcal{P}_T)$ is a one-to-one mapping from $T$ onto $T$. As soon as the ansatz~\eqref{candidateCert0} is well defined, by construction we immediately have $\mathcal{P}_T(\bs Y) = \bs h\bs m^*$, and the first condition in~\eqref{dualCertConditions} is immediately satisfied. Injectivity of this map is the point of lemma~\eqref{lemmaT} below. This lemma is proved in section~\eqref{proofT}.
\begin{restatable}{lemma}{lemmaT}
\label{lemmaT}
Let $T$ be defined as in~\eqref{tangentSpaceNN} with corresponding projector $\mathcal{P}_T$ defined as in~\eqref{projectorTBD}. Further let $\mu_m^2$ and $\mu_h^2$ be defined as in~\eqref{coherencem} and~\eqref{coherenceh} respectively. Then for any constant $\delta_1,\beta_1$, as soon as $L\gtrsim \beta_1(1/\delta_1) K\mu_h^2 \log(LN)$ and $N\gtrsim \beta_1 (1/\delta_1)  \mu_m^2\log (LN) $,
\begin{align}
\|\mathcal{P}_T\mathcal{A}^*\mathcal{A}\mathcal{P}_T - \mathcal{P}_T\|\lesssim \delta_1
\label{lemmaboundPalpha}
\end{align}
with probability at least $1-(L\vee N)^{-\beta_1}$ .
\end{restatable}
Lemma~\ref{lemmaT} shows that the mapping $\mathcal{P}_T\mathcal{A}^*\mathcal{A}\mathcal{P}_T -\mathcal{P}_T$ is a contraction. For a sufficiently small $\delta_1$ (see for example~\cite{kubrusly2012spectral}), anf for $\bs h\bs m^*\in T$, one can thus express the element $(\mathcal{P}_T\mathcal{A}^*\mathcal{A}\mathcal{P}_T)^{-1}\bs h\bs m^*$ as
\begin{align}
\left(\mathcal{P}_T - (\mathcal{P}_T - \mathcal{P}_T\mathcal{A}^*\mathcal{A}\mathcal{P}_T)\right)^{-1}\bs h\bs m^* &= \left(\mathcal{P}_T\mathcal{A}^*\mathcal{A}\mathcal{P}_T\right)^{-1} \bs h\bs m^*\\
&= \bs h\bs m^*+ \left(\mathcal{P}_T - \mathcal{P}_T\mathcal{A}^*\mathcal{A}\mathcal{P}_T \right)\bs h\bs m^* + \left(\mathcal{P}_T - \mathcal{P}_T\mathcal{A}^*\mathcal{A}\mathcal{P}_T \right)^2 \bs h\bs m^* + \ldots \label{seriesNeumann1}
\end{align}
Adding the second factor from~\eqref{candidateCert0} gives the following expansion for the ansatz~\eqref{candidateCert0}  
\begin{align}
\mathcal{A}^*\mathcal{A}\mathcal{P}_T\left(\mathcal{P}_T\mathcal{A}^*\mathcal{A}\mathcal{P}_T\right)^{-1}\bs h\bs m^* &= \mathcal{A}^*\mathcal{A}\mathcal{P}_T\left(\mathcal{P}_T\mathcal{A}^*\mathcal{A}\mathcal{P}_T\right)^{-1}\bs h\bs m^* \\
&= \mathcal{A}^*\mathcal{A}\mathcal{P}_T\bs h\bs m^*+\mathcal{A}^*\mathcal{A}\mathcal{P}_T \left(\mathcal{P}_T-\mathcal{P}_T\mathcal{A}^*\mathcal{A}\mathcal{P}_T \right)\bs h\bs m^* + \ldots \label{seriesNeumann2}
\end{align}
In order to show the second condition in~\eqref{ConditionforDualCert}, we are thus left with showing that the projection onto $T^\perp$ of each of the terms in~\eqref{seriesNeumann2} can be controlled. Following the approach in~\cite{candes2009exact} we will only use explicit concentration results to bound the first two terms in the series in order to reduce the general sample complexity and then use a more general argument for the remaining terms. We start with the first term
\begin{align}
\|\mathcal{P}_T^\perp \left(\mathcal{A}^*\mathcal{A}\mathcal{P}_T\bs h\bs m^*\right)\| &= \|\mathcal{P}_T^\perp \left(\mathcal{A}^*\mathcal{A}\mathcal{P}_T\bs h\bs m^* -\mathbb{E}\mathcal{A}^*\mathcal{A}\mathcal{P}_T\bs h\bs m^*\right)\|,
\end{align}
noting that $\mathbb{E}\mathcal{A}^*\mathcal{A}\mathcal{P}_T\bs h\bs m^* \in T$. This first term can be bounded through lemma~\ref{lemmaTperp} below whose proof is given in section~\ref{secSizeProperty}.
\begin{restatable}{lemma}{lemmaTperp} 
Let $\hat{\bs c}_{\ell,n}$ be defined as in section~\ref{notations} (including the scaling $\sqrt{L}$), where $c_{\ell,n}[k]\sim \mathcal{N}(0,1/L)$ are i.i.d. gaussian,. and $\bs A_{\ell,n} = \bs f_\ell\hat{\bs c}_{\ell,n}^*$. Let the operator $\mathcal{A}\;:\;\mathbb{C}^{L\times KN} \mapsto \mathbb{C}^{L\times N}$ be defined from the matrices $\bs A_{\ell,n}$ as in~\eqref{DefinitionLinearMap}. The coherences $\mu_m^2$ and $\mu_h^2$ are defined as in~\eqref{coherencem} and~\eqref{coherenceh}. Then for any constants $\delta_2,\beta_2$, as soon as $L\gtrsim \beta_2(1/\delta_2)K\mu_h^2$ and $N\gtrsim\beta_2(1/\delta_2)\mu_m^2$
\label{lemmaTperp} 
\begin{align}
\|\mathcal{P}_T^\perp \left(\mathcal{A}^*\mathcal{A}\mathcal{P}_T\bs h\bs m^*\right)\|\lesssim \delta_2
%\beta\max\left\{\sqrt{P \left(\frac{\mu^2_{h}K}{L}  + \frac{\mu^2_{m}}{N} \right)}\sqrt{\log(LKN)}, P \frac{\mu^2_{h}\mu^2_m}{LN} K \log(LKN) (\log(LKN))\right\}
\end{align}\label{secondTermNeumann}
with probability at least $1-\left(LN\right)^{-\beta_2}$. 
\end{restatable}
For the second term in~\eqref{seriesNeumann2}, we want to show that the event $\mathcal{E}_3$ defined as 
\begin{align}
\mathcal{E}_3\equiv \left\{\left\|\mathcal{P}_T^\perp\mathcal{A}^*\mathcal{A}\left(\mathcal{P}_T\mathcal{A}^*\mathcal{A}\mathcal{P}_T - \mathcal{P}_T\right)\bs h\bs m^*\right\|\lesssim \delta_3\right\}\label{4thorderNorm}
\end{align}
holds with sufficient probability, for a sufficiently small constant $\delta_3$. We prove this result through lemma~\ref{lemmaTotalSecondTerm} below. The proof of this lemma is given in section~\ref{secSizeProperty}.
\begin{restatable}{lemma}{lemmaTotalSecondTerm} 
Let $\hat{\bs c}_{\ell,n}$ be defined as in section~\ref{notations} (including the scaling $\sqrt{L}$) where $c_{\ell,n}[k]\sim \mathcal{N}(0,1/L)$. Let $\bs A_{\ell,n} = \bs f_\ell\hat{\bs c}_{\ell,n}^*$ and let $\mathcal{A}$ denote the linear map defined from the matrices $\bs A_{\ell,n}$ as in~\eqref{DefinitionLinearMap}. Then for any constant $\delta_3$, for coherences $\mu_m^2$ and $\mu_h^2$ defined as in~\eqref{coherencem} and~\eqref{coherenceh},  
\label{lemmaTotalSecondTerm} 
 \begin{align}
\Bigg\|\mathcal{P}_T^\perp\mathcal{A}^*\mathcal{A}\left(\mathcal{P}_T\mathcal{A}^*\mathcal{A}\mathcal{P}_T - \mathcal{P}_T\right)\bs h\bs m^*\Bigg\|&\lesssim \delta_3\label{normLemmaEbound}
\end{align}
with probability at least $1-(1/\delta_3)\sqrt{\left(\frac{K\mu_h^2}{L} + \frac{\mu_m^2}{N}\right)} - c_3(LN)^{-\beta}$ as soon as $L\gtrsim \beta(1/\delta_3) K\mu_h^2$, $N\gtrsim \beta (1/\delta_3)\mu_m^2$. 
\end{restatable}
The last lemma below concludes the proof by bounding the remaining terms in the series~\eqref{seriesNeumann2}. The proof of this lemma, which is almost identical to the proof given in~\cite{candes2009exact} is recalled for clarity in section~\ref{finNeumann}. 
\begin{restatable}{lemma}{lemmafinNeumann}
\label{lemmafinNeumann} 
Under the assumptions of Theorem~\ref{theorem:BlindDeconvTh}, there exists a constant $C_{k_0}$ such that 
\begin{align}
\sum_{k=k_0}^\infty\|\mathcal{P}_T^\perp\mathcal{A}^*\mathcal{A}\mathcal{P}_T \left(\mathcal{P}_T - \mathcal{P}_T\mathcal{A}^*\mathcal{A}\mathcal{P}_T\right)^{k}\bs h\bs m^*\| \leq C_{k_0}%(\delta_4) \sqrt{K \log(L\vee N)} \beta^k\left(\frac{\mu_h^2 K}{L}+ \frac{\mu_m^2}{N}\right)^k
\end{align}
with probability at least $1-(LN)^{-\beta}$. In particular, the constant $C_{k_0}$ can be made smaller than $\delta$ as soon as $L\gtrsim (1/\delta)K^{1+\frac{1}{k_0}}\mu_h^2$, $N\gtrsim (1/\delta)K^{1/k_0}$.

%
%\begin{align}
%\|\mathcal{P}_T^\perp\mathcal{A}^*\mathcal{A}\mathcal{P}_T \left(\mathcal{P}_T - \mathcal{P}_T\mathcal{A}^*\mathcal{A}\mathcal{P}_T\right)^{k}\bs h\bs m^*\| &\leq KC_{k_0}  \left(\frac{\delta_4}{K^{1/k_0}}+ \frac{\delta_4}{K^{1/k_0}}\right)^{k_0}\left(\frac{\delta_4}{K^{1/k_0}}+ \frac{\delta_4}{K^{1/k_0}}\right)^{k-k_0}\nonumber \\
%&\lesssim C_{k_0} \delta_4^{k_0}\left(\frac{\delta_4}{K^{1/k_0}}\right)^{k-k_0}
%\end{align}

\end{restatable}

As we only need to bound the terms corresponding to $k>1$, we can just take $k_0=2$ which gives the sample complexities of Theorem~\ref{theorem:BlindDeconvTh}.

Combining the results of lemmas~\ref{lemmaT} to~\ref{lemmafinNeumann}, and choosing $\delta$ such that $(1/\delta)\geq \left\{(1/\delta_i)\right\}_{i=1}^4$ as well as $\beta\geq \left\{\beta_1,\beta_2,\beta_3,\beta_4\right\}$, we have that as soon as $L\gtrsim (1/\delta)\beta K^{3/2}\mu_h^2$, $N\gtrsim (1/\delta)\beta\sqrt{K} \mu_m^2$ for any $\beta_1$, $\beta_2<1$, with probability at least $1 - (LN)^{-\sum\beta_i}  - \sqrt{\left(\frac{K\mu_h^2}{L} + \frac{\mu_m^2}{N}\right)}$, we have
\begin{align}
\|\mathcal{P}_T^\perp(\bs Y)\|&< \delta_1 + \delta_2 +\delta_3 + C_{k_0}\end{align}
which can be made sufficiently smaller than $1$ for sufficiently small constants $\delta_i, C_{k_0}$. This concludes the proof of Theorem~\ref{theorem:BlindDeconvTh}. The remaining sections proceed with the proofs of each of the lemmas mentioned above.

\subsection{\label{proofT}The injectivity property}
In this section, we prove injectivity of the normal operator $\mathcal{A}^*\mathcal{A}$ on $T$. This condition certifies that the ansatz~\eqref{candidateCert0} is well defined. We start by recalling lemma~\ref{lemmaT} below. 
\lemmaT*
\begin{proof}

To bound the operator norm of $\mathcal{P}_T\mathcal{A}_p^*\mathcal{A}_p\mathcal{P}_T$ we will use proposition~\eqref{bernstein}. We start by bounding the variance~\eqref{varianceBernstein}. We then derive a corresponding bound on the Orlicz norm of each of the terms appearing within the norm~\eqref{lemmaboundPalpha}. 

The operator $\mathcal{P}_T\mathcal{A}^*\mathcal{A}\mathcal{P}_T :\bm X\mapsto \mathcal{P}_T\mathcal{A}^*\mathcal{A}\mathcal{P}_T(\bm X)$ expands as
\begin{align}\begin{split}
\mathcal{P}_T\mathcal{A}^*\mathcal{A}\mathcal{P}_T(\bm X) &= \sum_{\ell=1}^L\sum_{n=1}^N \mathcal{P}_T (\bm A_{\ell,n})\langle \bm A_{\ell,n}, \mathcal{P}_T(\bm X) \rangle \\
& =  \sum_{\ell=1}^L\sum_{n=1}^N  \mathcal{P}_T (\bm A_{\ell,n})\langle \mathcal{P}_T(\bm A_{\ell,n}),\bm X \rangle \\
& = \sum_{\ell=1}^L\sum_{n=1}^N \vec(\mathcal{P}_T(\bm A_{\ell,n}))\otimes \vec(\mathcal{P}_T(\bm A_{\ell,n}))^* \vec(\bm X).
\end{split}\label{OperatorBernstein1}
\end{align}
Now taking the operator norm, we get,
\begin{align*}
&\left\|\sum_{\ell=1}^L\sum_{n=1}^N \mathcal{P}_T\mathcal{A}_{\ell,n}^*\mathcal{A}_{\ell,n}\mathcal{P}_T - \mathbb{E}\{\mathcal{P}_T\mathcal{A}_{\ell,n}^*\mathcal{A}_{\ell,n}\mathcal{P}_T\}\right\|\\
=& \left\|\sum_{\ell=1}^L\sum_{n=1}^N \mathcal{P}_T\mathcal{A}_{\ell,n}^*\mathcal{A}_{\ell,n}\mathcal{P}_T - \mathcal{P}_T\right\|
\end{align*}

%Following Bernoulli sampling, we let $\bs Z_{\ell,n}$ for $(\ell,n)\in |\Gamma_p|$ define the terms in~\eqref{OperatorBernstein1}, after centering such that $\mathbb{E}\{\mathcal{Z}_{\ell,n}\} = 0$,
%
%\begin{align*}
%\mathcal{Z}_{\ell,n} &= \mathcal{P}_T\mathcal{A}^*_{\ell,n}\mathcal{A}_{\ell,n}\mathcal{P}_T - \mathbb{E}\left\{\mathcal{P}_T\mathcal{A}^*_{\ell,n}\mathcal{A}_{\ell,n}\mathcal{P}_T\right\}\\
%& = \mathcal{P}_T\mathcal{A}^*_{\ell,n}\mathcal{A}_{\ell,n}\mathcal{P}_T - \mathcal{I}_{\ell,n}
%\end{align*}

where we used $\mathbb{E}\sum_{\ell=1}^L\sum_{n=1}^N\mathcal{A}_{\ell,n}^*\mathcal{A}_{\ell,n} = \mathcal{I}$.
%Note that this corresponds to picking $P$ indices at random with replacement and considering the measurements associated to the indices that appear several times only once.  CHECK THIS
We start by deriving the bound for the variance~\eqref{varianceBernstein}. Let us use $\mathcal{Z}_{\ell,n}$ to denote the operators defined as 
\begin{align*}
\mathcal{Z}_{\ell,n} &\equiv\mathcal{P}_T\mathcal{A}^*_{\ell,n}\mathcal{A}_{\ell,n}\mathcal{P}_T - \mathbb{E}\left\{\mathcal{P}_T\mathcal{A}^*_{\ell,n}\mathcal{A}_{\ell,n}\mathcal{P}_T\right\}\end{align*}
Recall that the variance is defined as 
\begin{align}
\sigma^2& = \max \left\{\left\|\sum_{\ell=1}^L\sum_{n=1}^N \mathbb{E}\mathcal{Z}_{\ell,n}^*\mathcal{Z}_{\ell,n}\right\|^2, \left\|\sum_{\ell=1}^L\sum_{n=1}^N \mathbb{E}\mathcal{Z}_{\ell,n}\mathcal{Z}_{\ell,n}^* \right\|^2\right\}\\
& = \max \left\{\sigma_1^2,\sigma_2^2\right\}\label{boundsVariance}
\end{align}
Since the variables $\mathcal{Z}_{\ell,n}$ are symmetric, the two bounds $\sigma_1^2$ and $\sigma_2^2$ in~\eqref{boundsVariance} are exactly the same in this case and we can thus focus on  either of them.
\begin{align}
\sigma^2& = \left\|\sum_{\ell=1}^L\sum_{n=1}^N \mathbb{E}\mathcal{P}_T\mathcal{A}^*_{\ell,n}\mathcal{A}_{\ell,n}\mathcal{P}_T\|\mathcal{P}_T(\bm A_{\ell,n})\|^2_F - \left(\mathbb{E} \mathcal{P}_T\mathcal{A}^*_{\ell,n}\mathcal{A}_{\ell,n}\mathcal{P}_T\right)^2\right\|^2\nonumber \\
&  =\left\|\sum_{\ell=1}^L\sum_{n=1}^N \mathbb{E}\left(\vec(\mathcal{P}_T(\bm A_{\ell,n}))\otimes \vec(\mathcal{P}_T(\bm A_{\ell,n}))\right)^*\|\vec(\mathcal{P}_T(\bm A_{\ell,n}))\|^2 - \left(\mathbb{E} \mathcal{P}_T\mathcal{A}^*_{\ell,n}\mathcal{A}_{\ell,n}\mathcal{P}_T\right)^2\right\|^2.\label{varianceDef1}
\end{align}
The Frobenius norm $\|\mathcal{P}_T\bm A_{\ell,n}\|$ can be bounded from the definition of the projector $\mathcal{P}_T$~\eqref{projectorTBD} and from the definition of the matrices $\bm A_{\ell,n}$ as 
\begin{align}\begin{split}
\|\mathcal{P}_T(\bm A_{\ell,n})\|_F^2 &= \|\bs h\bs h^*\bs f_\ell \bs c_{\ell,n}^* + \bs f_\ell \bs c_{\ell,n}^*\bs m\bs m^* - \bs h\bs h^*\bs f_\ell\bs c_{\ell,n}^*\bs m\bs m^*\|_F^2\\
& = \|\bs h\|^2|\hat{h}[\ell]|^2\|\bs c_{\ell,n}\|^2 + \|\bs m\|^2|\langle\bs c_{\ell,n}, \bs m \rangle |^2\|\bs f_\ell\|^2+ \|\bs h\|^2|\hat{h}[\ell]|^2|\langle \bs m, \bs c_{\ell,n}\rangle |^2\\
&   + 2\langle \bs f_\ell\bs c_{\ell,n}^*\bs m\bs m^*, \bs h \bs h^*\bs f_\ell\bs c_{\ell,n} \rangle 
- 2\langle \bs h\bs h^*\bs f_\ell\bs c_{\ell,n}^*\bs m\bs m^*, \bs h\bs h^*\bs f_\ell \rangle \\
&   - 2\langle \bs h\bs h^*\bs f_\ell\bs c_{\ell,n}^*\bs m\bs m^*, \bs f_\ell \bs c_{\ell,n}\bs m\bs m^*\rangle \\
& = |\hat{h}[\ell]|^2\|\bs c_{\ell,n}\|^2 + |\langle\bs c_{\ell,n}, \bs m \rangle |^2- |\hat{h}[\ell]|^2 |\langle \bs m, \bs c_{\ell,n}\rangle |^2 \\
&\leq |\hat{h}[\ell]|^2\|\bs c_{\ell,n}\|^2 + |\langle\bs c_{\ell,n}, \bs m \rangle |^2.
\end{split}\label{PTAfrobenius}
\end{align}
In~\eqref{PTAfrobenius} we use the fact that $\|\bs h\| = \|\bs m\| = 1$ as well as $\|\bs f_\ell\|=1$. Now note that for any projection operator $\mathcal{P}$, we always have $\|\mathcal{P}\mathcal{A}\|\leq \|\mathcal{A}\|$ and $\|\mathcal{A}\mathcal{P}\|\leq \|\mathcal{A}\|$. The norm~\eqref{varianceDef1} can thus simplify to
\begin{align}
\sigma^2 &= \left\|\sum_{\ell=1}^L\sum_{n=1}^N \mathbb{E}\left(\vec(\mathcal{P}_T(\bm A_{\ell,n}))\otimes \vec(\mathcal{P}_T(\bm A_{\ell,n}))\right)\|\vec(\mathcal{P}_T(\bm A_{\ell,n}))\|^2 - \left(\mathbb{E} \mathcal{P}_T\mathcal{A}^*_{\ell,n}\mathcal{A}_{\ell,n}\mathcal{P}_T\right)^2\right\|^2\\
& \leq \left\|\mathcal{P}_T\left(\sum_{\ell=1}^L\sum_{n=1}^N \mathbb{E}\mathcal{A}^*_{\ell,n}\mathcal{A}_{\ell,n}\|\mathcal{P}_T(\bm A_{\ell,n})\|^2_F \right)\mathcal{P}_T\right\|^2\\
&\leq \left\|\sum_{\ell=1}^L\sum_{n=1}^N \mathbb{E}\mathcal{A}^*_{\ell,n}\mathcal{A}_{\ell,n}\|\mathcal{P}_T(\bm A_{\ell,n})\|^2_F \right\|^2\label{varianceWithoutP}
\end{align}
In the second line, we use the positive semidefiniteness of the variance and the fact that for matrices $\bm A$ and $\bm B$ with $\bm A-\bm B\succeq 0$, $\|\bm A - \bm B\|\leq \|\bm A\|$.
The operator $\mathcal{A}_{\ell,n}^*\mathcal{A}_{\ell,n}$ can be written in matrix form as 
\begin{align}
\mathcal{A}_{\ell,n}^*\mathcal{A}_{\ell,n}(\bm X) & = \left(\vec(\bm A_{\ell,n})\vec(\bm A_{\ell,n})^*\right)\vec(\bm X)\\
& =(\bs e_n\bs e _n^*)\otimes  \left(\bs c_{\ell,n}\bs c_{\ell,n}^*\otimes \bs f_\ell\bs f_\ell^*\right)\vec(\bm X).\label{OperatorAln}
\end{align}
Substituting the last line of~\eqref{PTAfrobenius} together with~\eqref{OperatorAln} into~\eqref{varianceWithoutP}, and using the expression for the moments of multivariate gaussian random variables, one can write,
\begin{align}
\left\|\sum_{\ell=1}^L\sum_{n=1}^N\mathbb{E}\vec(\bm A_{\ell,n})\otimes \vec(\bm A_{\ell,n})\|\mathcal{P}_T(\bs A_{\ell,n})\|_F^2\right\| &\leq \left\|\sum_{\ell=1}^L\sum_{n=1}^N(\bs e_n\bs e_n^*)\otimes \left(K\bm I_K\otimes |\hat{h}[\ell]|^2\bs f_\ell\bs f_\ell^* \right)\right\|\nonumber\\
&+\left\|\sum_{\ell=1}^L\sum_{n=1}^N 2(\bs e_n\bs e_n^*)\otimes \left(\bs m_n\bs m_n^* \otimes \bs f_\ell\bs f_\ell^* \right)\right\|\label{boundIntermediate}
\end{align}
Now using the definitions of $\mu_h^2$ and $\mu_m^2$ from~\eqref{coherencem} and~\eqref{coherenceh}, and noting that $\sum_\ell \bs f_\ell\bs f_\ell^* = \bs F\bs F^* =  \bm I$, the two terms of expression~\eqref{boundIntermediate} can be upper bounded respectively as
\begin{align}
\left\|\sum_{\ell=1}^L\sum_{n=1}^N(\bs e_n\bs e_n^*)\otimes \left(K\bm I_K\otimes |\hat{h}[\ell]|^2\bs f_\ell\bs f_\ell^* \right)\right\|\leq \frac{K\mu_h^2}{L},\label{Bound1}
\end{align}
and 
\begin{align}
\left\|\sum_{\ell=1}^L\sum_{n=1}^N 2(\bs e_n\bs e_n^*)\otimes \left(\bs m_n\bs m_n^* \otimes \bs f_\ell\bs f_\ell^* \right)\right\|\leq 2\frac{\mu^2_m}{N}. \label{Bound2}
\end{align}
For the first bound, we use the fact that $\|\sum_n (\bs e_n\bs e_n)\otimes\bm A_n  \|\leq \sup_n \|\bm A_n\|$ for any given matrix $\bs A_n$ and $\|\sum_\ell \bs f_\ell \bs f_\ell^*\| = 1$. For the second bound we use $\|\bs m\bs m^*\otimes \bm I\|\leq \|\bs m\bs m^*\|$. Combining~\eqref{Bound1} and~\eqref{Bound2} into~\eqref{boundIntermediate} gives
\begin{align}
\left\|\sum_{\ell,n}\mathbb{E}\vec(\bm A_{\ell,n})\otimes \vec(\bm A_{\ell,n})^*\|\mathcal{P}_T\mathcal{A}_{\ell,n}\|_F^2\right\| \leq  \left(\frac{\mu_h^2 K}{L}  + 2\frac{\mu^2_m}{N} \right)
\end{align}
As explained earlier, the exact same result holds for the $\mathcal{Z}_{\ell,n}^*\mathcal{Z}_{\ell,n}$ since $\mathcal{Z}_{\ell,n} = \mathcal{Z}_{\ell,n}^*$. We now derive a bound on the Orlicz norm of the $\mathcal{Z}_{\ell,n}$. We use the norms derived from the functions $\Psi_2(z) = \exp(z^2) - 1$ and $\Psi_1(z)=\exp(x)-1$ respectively for sub-gaussian and sub-exponential random variables. For a gaussian random vector $\bs c_{\ell,n}$, we have shown above that $\|\bs c_{\ell,n}\|^2$ follows a $\chi^2$ distribution with $K$ degrees of freedom and is therefore subexponential.

For two matrices $\bm A, \bm B$, with $\|\bm A\|\leq \|\bm B\|$, we also have $\|\bm A\|_{\Psi_1}\leq \|\bm B\|_{\Psi_1}$ and so from~\eqref{PTAfrobenius}, 
\begin{align}
\|\mathcal{Z}_{\ell,n}\|_{\Psi_1}&\leq \|\|\vec(\mathcal{P}_T(\bm A_{\ell,n})\otimes \vec(\mathcal{P}_T(\bm A_{\ell,n})\|\|_{\Psi_1} \nonumber \\
&\leq \|\|\mathcal{P}_T(\bm A_{\ell,n})\|_F^2\|_{\Psi_1}\nonumber \\
&\leq \| |\hat{h}[\ell]|^2\|\bs c_{\ell,n}\|^2 + |\langle\bs c_{\ell,n}, \bs m \rangle |^2 \|_{\Psi_1}\nonumber \\
&\leq |\hat{h}[\ell]|^2\|\|\bs c_{\ell,n}\|^2\|_{\Psi_1} + \||\langle\bs c_{\ell,n}, \bs m \rangle |^2 \|_{\Psi_1}\label{Psi1Bound}
\end{align}
The first term in~\eqref{Psi1Bound} is a (subexponential) chi-squared distribution with $K$ degrees of freedom. Following the discussion in section~\ref{subgaussianAndSubexponential}, the first term of~\eqref{Psi1Bound} can be bounded by 
\begin{align}
|\hat{h}[\ell]|^2\|\|\bs c_{\ell,n}\|^2\|_{\Psi_1} &\lesssim \frac{K\mu_h^2 }{L}\label{Psi1ZlnB1}
\end{align}
To bound the second term, we use the fact that $\langle \bs c_{\ell,n},\bs m_n\rangle$ is a sum of zero-mean gaussian random variables (i.e gaussian mixture) for which the variance is simply given by the mixture of the variances. The square of this mixture is thus a chi-squared. Applying an argument similar to~\eqref{chiSquaredRV1}, one can show that the sub-exponential parameters of this chi-squared are given by $(b,\nu) = (\|m_n\|^2, 2)$ which gives the following bound,
\begin{align}
\||\langle \bs c_{\ell,n},\bs m_n\rangle |^2\|_{\Psi_1} & \lesssim \|\bs m_n\|^2\lesssim \frac{\mu^2_m}{N}.\label{Psi1ZlnB2}
\end{align}
The bound on the $\Psi_1$-norm of $\mathcal{Z}_{\ell,n}$ is thus finally given by combining~\eqref{Psi1ZlnB1} and~\eqref{Psi1ZlnB2} into
\begin{align}
\|\mathcal{Z}_{\ell,n}\|_{\Psi_1} \lesssim \frac{K\mu_h^2 }{L}+ \frac{\mu^2_m}{N}
\end{align}
We can now apply proposition~\ref{bernstein} which gives 
\begin{multline*}
\left\|\sum_{\ell=1}^L\sum_{n=1}^N \mathcal{P}_T\mathcal{A}_{\ell,n}^*\mathcal{A}_{\ell,n}\mathcal{P}_T - \mathcal{P}_T\right\|\lesssim \\
\max\left\{\sqrt{\left(\frac{\mu_h^2 K}{L}  + 2\frac{\mu^2_m}{N} \right)}\sqrt{t+\log(LKN)}, \left(\frac{\mu^2_hK}{L}  + 2\frac{\mu^2_m}{N} \right) \log(LKN) (t+ \log(LKN))\right\}\label{eq:BernsteinSum}
\end{multline*}
with probability $1-e^{-t}$. Taking $t \gtrsim \beta_1  \log LN$ for a constant $\beta_1$, and $L\gtrsim (1/\delta_1)K\mu^2_h $ and $N\gtrsim (1/\delta_1)\beta_1\mu^2_m $ up to log factors concludes the proof.
\end{proof}

\subsection{\label{secSizeProperty}The size property}

In this section, we prove the second condition from~\eqref{ConditionforDualCert}. To do so, we will bound each of the terms in~\eqref{seriesNeumann2}. We start with the first one. This term is bounded through lemma~\eqref{lemmaTperp},

\lemmaTperp*

\begin{proof}

The proof of lemma~\ref{lemmaTperp} relies, once again on proposition~\ref{bernstein}. We want to bound the quantity 
\begin{align}
\|\mathcal{P}_T^\perp(\mathcal{A}^*\mathcal{A})\bs h\bs m^*\| &= \|\mathcal{P}_T^\perp \left(\mathcal{A}^*\mathcal{A}(\bs h\bs m^*) - \mathbb{E}\mathcal{A}^*\mathcal{A}(\bs h\bs m^*)\right)\|\\
&\leq \|\mathcal{A}^*\mathcal{A}(\bs h\bs m^*) - \mathbb{E}\mathcal{A}^*\mathcal{A}(\bs h\bs m^*)\|\label{normToBound}
\end{align}
Where we use the fact that $\mathbb{E}\mathcal{A}^*\mathcal{A} = \mathcal{I}$. Again we introduce variables $\mathcal{Z}_{\ell,n}$ to denote each of the terms in the sum~\eqref{normToBound}. 
\begin{align}
\mathcal{Z}_{\ell,n}&\equiv\mathcal{P}_T^\perp(\mathcal{A}_{\ell,n}^*\mathcal{A}_{\ell,n})\bs h\bs m^* -\mathbb{E}\mathcal{P}_T^\perp(\mathcal{A}_{\ell,n}^*\mathcal{A}_{\ell,n})\bs h\bs m^*
\end{align}
The norm~\eqref{normToBound} expands as
\begin{align}
\left\|\sum_{\ell=1}^L\sum_{n=1}^N\mathcal{Z}_{\ell,n}\right\|&= \left\|\sum_{\ell=1}^L\sum_{n=1}^N\bm  A_{\ell,n}\langle \bm A_{\ell,n},\bs h\bs m^*\rangle - \mathbb{E}\bm A_{\ell,n}\langle \bm A_{\ell,n},\bs h\bs m^* \rangle\right\|\label{sumZlnTperp}
\end{align}
In order to apply proposition~\ref{bernstein}, we again start by deriving the bound for the variance. The $\mathcal{Z}_{\ell,n}$ are not hermitian anymore as is shown by $\bm A_{\ell,n}\bm A_{\ell,n}^* = \|\bs c_{\ell,n}\|^2\bs f_\ell\bs f_\ell^*$ and $\bm A_{\ell,n}^*\bm A_{\ell,n} = \|\bs f_\ell\|^2\bs c_{\ell,n}\bs c_{\ell,n}^*$ and we need to consider the two variance bounds in~\eqref{varianceBernstein} separately. For the first bound, we have
\begin{align}
\left\|\sum_{\ell=1}^L\sum_{n=1}^N\mathbb{E}\mathcal{ Z}_{\ell,n}\mathcal{Z}_{\ell,n}^*\right\|& = \left\|\sum_{\ell=1}^L\sum_{n=1}^N\mathbb{E}\bm A_{\ell,n}\bm A_{\ell,n}^* |\langle \bm A_{\ell,n}, \bs h\bs m^*\rangle |^2 - \left|\mathbb{E}\bm A_{\ell,n}\langle \bm A_{\ell,n},\bs h\bs m^*\rangle \right|^2\right\|\nonumber\\
&\leq \left\|\sum_{\ell=1}^L\sum_{n=1}^N\mathbb{E}\bm A_{\ell,n}\bm A_{\ell,n}^* |\langle \bm A_{\ell,n}, \bs h\bs m^*\rangle |^2\right\|\nonumber\\
&\leq \left\|\sum_{\ell=1}^L\sum_{n=1}^N \mathbb{E}\bs f_\ell\bs f_\ell^*\|\bs c_{\ell,n}\|^2 |\langle \bs c_{\ell,n}, \hat{\bs h}[\ell]\bs m^*_n\rangle |^2\right\|\nonumber\\
&\leq K\left\|\sum_{\ell=1}^L\sum_{n=1}^N\bs f_\ell\bs f_\ell^* \|\hat{\bs h}[\ell]\bs m^*_n\|^2\right\|\label{boundlemma2variance1aa}
\end{align}
 Now using the fact that $\sum_{\ell=1}^L \bs f_\ell \bs f_\ell^* = \bs I$ as well as $\|\bs m\|=1$, the last line of~\eqref{boundlemma2variance1aa} can be reduced to 
\begin{align}
\left\|\sum_{\ell=1}^L\sum_{n=1}^N\mathbb{E}\mathcal{Z}_{\ell,n}\mathcal{Z}_{\ell,n}^*\right\|&\leq K \cdot \frac{\mu^2_{h}}{L}\label{boundlemma2variance1baa}
\end{align}
For the second bound, a similar argument gives,
\begin{align}
\left\|\mathbb{E}\sum_{\ell=1}^L\sum_{n=1}^N  \mathcal{Z}_{\ell,n}^*\mathcal{Z}_{\ell,n}\right\|& = \left\|\mathbb{E}\sum_{\ell=1}^L\sum_{n=1}^N\|\bs f_\ell\|^2\bs c_{\ell,n}\bs c_{\ell,n}^* |\langle \bs f_\ell\bs c_{\ell,n}, \bs h\bs m^*\rangle |^2 \right\|,\nonumber \\
&\leq \nonumber \left\| \mathbb{E}\sum_{\ell=1}^L\sum_{n=1}^N\|\bs f_\ell\|^2\bs c_{\ell,n}\bs c_{\ell,n}^* |\langle\bs c_{\ell,n}, \hat{\bs h}[\ell]\bs m^*_n\rangle |^2\right\|\\
&\lesssim \left\|\sum_{\ell=1}^L\sum_{n=1}^N\|\hat{\bs h}[\ell]\bs m^*_n\|^2 (\bs e_n\bs e_n^*)\otimes \bm I_K\right\|\nonumber + \left\|\sum_{(\ell,n)\in [L]\times [N]}\hat{\bs h}[\ell]\bs m_n\hat{\bs h}[\ell]^*\bs m^*_n\right\|\nonumber \\
&\lesssim \frac{\mu^2_{m}}{N}. \label{boundlemma2variance2aa}
\end{align}
Combining~\eqref{boundlemma2variance1baa} with~\eqref{boundlemma2variance2aa}, we get the following bound on the variance $\sigma$,
\begin{align}
\sigma \lesssim \left( \frac{\mu^2_{h}K}{L} + \frac{\mu^2_{m}}{N}\right)\label{boundVarianceLemmaN1a}
\end{align}
We now bound the Orlicz $1$-norm of each of the variables $\mathcal{Z}_{\ell,n}$ in~\eqref{sumZlnTperp}. We start by establishing the distribution of $\|\bm A_{\ell,n}\langle \bm A_{\ell,n}, \bs h\bs m^* \rangle - \mathbb{E} \bm A_{\ell,n}\langle \bm A_{\ell,n}, \bs h\bs m^* \rangle\|$
\begin{align}
\|\bm A_{\ell,n}\langle \bm A_{\ell,n}, \bs h\bs m^* \rangle\|^2& = \|\bs f_{\ell}\bs c_{\ell,n}^*\langle \bs c_{\ell,n},\hat{h}[\ell]\bs m_n\rangle -\hat{h}[\ell] \bs f_\ell\bs m_n^*\|^2\\
&\leq \frac{\mu_h^2}{L}\|\bs c_{\ell,n}^*\langle \bs c_{\ell,n},\hat{h}[\ell]\bs m_n\rangle -\bs m_n^*\|\\
 &\leq \frac{\mu_h^2}{L} \|\bs c_{\ell,n}\|^2 |\langle \bs c_{\ell,n},\bs m_n\rangle |^2\label{psi1tmpLemma2aa}
\end{align}
Both $\|\bs c_{\ell,n}\|^2$ as well as $|\langle \bs c_{\ell,n},\hat{h}[\ell]\bs m_n\rangle |^2$ are subexponential (chi-squared) variables for which the Orlicz-$1$ norm can be bounded by using the discussion in section~\eqref{subgaussianAndSubexponential} (apply~\eqref{chiSquaredRV1} for a general gaussian $X$ with mean $0$ and variance $\sigma^2$ or alternatively use lemma 7 in~\cite{ahmed2014blind} together with lemma 2.2.1 in~\cite{vandervaart96we}). Then using proposition~\ref{EqNormPsi}, note that $\|\|\bs c_{\ell,n}\|\|_{\psi_2}\lesssim \sqrt{K}$ and $\||\langle \bs c_{\ell,n},\bs m_n\rangle |\|_{\psi_2}\lesssim \|\bs m_n\|$. Finally note that for two subgaussian random variables $X$ and $Y$ with $\|X\|_{\psi_2},\|Y\|_{\psi_2}<\infty$, we have $\|XY\|_{\psi_1}\leq \|X\|_{\psi_2}\|Y\|_{\psi_2}$ (see lemma 7 in~\cite{ahmed2015compressive2}). 
\begin{align}
\|\|\bs c_{\ell,n}\| |\langle \bs c_{\ell,n},\bs m_n\rangle |\|_{\Psi_1}\lesssim \frac{\mu_h\mu_m\sqrt{K}}{\sqrt{LN}}
\end{align}
\begin{align}
\|\|\bm A_{\ell,n}\|\langle \bm A_{\ell,n}, \bs h\bs m^*\rangle\| \|_{\Psi_1}&\leq \frac{\mu_h\mu_m\sqrt{K}}{\sqrt{LN}}\label{boundOrliczLemmaN1a}.
\end{align}
Substituting~\eqref{boundOrliczLemmaN1a} and~\eqref{boundVarianceLemmaN1a} into proposition~\ref{bernstein}, we have
\begin{multline*}
\left\|\sum_{\ell=1}^L\sum_{n=1}^N \mathcal{A}_{\ell,n}^*\mathcal{A}_{\ell,n}\bs h\bs m^* - \bs h\bs m^*\right\|\lesssim \\
\max\left\{\sqrt{\left(\frac{\mu^2_h K}{L}  + \frac{\mu_m^2}{N} \right)}\sqrt{t+\log(LKN)},  \frac{\mu_h\mu_m\sqrt{K}}{\sqrt{LN}} \log(LKN) (t+ \log(LKN))\right\}
\end{multline*}
with probability at least $1-e^{-t}$. Taking $t=\beta_2\log(LN)$ gives the desired result.
\end{proof}

We now bound the second term in the series~\eqref{seriesNeumann2}. The general idea behind this second bound is summarized through lemma~\ref{lemmaTotalSecondTerm} which we recall below.

\lemmaTotalSecondTerm* 

The second term in the Neumann series reads as 
\begin{align}
\mathcal{A}^*\mathcal{A}\mathcal{P}_T \left(\mathcal{P}_T\mathcal{A}^*\mathcal{A}\mathcal{P}_T-\mathcal{P}_T \right)\bs h\bs m^*\label{recallExpressionNeumann2}
\end{align}
This term is a $4^{th}$ order gaussian chaos involving products of dependent gaussians. It is not sub-exponential anymore and we will thus need to turn to a generalization of the Bernstein concentration bound~\eqref{bernstein}. To derive a tail bound on the event $\mathcal{E}_3$ in~\eqref{4thorderNorm}, we proceed as follows. Let $\bs E_{\ell,n}$ denote the expectation,
\begin{align}
\bs E_{\ell,n} & \equiv \mathbb{E}\mathcal{P}_T(\bs A_{\ell,n})\langle \mathcal{P}_T(\bs A_{\ell,n}), \bs h \bs m^*\rangle \\
& = \bs h|\hat{h}[\ell]|^2\bs m_n^* + \bs f_\ell |\hat{h}[\ell]| \|\bs m_n\|^2\bs m^* - \bs h|\hat{h}[\ell]|^2\|\bs m_n\|^2\bs m^*\label{definitionEelln}
\end{align}
so that clearly, $\sum_{\ell,n } \bs E_{\ell,n} = \bs h\bs m^*$. We start by decomposing the norm in~\eqref{4thorderNorm} as a sum of four contributions contributions.
\begin{align}
&\Bigg\|\mathcal{P}_T^\perp \mathcal{A}^*\mathcal{A}\left(\mathcal{P}_T\mathcal{A}^*\mathcal{A}\mathcal{P}_T\left(\bs h\bs m^*\right) - \mathcal{P}_T\left(\bs h\bs m^*\right) \right)\Bigg\|\\
\begin{split}
&=\left\|\mathcal{P}_T^\perp\sum_{(\ell,n)} \bs A_{\ell,n}\left\langle \bs A_{\ell,n}, \left(\sum_{\ell',n'}\mathcal{P}_T(\bs A_{\ell',n'})\langle \mathcal{P}_T(\bs A_{\ell',n'}), \bs h \bs m^*\rangle - \bs E_{\ell',n'}\right) \right\rangle \right.\\
&\left.-\mathcal{P}_T^\perp\mathbb{E}_{(\ell,n)}  \sum_{(\ell,n)} \bs A_{\ell,n}\left\langle \bs A_{\ell,n},\mathbb{E}_{\ell',n'}\left(\sum_{\ell',n'}\mathcal{P}_T(\bs A_{\ell',n'})\langle \mathcal{P}_T(\bs A_{\ell',n'}), \bs h \bs m^*\rangle - \bs E_{\ell',n'}\right)\right\rangle\right\|
\end{split}\\
\label{decompositionNormSecondTermlastbeforeCovZero}
\begin{split}
&\leq \left\|\mathcal{P}_T^\perp\sum_{(\ell,n)}\bs A_{\ell,n}\left\langle \bs A_{\ell,n}, \left(\sum_{(\ell',n')=(\ell,n)} \mathcal{P}_T(\bs A_{\ell',n'})\langle \mathcal{P}_T(\bs A_{\ell',n'}), \bs h \bs m^*\rangle - \bs E_{\ell',n'}\right)  \right\rangle\right.\\
&\left.-\mathcal{P}_T^\perp\mathbb{E}_{(\ell,n)}  \sum_{(\ell,n)} \bs A_{\ell,n}\left\langle \bs A_{\ell,n},\mathbb{E}_{\ell',n'}\left(\sum_{(\ell',n')=(\ell,n)} \mathcal{P}_T(\bs A_{\ell',n'})\langle \mathcal{P}_T(\bs A_{\ell',n'}), \bs h \bs m^*\rangle - \bs E_{\ell',n'}\right)  \right\rangle\right\|
\end{split}\\
\label{decompositionNormSecondTermlastbeforeCovOne}
\begin{split}
& + \left\|\mathcal{P}_T^\perp\sum_{(\ell,n)}\bs A_{\ell,n}\left\langle \bs A_{\ell,n}, \left(\sum_{(\ell',n')\neq (\ell,n)} \mathcal{P}_T(\bs A_{\ell',n'})\langle \mathcal{P}_T(\bs A_{\ell',n'}), \bs h \bs m^*\rangle - \bs E_{\ell',n'}\right)  \right\rangle\right.\\
&\left.-\mathcal{P}_T^\perp\mathbb{E}_{(\ell,n)}  \sum_{(\ell,n)} \bs A_{\ell,n}\left\langle \bs A_{\ell,n}, \mathbb{E}_{\ell',n'}\left(\sum_{(\ell',n')\neq (\ell,n)} \mathcal{P}_T(\bs A_{\ell',n'})\langle \mathcal{P}_T(\bs A_{\ell',n'}), \bs h \bs m^*\rangle - \bs E_{\ell',n'}\right)  \right\rangle\right\|\end{split}
\end{align}
\begin{align}
&\Bigg\|\mathcal{P}_T^\perp \mathcal{A}^*\mathcal{A}\left(\mathcal{P}_T\mathcal{A}^*\mathcal{A}\mathcal{P}_T\left(\bs h\bs m^*\right) - \mathcal{P}_T\left(\bs h\bs m^*\right) \right)\Bigg\|\\
\begin{split}
&\leq \left\|\mathcal{P}_T^\perp\sum_{(\ell,n)}\bs A_{\ell,n}\left\langle \bs A_{\ell,n}, \left(\sum_{(\ell',n')=(\ell,n)} \mathcal{P}_T(\bs A_{\ell',n'})\langle \mathcal{P}_T(\bs A_{\ell',n'}), \bs h \bs m^*\rangle \right)  \right\rangle\right.\\
&\left.-\mathcal{P}_T^\perp\mathbb{E}_{(\ell,n)}  \sum_{(\ell,n)} \bs A_{\ell,n}\left\langle \bs A_{\ell,n},\mathbb{E}_{\ell',n'}\left(\sum_{(\ell',n')=(\ell,n)} \mathcal{P}_T(\bs A_{\ell',n'})\langle \mathcal{P}_T(\bs A_{\ell',n'}), \bs h \bs m^*\rangle \right)  \right\rangle\right\|
\end{split}\label{lastTermDecZero}\\
\begin{split}
& + \left\|\mathcal{P}_T^\perp\sum_{(\ell,n)}\bs A_{\ell,n}\left\langle \bs A_{\ell,n}, \left(\sum_{(\ell',n')\neq (\ell,n)} \mathcal{P}_T(\bs A_{\ell',n'})\langle \mathcal{P}_T(\bs A_{\ell',n'}), \bs h \bs m^*\rangle - \bs E_{\ell',n'}\right)  \right\rangle\right.\\
&\left.-\mathcal{P}_T^\perp\mathbb{E}_{(\ell,n)}  \sum_{(\ell,n)} \bs A_{\ell,n}\left\langle \bs A_{\ell,n}, \mathbb{E}_{\ell',n'}\left(\sum_{(\ell',n')\neq (\ell,n)} \mathcal{P}_T(\bs A_{\ell',n'})\langle \mathcal{P}_T(\bs A_{\ell',n'}), \bs h \bs m^*\rangle - \bs E_{\ell',n'}\right)  \right\rangle\right\|\end{split}\label{lastTermDecOne}\\
& + \left\|\sum_{\ell,n}\bs A_{\ell,n}\langle \bs A_{\ell,n}, \bs E_{\ell,n}\rangle - \mathbb{E} \sum_{\ell,n}\bs A_{\ell,n}\langle \bs A_{\ell,n}, \bs E_{\ell,n}\rangle\right\|
\end{align}
Now let $\mbox{Cov}(\mathcal{A}^*\mathcal{A},\mathcal{P}_T\mathcal{A}^*\mathcal{A}\mathcal{P}_T\bs h\bs m^*)$ be defined as
\begin{multline}
\mbox{Cov}(\mathcal{A}^*\mathcal{A},\mathcal{P}_T\mathcal{A}^*\mathcal{A}\mathcal{P}_T\bs h\bs m^* )\equiv\mathcal{P}_T^\perp\mathbb{E}_{\ell,n} \sum_{(\ell,n)} \bs A_{\ell,n}\langle \bs A_{\ell,n},\left(\mathcal{P}_T(\bs A_{\ell,n})\langle \mathcal{P}_T(\bs A_{\ell,n}), \bs h \bs m^*\rangle \right)\rangle \\
 -  \mathcal{P}_T^\perp\mathbb{E}_{(\ell,n)}  \sum_{(\ell,n)} \bs A_{\ell,n}\left\langle \bs A_{\ell,n},\left(\sum_{(\ell',n')=(\ell,n)}\mathbb{E}_{\ell',n'} \mathcal{P}_T(\bs A_{\ell',n'})\langle \mathcal{P}_T(\bs A_{\ell',n'}), \bs h \bs m^*\rangle \right)\right\rangle\label{expandedCovariance}
\end{multline}
From~\eqref{expandedCovariance}, we can expand~\eqref{decompositionNormSecondTermlastbeforeCovZero} into
\begin{align}\begin{split}
&\leq \left\|\mathcal{P}_T^\perp\sum_{(\ell,n)}\bs A_{\ell,n}\left\langle \bs A_{\ell,n}, \left( \mathcal{P}_T(\bs A_{\ell,n})\langle \mathcal{P}_T(\bs A_{\ell,n}), \bs h \bs m^*\rangle\right)  \right\rangle\right.\\
&\left.-\mathcal{P}_T^\perp\mathbb{E}_{(\ell,n)}  \sum_{(\ell,n)} \bs A_{\ell,n}\left\langle \bs A_{\ell,n},\left(\mathcal{P}_T(\bs A_{\ell,n})\langle \mathcal{P}_T(\bs A_{\ell,n}), \bs h \bs m^*\rangle \right)  \right\rangle\right\|\end{split}
\label{firstTermSecondOrderN0}\\
& +\Bigg\|\mbox{Cov}(\mathcal{A}^*\mathcal{A},\mathcal{P}_T\mathcal{A}^*\mathcal{A}\mathcal{P}_T\bs h\bs m^*)\Bigg\|\label{covarianceTerm0}\\
\begin{split}
& + \left\|\mathcal{P}_T^\perp\sum_{(\ell,n)}\bs A_{\ell,n}\left\langle \bs A_{\ell,n}, \left(\sum_{(\ell',n')\neq (\ell,n)} \mathcal{P}_T(\bs A_{\ell',n'})\langle \mathcal{P}_T(\bs A_{\ell',n'}), \bs h \bs m^*\rangle - \bs E_{\ell',n'}\right)  \right\rangle\right.\\
&\left.-\mathcal{P}_T^\perp\mathbb{E}_{(\ell,n)}  \sum_{(\ell,n)} \bs A_{\ell,n}\left\langle \bs A_{\ell,n}, \left(\sum_{(\ell',n')\neq (\ell,n)} \mathcal{P}_T(\bs A_{\ell',n'})\langle \mathcal{P}_T(\bs A_{\ell',n'}), \bs h \bs m^*\rangle - \bs E_{\ell',n'}\right)  \right\rangle\right\|\end{split}\label{mixedLastTermSecondOrderNeumann0}\\
& + \left\|\sum_{\ell,n}\bs A_{\ell,n}\langle \bs A_{\ell,n}, \bs E_{\ell,n}\rangle - \mathbb{E} \sum_{\ell,n}\bs A_{\ell,n}\langle \bs A_{\ell,n}, \bs E_{\ell,n}\rangle\right\|\label{lastNormEtoBound}
\end{align}
We will successively bound each of terms~\eqref{firstTermSecondOrderN0},~\eqref{covarianceTerm0},~\eqref{mixedLastTermSecondOrderNeumann0} and~\eqref{lastNormEtoBound} through corresponding lemmas~\ref{lemmaEbound}, ~\eqref{lemmaBoundCovariance}, ~\eqref{FourthOrderlemma}, and~\eqref{lemmaTperpSecondTerm} below. We start with the last term. This term is a sum of sub-exponential random variables and can be bounded through proposition~\ref{bernstein}. This idea is summarized by lemma~\ref{lemmaEbound} below which is proved in section~\ref{proofLemmaE},
\begin{restatable}{lemma}{lemmaEbound} 
Let $\hat{\bs c}_{\ell,n}$ be defined as in section~\ref{notations} (including the scaling $\sqrt{L}$) with $c_{\ell,n}[k]\sim \mathcal{N}(0,1/L)$ i.i.d. gaussian and $\bs A_{\ell,n} = \bs f_\ell\hat{\bs c}_{\ell,n}^*$. Let $\bs E_{\ell,n}$ be defined as in~\eqref{definitionEelln}. Then as soon as $L\gtrsim (1/\delta_5)\beta_5 K\mu_h^2$, $N\gtrsim (1/\delta_5)\beta_5 \mu_m^2$ for coherences $\mu_m^2$ and $\mu_h^2$ defined as in~\eqref{coherencem} and~\eqref{coherenceh},  
\label{lemmaEbound} 
 \begin{align}
\left\|\sum_{\ell,n}\bs A_{\ell,n}\langle \bs A_{\ell,n}, \bs E_{\ell,n}\rangle - \mathbb{E} \sum_{\ell,n}\bs A_{\ell,n}\langle \bs A_{\ell,n}, \bs E_{\ell,n}\rangle\right\|&\lesssim \delta_5\label{normLemmaEbound}
\end{align}
with probability at least $1 - (LN)^{-\beta_5}$.
\end{restatable}
The covariance term~\eqref{covarianceTerm0} is a purely deterministic term. It can be controlled through lemma~\eqref{lemmaBoundCovariance} below which is proved in section~\ref{lemmaCovarianceSec}.
\begin{restatable}{lemma}{lemmaBoundCovariance} 
\label{lemmaBoundCovariance}
Let the covariance $\mbox{\upshape Cov}(\mathcal{A}^*\mathcal{A},\mathcal{P}_T\mathcal{A}^*\mathcal{A}\mathcal{P}_T\bs h\bs m^*)$ be defined as in~\eqref{expandedCovariance}. Then for any constant $\delta_6$, as soon as $L\gtrsim (1/\delta_6)K\mu_h^2$ and $N\gtrsim (1/\delta_6)\mu_m^2$,
\begin{align}
\Bigg\|\mbox{\upshape Cov}(\mathcal{A}^*\mathcal{A},\mathcal{P}_T\mathcal{A}^*\mathcal{A}\mathcal{P}_T\bs h\bs m^*)\Bigg\|\lesssim \delta_6
\end{align}
\end{restatable}
The first term~\eqref{firstTermSecondOrderN0} is a sum of $LN$ independent random variables, each defined from fourth order monomials in the gaussian vectors $\bs c_{\ell,n}$, $1\leq \ell\leq L$, $1\leq n\leq N$. We will bound this first term through the Rosenthal-Pinelis inequality (see (3.1) in~\cite{gine2012high} as well as~\cite{tropp2016expected}, Theorem 1, for the matrix version.) which extends proposition~\ref{bernstein}. This inequality bounds the expectation of the operator norm of a sum of independent random matrices from the knowledge of the variance and a bound on the expectation of the largest operator norm among those matrices. It is recalled through proposition~\ref{rosenthalTropp} below,
\begin{proposition}[Rosenthal-Pinelis~\cite{tropp2016expected}]\label{rosenthalTropp}
Let $\bs Z_1,\ldots, \bs Z_n$ be i.i.d. random matrices of dimension $d_1\times d_2$ with $\mathbb{E}\{\bs Z_i\} = 0$. Let
\begin{equation}
\sigma_Z = \max \left\{\left\|\sum_{i=1}^n \mathbb{E}\left(\bs Z_i\bs Z_i^*\right)\right\|, \left\|\sum_{i=1}^n \mathbb{E}\left(\bs Z_i^*\bs Z_i\right)\right\|\right\},\label{varianceRosenthal}
\end{equation}
as well as 
\begin{align}
B = \left(\mathbb{E} \max_i \|\bs Z_i\|^2\right)^{1/2}, \quad \mbox{and} \quad C(d_1,d_2)= 4\left(1+2\lceil \log(d_1+d_2)\rceil\right).\label{expectedOpNormRosP}
\end{align}
Then the expectation of the norm of $\sum_{i=1}^n \bs Z_i$ can be bounded as 
\begin{align}
\mathbb{E}\|\bs Z\|^2\leq \left(\sqrt{C(d_1,d_2)\sigma} + C(d_1,d_2)B\right)^2.\label{rosenthalTroppBound1}
\end{align} 
\end{proposition}
Let $Q$ denote the norm
\begin{align}
Q &\equiv \left\|\mathcal{P}_T^\perp\sum_{(\ell,n)}\bs A_{\ell,n}\left\langle \bs A_{\ell,n}, \left(\mathcal{P}_T(\bs A_{\ell,n})\langle \mathcal{P}_T(\bs A_{\ell,n}), \bs h \bs m^*\rangle \right)  \right\rangle\right.\\
&\left.-\mathcal{P}_T^\perp\mathbb{E}_{(\ell,n)}  \sum_{(\ell,n)} \bs A_{\ell,n}\left\langle \bs A_{\ell,n},\left(\mathcal{P}_T(\bs A_{\ell,n})\langle \mathcal{P}_T(\bs A_{\ell,n}), \bs h \bs m^*\rangle \right)  \right\rangle\right\|
\end{align}
As soon as one can control the expectation of the operator norm, one can then use standard concentration tools such as Markov's inequality to derive a tail bound on the event $\mathcal{E}_5$ defined as $\mathcal{E}_5 \equiv\left\{ Q \geq \delta\right\}$ as 
\begin{align}
&\pr\left(Q \geq t \right)\leq \frac{\mathbb{E} Q}{t}
\end{align}
Lemma~\ref{FourthOrderlemma} below precisely derives such a bound on the expectation $\mathbb{E}Q$. This lemma is proved in section~\eqref{prooflemmaFourthOrderDep}.
\begin{restatable}[Fourth order dependence]{lemma}{FourthOrderlemma}
\label{FourthOrderlemma}
Let $\hat{\bs c}_{\ell,n}$ be defined as in section~\ref{notations} (including the scaling $\sqrt{L}$), where $C_n[\ell,k]\sim \mathcal{N}(0,1/L)$ are i.i.d. gaussian. Let $\bs A_{\ell,n} = \bs f_\ell\bs c_{\ell,n}^*$ and let the operators $\mathcal{A}_{\ell,n}\;:\;\mathbb{C}^{L\times KN} \mapsto \mathbb{C}^{L\times N}$ be defined from those matrices $\bs A_{\ell,n}$ as in~\eqref{DefinitionLinearMap}. The coherences $\mu_m^2$ and $\mu_h^2$ are defined as in~\eqref{coherencem} and~\eqref{coherenceh}. Then,
\begin{align}
\mathbb{E}\left\|\sum_{\ell,=1}^L\sum_{n=1}^N \mathcal{A}_{\ell,n}^*\mathcal{A}_{\ell,n}(\mathcal{P}_T\mathcal{A}_{\ell,n}^*\mathcal{A}_{\ell,n}\mathcal{P}_T)\bs h\bs m^*  - \mathbb{E} \mathcal{A}_{\ell,n}^*\mathcal{A}_{\ell,n}(\mathcal{P}_T\mathcal{A}_{\ell,n}^*\mathcal{A}_{\ell,n}\mathcal{P}_T)\bs h\bs m^*\right\| \lesssim \max \left(\frac{\mu_L^2}{L},\frac{\mu_m^2}{N}\right) \label{fourthOrderequation1}
\end{align}
\end{restatable}
Finally, the third term~\eqref{mixedLastTermSecondOrderNeumann0} is of the form
\begin{align}
\left\|\sum_{(\ell,n)\neq (\ell',n')}h(\mbox{vec}(\bs A_{\ell,n})\otimes\mbox{vec}(\bs A_{\ell,n})^*, \mbox{vec}(\bs A_{\ell',n'})\otimes\mbox{vec}(\bs A_{\ell',n'})^*)\right\|
\end{align}
for appropriate measurable functions $h(\bs X,\bs Y)$. This is a special case of a $U$-statistics (see~\cite{de2012decoupling}, chapter 3. as well as~\cite{hoeffding1948class}). We bound this last term by using a decoupling argument from de la Pen\~a et al.~\cite{de1995decoupling}. This result is summarized by the following lemma which is proved in section~\ref{sectionUstatistics} below.
\begin{restatable}{lemma}{lemmaTperpSecondTerm} 
Let $\hat{\bs c}_{\ell,n}$ be defined as in section~\ref{notations} (including the scaling $\sqrt{L}$) where $c_{\ell,n}[k]\sim \mathcal{N}(0,1/L)$ are i.i.d gaussian. Let $\bs A_{\ell,n} = \bs f_\ell\bs c_{\ell,n}^*$. And $\mathcal{A}_{\ell,n}$ be defined from the matrices $\bs A_{\ell,n}$ as the corresponding operators $\mathcal{A}_{\ell,n}(\bs X) = \langle \bs A_{\ell,n},\bs X\rangle$. The coherences $\mu_m^2$ and $\mu_h^2$ are defined as in~\eqref{coherencem} and~\eqref{coherenceh}. Then for any constant $\delta_8, \beta_8$, as soon as $L\gtrsim \beta_8 (1/\delta_8)K\mu_h^2$ and $N\gtrsim (1/\delta_8)\beta_8 \mu_m^2$
\label{lemmaTperpSecondTerm} 
\begin{align}
& \left\|\mathcal{P}_T^\perp\sum_{(\ell,n)}\bs A_{\ell,n}\left\langle \bs A_{\ell,n}, \left(\sum_{(\ell',n')\neq (\ell,n)} \mathcal{P}_T(\bs A_{\ell',n'})\langle \mathcal{P}_T(\bs A_{\ell',n'}), \bs h \bs m^*\rangle - \bs E_{\ell',n'}\right)  \right\rangle\right.\\
&\left.-\mathcal{P}_T^\perp\mathbb{E}  \sum_{(\ell,n)} \bs A_{\ell,n}\left\langle \bs A_{\ell,n}, \left(\sum_{(\ell',n')\neq (\ell,n)} \mathcal{P}_T(\bs A_{\ell',n'})\langle \mathcal{P}_T(\bs A_{\ell',n'}), \bs h \bs m^*\rangle - \bs E_{\ell',n'}\right)  \right\rangle\right\|\lesssim \delta_8\label{normMixed2UStatistics}
\end{align}
with probability at least $1-(LN)^{-\beta_8}$
\end{restatable}
Combining the results of lemmas~\ref{lemmaEbound} to~\ref{lemmaTperpSecondTerm} with the discussion above, with $(1/\delta_3)\gtrsim \max\{(1/\delta_i)\}_{i=4}^8$, $\beta\geq \max\{\beta_i\}$, we have that as soon as $L\gtrsim \beta(1/\delta_3)\beta_3 K\mu_h^2$ and $N\gtrsim \beta_3(1/\delta_3)\mu_m^2$, the bound of lemma~\ref{lemmaTotalSecondTerm} is satisfied with probability at least $1 -  c(LN)^{-\beta_3}$ where $c$ is a constant. The next section proceeds with the proofs of each lemma.

\section{\label{auxiliaryLemmasProof}Proofs of Auxiliary lemmas}

\subsection{\label{proofLemmaE}Proof of lemma~\ref{lemmaEbound}}
\lemmaEbound*
The norm on the LHS of~\eqref{normLemmaEbound} is the norm of a sum of subexponential random variables and we can thus use proposition~\ref{bernstein}. We start by deriving the bound on the variance. Recall that $\bs E_{\ell,n}$ is defined as $\bs E_{\ell,n} = \mathbb{E}\mathcal{P}_T(\bs A_{\ell,n})\langle \mathcal{P}_T(\bs A_{\ell,n}), \bs h\bs m^*\rangle$. Let $\mathcal{Z}_{\ell,n}$ be defined as 
\begin{align}
\mathcal{Z}_{\ell,n} &\equiv \sum_{\ell,n} \bs A_{\ell,n}\langle \bs A_{\ell,n}, \bs E_{\ell,n}\rangle - \mathbb{E}\sum_{\ell,n} \bs A_{\ell,n}\langle \bs A_{\ell,n}, \bs E_{\ell,n}\rangle
\end{align}
Using~\eqref{definitionEelln},  $|\langle \bs A_{\ell,n},\bs E_{\ell,n}\rangle |^2$ expands as 
\begin{align}
|\langle \bs A_{\ell,n},\bs E_{\ell,n}\rangle |^2 &= \left|\hat{h}[\ell]^3\langle \bs m_n,\bs c_{\ell,n}\rangle + \|\bs f_{\ell}\|^2\hat{h}[\ell]\|\bs m_n\|^2 \langle \bs m_n,\bs c_{\ell,n}\rangle - \hat{h}[\ell]^3\|\bs m_n\|^2\langle \bs m_n,\bs c_{\ell,n}\rangle\right|^2\\
&\lesssim|\hat{h}[\ell]|^6|\langle \bs m_n,\bs c_{\ell,n}\rangle|^2 + |\hat{h}[\ell]|^2 \|\bs m_n\|^4 |\langle \bs m_n,\bs c_{\ell,n}\rangle|^2 + |\hat{h}[\ell]|^6\|\bs m_n\|^4|\langle \bs m_n,\bs c_{\ell,n}\rangle|^2
\end{align}
Following the proofs of lemmas~\ref{lemmaT} and~\ref{lemmaTperp}, each of the variance bounds in~\eqref{varianceBernstein} can be expressed as
\begin{align}
\sigma & \leq \max \left\{\left\|\mathbb{E}\sum_{\ell,n}\mathcal{Z}_{\ell,n}^*\mathcal{Z}_{\ell,n}\right\|, \left\|\mathbb{E}\sum_{\ell,n}\mathcal{Z}_{\ell,n}\mathcal{Z}_{\ell,n}^*\right\|\right\}\\
&\leq \max \left\{\left\|\mathbb{E}\sum_{\ell,n}\bs f_\ell\bs f_{\ell}^* \|\bs c_{\ell,n}\|^2 |\langle \bs A_{\ell,n},\bs E_{\ell,n}\rangle |^2  \right\|, \left\|\mathbb{E}\sum_{\ell,n}\bs c_{\ell,n}\bs c_{\ell,n}^* \|\bs f_{\ell}\|^2 |\langle \bs A_{\ell,n},\bs E_{\ell,n}\rangle |^2 \right\|\right\}
\end{align}
The first term can be bounded as 
\begin{align}
\left\|\mathbb{E}\sum_{\ell,n}\bs f_\ell\bs f_{\ell}^* \|\bs c_{\ell,n}\|^2 |\langle \bs A_{\ell,n},\bs E_{\ell,n}\rangle |^2  \right\|&\lesssim \sup_{\ell }\left\|\mathbb{E}\sum_{n}\bs f_\ell\bs f_{\ell}^* \|\bs c_{\ell,n}\|^2 |\langle \bs A_{\ell,n},\bs E_{\ell,n}\rangle |^2\right\|\\
&\lesssim K \frac{\mu_h^6}{L^3}\|\bs m_n\|^2 + K\frac{\mu_h^2}{L} \|\bs m_n\|^6 + K\frac{\mu_h^6}{L^3}\|\bs m_n\|^6\\
&\lesssim K\frac{\mu_h^2}{L}\left(\frac{\mu_h^4}{L^2}\frac{\mu_m^2}{N} + \frac{\mu_m^6}{N^3} + \frac{\mu_h^4}{L^2}\frac{\mu_m^6}{N^3}\right) 
\end{align}
For the second term, we can similarly write,
\begin{align}
\left\|\mathbb{E}\sum_{\ell,n}\bs c_{\ell,n}\bs c_{\ell,n}^* \|\bs f_{\ell}\|^2 |\langle \bs A_{\ell,n},\bs E_{\ell,n}\rangle |^2 \right\|& \lesssim \left\|\sum_{\ell,n}\left(\frac{\mu_h^4}{L^2}|\hat{h}[\ell]|^2 + |\hat{h}[\ell]|^2\|\bs m_n\|^4 + \frac{\mu_h^4}{L^2}|\hat{h}[\ell]|^2\|\bs m_n\|^4\right)\bs m_n\bs m_n^*\right\|\label{secondBoundLemmaEvara}\\
& +  \left\|\sum_{\ell,n}\left(|\hat{h}[\ell]|^2\frac{\mu_h^4}{L^2} + |\hat{h}[\ell]|^2\|\bs m_n\|^4 + \frac{\mu_h^4}{L^2}|\hat{h}[\ell]|^2\|\bs m_n\|^4\right)\bs \|\bs m_n\|^2\bs I\right\|\label{secondBoundLemmaEvarb}
\end{align}
Equations~\eqref{secondBoundLemmaEvara} and~\eqref{secondBoundLemmaEvarb} can further be reduced to 
\begin{align}
\left\|\mathbb{E}\sum_{\ell,n}\bs c_{\ell,n}\bs c_{\ell,n}^* \|\bs f_{\ell}\|^2 |\langle \bs A_{\ell,n},\bs E_{\ell,n}\rangle |^2 \right\|& \lesssim\sup_n \left(\frac{\mu_h^4}{L^2} +\|\bs m_n\|^4 + \frac{\mu_h^4}{L^2}\|\bs m_n\|^4\right)\|\bs m_n\|^2\\
& +  \left(\frac{\mu_h^4}{L^2} + \|\bs m_n\|^4 + \frac{\mu_h^4}{L^2}\|\bs m_n\|^4\right)\bs \|\bs m_n\|^2\\
&\lesssim \frac{\mu_h^4}{L^2} + \frac{\mu_m^4}{N^2}
\end{align}
Finally the Orlicz norm is bounded by noting that the variables are subexponentials and by developing $\|\mathcal{Z}_{\ell,n}\|_{\Psi_1}$. Following the same reasoning as in the proof of lemmas~\ref{lemmaTperp} and~\ref{lemmaT}, we can write,
\begin{align}
&\|\bs f_{\ell}\bs c_{\ell,n}\left(\hat{h}[\ell]^3 |\langle \bs c_{\ell,n},\bs m_n\rangle + |\hat{h}[\ell]|\|\bs m_n\|^2\langle \bs m_n,\bs c_{\ell,n}\rangle - \hat{h}[\ell]^3\|\bs m_n\|^2\langle \bs m_n,\bs c_{\ell,n}\rangle\right)\|\\
&\lesssim \|\|\bs c_{\ell,n}\|\|_{\psi_2} \||\langle \bs c_{\ell,n},\bs m_n\rangle|\|_{\psi_2}\left(\frac{\mu_h^3}{L^{3/2}}+ \frac{\mu_h}{\sqrt{L}}\frac{\mu_m^2}{N} \right) \\
&\lesssim \sqrt{K} \frac{\mu_m}{\sqrt{N}}\left(\frac{\mu_h^3}{L^{3/2}}+ \frac{\mu_h}{\sqrt{L}}\frac{\mu_m^2}{N} \right)
\end{align}
\begin{align}
\|\|\bs f_\ell\bs c_{\ell,n}\|\|_{\Psi_1}&\lesssim \sqrt{K}\left(\frac{\mu_h^3}{L^{3/2}}\frac{\mu_m}{\sqrt{N}} + \frac{\mu_h}{\sqrt{L}}\frac{\mu_m^{3}}{N^{3/2}}\right) \\
&\lesssim \frac{\mu_h\mu_m\sqrt{K}}{\sqrt{LN}}
\end{align}
Using the fact that $\|\bs c_{\ell,n}\|$ and $|\langle \bs c_{\ell,n}, \bs m_n\rangle |$ are sub-gaussians. We can now apply proposition~\ref{bernstein} from which the conclusion follows,

\subsection{\label{lemmaCovarianceSec}Proof of lemma~\ref{lemmaBoundCovariance}}

The lemma below first shows that the covariance~\eqref{expandedCovariance} can be made arbitrarily small.

\lemmaBoundCovariance*

\begin{proof}
The covariance in~\eqref{expandedCovariance} expands as 
\begin{align}\begin{split}
\mbox{Cov}(\mathcal{A}^*\mathcal{A},\left(\mathcal{P}_T - \mathcal{P}_T\mathcal{A}^*\mathcal{A}\mathcal{P}_T\right)\bs h\bs m^*) = &\mathcal{P}_T^\perp\mathbb{E}\sum_{\ell,n} \bs f_\ell\bs c_{\ell,n}^*\langle \bs f_{\ell}\bs c_{\ell,n}^*,\mathcal{P}_T(\bs f_\ell\bs c_{\ell,n}^*)\langle \mathcal{P}_T(\bs f_\ell\bs c_{\ell,n}^*) ,\bs h\bs m^*\rangle \rangle\\
&- \mathcal{P}_T^\perp\sum_{\ell,n} \mathbb{E}\bs f_\ell\bs c_{\ell,n}^*\langle \bs f_{\ell}\bs c_{\ell,n}^*,\mathbb{E}\mathcal{P}_T(\bs f_\ell\bs c_{\ell,n}^*)\langle \mathcal{P}_T(\bs f_\ell\bs c_{\ell,n}^*),\bs h\bs m^* \rangle \rangle
\end{split} \label{dependentSum} \end{align}
As explained above, the second term vanishes. Developing the first term in~\eqref{dependentSum}, we have,
\begin{align}
&\mbox{\upshape Cov}(\mathcal{A}^*\mathcal{A},\left( \mathcal{P}_T\mathcal{A}^*\mathcal{A}\mathcal{P}_T\right))\nonumber \\
&=\mathbb{E}\sum_{\ell=1}^L\sum_{n=1}^N\bs f_\ell \bs c_{\ell,n}^* \left(|\hat{h}[\ell]|^2 \|\bs c_{\ell,n}\|^2 + |\langle \bs c_{\ell,n},\bs m_n\rangle |^2 - |\hat{h}[\ell]|^2|\langle \bs c_{\ell,n},\bs m_n\rangle |^2\right)\hat{h}[\ell]\langle \bs c_{\ell,n},\bs m_n\rangle  \\
&= \mathbb{E}\sum_{\ell=1}^L\sum_{n=1}^N \bs f_\ell\bs c_{\ell,n}^* \left(|\hat{h}[\ell]|^2\hat{h}[\ell] \|\bs c_{\ell,n}\|^2 \langle \bs c_{\ell,n},\bs m_n\rangle + |\langle \bs c_{\ell,n},\bs m_n\rangle |^2 \langle \bs c_{\ell,n},\bs m_n\rangle \hat{h}[\ell] \right)\\
& -\mathbb{E}\sum_{\ell=1}^L\sum_{n=1}^N \bs f_\ell\bs c_{\ell,n}^*\left( |\hat{h}[\ell]|^2\hat{h}[\ell] |\langle \bs c_{\ell,n},\bs m_n \rangle|^2\langle \bs c_{\ell,n},\bs m_n\rangle  \right)\\
&= \sum_{\ell=1}^L\sum_{n=1}^N K|\hat{h}[\ell]|^2\hat{h}[\ell]\bs f_\ell\bs m_n^*  + \mathbb{E}\bs f_\ell\bs c_{\ell,n}^*|\langle \bs c_{\ell,n},\bs m_n\rangle |^2 \langle \bs c_{\ell,n},\bs m_n\rangle \hat{h}[\ell]\\
& -\mathbb{E}\sum_{\ell=1}^L\sum_{n=1}^N\bs f_\ell\bs c_{\ell,n}^* |\hat{h}[\ell]|^2\hat{h}[\ell] |\langle \bs c_{\ell,n},\bs m_n \rangle|^2\langle \bs c_{\ell,n},\bs m_n\rangle \label{convTempBound}
\end{align}
In the last line we use $\mathbb{E}(\bs f_\ell \bs c_{\ell,n}^*)\|\bs c_{\ell,n}\|^2 \langle \bs c_{\ell,n},\bs m_n\rangle = K \bs f_\ell \bs m_n^*$. When considering the remaining terms, $|\langle \bs c_{\ell,n},\bs m_n\rangle|^2\langle \bs c_{\ell,n},\bs m_n\rangle $,  those terms can be decomposed as 
\begin{equation}
\begin{aligned}
\label{STOrderb}
\displaystyle |\langle \bs c_{\ell,n},\bs m_n\rangle|^2\langle \bs c_{\ell,n},\bs m_n\rangle& \displaystyle = \sum_{k=1}^K (c_{\ell,n}[k])^3 m_n^3[k] + \sum_{k=1}^K 3c_{\ell,n}^2[k] m_n^2[k] \sum_{j=1}^K c_{\ell,n}[j]m_n[j]\\
& \displaystyle+ 6\sum_{i\neq j\neq k} c_{\ell,n}[i]c_{\ell,n}[j]c_{\ell,n}[k]m_{n}[i]m_{n}[j]m_{n}[k].
\end{aligned}
\end{equation}
When multiplying~\eqref{STOrderb} by $\bs f_\ell\bs c_{\ell,n}^*$, the expectation of the sum reduces to 
\begin{align}
\mathbb{E}\bs f_\ell \bs c_{\ell,n}^*|\langle \bs c_{\ell,n},\bs m_n\rangle|^2\langle \bs c_{\ell,n},\bs m_n\rangle& = \mathbb{E}\bs f_\ell \bs c_{\ell,n}\left(\sum_{k=1}^K (c_{\ell,n}[k])^3 m_n^3[k]+ \sum_{k=1}^K 3c_{\ell,n}^2[k] m_n^2[k] \sum_{j=1}^K c_{\ell,n}[j]m_n[j]\right)\\
& + \mathbb{E}\bs f_\ell \bs c_{\ell,n}^* 6\sum_{i\neq j\neq k} c_{\ell,n}[i]c_{\ell,n}[j]c_{\ell,n}[k]m_{n}[i]c_{n}[j]c_{n}[k]\\
& = \sum_{k=1}^K \bs f_\ell m_n[k]^3(\bs e_{n}\otimes \bs e_k)^* + 3 \|\bs m_n\|^2 \bs f_\ell \bs m_n^*\label{expandedExp}
\end{align}
Plugging~\eqref{expandedExp} into~\eqref{convTempBound}, gives 
\begin{align}
\mbox{\upshape Cov}(\mathcal{A}^*\mathcal{A},\left( \mathcal{P}_T\mathcal{A}^*\mathcal{A}\mathcal{P}_T\right))= &\sum_{\ell=1}^L\sum_{n=1}^N K|\hat{h}[\ell]|^2\hat{h}[\ell]\bs f_\ell\bs m_n^*\label{firstTermExpectedCov}\\
&  + \sum_{\ell=1}^L\sum_{n=1}^N\left(\sum_{k=1}^K \bs f_\ell m_n[k]^3 (\bs e_{\ell,n}\otimes \bs e_k)^* + \bs f_\ell \bs m_n^* \|\bs m_n\|^2\right) \left(\hat{h}[\ell] + |\hat{h}[\ell]|^2\hat{h}[\ell]\right) \label{secondTermExpectedCov} \\
\end{align}
Each of the two terms in~\eqref{firstTermExpectedCov} and~\eqref{secondTermExpectedCov} have operator norms respectively bounded as
\begin{align}
\left\|\sum_{\ell=1}^L\sum_{n=1}^N K|\hat{h}[\ell]|^2\hat{h}[\ell]\bs f_\ell\bs m_n^*\right\|&\lesssim K\frac{\mu_h^2}{L}\left\|\sum_{\ell=1}^L\bs f_\ell\hat{h}[\ell]\right\|\|\bs m\|\\
& \lesssim \frac{K\mu_h^2}{L}\label{boundCovariance1}
\end{align}
where we use $\sum_\ell \hat{h}[\ell]\bs f_\ell = \bs h$ as well as $\|\bs h\| = \|\bs m\|=1$. 
\begin{align}
\left\|\sum_{\ell=1}^L\sum_{n=1}^N\left(\sum_{k=1}^K \bs f_\ell m_n[k]^3 (\bs e_{n}\otimes \bs e_k)^* + \bs f_\ell \bs m_n^* \|\bs m_n\|^2\right)\left(\hat{h}[\ell] + |\hat{h}[\ell]|^2\hat{h}[\ell]\right)  \right\|\lesssim \frac{\mu_m^2}{N}\label{boundCovariance2}
\end{align}
where we use $\sum_{\ell =1}^L \hat{h}[\ell]\bs f_\ell = \bs h $ and $\left\|\sum_n \sum_k m_n[k]^3 (\bs e_{n}\otimes \bs e_k)^* \right\|\leq \sup_{n,k}|m_n[k]|^2\left\|\sum_n \sum_k m_n[k] (\bs e_{\ell,n}\otimes \bs e_k)^* \right\|$ as $\|\bs m_n\|\leq 1$. Finally,~\eqref{secondTermExpectedCov} can thus be bounded as
\begin{align}
\left\|\sum_{\ell,n}\left(\sum_k \bs f_\ell m_n[k]^3 (\bs e_{\ell,n}\otimes \bs e_k)^* + \sum_{\ell,n}\bs f_\ell \bs m_n^* \|\bs m_n\|^2\right) |\hat{h}[\ell]|^2\hat{h}[\ell]\right\|\lesssim \frac{\mu_h^2}{L} \frac{\mu_m^2}{N}\label{boundCovariance3}
\end{align}
Note that each of the bounds~\eqref{boundCovariance1},~\eqref{boundCovariance2} and~\eqref{boundCovariance3} can be made smaller than $\delta$ for any constant $\delta$ as soon as $L\gtrsim (1/\delta)\mu_h^2 K$ and $N\gtrsim \mu_m^2  (1/\delta)$. This concludes the proof of lemma~\ref{lemmaBoundCovariance}.
\end{proof}

\subsection{\label{prooflemmaFourthOrderDep}Proof of lemma~\ref{FourthOrderlemma}}

Before giving the proof of the lemma, we recall the statement for clarity.

\FourthOrderlemma*

\begin{proof}

Let us first develop the sum $\sum_{\ell=1}^L\sum_{n=1}^N \mathcal{A}_{\ell,n}^*\mathcal{A}_{\ell,n}(\mathcal{P}_T\mathcal{A}_{\ell,n}^*\mathcal{A}_{\ell,n}\mathcal{P}_T)\bs h\bs m^*$. We have 
\begin{align}
&\sum_{\ell=1}^L\sum_{n=1}^N \mathcal{A}_{\ell,n}^*\mathcal{A}_{\ell,n}(\mathcal{P}_T\mathcal{A}_{\ell,n}^*\mathcal{A}_{\ell,n}\mathcal{P}_T)\bs h\bs m^*\\
& = \sum_{\ell=1}^L\sum_{n=1}^N \bs f_\ell \bs c_{\ell,n}^* \left(\langle \bs f_\ell\bs c_{\ell,n}^*, \bs h\hat{h}[\ell]\bs c_{\ell,n}^* +  \bs f_\ell \langle \bs c_{\ell,n}, \bs m_n\rangle \bs m^* - \bs h\hat{h}[\ell] \langle \bs c_{\ell,n}, \bs m_n\rangle \bs m^*  \rangle \right)\hat{h}[\ell] \langle \bs c_{\ell,n},\bs m_n\rangle \\
& = \sum_{\ell=1}^L\sum_{n=1}^N \bs f_\ell \bs c_{\ell,n}^* \left(|\hat{h}[\ell]|^2 \|\bs c_{\ell,n}\|^2 +  |\langle \bs c_{\ell,n}, \bs m_n\rangle|^2  - |\hat{h}[\ell]|^2 |\langle \bs c_{\ell,n}, \bs m_n\rangle|^2 \right)\hat{h}[\ell] \langle \bs c_{\ell,n},\bs m_n\rangle \label{sumConcentrationCovariancetmp1}
\end{align}
To prove~\eqref{fourthOrderequation1} through proposition~\eqref{rosenthalTropp}, we need to bound the variance and the expected operator norm of each of the variables within the norm. We first compute the variance. Deriving the bound on the variance is the point of section~\ref{secVarianceFourthOrder} below. Section~\ref{secOpNormFourthOrder} then derives a bound on the expectation of any of the variable operator norm.

\subsubsection{\label{secVarianceFourthOrder}Bound on the variance}
Squaring the weights of each of the terms in~\eqref{sumConcentrationCovariancetmp1}, we get
\begin{align}
&\left||\hat{h}[\ell]|^2 \|\bs c_{\ell,n}\|^2 +  |\langle \bs c_{\ell,n}, \bs m_n\rangle|^2  - |\hat{h}[\ell]|^2 |\langle \bs c_{\ell,n}, \bs m_n\rangle|^2 \right|^2 |\hat{h}[\ell]|^2 |\langle \bs c_{\ell,n},\bs m_n\rangle|^2 \\
& \lesssim |\hat{h}[\ell]|^6 \|\bs c_{\ell,n}\|^4 |\langle \bs c_{\ell,n},\bs m_n\rangle|^2 + |\hat{h}[\ell]|^2 |\langle \bs c_{\ell,n}, \bs m_n\rangle|^6  + |\hat{h}[\ell]|^6 |\langle \bs c_{\ell,n}, \bs m_n\rangle|^6.\label{scalarWeightCovarianceConc}
\end{align}
Let us, once again, use $\mathcal{Z}_{\ell,n}$ to denote each of the variables 
\begin{align}
\mathcal{Z}_{\ell,n} &\equiv \mathcal{A}_{\ell,n}^*\mathcal{A}_{\ell,n}(\mathcal{P}_T\mathcal{A}_{\ell,n}^*\mathcal{A}_{\ell,n}\mathcal{P}_T)\bs h\bs m^* - \mathbb{E}\mathcal{A}_{\ell,n}^*\mathcal{A}_{\ell,n}(\mathcal{P}_T\mathcal{A}_{\ell,n}^*\mathcal{A}_{\ell,n}\mathcal{P}_T)\bs h\bs m^*.
\end{align}
Recall that the variance bound is defined as
\begin{align}
\sigma = \max \left\{\mathbb{E}\sum_{\ell=1}^L\sum_{n=1}^N\mathcal{Z}_{\ell,n}\mathcal{Z}_{\ell,n}^*, \mathbb{E}\sum_{\ell=1}^L\sum_{n=1}^N\mathcal{Z}^*_{\ell,n}\mathcal{Z}_{\ell,n}\right\}\label{varianceBoundreminder}
\end{align}
Using~\eqref{scalarWeightCovarianceConc}, the first term in~\eqref{varianceBoundreminder} can be expressed as 
\begin{align}
&\mathbb{E}\sum_{\ell=1}^L\sum_{n=1}^N\left(\mathcal{A}_{\ell,n}^*\mathcal{A}_{\ell,n}(\mathcal{P}_T\mathcal{A}_{\ell,n}^*\mathcal{A}_{\ell,n}\mathcal{P}_T)\bs h\bs m^*\right)\left(\mathcal{A}_{\ell,n}^*\mathcal{A}_{\ell,n}(\mathcal{P}_T\mathcal{A}_{\ell,n}^*\mathcal{A}_{\ell,n}\mathcal{P}_T)\bs h\bs m^*\right)^*\nonumber \\
& = \mathbb{E}\sum_{\ell=1}^L\sum_{n=1}^N \bs f_{\ell}\bs f_\ell^* \|\bs c_{\ell,n}\|^2 \left(|\hat{h}[\ell]|^6 \|\bs c_{\ell,n}\|^4 |\langle \bs c_{\ell,n},\bs m_n\rangle|^2 + |\hat{h}[\ell]|^2 |\langle \bs c_{\ell,n}, \bs m_n\rangle|^6  + |\hat{h}[\ell]|^6 |\langle \bs c_{\ell,n}, \bs m_n\rangle|^6\right)\\
& = \mathbb{E}\sum_{\ell=1}^L\sum_{n=1}^N \bs f_{\ell}\bs f_\ell^* \|\bs c_{\ell,n}\|^6 |\hat{h}[\ell]|^6 |\langle \bs c_{\ell,n},\bs m_n\rangle|^2 + \mathbb{E}\sum_{\ell=1}^L\sum_{n=1}^N \bs f_{\ell}\bs f_\ell^* \|\bs c_{\ell,n}\|^2 |\langle \bs c_{\ell,n},\bs m_n\rangle|^6 |\hat{h}[\ell]|^2\label{firstVarianceBoundCov0}\\
& + \mathbb{E}\sum_{\ell=1}^L\sum_{n=1}^N\bs f_{\ell}\bs f_\ell^* \|\bs c_{\ell,n}\|^2  |\hat{h}[\ell]|^6 |\langle \bs c_{\ell,n}, \bs m_n\rangle|^6\label{firstVarianceBoundCov1}
\end{align}
We now bound the operator norm of each of the terms in~\eqref{firstVarianceBoundCov0} and~\eqref{firstVarianceBoundCov1}. For the first term, first note that 
\begin{align}
\|\bs c_{\ell,n}\|^6 &= \left(\sum_{k=1}^K |c_{\ell,n}[k]|^2\right)^3\\
& = \sum_{k=1}^K |c_{\ell,n}[k]|^6 + 3\sum_{k\neq k'} |c_{\ell,n}[k]|^4 |c_{\ell,n}[k']|^2 + 6\sum_{k_1\neq k_2\neq k_3} |c_{\ell,n}[k_1]|^2|c_{\ell,n}[k_2]|^2|c_{\ell,n}[k_3]|^2 \label{expandedNormFour}
\end{align}
as well as 
\begin{align}
|\langle \bs c_{\ell,n},\bs m_n\rangle |^2 = \sum_{k=1}^K |c_{\ell,n}[k]|^2 m_n^2[k]+ 2\Re e \left\{ \sum_{k,k'}c_{\ell,n}[k]c^*_{\ell,n}[k'] m_n[k]m_n[k']\right\}.\label{expandedInnerProdCov}
\end{align}
When multiplying~\eqref{expandedNormFour} by~\eqref{expandedInnerProdCov} and taking the expectation, we thus get 
\begin{align}
\left\|\mathbb{E}\sum_{\ell=1}^L\sum_{n=1}^N \bs f_\ell\bs f_\ell^* \|\bs c_{\ell,n}\|^6 |\langle \bs c_{\ell,n},\bs m_n\rangle |^2 |\hat{h}[\ell]|^6\right\| \lesssim  \left\|\sum_{\ell=1}^L\sum_{n=1}^N \bs f_\ell\bs f_\ell^* \left(K^3\|\bs m_n\|^2\right) \frac{\mu_h^6}{L^3} \right\| \lesssim \frac{K^3\mu_h^6}{L^3}\label{boundFirstTerm4OrderVar1}
\end{align}
Noting that for $X\sim\mathcal{N}(0,\sigma)$, $\mathbb{E} \left\{X^n\right\} = \sigma (n-1)!! = c$, where $c$ is a constant. For the second term, we can expand the inner product $|\langle \bs c_{\ell,n},\bs m_n\rangle|^6$ in a similar way, 
\begin{align}
&\mathbb{E} \|\bs f_{\ell}\|^2\|\bs c_{\ell,n}\|^2 |\hat{h}[\ell]|^4 |\langle \bs c_{\ell,n},\bs m_n\rangle|^6 |\hat{h}[\ell]|^2\label{firstTermBound00}  \\
&\leq \mathbb{E} \|\bs c_{\ell,n}\|^2 \frac{\mu_h^6}{L^3} |\langle \bs c_{\ell,n},\bs m_n\rangle|^6\\
&\lesssim \frac{\mu_h^6}{L^3}\left(\sum_{k=1}^K |c_{\ell,n}[k]|^2\right) |\langle \bs c_{\ell,n},\bs m_n\rangle |^6\\
&\lesssim  \frac{\mu_h^6}{L^3}\left(\sum_{k=1}^K |c_{\ell,n}[k]|^2\right) \left(\sum_{i,j}c_{\ell,n}[i]c_{\ell,n}[j]m_n[i]m_n[j] + \sum_{i=1}^K|c_{\ell,n}[i]|^2m_n[i]^2\right)^3
\end{align}
Developing the right factor in the expression above gives
\begin{align}
&\mathbb{E} \|\bs f_{\ell}\|^2\|\bs c_{\ell,n}\|^2 |\hat{h}[\ell]|^4 |\langle \bs c_{\ell,n},\bs m_n\rangle|^6 |\hat{h}[\ell]|^2\label{firstTermOpN4} \\
&\lesssim  \mathbb{E}\frac{\mu_h^6}{L^3}\left(\sum_{k=1}^K |c_{\ell,n}[k]|^2\right) \left( \sum_{i=1}^K|c_{\ell,n}[i]|^2m_n[i]^2\right)^3\label{secondTermOpN4}\\
&+ \frac{\mu_h^6}{L^3}\mathbb{E} \left(\sum_{k=1}^K |c_{\ell,n}[k]|^2\right)\left(\sum_{i,j}^Kc_{\ell,n}[i]c_{\ell,n}[j]m_n[i]m_n[j]\right)^3\label{thirdTermOpN4}\\
& +\frac{\mu_h^6}{L^3}\mathbb{E} \left(\sum_{k=1}^K |c_{\ell,n}[k]|^2\right)\left(\sum_{i,j}^Kc_{\ell,n}[i]c_{\ell,n}[j]m_n[i]m_n[j]\right)^2 \left(\sum_{i=1}^K|c_{\ell,n}[i]|^2m_n[i]^2\right)\label{fourthTermOpN4}\\
& +\frac{\mu_h^6}{L^3}\mathbb{E} \left(\sum_{k=1}^K |c_{\ell,n}[k]|^2\right)\left(\sum_{i,j}^Kc_{\ell,n}[i]c_{\ell,n}[j]m_n[i]m_n[j]\right) \left(\sum_{i=1}^K|c_{\ell,n}[i]|^2m_n[i]^2\right)^2\label{fifthTermOpN5}
\end{align}
The first term~\eqref{secondTermOpN4} above can read as 
\begin{align}
&\mathbb{E}\frac{\mu_h^6}{L^3}\left(\sum_{k=1}^K |c_{\ell,n}[k]|^2\right) \left( \sum_{i=1}^K|c_{\ell,n}[i]|^2m_n[i]^2\right)^3\lesssim\\
&\mathbb{E}\frac{\mu_h^6}{L^3}\left(\sum_{k=1}^K |c_{\ell,n}[k]|^2\right) \left( \sum_{i=1}^K|c_{\ell,n}[i]|^6m_n[i]^6 + \sum_{i,j}^K|c_{\ell,n}[i]|^4c_{\ell,n}[j]|^2 m_n[i]^4m_n[j]^2\right)\\
& + \mathbb{E}\frac{\mu_h^6}{L^3}\left(\sum_{k=1}^K |c_{\ell,n}[k]|^2\right) \left(6\sum_{i_1,i_2,i_3} |c_{\ell,n}[i_1]|^2 m_n[i_1]^2|c_{\ell,n}[i_2]|^2 m_n[i_2]^2|c_{\ell,n}[i_3]|^2 m_n[i_3]^2\right)\\
&\lesssim K \frac{\mu_h^6}{L^3} \left(\sum_{i=1}^K m_n[i]^6 +  \sum_{i,j}^K m_n[i]^4m_n[j]^2 +\sum_{i_1,i_2,i_3}  m_n[i_1]^2 m_n[i_2]^2 m_n[i_3]^2 \right)\\
&\lesssim K \frac{\mu_h^6}{L^3} \left(\frac{\mu_m^6}{N^3} + \frac{\mu_m^4}{N^2} \frac{\mu_m^2}{N} + \left(\frac{\mu_m^2}{N}\right)^3 \right)\lesssim \frac{\mu_m^6}{N^3}K \frac{\mu_h^6}{L^3}
\end{align}
For~\eqref{thirdTermOpN4}, we get,
\begin{align}
&\frac{\mu_h^6}{L^3}\mathbb{E} \left(\sum_{k=1}^K |c_{\ell,n}[k]|^2\right)\left(\sum_{i,j}^Kc_{\ell,n}[i]c_{\ell,n}[j]m_n[i]m_n[j]\right)^3\lesssim \\
&\mathbb{E}\frac{\mu_h^6}{L^3}\left(\sum_{k=1}^K |c_{\ell,n}[k]|^2\right) \left(\sum_{i,j}^Kc_{\ell,n}^3[i]c_{\ell,n}^3[j]m_n^3[i]m_n^3[j]\right)\label{line1tmpa}\\
& + \mathbb{E}\frac{\mu_h^6}{L^3}\left(\sum_{k=1}^K |c_{\ell,n}[k]|^2\right) \left(6\sum_{i_1,j_1,i_2,j_2,i_3,j_3} \prod_{k=1}^3 c_{\ell,n}[i_k]c_{\ell,n}[j_k]m_n[i_k]m_n[j_k] \right)\label{line1tmpb} \\
& + \mathbb{E}\frac{\mu_h^6}{L^3}\left(\sum_{k=1}^K |c_{\ell,n}[k]|^2\right) \left(\sum_{i_1,j_1,i_2,j_2}^K |c_{\ell,n}[i_1]|^2|c_{\ell,n}[j_1]|^2m_n[i_1]^2m_n[j_2]^2 c_{\ell,n}[i_2]c_{\ell,n}[j_2]m_n[i_2]m_n[j_2]\right)\label{line1tmpc}
\end{align}
The first and last terms always vanish and~\eqref{line1tmpb} contributes to the expectation only for chains of the form $i_1,i_1,i_2,i_2,i_3,i_3$ for some $i_1,i_2,i_3$. We can thus write
\begin{align}
&\frac{\mu_h^6}{L^3}\mathbb{E} \left(\sum_{k=1}^K |c_{\ell,n}[k]|^2\right)\left(\sum_{i,j}^Kc_{\ell,n}[i]c_{\ell,n}[j]m_n[i]m_n[j]\right)^3\\
&\lesssim  \frac{\mu_h^6}{L^3} \left(\sum_{k=1}^K |c_{\ell,n}[k]|^2\right)\left(\sum_{i_1,i_2,i_3} m_n[i_1]^2m_n[i_2]^2m_n[i_3]^2\right)\lesssim  K\frac{\mu_h^6}{L^3} \frac{\mu_m^6}{N^3}
\end{align}
For the last two terms~\eqref{fourthTermOpN4} and~\eqref{fifthTermOpN5}, we can respectively write,
\begin{align}
&\frac{\mu_h^6}{L^3}\mathbb{E} \left(\sum_{k=1}^K |c_{\ell,n}[k]|^2\right)\left(\sum_{i,j}^Kc_{\ell,n}[i]c_{\ell,n}[j]m_n[i]m_n[j]\right)^2 \left(\sum_{i=1}^K|c_{\ell,n}[i]|^2m_n[i]^2\right)\\
&\lesssim \frac{\mu_h^6}{L^3}\mathbb{E} \left(\sum_{k=1}^K |c_{\ell,n}[k]|^2\right)\left(\sum_{i,j}^K |c_{\ell,n}[i]|^2|c_{\ell,n}[j]|^2m^2_n[i]m^2_n[j]\right) \left(\sum_{i=1}^K|c_{\ell,n}[i]|^2m_n[i]^2\right) \\
&+ 2\frac{\mu_h^6}{L^3}\mathbb{E} \left(\sum_{k=1}^K |c_{\ell,n}[k]|^2\right)\left( \sum_{i_1,j_1,i_2,j_2}^K \prod_{k=1}^2 c_{\ell,n}[i_k]c_{\ell,n}[j_k] m_n[i_k]m_n[j_k]\right) \left(\sum_{i=1}^K|c_{\ell,n}[i]|^2m_n[i]^2\right)\label{bound1FirstofLastTwoTerms}\\
&\lesssim K\frac{\mu_h^6}{L^3}\frac{\mu_m^2}{N}\left(\|\bs m_n\|^4\right)\lesssim K\frac{\mu_h^6}{L^3}\frac{\mu_m^6}{N^3}
\end{align}
as well as 
\begin{align}
&\frac{\mu_h^6}{L^3}\mathbb{E} \left(\sum_{k=1}^K |c_{\ell,n}[k]|^2\right)\left(\sum_{i,j}^Kc_{\ell,n}[i]c_{\ell,n}[j]m_n[i]m_n[j]\right) \left(\sum_{i=1}^K|c_{\ell,n}[i]|^2m_n[i]^2\right)^2\\
&\lesssim \frac{\mu_h^6}{L^3}\mathbb{E} \left(\sum_{k=1}^K |c_{\ell,n}[k]|^2\right)\left(\sum_{i,j}^Kc_{\ell,n}[i]c_{\ell,n}[j]m_n[i]m_n[j]\right) \left(\sum_{i=1}^K|c_{\ell,n}[i]|^4m_n[i]^4\right)\\
&+ \frac{\mu_h^6}{L^3}\mathbb{E} \left(\sum_{k=1}^K |c_{\ell,n}[k]|^2\right)\left(\sum_{i,j}^Kc_{\ell,n}[i]c_{\ell,n}[j]m_n[i]m_n[j]\right) \left(\sum_{i,j}^K|c_{\ell,n}[i]|^2m_n[i]^2|c_{\ell,n}[j]|^2m_n[j]^2\right)\label{lastBoundClnSix}
\end{align}
Both of these lines vanish when taking the expectation because of the second factor. The total bound on~\eqref{firstTermBound00} is thus given by 
\begin{align}
\mathbb{E} \|\bs c_{\ell,n}\|^2 |\hat{h}[\ell]|^6 |\langle \bs c_{\ell,n},\bs m_n\rangle|^6 \lesssim K\frac{\mu_h^6}{L^3}\frac{\mu_m^6}{N^3}\label{lasBoundOpNorm001}
\end{align}
For the second term in~\eqref{boundOpnormCase1}, one can derive a similar bound by dividing~\eqref{lasBoundOpNorm001} by $\mu_h^4/L^2$ as,
\begin{align}
&\mathbb{E}\|\bs c_{\ell,n}\|^2 |\hat{h}[\ell]|^2 |\langle \bs c_{\ell,n},\bs m_n\rangle |^6 \\
&\lesssim \frac{\mu_h^2}{L} K\frac{\mu_m^6}{N^3}\label{bboundFOrder42}
\end{align}
All three bounds~\eqref{boundFirstTerm4OrderVar1},~\eqref{lasBoundOpNorm001} and~\eqref{bboundFOrder42} can be made sufficiently small as soon as $L\gtrsim K\mu_h^2$ and $N\gtrsim \mu_m^2$. We now bound the second term in~\eqref{varianceBoundreminder}. For this term we have to bound the sum,
 \begin{align}
&\sum_{\ell=1}^L\sum_{n=1}^N\mathcal{Z}^*_{\ell,n}\mathcal{Z}_{\ell,n}\nonumber \\
&= \sum_{\ell=1}^L\sum_{n=1}^N \|\bs f_\ell\|^2 \bs c_{\ell,n}\bs c_{\ell,n}^*\left(|\hat{h}[\ell]|^6 \|\bs c_{\ell,n}\|^4 |\langle \bs c_{\ell,n},\bs m_n\rangle|^2 + |\hat{h}[\ell]|^2 |\langle \bs c_{\ell,n}, \bs m_n\rangle|^6  + |\hat{h}[\ell]|^6 |\langle \bs c_{\ell,n}, \bs m_n\rangle|^6\right)\label{sumVariance2a}
\end{align}
The sum~\eqref{sumVariance2a} is a block diagonal matrix. When taking the operator norm, one can thus only take the supremum over the $n$ indices as  
\begin{align}
&\left\|\mathbb{E}\sum_{\ell=1}^L\sum_{n=1}^N\mathcal{Z}^*_{\ell,n}\mathcal{Z}_{\ell,n}\right\|\nonumber \\
&\lesssim \sup_n \left\|\mathbb{E}\sum_{\ell=1}^L \|\bs f_\ell\|^2 \bs c_{\ell,n}\bs c_{\ell,n}^*\left(|\hat{h}[\ell]|^6 \|\bs c_{\ell,n}\|^4 |\langle \bs c_{\ell,n},\bs m_n\rangle|^2 + |\hat{h}[\ell]|^2 |\langle \bs c_{\ell,n}, \bs m_n\rangle|^6  + |\hat{h}[\ell]|^6 |\langle \bs c_{\ell,n}, \bs m_n\rangle|^6\right)\right\|\label{sumVariance2b}
\end{align}
Now, as before, we split the sum into the diagonal contribution and the off-diagonal terms.  For the diagonal, each of the three terms in~\eqref{sumVariance2b} read as 
\begin{align}
&\sup_n \mathbb{E}\sum_{\ell=1}^L \|\bs f_\ell\|^2 \diag(|\bs c_{\ell,n}[k]|^2)\left(|\hat{h}[\ell]|^6 \|\bs c_{\ell,n}\|^4 |\langle \bs c_{\ell,n},\bs m_n\rangle|^2\right) \label{diagonalContribSecondVariance1}\\
&\sup_n \mathbb{E}\sum_{\ell=1}^L \|\bs f_\ell\|^2 \diag(|\bs c_{\ell,n}[k]|^2) \left(|\hat{h}[\ell]|^2 |\langle \bs c_{\ell,n}, \bs m_n\rangle|^6\right)\label{diagonalContribSecondVariance2} \\
& \sup_n \mathbb{E}\sum_{\ell=1}^L \|\bs f_\ell\|^2 \diag(|\bs c_{\ell,n}[k]|^2) \left(|\hat{h}[\ell]|^6 |\langle \bs c_{\ell,n}, \bs m_n\rangle|^6\right)\label{diagonalContribSecondVariance3}
\end{align}
In~\eqref{diagonalContribSecondVariance1},~\eqref{diagonalContribSecondVariance2} and~\eqref{diagonalContribSecondVariance3}, we use $\diag(|\bs c_{\ell,n}[k]|^2)$ to denote the diagonal matrix built from the diagonal of $\bs c_{\ell,n}\bs c_{\ell,n}^*$. Following the same reasoning as above, we get 
\begin{align}\begin{split}
\left\|\mathbb{E}\sup_n \sum_{\ell=1}^L \|\bs f_\ell\|^2 \diag(|\bs c_{\ell,n}[k]|^2)\left(|\hat{h}[\ell]|^6 \|\bs c_{\ell,n}\|^4 |\langle \bs c_{\ell,n},\bs m_n\rangle|^2\right) \right\|
&\lesssim \frac{K^2\mu_h^4}{L^2}\|\bs m_n\|^2\\
&\lesssim \frac{K^2 \mu_h^4}{L^2}\frac{\mu^2_m}{N}\end{split}\label{diagonal1}\\
\begin{split}
\left\|\mathbb{E}\sup_n \sum_{\ell=1}^L \|\bs f_\ell\|^2 \diag(|\bs c_{\ell,n}[k]|^2) \left(|\hat{h}[\ell]|^2 |\langle \bs c_{\ell,n}, \bs m_n\rangle|^6\right)\right\|
&\lesssim \frac{\mu_m^2}{N}\end{split}\label{diagonal2}\\
\begin{split}
\left\|\mathbb{E}\sup_n \sum_{\ell=1}^L \|\bs f_\ell\|^2 \diag(|\bs c_{\ell,n}[k]|^2) \left(|\hat{h}[\ell]|^6 |\langle \bs c_{\ell,n}, \bs m_n\rangle|^6\right)\right\|
&\lesssim \frac{\mu^4_h}{L^2}\frac{\mu_m^2}{N}\end{split}\label{diagonal3}
\end{align}
In~\eqref{diagonal1} we use $\sum_{\ell} |\hat{h}[\ell]|^2\|\bs f_{\ell}\|^2 = \|\bs h\|^2=1$. For~\eqref{diagonal2} and~\eqref{diagonal3} we use the same reasoning as the one used to derive~\eqref{lasBoundOpNorm001} except that we now have
\begin{align}
&\left\|\mathbb{E}\sup_n \sum_{\ell=1}^L \|\bs f_\ell\|^2 \diag(|\bs c_{\ell,n}[k]|^2) \left(|\hat{h}[\ell]|^2 |\langle \bs c_{\ell,n}, \bs m_n\rangle|^6\right)\right\|\\
&\leq \left|\sum_{\ell=1}^L \|\bs f_\ell\|^2|\hat{h}[\ell]|^2\right| \sup_{\ell,k}\left| \mathbb{E}\sup_n \diag(|\bs c_{\ell,n}[k]|^2) \left( |\langle \bs c_{\ell,n}, \bs m_n\rangle|^6\right)\right|
\end{align}
Since the inner product $|\langle \bs c_{\ell,n}, \bs m_n\rangle|^6$ is now multiplied with a single element $|\bs c_{\ell,n}[k]|^2$ from the diagonal $\diag(c_{\ell,n}[k]^2)$, the bounds in~\eqref{lasBoundOpNorm001},~\eqref{bboundFOrder42} also get divided by $K$. For the off-diagonal contributions in~\eqref{sumVariance2b}, only the terms exhibiting an even power in the $c_{\ell,n}[k]$ monomials corresponding to the off-diagonal entry will remain when taking the expectation. Focusing on those off-diagonal entries, the three contributions of~\eqref{diagonalContribSecondVariance1},~\eqref{diagonalContribSecondVariance2}~\eqref{diagonalContribSecondVariance1} can now read as, 
\begin{align}
&\mathbb{E}\sum_{n=1}^N \sum_{\ell=1}^L \|\bs f_\ell\|^2 (\bs c_{\ell,n}\bs c_{\ell,n}^* - \diag(|\bs c_{\ell,n}[k]|^2))\left(|\hat{h}[\ell]|^6 \|\bs c_{\ell,n}\|^4 |\langle \bs c_{\ell,n},\bs m_n\rangle|^2\right) \label{diagonalContribSecondVarianceOffDiag1c}\\
&\mathbb{E}\sum_{n=1}^N \sum_{\ell=1}^L \|\bs f_\ell\|^2 (\bs c_{\ell,n}\bs c_{\ell,n}^* - \diag(|\bs c_{\ell,n}[k]|^2)) \left(|\hat{h}[\ell]|^2 |\langle \bs c_{\ell,n}, \bs m_n\rangle|^6\right)\label{diagonalContribSecondVarianceOffDiag2c} \\
&\mathbb{E} \sum_{n=1}^N \sum_{\ell=1}^L \|\bs f_\ell\|^2 (\bs c_{\ell,n}\bs c_{\ell,n}^* - \diag(|\bs c_{\ell,n}[k]|^2)) \left(|\hat{h}[\ell]|^6 |\langle \bs c_{\ell,n}, \bs m_n\rangle|^6\right)\label{diagonalContribSecondVarianceOffDiag3c}
\end{align}
In~\eqref{diagonalContribSecondVarianceOffDiag1c}, all the even powers generated from the weight $\left(|\hat{h}[\ell]|^6 \|\bs c_{\ell,n}\|^4 |\langle \bs c_{\ell,n},\bs m_n\rangle|^2\right) $ will vanish. The only remaining term can thus be bounded as  
\begin{align}
&\left\|\sup_n \sum_{\ell=1}^L \|\bs f_\ell\|^2 \bs m_n\bs m_n^*|\hat{h}[\ell]|^6 \right\|\lesssim \frac{\mu_h^4 K^2}{L^2}\frac{\mu_m^2}{N} \label{diagonalContribSecondVarianceOffDiag1d}
\end{align}
For the last two terms, we again use~\eqref{firstTermBound00} to~\eqref{lastBoundClnSix}. Retaining only~\eqref{secondTermOpN4} to~\eqref{fifthTermOpN5}, we can thus write 
\begin{align}
&\mathbb{E} \sum_{n=1}^N \sum_{\ell=1}^L \|\bs f_\ell\|^2 (\bs c_{\ell,n}\bs c_{\ell,n}^* - \diag(|\bs c_{\ell,n}[k]|^2))  |\langle \bs c_{\ell,n}, \bs m_n\rangle|^6\label{offDiagonalVar1}\\
&=   \sum_{n=1}^N\sum_{\ell=1}^L \|\bs f_\ell\|^2 \sum_{k,k'} m_n^3[k]m_n^3[k'](\bs e_{n}\otimes \bs e_k) (\bs e_{n}\otimes \bs e_{k'})^*\label{offDiagonalVar2}\\
& + \sum_{n=1}^N\sum_{\ell=1}^L \|\bs f_\ell\|^2 \left(\bs m_n\bs m_n^*\right)\sum_{k_1, k_2} m_n[k_1]^2m_n[k_2]^2\label{offDiagonalVar3}\\
& +\sum_{n=1}^N \sum_{\ell=1}^L  \|\bs f_\ell\|^2 \left(\|\bs m_n\|^4 \bs m_n\bs m_n^*\right)\label{offDiagonalExp4}\\
&+ \sum_{n=1}^N \sum_{\ell=1}^L \|\bs f_\ell\|^2\left(\sum_{k=1}^K |m_n[k]|^2\right)\left(\bs m_n\bs m_n^*\right)\|\bs m_n\|^2\label{offDiagonalVar6}\\
&+ \sum_{n=1}^N \sum_{\ell=1}^L \|\bs f_\ell\|^2\left(\sum_{k=1}^K |m_n[k]|^2\right)^2\bs m_n\bs m_n^*\label{offDiagonalVar5} 
\end{align}
For~\eqref{offDiagonalVar3}, whenever all three factors are different, the expectation reduces to zero. When at least two of them are the same, such as in $(c_ic_j^*)(c_jc_i^*)(c_kc_\ell)$ or in chains of the form $(c_ic_j^*)(c_jc_k^*)(c_kc_\ell^*)$. Those chains will contribute to the corresponding entries of $\bs c_{\ell,n}\bs c_{\ell,n}^*$. In other words, we have the freedom to choose the first four indices arbitrarily and the last two are fixed. As a consequence, the only part of~\eqref{line1tmpb} that will contribute to the variance bound when multiplying by $(\bs c_{\ell,n}\bs c_{\ell,n}^* - \diag(|\bs c_{\ell,n}[k]|^2))$ and taking the expectation is of the form
\begin{align}
\sum_{i_1,i_2,i_3,i_4} c_{\ell,n}[i_1]|c_{\ell,n}[i_2]|^2|c_{\ell,n}[i_2]|^2|c_{\ell,n}[i_3]|^2c_{\ell,n}[i_4] m_{n}[i_1]|m_{n}[i_2]|^2|m_n[i_2]|^2|m_{n}[i_3]|^2m_{n}[i_4]. 
\end{align}
When multiplying this sum by $(\bs c_{\ell,n}\bs c_{\ell,n}^* - \diag(|\bs c_{\ell,n}[k]|^2))$, each of the $(i_1,i_4)$ terms contribute to one off-diagonal entry in the matrix, and we thus get 
\begin{align}
\mathbb{E}(\bs c_{\ell,n}\bs c_{\ell,n}^* - \diag(|\bs c_{\ell,n}[k]|^2))\sum_{\substack{k_1\neq k_1'\\k_2\neq k'_2\\k_3\neq k_3'} } \prod_{i=1}^3\left(c_{\ell,n}[k_i]c_{\ell,n}[k_i'] m_n[k_i] m_n[k_i']\right)\\
\asymp \sum_{k_2,k_3} \bs m_n\bs m_n^* |\bs m_n[k_2]|^2|\bs  m_n[k_3]|^2.
\end{align}
From this, we get the following bound on the operator norm of~\eqref{offDiagonalVar3},
\begin{align}
\left\|\sum_{\ell=1}^L\sum_{n=1}^N\|\bs f_\ell\|^2 |\hat{h}[\ell]|^2\mathbb{E}(\bs c_{\ell,n}\bs c_{\ell,n}^* - \diag(|\bs c_{\ell,n}[k]|^2))\sum_{\substack{k_1\neq k_1', k_2\neq k'_2\\k_3\neq k'_3} } \prod_{i=1}^3\left(c_{\ell,n}[k_i]c_{\ell,n}[k_i'] m_n[k_i] m_n[k_i']\right)\right\|\lesssim  \frac{\mu_m^6}{N^3}\label{boundTemp01}
\end{align}
Similarly, For~\eqref{offDiagonalVar6}, simply note that the only possibility for~\eqref{bound1FirstofLastTwoTerms} to contribute to the off-diagonal entries, when taking the expectation, is for the second factor in this expression to exhibit a chain of the form $c_{\ell,n}[i]c_{\ell,n}[j] c_{\ell,n}^*[j] c_{\ell,n}^*[k] = c_{\ell,n}[i]|c_{\ell,n}[j]|^2c_{\ell,n}[k]^*$. In this case, the whole second factor contributes to each off-diagonal entry and we can write,
\begin{align}
&\mathbb{E} (\bs c_{\ell,n}\bs c_{\ell,n}^* - \diag(|\bs c_{\ell,n}[k]|^2))\|\bs m_n\|^2\left(\sum_{k_1=1}^K |c_{\ell,n}[k_1]|^2 m_n[k_1]^2 \sum_{k_2,k_3}\prod_{j=2}^3c_{\ell,n}[k_j] m_n[k_j]\right)\\
& = \|\bs m_n\|^2 \left(\sum_{k} |m_n[k]|^2\right) \bs m_n\bs m_n^*
\end{align}
which gives 
\begin{align}
\left\|\sum_{\ell,n} |\hat{h}[\ell]|^2\|\bs f_\ell\|^2 \|\bs m_n\|^4 \bs m_n\bs m_n^*\right\| \lesssim \frac{\mu_m^6}{N^3}\label{boundTemp02}
\end{align}
%
%We now bound the norms arising from the contributions from~\eqref{offDiagonalVar1},~\eqref{} and~\eqref{cm65}. Each of the terms in~\eqref{cm64} only contribute to the corresponding off-diagonal entry and we thus have 
%
%\begin{align}
%\left\|\sum_{\ell,n} \|\bs f_\ell\|^2 \bs m_n\bs m_n^* \right\|\lesssim \left\|\sum_n\sum_\ell |\hat{h}[\ell]|^2\sum_{k,k'} m_n[k]^3m_n[k']^3(\bs e_{\ell,n}\otimes \bs e_k)(\bs e_{\ell,n}\otimes \bs e_{k'})^*\right\|\\
%&\lesssim L\frac{\mu_h^2}{L}\frac{\mu_m^2}{N}
%\end{align}
%
%We now bound~\eqref{cm65}. For this last term, all the $(k_1,k_2)$ terms contribute to the off-diagonal entry $(k_1',k_2')$. We thus have,
%
%\begin{align}
%\sum  \asymp \|\bs m_n\|^4 \bs m_n\bs m_n^*
%\end{align}
%
%from which we get the final bound 
%
%\begin{align}
%&\left\|\mathbb{E} \sum_{\ell,n} \|\bs f_\ell\|^2 \bs c_{\ell,n}\bs c_{\ell,n}^* \sum_{k_1\neq k_2} \sum_{k_1'\neq k_2'} \left(c_{\ell,n}[k_1]^2c_{\ell,n}[k_2]^2 m_n[k_1]^2m_n[k_2]^2 c_{\ell,n}[k'_1]c_{\ell,n}[k'_2] m_n[k'_1]m_n[k'_2]\right)\right\|\\
%&\lesssim L\frac{\mu_h^2}{L} \frac{\mu_m^6}{N^3} 
%\end{align} 
%
Multiplying by $|\hat{h}[\ell]|^2$ and taking the operator norm for the terms in~\eqref{offDiagonalVar1} to~\eqref{offDiagonalVar5}, gives 
\begin{align}
\left\|\mathbb{E}\sum_{n=1}^N \sum_{\ell=1}^L \|\bs f_\ell\|^2 (\bs c_{\ell,n}\bs c_{\ell,n}^* - \diag(|\bs c_{\ell,n}[k]|^2)) \left(|\hat{h}[\ell]|^2 |\langle \bs c_{\ell,n}, \bs m_n\rangle|^6\right)\right\|\lesssim  \frac{\mu_m^6}{N^3}.\label{diagonalContribSecondVarianceOffDiag1dd}
\end{align}
This first contribution can thus be made sufficiently small as soon as $N\gtrsim \mu_m^2$. The same bound applies to~\eqref{diagonalContribSecondVarianceOffDiag3c} up to multiplication by $|\hat{h}[\ell]|^4$ from which we immediately get,
\begin{align}
\left\|\mathbb{E} \sum_{n=1}^N \sum_{\ell=1}^L \|\bs f_\ell\|^2 (\bs c_{\ell,n}\bs c_{\ell,n}^* - \diag(|\bs c_{\ell,n}[k]|^2)) \left(|\hat{h}[\ell]|^6 |\langle \bs c_{\ell,n}, \bs m_n\rangle|^6\right)\right\|\lesssim \frac{\mu_m^6}{N^3}\frac{\mu_h^4}{L^2}.\label{diagonalContribSecondVarianceOffDiag1ddd}
\end{align}
Grouping~\eqref{diagonalContribSecondVarianceOffDiag1ddd} together with~\eqref{diagonalContribSecondVarianceOffDiag1dd} and~\eqref{diagonalContribSecondVarianceOffDiag1d} and adding the diagonal contributions~\eqref{diagonal1} to~\eqref{diagonal3}, we get 
\begin{align}
\left\|\sum_{\ell=1}^L\sum_{n=1}^N\mathcal{Z}^*_{\ell,n}\mathcal{Z}_{\ell,n}\right\|&\lesssim \left( \frac{\mu_m^6}{N^3}\frac{\mu_h^4}{L^2} + \frac{\mu_m^6}{N^3} + \frac{\mu_h^4 K^2}{L^2}\frac{\mu_m^2}{N}\right) +\left(\frac{K^2 \mu_h^4}{L^2}\frac{\mu^2_m}{N}+ \frac{\mu_m^2}{N} + \frac{\mu^4_h}{L^2}\frac{\mu_m^2}{N}\right)\\
&\lesssim \frac{\mu_m^2}{N} + \frac{\mu_h^4K^2}{L^2}\label{boundZstartZ}
\end{align}
Combining the bound~\eqref{boundZstartZ} with~\eqref{boundFirstTerm4OrderVar1},~\eqref{lasBoundOpNorm001} and~\eqref{bboundFOrder42}, we get the final variance bound for the fourth order terms,
\begin{align}
\sigma& \lesssim \max\left\{\frac{\mu_m^2}{N} + \frac{\mu_h^4K^2}{L^2}, \left(K^3\frac{\mu_h^6}{L^3}+\frac{\mu_m^4}{N^2}  \frac{K^3\mu_h^6}{L^3} + \frac{\mu_h^2}{L} K\frac{\mu_m^4}{N^2}\right)\right\}
%
%\max\left\{\left(\frac{\mu_m^4}{N^2} +  \frac{\mu_m^6}{N^3}\right)\frac{\mu_h^2}{L} + \frac{\mu_m^4}{N^2} + \frac{\mu_m^6}{N^3} + \frac{\mu_h^4 K^2}{L^2}, \frac{K\mu_h^2\mu_m^2}{LN}\right\}
\label{finalVarianceFourthOrder}
\end{align}
This bound can be made less than $\delta$ by taking $L\gtrsim (1/\delta)K\mu_h^2$ and $N\gtrsim (1/\delta)\mu_m^2$. In order to use proposition~\eqref{rosenthalTropp} we still need to bound the expected maximal operator norm among each of the terms in~\eqref{fourthOrderequation1}. This is the point of section~\ref{secOpNormFourthOrder} below.

\subsubsection{\label{secOpNormFourthOrder}Expected maximal operator norm}

%
%\subsubsection{\label{OpNormRosenthal4dep}Fourth order dependence}
%
We now derive a bound on the expected largest operator norm~\eqref{expectedOpNormRosP} for each of the terms in~\eqref{fourthOrderequation1}. The operator norm of any of the term in~\eqref{fourthOrderequation1} reads as,
\begin{align}\begin{split}
\|\mathcal{Z}_{\ell,n}\|  = &\left\|\bs f_\ell \bs c_{\ell,n}^* \langle \bs f_\ell \bs c_{\ell,n}^*,\mathcal{P}_T(\bs f_\ell\bs c_{\ell,n}^*)\langle \mathcal{P}_T(\bs f_\ell\bs c_{\ell,n}^*, \bs h\bs m^*\rangle \rangle\right. \\
&-\left. \mathbb{E}\bs f_\ell \bs c_{\ell,n}^* \langle \bs f_\ell \bs c_{\ell,n}^*,\mathcal{P}_T(\bs f_\ell\bs c_{\ell,n}^*)\langle \mathcal{P}_T(\bs f_\ell\bs c_{\ell,n}^*, \bs h\bs m^*\rangle \rangle\right\|\end{split}\label{boundOpNormRosenthal}
\end{align}
The first term can be bounded as follows,
\begin{align}
&\mathbb{E}\left\|\bs f_\ell \bs c_{\ell,n}^* \langle \bs f_\ell \bs c_{\ell,n}^*,\mathcal{P}_T(\bs f_\ell\bs c_{\ell,n}^*)\langle \mathcal{P}_T(\bs f_\ell\bs c_{\ell,n}^*), \bs h\bs m^*\rangle \rangle \right\|^2\nonumber \\
&\leq \mathbb{E} \|\bs f_{\ell}\|^2\|\bs c_{\ell,n}\|^2 |\hat{h}[\ell]|^2 |\langle \bs c_{\ell,n},\bs m_n\rangle|^2 \left(|\hat{h}[\ell]|^2 \|\bs c_{\ell,n}\|^2 + \|\bs f_\ell\|^2 |\langle \bs m_n,\bs c_{\ell,n}\rangle |^2 + |\hat{h}[\ell]|^2|\langle \bs m_n, \bs c_{\ell,n}\rangle|^2\right)^2\\
&\lesssim \mathbb{E} \|\bs f_{\ell}\|^2\|\bs c_{\ell,n}\|^2 |\hat{h}[\ell]|^2 |\langle \bs c_{\ell,n},\bs m_n\rangle|^2\left(|\hat{h}[\ell]|^4 \|\bs c_{\ell,n}\|^4 + \|\bs f_\ell\|^4 |\langle \bs m_n,\bs c_{\ell,n}\rangle |^4 + |\hat{h}[\ell]|^4|\langle \bs m_n, \bs c_{\ell,n}\rangle|^4\right) \\
&\lesssim \mathbb{E} \|\bs f_{\ell}\|^2
\left(|\hat{h}[\ell]|^6 \|\bs c_{\ell,n}\|^6|\langle \bs m_n,\bs c_{\ell,n}\rangle |^2 + \|\bs f_\ell\|^4 |\hat{h}[\ell]|^2\|\bs c_{\ell,n}\|^2 |\langle \bs m_n,\bs c_{\ell,n}\rangle |^6\right)\\
& + \mathbb{E} \|\bs f_{\ell}\|^2\left(\|\bs c_{\ell,n}\|^2|\hat{h}[\ell]|^6|\langle \bs m_n, \bs c_{\ell,n}\rangle|^6\right) 
\label{boundOpnormCase1}
\end{align}
Note that the three terms above are in fact very similar to the ones (~\eqref{boundFirstTerm4OrderVar1}~\eqref{lasBoundOpNorm001} and~\eqref{bboundFOrder42}) that appeared in the derivation of the variance. Using those results, we can bound~\eqref{boundOpnormCase1} as
\begin{align}
\mathbb{E}\left\|\bs f_\ell \bs c_{\ell,n}^* \langle \bs f_\ell \bs c_{\ell,n}^*,\mathcal{P}_T(\bs f_\ell\bs c_{\ell,n}^*)\langle \mathcal{P}_T(\bs f_\ell\bs c_{\ell,n}^*), \bs h\bs m^*\rangle \rangle \right\|^2 \lesssim \frac{\mu_m^2}{N}\left(K^3\frac{\mu_h^6}{L^3}+\frac{\mu_m^4}{N^2}  \frac{K^3\mu_h^6}{L^3} + \frac{\mu_h^2}{L} K\frac{\mu_m^4}{N^2}\right)\label{boundOperatorNormFourthOrder}
\end{align}
where we again used the fact that
\begin{align}
|\langle \bs f_\ell\bs c_{\ell,n}^*,\mathcal{P}_T(\bs f_\ell\bs c_{\ell,n}^*)\rangle| \leq |\hat{h}[\ell]|^2 \|\bs c_{\ell,n}\|^2 + \|\bs f_\ell\|^2 |\langle \bs m_n,\bs c_{\ell,n}\rangle |^2 + |\hat{h}[\ell]|^2|\langle \bs m_n, \bs c_{\ell,n}\rangle|^2
\end{align}
as well as 
\begin{align}
|\langle \mathcal{P}_T(\bs f_\ell\bs c_{\ell,n}^*),\bs h\bs m^*\rangle |^2 = |\langle \bs f_\ell\bs c_{\ell,n}^*,\bs h\bs m^*\rangle |^2 \leq |\hat{h}[\ell]|^2 |\langle \bs c_{\ell,n},\bs m_n\rangle|^2
\end{align}
For the expectation, simply note that for a random variable $\bs X$, $\mbox{Var}(\bs X) = \langle \bs X-\mathbb{E}\bs X, \bs X-\mathbb{E}\bs X\rangle = \mathbb{E}\|\bs X\|^2 - \|\mathbb{E}(\bs X)\|^2\geq 0$. Again,~\eqref{boundOperatorNormFourthOrder} can be made less than $\delta$ for any constant $\delta$ as soon as $L\gtrsim (1/\delta) K\mu_h^2$ and $N\gtrsim (1/\delta)\mu_m^2$.

\subsubsection{\label{rosenthalPcclorder4}Conclusion through Rosenthal-Pinelis}

We now use the results of section~\ref{secVarianceFourthOrder} and~\ref{secOpNormFourthOrder} to bound the norm~\eqref{fourthOrderequation1}. Using the bound on the operator norm~\eqref{boundOperatorNormFourthOrder} together with the bound on the variance~\eqref{finalVarianceFourthOrder}, as well as proposition~\eqref{rosenthalTropp}, as soon as $L\gtrsim K\mu_h^2$ and $N\gtrsim \mu_m^2$, we have that 
\begin{align}
\mathbb{E}\left\|\sum_{\ell,=1}^L\sum_{n=1}^N \mathcal{A}_{\ell,n}^*\mathcal{A}_{\ell,n}(\mathcal{P}_T\mathcal{A}_{\ell,n}^*\mathcal{A}_{\ell,n}\mathcal{P}_T)\bs h\bs m^*  - \mathbb{E} \mathcal{A}_{\ell,n}^*\mathcal{A}_{\ell,n}(\mathcal{P}_T\mathcal{A}_{\ell,n}^*\mathcal{A}_{\ell,n}\mathcal{P}_T)\bs h\bs m^*\right\|^2 \lesssim  \left(\frac{K\mu_h^2}{L} + \frac{\mu_m^2}{N}\right) 
\end{align}
As explained above, we can then use Markov's inequality, 
\begin{align}
\left\|\sum_{\ell,=1}^L\sum_{n=1}^N \mathcal{A}_{\ell,n}^*\mathcal{A}_{\ell,n}(\mathcal{P}_T\mathcal{A}_{\ell,n}^*\mathcal{A}_{\ell,n}\mathcal{P}_T)\bs h\bs m^*  - \mathbb{E} \mathcal{A}_{\ell,n}^*\mathcal{A}_{\ell,n}(\mathcal{P}_T\mathcal{A}_{\ell,n}^*\mathcal{A}_{\ell,n}\mathcal{P}_T)\bs h\bs m^*\right\|\lesssim t
\end{align}
with probability at least $1 - \frac{1}{t}\left(\frac{K\mu_h^2}{L} + \frac{\mu_m^2}{N}\right)^{1/2} $. Taking $t = \delta$ gives the desired result.
\end{proof}

\subsection{\label{sectionUstatistics}Proof of lemma~\ref{lemmaTperpSecondTerm}}

The proof of of lemma~\ref{lemmaTperpSecondTerm} relies on an argument developed in the proof of Theorem 3.4.1 in~\cite{de2012decoupling}, we recall this argument, as well as its proof, in section~\ref{decouplingStrategySec}. The idea is that "just" splitting the norm into independent components is not enough to get a sufficient level of concentration. Fortunately, the sum inside the norm~\eqref{normMixed2UStatistics} is a special case of a $U$-statistics for which it is possible to derive efficient tail bounds through a decoupling argument. Let us start by recalling the statement of  lemma~\ref{lemmaTperpSecondTerm}. 

\lemmaTperpSecondTerm*

\begin{proof}
The proof of lemma~\ref{lemmaTperpSecondTerm} relies on a $3$ steps argument. The first step uses the decoupling inequality due to de la Pen\~a and Montgomery-Smith~\cite{de1995decoupling} which is summarized by proposition~\ref{decouplingProp}. We provide an adapted version of the proof of this proposition in section~\ref{decouplingStrategySec} for completeness although this proof is essentially the same as the proof of Theorem 3.4.1 in~\cite{de2012decoupling}. Proposition~\ref{decouplingProp} basically shows that the the probability of success for the event $\mathcal{E}_8(\bs A_{\ell,n})$ corresponding to the $U$-statistics and defined as 
\begin{align}\begin{split}
&\mathcal{E}_8(\bs A_{\ell,n})\equiv \left\{\left\|\mathcal{P}_T^\perp\sum_{(\ell,n)}\bs A_{\ell,n}\left\langle \bs A_{\ell,n}, \left(\sum_{(\ell',n')\neq (\ell,n)} \mathcal{P}_T(\bs A_{\ell',n'})\langle \mathcal{P}_T(\bs A_{\ell',n'}), \bs h \bs m^*\rangle - \bs E_{\ell',n'}\right)  \right\rangle\right.\right.\\
&\left.\left.-\mathcal{P}_T^\perp\mathbb{E}  \sum_{(\ell,n)} \bs A_{\ell,n}\left\langle \bs A_{\ell,n}, \left(\sum_{(\ell',n')\neq (\ell,n)} \mathcal{P}_T(\bs A_{\ell',n'})\langle \mathcal{P}_T(\bs A_{\ell',n'}), \bs h \bs m^*\rangle - \bs E_{\ell',n'}\right)  \right\rangle\right\|\geq  \delta\right\}
\end{split}\label{defEventE3}
\end{align}
can be upper bounded by the probability of success of the decoupled event $\mathcal{E}_8^d(\bs A_{\ell,n}, \tilde{\bs A}_{\ell,n})$ where the $\tilde{\bs A}_{\ell,n}$ are independent copies of the $\bs A_{\ell,n}$ and defined as
\begin{align}\begin{split}
&\mathcal{E}_8^d(\bs A_{\ell,n},\tilde{\bs A}_{\ell,n})\equiv \left\{\left\|\mathcal{P}_T^\perp\sum_{(\ell,n)}\bs A_{\ell,n}\left\langle \bs A_{\ell,n}, \left(\sum_{(\ell',n')\neq (\ell,n)} \mathcal{P}_T(\tilde{\bs A}_{\ell',n'})\langle \mathcal{P}_T(\tilde{\bs A}_{\ell',n'}), \bs h \bs m^*\rangle - \bs E_{\ell',n'}\right)  \right\rangle\right.\right.\\
&\left.\left.-\mathcal{P}_T^\perp\mathbb{E}  \sum_{(\ell,n)} \bs A_{\ell,n}\left\langle \bs A_{\ell,n}, \left(\sum_{(\ell',n')\neq (\ell,n)} \mathcal{P}_T(\tilde{\bs A}_{\ell',n'})\langle \mathcal{P}_T(\tilde{\bs A}_{\ell',n'}), \bs h \bs m^*\rangle - \bs E_{\ell',n'}\right)  \right\rangle\right\|\geq  \delta\right\}.
\end{split}\label{defEventE4}
\end{align}
Once we know that $\pr(\mathcal{E}_8)\leq \pr(\mathcal{E}_8^d)$, we can use the following argument. First note that,
\begin{align}
\mathcal{P}_T^\perp\mathbb{E}  \sum_{(\ell,n)} \bs A_{\ell,n}\left\langle \bs A_{\ell,n}, \left(\sum_{(\ell',n')\neq (\ell,n)} \mathcal{P}_T(\tilde{\bs A}_{\ell',n'})\langle \mathcal{P}_T(\tilde{\bs A}_{\ell',n'}), \bs h \bs m^*\rangle - \bs E_{\ell',n'}\right)  \right\rangle= 0.
\end{align}
as the variables $\mathcal{P}_T(\tilde{\bs A}_{\ell',n'})\langle \mathcal{P}_T(\tilde{\bs A}_{\ell',n'}), \bs h \bs m^*\rangle - \bs E_{\ell',n'}$ are centered,
\begin{align}
\mathbb{E}\left(\sum_{(\ell',n')\neq (\ell,n)} \mathcal{P}_T(\tilde{\bs A}_{\ell',n'})\langle \mathcal{P}_T(\tilde{\bs A}_{\ell',n'}), \bs h \bs m^*\rangle - \bs E_{\ell',n'}\right) = 0
\end{align}
In particular, we thus have 
\begin{align}\begin{split}
&\left\|\mathcal{P}_T^\perp\sum_{(\ell,n)}\bs A_{\ell,n}\left\langle \bs A_{\ell,n}, \left(\sum_{(\ell',n')\neq (\ell,n)} \mathcal{P}_T(\tilde{\bs A}_{\ell',n'})\langle \mathcal{P}_T(\tilde{\bs A}_{\ell',n'}), \bs h \bs m^*\rangle - \bs E_{\ell',n'}\right)  \right\rangle\right.\\
&\left.-\mathcal{P}_T^\perp\mathbb{E}  \sum_{(\ell,n)} \bs A_{\ell,n}\left\langle \bs A_{\ell,n}, \left(\sum_{(\ell',n')\neq (\ell,n)} \mathcal{P}_T(\tilde{\bs A}_{\ell',n'})\langle \mathcal{P}_T(\tilde{\bs A}_{\ell',n'}), \bs h \bs m^*\rangle - \bs E_{\ell',n'}\right)  \right\rangle\right\|\end{split}\label{Sequence01}\\
\leq &\left\|\mathcal{P}_T^\perp\sum_{(\ell,n)}\bs A_{\ell,n}\left\langle \bs A_{\ell,n}, \left(\sum_{(\ell',n')\neq (\ell,n)} \mathcal{P}_T(\tilde{\bs A}_{\ell',n'})\langle \mathcal{P}_T(\tilde{\bs A}_{\ell',n'}), \bs h \bs m^*\rangle - \bs E_{\ell',n'}\right)  \right\rangle\right.\nonumber\\
&\left.-\mathcal{P}_T^\perp\mathbb{E}_{\bs A_{\ell,n}}  \sum_{(\ell,n)} \bs A_{\ell,n}\left\langle \bs A_{\ell,n}, \left(\sum_{(\ell',n')\neq (\ell,n)} \mathcal{P}_T(\tilde{\bs A}_{\ell',n'})\langle \mathcal{P}_T(\tilde{\bs A}_{\ell',n'}), \bs h \bs m^*\rangle - \bs E_{\ell',n'}\right)  \right\rangle\right\|\label{beforeLastHopeNorm}\\
+&\left\|\mathcal{P}_T^\perp\mathbb{E}_{\bs A_{\ell,n}}  \sum_{(\ell,n)} \bs A_{\ell,n}\left\langle \bs A_{\ell,n}, \left(\sum_{(\ell',n')\neq (\ell,n)} \mathcal{P}_T(\tilde{\bs A}_{\ell',n'})\langle \mathcal{P}_T(\tilde{\bs A}_{\ell',n'}), \bs h \bs m^*\rangle - \bs E_{\ell',n'}\right)  \right\rangle\right\|\label{lastHopeNorm} 
\end{align}
Where the expectation in~\eqref{beforeLastHopeNorm} is now taken with respect to the outer matrices $\bs A_{\ell,n}$ only. That is the second factor in~\eqref{lastHopeNorm} remains a random variable. Moreover, as we have 
\begin{align}
\left\|\mathcal{P}_T^\perp\mathbb{E}_{\bs A_{\ell,n}}  \sum_{(\ell,n)} \bs A_{\ell,n}\left\langle \bs A_{\ell,n}, \left(\sum_{(\ell',n')} \mathcal{P}_T(\tilde{\bs A}_{\ell',n'})\langle \mathcal{P}_T(\tilde{\bs A}_{\ell',n'}), \bs h \bs m^*\rangle - \bs E_{\ell',n'}\right)  \right\rangle\right\| = 0,
\end{align}
one can further bound the norm~\eqref{lastHopeNorm} as
\begin{align}
&\left\|\mathcal{P}_T^\perp\mathbb{E}_{\bs A_{\ell,n}}  \sum_{(\ell,n)} \bs A_{\ell,n}\left\langle \bs A_{\ell,n}, \left(\sum_{(\ell',n')\neq (\ell,n)} \mathcal{P}_T(\tilde{\bs A}_{\ell',n'})\langle \mathcal{P}_T(\tilde{\bs A}_{\ell',n'}), \bs h \bs m^*\rangle - \bs E_{\ell',n'}\right)  \right\rangle\right\|\label{sequence11}\\
&= \left\|\mathcal{P}_T^\perp\mathbb{E}_{\bs A_{\ell,n}}  \sum_{(\ell,n)} \bs A_{\ell,n}\left\langle \bs A_{\ell,n}, \left(\sum_{(\ell',n')} \mathcal{P}_T(\tilde{\bs A}_{\ell',n'})\langle \mathcal{P}_T(\tilde{\bs A}_{\ell',n'}), \bs h \bs m^*\rangle - \bs E_{\ell',n'}\right)  \right\rangle\right.\\
& -\left.\mathcal{P}_T^\perp\mathbb{E}_{\bs A_{\ell,n}}  \sum_{(\ell,n)} \bs A_{\ell,n}\left\langle \bs A_{\ell,n}, \left( \mathcal{P}_T(\tilde{\bs A}_{\ell,n})\langle \mathcal{P}_T(\tilde{\bs A}_{\ell,n}), \bs h \bs m^*\rangle - \bs E_{\ell,n}\right)  \right\rangle\right\|\\
&\leq \left\|\mathcal{P}_T^\perp\mathbb{E}_{\bs A_{\ell,n}}  \sum_{(\ell,n)} \bs A_{\ell,n}\left\langle \bs A_{\ell,n}, \left( \mathcal{P}_T(\tilde{\bs A}_{\ell,n})\langle \mathcal{P}_T(\tilde{\bs A}_{\ell,n}), \bs h \bs m^*\rangle - \bs E_{\ell,n}\right)  \right\rangle\right\|\label{normSwappedExp}
\end{align}
Given those comments. One can proceed with the rest of the proof. The second norm~\eqref{normSwappedExp} is bounded through lemma~\ref{lemmaLemmaExpectationSwap} below. This lemma follows from a direct application of proposition~\ref{bernstein} as the outer matrices are averaged out. It is proved in section~\ref{sectionExpectationSwap}.
\begin{restatable}{lemma}{lemmaLemmaExpectationSwap} 
Let $\hat{\bs c}_{\ell,n}$ be defined as in section~\ref{notations} (including the scaling $\sqrt{L}$), with $C_n[\ell,k]\sim \mathcal{N}(0,1/L)$ i.i.d. gaussian. Let $\bs A_{\ell,n} = \bs f_\ell\hat{\bs c}_{\ell,n}^*$ and let the operator $\mathcal{A}\;:\;\mathbb{C}^{L\times KN} \mapsto \mathbb{C}^{L\times N}$ be defined from those matrices $\bs A_{\ell,n}$ as in~\eqref{DefinitionLinearMap}. The coherences $\mu_m^2$ and $\mu_h^2$ are defined as in~\eqref{coherencem} and~\eqref{coherenceh}. Then for any constants $\delta_9,\beta_9$, as soon as $L\gtrsim \beta_9(1/\delta_9)K\mu_h^2$ and $N\gtrsim\beta_9(1/\delta_9)\mu_m^2$
\label{lemmaLemmaExpectationSwap} 
\begin{align}
\left\|\mathcal{P}_T^\perp\mathbb{E}_{\bs A_{\ell,n}}  \sum_{(\ell,n)} \bs A_{\ell,n}\left\langle \bs A_{\ell,n}, \left( \mathcal{P}_T(\tilde{\bs A}_{\ell,n})\langle \mathcal{P}_T(\tilde{\bs A}_{\ell,n}), \bs h \bs m^*\rangle - \bs E_{\ell,n}\right)  \right\rangle\right\|\leq \delta_9 \label{normLemmaLemmaExpectationSwap}
\end{align}
with probability at least $1-\left(LN\right)^{-\beta_9}$. 
\end{restatable}
The first norm in~\eqref{beforeLastHopeNorm} will be bounded by lemmas~\ref{lemmaSecondTermIndependent} and~\ref{convergenceCoherence}. The idea behind these two lemmas is as follows. Since the variables $\bs A_{\ell,n}$ are independent from the variables $\tilde{\bs A}_{\ell,n}$, we can now condition the probability of having a small norm~\eqref{beforeLastHopeNorm} on the $\tilde{\bs A}_{\ell,n}$ and therefore consider those variables as fixed over a first phase. In particular, if we let 
\begin{align}
\bs H_{\ell,n} \equiv \left(\sum_{(\ell',n')\neq (\ell,n)} \mathcal{P}_T(\tilde{\bs A}_{\ell',n'})\langle \mathcal{P}_T(\tilde{\bs A}_{\ell',n'}), \bs h \bs m^*\rangle - \bs E_{\ell',n'}\right).
\end{align}
Let us now define the coherences $\mu^2_{\bs H_{\ell,n}}$, $\rho^2_{\bs H_{\ell,n}}$ and $\nu^2_{\bs H_{\ell,n}}$, for given matrices $\bs H_{\ell,n}$ as
\begin{equation}
\mu^2_{\bs H_{\ell,n}} = L\cdot \sup_{\ell_2} \sum_{n_2=1}^N \|\bs H_{\ell,n}[\ell_2,\sim n_2]\|^2,\label{coherenceHmu}
\end{equation}
\begin{equation}
\rho^2_{\bs H_{\ell,n}} = N\cdot \sup_{n_2} \sum_{\ell_2=1}^L\|\bs H_{\ell,n}[\ell_2,\sim n_2]\|^2\label{coherenceHrho}
\end{equation}
\begin{equation}
\nu^2_{\bs H_{\ell,n}} = LN\cdot\sup_{\ell_2,n_2} \|\bs H_{\ell,n}[\ell_2,\sim n_2]\|^2\label{coherenceHnu}
\end{equation}
and let the event $\mathcal{E}_9$ be defined as   
\begin{align}
\mathcal{E}_9\equiv\left\{\mu^2_{\bs H_{\ell,n}}\lesssim \mu_h^2,\; \rho^2_{\bs H_{\ell,n}}\lesssim \mu_m^2,\quad \nu^2_{\bs H_{\ell,n}}\lesssim \mu_h^2\mu_m^2\right\}
\end{align}
one can now bound the probability 
\begin{align}
&\pr \left(\left\|\mathcal{P}_T^\perp  \sum_{(\ell,n)} \bs A_{\ell,n}\left\langle \bs A_{\ell,n}, \bs H_{\ell,n} \right\rangle -\mathcal{P}_T^\perp\mathbb{E}_{\bs A_{\ell,n}}  \sum_{(\ell,n)} \bs A_{\ell,n}\left\langle \bs A_{\ell,n}, \bs H_{\ell,n}  \right\rangle\right\|\geq \delta \right)
\end{align}
as
\begin{align}\begin{split}
&\pr \left(\left\|\mathcal{P}_T^\perp\mathbb{E}_{\bs A_{\ell,n}}  \sum_{(\ell,n)} \bs A_{\ell,n}\left\langle \bs A_{\ell,n}, \bs H_{\ell,n} \right\rangle-\mathcal{P}_T^\perp\mathbb{E}_{\bs A_{\ell,n}}  \sum_{(\ell,n)} \bs A_{\ell,n}\left\langle \bs A_{\ell,n}, \bs H_{\ell,n} \right\rangle\right\|\geq \delta \right)\\
&\leq \pr \left(\left\|\mathcal{P}_T^\perp\mathbb{E}_{\bs A_{\ell,n}}  \sum_{(\ell,n)} \bs A_{\ell,n}\left\langle \bs A_{\ell,n}, \bs H_{\ell,n} \right\rangle-\mathcal{P}_T^\perp\mathbb{E}_{\bs A_{\ell,n}}  \sum_{(\ell,n)} \bs A_{\ell,n}\left\langle \bs A_{\ell,n}, \bs H_{\ell,n} \right\rangle\right\|\geq \delta \Bigg|  \mathcal{E}_9\right)
+\pr\left(\mathcal{E}_9^c\right)\end{split}\label{inequalityProbabilityFinal}
\end{align}
Because of the decoupling argument, the event $\mathcal{E}_9$ is now independent of the matrices $\bs A_{\ell,n}$ and can be fixed while bounding the norm. Bounding the two probabilities on the RHS of~\eqref{inequalityProbabilityFinal} is the point of lemmas~\ref{lemmaSecondTermIndependent} and~\ref{convergenceCoherence} below which are respectively proved in sections~\ref{proofIndependentSecondOrder} and~\ref{proofconvergenceCoherence}. 
\begin{restatable}{lemma}{lemmaSecondTermIndependent}
\label{lemmaSecondTermIndependent} 
Let the linear map $\mathcal{A}$ be defined as in~\eqref{DefinitionLinearMap} with $\mathcal{A}_{\ell,n}(\bs X) = \langle \bs A_{\ell,n}, \bs X\rangle =\langle \bs f_\ell\bs c_{\ell,n}^*, \bs X\rangle $. For any fixed matrices $\bm X_{\ell,n}$ independent of the $\mathcal{A}_{\ell,n}$. Let the coherences of the $\bs X_{\ell,n}$ be defined as,
\begin{equation}
\mu^2_{\bm X_{\ell,n}} = L\cdot \sup_{\ell_2} \sum_{n_2=1}^N \|\bm X_{\ell,n}[\ell_2,\sim n_2]\|^2,\label{coherenceXmu}
\end{equation}
\begin{equation}
\rho^2_{\bm X_{\ell,n}} = N\cdot \sup_{n_2} \sum_{\ell_2=1}^L\|\bm X_{\ell,n}[\ell_2,\sim n_2]\|^2\label{coherenceXrho}
\end{equation}
\begin{equation}
\nu^2_{\bm X_{\ell,n}} = LN\cdot\sup_{\ell_2,n_2} \|\bm X_{\ell,n}[\ell_2,\sim n_2]\|^2\label{coherenceXnu}
\end{equation}
Correspondingly, let $\overline{\mu}_{\bs X}^2,\overline{\rho}_{\bs X}^2$ and $\overline{\nu}^2_{\bs X}$ denote the supremas over $(\ell,n)$ of these quantities, 
\begin{equation}\label{supremasCoherence}
\overline{\mu}_{\bs X}^2 = \sup_{\ell,n} \mu^2_{\bm X_{\ell,n}},\quad \overline{\rho}_{\bs X}^2 =  \sup_{\ell,n} \rho^2_{\bm X_{\ell,n}},\quad\mbox{and}\quad \overline{\nu}_{\bs X}^2 =  \sup_{\ell,n} \nu^2_{\bm X_{\ell,n}}.\end{equation}
One can write, 
\begin{multline}
\left\|\sum_{\ell=1}^L\sum_{n=1}^N \mathcal{A}_{\ell,n}^*\mathcal{A}_{\ell,n}(\bm X_{\ell,n}) - \mathbb{E}\sum_{\ell=1}^L\sum_{n=1}^N\mathcal{A}_{\ell,n}^*\mathcal{A}_{\ell,n}(\bm X_{\ell,n})\right\|\lesssim \\
\beta_{10}\max\left\{\sqrt{\left(\frac{\overline{\mu}^2_{\bm X}K}{L}  + \frac{\overline{\rho}^2_{\bm X}}{N} \right)}\sqrt{\log(LKN)},  \frac{\overline{\nu}^2_{\bm X}}{LN} K \log(LKN))\right\}\label{tailBoundTperDecoupled}
\end{multline}
with probability at least $1 - (LN)^{-\beta_{10}}$. In particular, letting $\bs X_{\ell,n} = \bs H_{\ell,n}$, assuming $\mu^2_{\bs H_{\ell,n}}\lesssim \mu_h^2$, $\nu^2_{\bs H_{\ell,n}}\lesssim\mu_h^2\mu_m^2$ and $\rho^2_{\bs H_{\ell,n}}\lesssim \mu_m^2$, and taking $L\gtrsim \beta_{10}(1/\delta_{10})K\mu_h^2$ as well as $N\gtrsim \mu_m^2$, one can make the norm on the LHS of~\eqref{tailBoundTperDecoupled} less than $\delta_{10}$. Consequently, the first probability on the RHS of~\eqref{inequalityProbabilityFinal} can be bounded by $(LN)^{-\beta_{10}}$.
\end{restatable}
\begin{restatable}{lemma}{convergenceCoherence}
\label{convergenceCoherence}
Let $\bs A_{\ell,n}$ be defined as in~\eqref{DefinitionLinearMap}, and let $\bs X_{\ell,n}$ (or equivalently $\bs H_{\ell,n}$) be defined as 
\begin{align}
\bs X_{\ell,n} &= \sum_{(\ell',n')\neq (\ell,n)}\mathcal{P}_T(\bs A_{\ell',n'}) \langle \mathcal{P}_T(\bs A_{\ell',n'}),\bs h\bs m^*\rangle - \mathbb{E} \sum_{(\ell',n')\neq (\ell,n)}\mathcal{P}_T(\bs A_{\ell',n'}) \langle \mathcal{P}_T(\bs A_{\ell',n'}),\bs h\bs m^*\rangle \end{align} 
where the matrices $\bs A_{\ell',n'}$ are defined as in~\eqref{deconv3}. The coherences $\mu^2_{\bs X_{\ell,n}}$, $\rho^2_{\bs X_{\ell,n}}$ and $\nu^2_{\bs X_{\ell,n}}$ are defined as in~\eqref{coherenceHmu},~\eqref{coherenceHrho} and~\eqref{coherenceHnu}. Then we have 
\begin{align}
\mu^2_{\bs X_{\ell,n}}\lesssim \mu_h^2, \quad \rho^2_{\bs X_{\ell,n}}\lesssim \mu_m^2,\quad \mbox{and}\quad \nu^2_{\bs X_{\ell,n}}\lesssim \mu_m^2\mu_h^2, \qquad \forall (\ell,n) \in [L]\times[N],
\end{align}
where each event holds with probability at least $1-(LN)^{-\beta_{11}}$ for any constant $\beta_{11}$ as soon as $L\gtrsim \beta_{11} K\mu_h^2$, $N\gtrsim \beta_{11}\mu_m^2$. 
\end{restatable}

Before going through the decoupling argument underlying the bound $\pr(\mathcal{E}_8)\leq \pr(\mathcal{E}_8^d)$, we combine all previous results and conclude the proof of lemma~\ref{lemmaTperpSecondTerm}. Combining the decoupling argument relating $\mathcal{E}_8$ and $\mathcal{E}_8^c$ together with the sequence of bounds~\eqref{Sequence01} to~\eqref{lastHopeNorm}, the sequence~\eqref{sequence11} to~\eqref{normSwappedExp}, the fact that $\pr(\|A + B\|\geq \delta)\leq \pr(\|A\|\geq \delta/2) + \pr(\|B\|\geq \delta/2) $, as well as the results of lemmas~\ref{lemmaSecondTermIndependent} and~\ref{convergenceCoherence} and~\ref{lemmaLemmaExpectationSwap}, we get,
\begin{align}
& \left\|\mathcal{P}_T^\perp\sum_{(\ell,n)}\bs A_{\ell,n}\left\langle \bs A_{\ell,n}, \left(\sum_{(\ell',n')\neq (\ell,n)} \mathcal{P}_T(\bs A_{\ell',n'})\langle \mathcal{P}_T(\bs A_{\ell',n'}), \bs h \bs m^*\rangle - \bs E_{\ell',n'}\right)  \right\rangle\right.\\
&\left.-\mathcal{P}_T^\perp\mathbb{E}  \sum_{(\ell,n)} \bs A_{\ell,n}\left\langle \bs A_{\ell,n}, \left(\sum_{(\ell',n')\neq (\ell,n)} \mathcal{P}_T(\bs A_{\ell',n'})\langle \mathcal{P}_T(\bs A_{\ell',n'}), \bs h \bs m^*\rangle - \bs E_{\ell',n'}\right)  \right\rangle\right\|\lesssim \delta
\end{align}
with probability at least $1 - c(LN)^{-\beta}$ where $c,\beta>0$ are positive constants. This concludes the proof of lemma~\ref{lemmaTperpSecondTerm}.
\end{proof}

\subsubsection{\label{decouplingStrategySec}Decoupling strategy}
In this section, we explain the decoupling argument used in the proof of lemma~\ref{lemmaTperpSecondTerm} to relate $\mathcal{E}_3$ and $\mathcal{E}_4$. This argument is derived from Theorem 3.1.1 and 3.4.1 in~\cite{de2012decoupling, de1995decoupling}. A similar result is given for the expectation in 
~\cite{vershynin2011simple} (Theorem 1.). We provide the proof of this result as well for completeness although the proof is essentially the same as the proof given in~\cite{de2012decoupling}.
\begin{proposition}[\label{decouplingProp}Decoupling inequality for U-statistics (matrix version)]
Let us recall the norm,
\begin{align}
&E = \left\|\mathcal{P}_T^\perp \sum_{(\ell,n)}\sum_{(\ell',n')\neq (\ell,n)}\mathcal{A}_{\ell,n} \langle \mathcal{A}_{\ell,n}, \mathcal{P}_T(\mathcal{A}_{\ell',n'}) \langle \mathcal{P}_T(\mathcal{A}_{\ell',n'}),\bs h\bs m^*\rangle \rangle\right.\\
&  \left.-\mathcal{P}_T^\perp \mathbb{E}\sum_{(\ell,n)\neq (\ell',n')}\sum_{(\ell',n')\neq (\ell,n)}\mathcal{A}_{\ell,n} \langle \mathcal{A}_{\ell,n}, \mathcal{P}_T(\mathcal{A}_{\ell',n'}) \langle \mathcal{P}_T(\mathcal{A}_{\ell',n'}),\bs h\bs m^*\rangle \rangle \right\|\label{inequalitySecondOrderbb}
\end{align}
and let us consider the following norm where the sequence of variables $\mathcal{A}_{\ell',n'}$ have been replaced by an independent copy of the sequence $\left\{\mathcal{A}_{\ell,n}\right\}_{(\ell,n)\in [L]\times[N]}$,
\begin{align}
&\tilde{E}= \left\|\mathcal{P}_T^\perp \sum_{(\ell,n)}\sum_{(\ell',n')\neq (\ell,n)}\mathcal{A}_{\ell,n} \langle \mathcal{A}_{\ell,n}, \mathcal{P}_T(\tilde{\mathcal{A}}_{\ell',n'}) \langle \mathcal{P}_T(\tilde{\mathcal{A}}_{\ell',n'}),\bs h\bs m^*\rangle \rangle\right.\\
&  \left.-\mathcal{P}_T^\perp \mathbb{E}\sum_{(\ell,n)\neq (\ell',n')}\sum_{(\ell',n')\neq (\ell,n)}\mathcal{A}_{\ell,n} \langle \mathcal{A}_{\ell,n}, \mathcal{P}_T(\tilde{\mathcal{A}}_{\ell',n'}) \langle \mathcal{P}_T(\tilde{\mathcal{A}}_{\ell',n'}),\bs h\bs m^*\rangle \rangle \right\|\label{inequalitySecondOrderbbb}
\end{align}
Then $\pr(E\geq t)\leq C\pr\left(\tilde{E}\geq t/C\right)$ where $C$ is a constant.
\end{proposition}
\begin{proof}
The proof is, in essence, the same as the proof given by de la Pen\~a and Gin\'e in~\cite{de2012decoupling}. The only difference lies in its transposition to matrix valued random variables. Let $\bs X_{\ell,n}$ and $\bs Y_{\ell',n'}$ be defined as follows,
\begin{align}
\bs X_{\ell,n}& = \left(\mbox{vec}(\bs A_{\ell,n})\otimes  \mbox{vec}(\bs A_{\ell,n})^*\right) \\
\bs Y_{\ell',n'}& = \left(\mbox{vec}(\mathcal{P}_T(\bs A_{\ell',n'})) \langle \mathcal{P}_T(\bs A_{\ell',n'}),\bs h\bs m^*\rangle \right)
\end{align}
The norm~\eqref{inequalitySecondOrderbbb} can read compactly as 
\begin{align}
&\sum_{(\ell,n)}\sum_{(\ell',n')\neq (\ell,n)}\mathcal{P}_T^\perp \left(\bs X_{\ell,n} \bs Y_{\ell',n'} \right) -\mathbb{E}\sum_{(\ell,n)}\sum_{(\ell',n')\neq (\ell,n)}\mathcal{P}_T^\perp \left(\bs X_{\ell,n} \bs Y_{\ell',n'} \right)\nonumber \\
& = \sum_{(\ell,n)}\sum_{(\ell',n')\neq (\ell,n)} \left(\mbox{vec}(\bs A_{\ell,n})\otimes  \mbox{vec}(\bs A_{\ell,n})\right)\left(\mbox{vec}(\mathcal{P}_T(\bs A_{\ell',n'})) \langle \mathcal{P}_T(\bs A_{\ell',n'}),\bs h\bs m^*\rangle \right)- \bs C_{\ell,n,\ell',n'}\\
& = \sum_{(\ell,n)}\sum_{(\ell',n')\neq (\ell,n)} \left(\mbox{vec}(\bs A_{\ell,n})\otimes  \mbox{vec}(\bs A_{\ell,n})\right)\mathcal{P}_T\left(\mbox{vec}(\bs A_{\ell',n'})\otimes \mbox{vec}(\bs A_{\ell',n'})^*\right)\mbox{vec}(\bs h\bs m^*)- \bs C_{\ell,n,\ell',n'}\\
&= h_{\ell,n,\ell',n'}\bigg(\mbox{vec}(\bs A_{\ell,n})\otimes  \mbox{vec}(\bs A_{\ell,n}), \mbox{vec}(\bs A_{\ell',n'})\otimes  \mbox{vec}(\bs A_{\ell',n'}) \bigg)\\
& = h_{\ell,n,\ell',n'}(\bs X_{\ell,n},\bs X_{\ell',n'})
\end{align}
where $\bs C_{\ell,n,\ell',n'}$ are constant (i.e. deterministic) matrices and $h_{\ell,n,\ell',n'}(\bs X,\bs Y)$ is a function defined as 
\begin{align}
h_{\ell,n,\ell',n'}(\bs X_{\ell,n},\bs X_{\ell',n'}) = \mathcal{P}_T^\perp\bs X_{\ell,n} \bs P_T \bs X_{\ell',n'}\mbox{vec}(\bs h\bs m^*) - \bs C_{\ell,n,\ell',n'}
\end{align}
$\bs P_T$ is the matrix formulation of the projector $\mathcal{P}_T$ defined in~\eqref{projectorTBD}. Let $\{\tilde{\bs X}_{\ell,n}\}_{\ell,n}$ denote an independent copy of the sequence $\{\bs X_{\ell,n}\}_{\ell,n}$.
Let $\{\varepsilon_{\ell,n}\}$ denote a sequence of independent Rademacher random variables and consider the accompanying sequences of matrices $\{\bs Z_{\ell,n}\}$ and $\{\tilde{\bs Z}_{\ell,n}\}$ defined as 
\begin{align}
{\bs Z}_{\ell,n} = \left\{\begin{array}{ll}
\bs X_{\ell,n} & \mbox{if $\varepsilon_{\ell,n}=1$}\\
\tilde{\bs X}_{\ell,n} & \mbox{if $\varepsilon_{\ell,n}=-1$}\\
\end{array}\right.,\qquad \tilde{\bs Z}_{\ell,n} = \left\{\begin{array}{ll}
\tilde{\bs X}_{\ell,n} & \mbox{if $\varepsilon_{\ell,n}=1$}\\
\bs X_{\ell,n} & \mbox{if $\varepsilon_{\ell,n}=-1$}\\
\end{array}\right.
\end{align}
The law $\mathcal{D}$ of $(\bs Z_{1,1},\ldots, \bs Z_{L,N})$ is the same as the law of $(\bs A_{(1,1)},\ldots,\bs A_{(L,N)})$ and similarly, 
$$\mathcal{D}(\bs Z_{1,1},\ldots, \bs Z_{L,N}, \tilde{\bs Z}_{1,1},\ldots, \tilde{\bs Z}_{L,N}) = \mathcal{D}(\bs A_{1,1},\ldots, \bs A_{L,N}, \tilde{\bs A}_{1,1},\ldots, \tilde{\bs A}_{L,N})$$ 
Both of these laws are given by $(\mathcal{P}_{1,1},\ldots, \mathcal{P}_{L,N}) = \mathcal{P}^{LN}$ and $(\mathcal{P}_{1,1},\ldots, \mathcal{P}_{L,N})^2 = \mathcal{P}^{2LN}$ where $\mathcal{P}$ is the law of each $\bs A_{\ell,n}$. If we let $\mathcal{X}$ denote the sigma algebra $\sigma(\bs A_{\ell,n}, \tilde{\bs A}_{\ell,n})$ generated by the sequences $\bs A_{\ell,n}$ and $\tilde{\bs A}_{\ell,n}$, we have 
\begin{align}\begin{split}
h(\bs Z_{\ell,n}, \bs Z_{\ell',n'}|\mathcal{X}) & = h(\bs X_{\ell,n}, \bs X_{\ell',n'})\delta(\varepsilon_{\ell,n}=1,\varepsilon_{\ell',n'}=1) \\
&+ h(\bs X_{\ell,n}, \tilde{\bs X}_{\ell',n'})\delta(\varepsilon_{\ell,n}=1,\varepsilon_{\ell',n'}=-1)\\
& + h(\tilde{\bs X}_{\ell,n}, \bs X_{\ell',n'})\delta(\varepsilon_{\ell,n}=-1,\varepsilon_{\ell',n'}=1)\\
& + h(\tilde{\bs X}_{\ell,n}, \tilde{\bs X}_{\ell',n'})\delta(\varepsilon_{\ell,n}=-1,\varepsilon_{\ell',n'}=-1)
\end{split}\label{decompositionh01}
\end{align}
which, when averaging over the Rademacher sequence, implies 
\begin{align}
\mathbb{E}_{\varepsilon} h(\bs Z_{\ell,n}, \bs Z_{\ell',n'}|\mathcal{X}) &= \frac{1}{4}\left(h(\bs X_{\ell,n}, \bs X_{\ell',n'})+ h(\bs X_{\ell,n}, \tilde{\bs X}_{\ell',n'})+ h(\tilde{\bs X}_{\ell,n}, \bs X_{\ell',n'})+ h(\tilde{\bs X}_{\ell,n}, \tilde{\bs X}_{\ell',n'})\right)
\end{align}
Equivalently, we will use the following relation later in the proof,
\begin{align}\begin{split}
4 h(\bs Z_{\ell,n}, \tilde{\bs Z}_{\ell',n'})& = (1+\varepsilon_{\ell,n})(1+\varepsilon_{\ell',n'})h(\bs X_{\ell,n}, \tilde{\bs X}_{\ell',n'}) \\
& + (1+\varepsilon_{\ell,n})(1-\varepsilon_{\ell',n'})h(\bs X_{\ell,n}, \bs X_{\ell',n'})\\
& + (1-\varepsilon_{\ell,n})(1+\varepsilon_{\ell',n'})h(\tilde{\bs X}_{\ell,n}, \tilde{\bs X}_{\ell',n'})\\
& + (1-\varepsilon_{\ell,n})(1-\varepsilon_{\ell',n'})h(\tilde{\bs X}_{\ell,n}, \bs X_{\ell',n'})
\end{split}\label{decompositionhZ}
\end{align}
Now we can follow the proof of Theorem 3.4.1 in~\cite{de2012decoupling} (Theorem 1 in~\cite{de1995decoupling}). We will need the following result which follows as a special case of Theorem 1.1.3. in~\cite{de2012decoupling}.  For $\bs X, \bs Y$ and $\bs Z$ i.i.d., we have  
\begin{align}
\pr\left(\|\bs X\|>t\right) & = \pr\left(\|(\bs X+\bs Y) +  (\bs X+\bs Z) - (\bs Y+\bs Z)\|>2t\right)\\
& \leq \pr\left(\|(\bs X+\bs Y)\| > \frac{2t}{3}\right) +  \pr\left(\|\bs X+\bs Z\|>\frac{2t}{3}\right) +\pr\left(\|\bs Y+\bs Z\|>\frac{2t}{3}\right)\\
&\leq 3\pr\left(\|\bs X + \bs Y\|>\frac{2t}{3}\right)\label{inequalityGine}
\end{align}
Using~\eqref{inequalityGine} with the independent sequence $\{\tilde{\bs X}_{\ell,n}\}_{\ell,n}$, one can write
\begin{align}
&\pr\left(\left\|\sum_{\ell,n}\sum_{(\ell',n')\neq (\ell,n)} h(\bs X_{\ell,n}, \bs X_{\ell',n'})\right\|>t\right) \\
&\leq 3\pr\left(\left\|\sum_{\ell,n}\sum_{(\ell',n')\neq (\ell,n)} h(\bs X_{\ell,n}, \bs X_{\ell',n'}) + h(\tilde{\bs X}_{\ell,n}, \tilde{\bs X}_{\ell',n'})\right\|>\frac{2t}{3}\right)\label{probaboundDecoup1}
\end{align}
Splitting the norm, 
\begin{align}
&\left\|\sum_{\ell,n}\sum_{(\ell',n')\neq (\ell,n)} h(\bs X_{\ell,n}, \bs X_{\ell',n'}) + h(\tilde{\bs X}_{\ell,n}, \tilde{\bs X}_{\ell',n'})\right\|\\
&\leq \left\|\sum_{\ell,n}\sum_{(\ell',n')\neq (\ell,n)} h(\bs X_{\ell,n}, \bs X_{\ell',n'}) + h(\tilde{\bs X}_{\ell,n}, \tilde{\bs X}_{\ell',n'}) - h(\tilde{\bs X}_{\ell,n}, \bs X_{\ell',n'}) - h(\bs X_{\ell,n}, \tilde{\bs X}_{\ell',n'})\right\|\\
&+\left\|\sum_{\ell,n}\sum_{(\ell',n')\neq (\ell,n)} h(\tilde{\bs X}_{\ell,n}, \bs X_{\ell',n'}) - h(\bs X_{\ell,n}, \tilde{\bs X}_{\ell',n'})\right\|
\end{align}
Susbtituting into~\eqref{probaboundDecoup1} yields,
\begin{align}
&\pr\left(\left\|\sum_{\ell,n}\sum_{(\ell',n')\neq (\ell,n)} h(\bs X_{\ell,n}, \bs X_{\ell',n'})\right\|>t\right)\\
&\leq 3\pr\left(\left\|\sum_{\ell,n}\sum_{(\ell',n')\neq (\ell,n)} h(\bs X_{\ell,n}, \bs X_{\ell',n'}) + h(\tilde{\bs X}_{\ell,n}, \tilde{\bs X}_{\ell',n'}) - h(\tilde{\bs X}_{\ell,n}, \bs X_{\ell',n'}) - h(\bs X_{\ell,n}, \tilde{\bs X}_{\ell',n'})\right\|> \frac{t}{3}\right)\label{decouple2}\\
& + 3\pr \left(\left\|\sum_{\ell,n}\sum_{(\ell',n')\neq (\ell,n)} h(\tilde{\bs X}_{\ell,n}, \bs X_{\ell',n'}) - h(\bs X_{\ell,n}, \tilde{\bs X}_{\ell',n'})\right\|>\frac{t}{3}\right)\label{decouple3}
\end{align}
\begin{align}
&\pr\left(\left\|\sum_{\ell,n}\sum_{(\ell',n')\neq (\ell,n)} h(\bs X_{\ell,n}, \bs X_{\ell',n'})\right\|>t\right)\\
&\leq 3\pr\left(\left\|\sum_{\ell,n}\sum_{(\ell',n')\neq (\ell,n)} h(\bs X_{\ell,n}, \bs X_{\ell',n'}) + h(\tilde{\bs X}_{\ell,n}, \tilde{\bs X}_{\ell',n'}) - h(\tilde{\bs X}_{\ell,n}, \bs X_{\ell',n'}) - h(\bs X_{\ell,n}, \tilde{\bs X}_{\ell',n'})\right\|> \frac{t}{3}\right)\label{decouple4}\\
& + 3 \pr \left(\left\|\sum_{\ell,n}\sum_{(\ell',n')\neq (\ell,n)} h(\tilde{\bs X}_{\ell,n}, \bs X_{\ell',n'}) \right\|>\frac{t}{6}\right)\label{decouple5}\\
& + 3 \pr \left(\left\|\sum_{\ell,n}\sum_{(\ell',n')\neq (\ell,n)} h(\bs X_{\ell,n}, \tilde{\bs X}_{\ell',n'}) \right\|>\frac{t}{6}\right)\label{decouple6}\\
& \leq  3\pr\left(\left\|\sum_{\ell,n}\sum_{(\ell',n')\neq (\ell,n)} h(\bs X_{\ell,n}, \bs X_{\ell',n'}) + h(\tilde{\bs X}_{\ell,n}, \tilde{\bs X}_{\ell',n'}) - h(\tilde{\bs X}_{\ell,n}, \bs X_{\ell',n'}) - h(\bs X_{\ell,n}, \tilde{\bs X}_{\ell',n'})\right\|> \frac{t}{3}\right)\label{decouple7}\\
& + 6 \pr \left(\left\|\sum_{\ell,n}\sum_{(\ell',n')\neq (\ell,n)} h(\tilde{\bs X}_{\ell,n}, \bs X_{\ell',n'}) \right\|>\frac{t}{6}\right)\label{decouple8}
\end{align}
%
%In~\eqref{decouple8} and~\eqref{decouple3} as well as~\eqref{decouple5} and~\eqref{decouple6}, we use the fact that whenever both $\|A\|$ and $\|B\|$ are less than $t$, $\|A+B\|<2t$. 
In~\eqref{decouple8} we use the fact that $h(\bs X_{\ell,n}, \tilde{\bs X}_{\ell',n'}) $ and $h(\tilde{\bs X}_{\ell,n}, \bs X_{\ell',n'})$ have the same distribution. We are thus left with bounding the $U$-statistics~\eqref{decouple7}. Let us use $U$ to denote the sum,
\begin{align}
U &= 4\sum_{\ell,n}\sum_{(\ell',n')\neq (\ell,n)} h(\bs Z_{\ell,n}, \tilde{\bs Z}_{\ell',n'})
\end{align}
and let $\mathcal{U} = \left\{h(\bs X_{\ell,n},\bs X_{\ell',n'}), h(\bs X_{\ell,n},\tilde{\bs X}_{\ell',n'}), h(\tilde{\bs X}_{\ell,n},\bs X_{\ell',n'}),h(\tilde{\bs X}_{\ell,n},\tilde{\bs X}_{\ell',n'}) \right\}$. 
Using~\eqref{decompositionhZ}, one can see that $U$ conditionned on $\mathcal{U}$ is a second order Rademacher chaos. finally, to conclude, we use a conditional version of Jensen's inequality following from Theorem 3.3.6~\cite{de2012decoupling}. Let us use $T$ to denote the U-statistics,
\begin{align}
T = \sum_{\ell,n}\sum_{\ell',n'\neq (\ell,n)}\left[h(\bs X_{\ell,n},\bs X_{\ell',n'}) + h(\bs X_{\ell,n},\tilde{\bs X}_{\ell',n'}) + h( \tilde{\bs X}_{\ell,n},\bs X_{\ell',n'}) + h(\tilde{\bs X}_{\ell,n},\tilde{\bs X}_{\ell',n'}) \right]
\end{align}
Clearly,~\eqref{decompositionh01} and~\eqref{decompositionhZ} show that we have $\mathbb{E}\left\{U|\mathcal{U}\right\} = T$. Following the proof of~\cite{de1995decoupling}, we now use the following Theorem from~\cite{de2012decoupling},
\begin{theorem}[\label{theoremOuterProba}Theorem 3.3.6 in~\cite{de2012decoupling}]
Let $\mathcal{F}$ be a normed linear space and let $n\in \mathbb{N}$. Let $\left\{\varepsilon_i\right\}_{1\leq i\leq n}$ denote a sequence of independent Rademacher variables. Let $\bs X = \left\{X_0,X_1,\ldots, X_i,\ldots,X_{i_1,\ldots,i_d}\right\}$, $1\leq i,i_k\leq n$ denote a sequence of random variables taking values in a normed vector space $\mathcal{F}$. Let $\bs \varepsilon$ and $\bs X$ be defined on different factors of a product probability space. Then for every $p$ there exists 
a constant $C_p^d$ such that the conditional Rademacher chaos 
\begin{align}
X = X_0 + \sum_{i=1}^n X_i \xi_i + \ldots + \sum_{1\leq i_1\neq i_2\neq \ldots\neq i_d\leq n}X_{i_1,\ldots,i_d}\varepsilon_1\ldots \varepsilon_{i_d}
\end{align}
satisfies 
\begin{align}
\pr^*\left(\mathbb{E}\left\{\|X\|^p|\bs X\right\}\geq 2t^p\right)\leq C_{p,d}\pr^*\left(\|X\|>t\right).\label{boundOuterProbability}
\end{align}
Here the notation $\pr^*\left(X\right)$ means the outer probability of the event $X$.
\end{theorem}
Using Theorem~\ref{theoremOuterProba} for $p=1$, one can write, 
\begin{align}
\pr\left(\|T\|>t\right) &= \pr\left(\|\mathbb{E}\{U|\mathcal{U}\}\|>t\right)\\
&\leq c\pr\left(2\|U\|>t\right)\\
&=c \pr\left(2\|4\sum_{(\ell,n)}\sum_{\ell',n'} h(\bs Z_{\ell,n},\tilde{\bs Z}_{\ell',n'})\|\right)\\
&=c \pr\left(2\|4\sum_{(\ell,n)}\sum_{\ell',n'} h(\bs X_{\ell,n},\tilde{\bs X}_{\ell',n'})\|\right)
\end{align}
In the last line, we use the fact that $(\bs Z_{\ell,n}, \tilde{\bs Z}_{\ell',n'})$ and $(\bs X_{\ell,n}, \tilde{\bs X}_{\ell',n'})$ have the same distribution. This concludes the proof of proposition~\ref{decouplingProp}. 
\end{proof}

\subsection{\label{proofIndependentSecondOrder}Proof of lemma~\ref{lemmaSecondTermIndependent}}
Building upon the result of proposition~\ref{decouplingProp} and the discussion of section~\ref{sectionUstatistics}, we now prove the remaining two lemmas~\ref{lemmaSecondTermIndependent} and~\ref{convergenceCoherence}
\lemmaSecondTermIndependent*
\begin{proof}
Let $\bm X_{\ell,n}$ denote arbitrary matrices of size $L\times NK$ and let $\bm X[\ell,\sim n]$ denote the submatrix given by considering the $\ell^{th}$ row of $\bm X$ and the columns $(n-1)K+1\leq k\leq nK$. I.e 
\begin{align}
\hat{\bs X}[\ell,\sim n]\equiv \sum_{k=1}^K \bs f_{\ell}^*\bs X[:,(n-1)K+k](\bs e_n\otimes \bs e_k).
\end{align}
$\bs e_{k}$ is a $K$-dimensional zero vector with its $k^{th}$ entry set to $1$ and $\bs e_n$ is the $N$-dimensional zero vector with its $n^{th}$ entry set to $1$. The norm on the LHS of~\eqref{tailBoundTperDecoupled} can expand as
\begin{align}
\left\|\sum_{(\ell,n)}\mathcal{Z}_{\ell,n}\right\|&= \left\|\sum_{(\ell,n)}\bm  A_{\ell,n}\langle \bm A_{\ell,n},\bm X_{\ell,n}\rangle - \mathbb{E}\bm A_{\ell,n}\langle \bm A_{\ell,n},\bm X_{\ell,n} \rangle\right\|\\
&= \left\|\sum_{(\ell,n)} \bm A_{\ell,n}\langle \bm A_{\ell,n},\bm X_{\ell,n}\rangle - \mathbb{E}\bm A_{\ell,n}\langle \bm A_{\ell,n},\bm X_{\ell,n} \rangle\right\|
\end{align}
We will use proposition~\eqref{bernstein} to derive the tail bound of lemma~\ref{lemmaSecondTermIndependent}. To use this proposition, we start by deriving the bound for the variance. Note that the $\mathcal{Z}_{\ell,n}$ are not hermitian as we have $\bm A_{\ell,n}\bm A_{\ell,n}^* = \|\bs c_{\ell,n}\|^2\bs f_\ell\bs f_\ell^*$ and $\bm A_{\ell,n}^*\bm A_{\ell,n}^* = \|\bs f_\ell\|^2\bs c_{\ell,n}\bs c_{\ell,n}^*$. The first variance bound gives
\begin{align}
\left\|\sum_{(\ell,n)}\mathbb{E}\bm Z_{\ell,n}\bm Z_{\ell,n}^*\right\|& = \left\|\sum_{(\ell,n)}\mathbb{E}\bm A_{\ell,n}\bm A_{\ell,n}^* |\langle \bm A_{\ell,n}, \bm X_{\ell,n}\rangle |^2 - \left|\mathbb{E}\bm A_{\ell,n}\langle \bm A_{\ell,n}, \bm X_{\ell,n}\rangle \right|^2\right\|\nonumber\\
&\leq \left\|\sum_{(\ell,n)}\mathbb{E}\bm A_{\ell,n}\bm A_{\ell,n}^* |\langle \bm A_{\ell,n}, \bm X_{\ell,n}\rangle |^2\right\|\nonumber\\
&\leq \left\|\sum_{(\ell,n)} \mathbb{E}\bs f_\ell\bs f_\ell^*\|\bs c_{\ell,n}\|^2 |\langle \bs c_{\ell,n}, \hat{\bm X}_{\ell,n}[\ell,\sim n]\rangle |^2\right\|\nonumber\\
&\leq K\left\|\sum_{(\ell,n)}\bs f_\ell\bs f_\ell^* \|\hat{\bm X}_{\ell,n}[\ell,\sim n]\|^2\right\|\label{boundlemma2variance1}
\end{align}
 Now using the coherences defined in~\eqref{coherenceXmu} to~\eqref{coherenceXnu} and their supremas defined in~\eqref{supremasCoherence}, the last line of~\eqref{boundlemma2variance1} can be reduced to 
\begin{align}
\left\|\sum_{(\ell,n)\in \Gamma_p}\mathbb{E}\bm Z_{\ell,n}\bm Z_{\ell,n}^*\right\|&\leq K\cdot \frac{\overline{\mu}^2_{\bm X}}{L}\label{boundlemma2variance1b}
\end{align}
For the second term, a similar argument gives,
\begin{align}
\left\|\mathbb{E}\sum_{(\ell,n)}  \bm Z_{\ell,n}^*\bm Z_{\ell,n}\right\|& = \left\|\mathbb{E}\sum_{(\ell,n)\in \Gamma_p}\|\bs f_\ell\|^2\bs c_{\ell,n}\bs c_{\ell,n}^* |\langle \bs f_\ell\bs c_{\ell,n}, \bm X_{\ell,n}\rangle |^2 \right\|,\nonumber \\
&\leq \nonumber \left\| \mathbb{E}\sum_{(\ell,n)}\|\bs f_\ell\|^2\bs c_{\ell,n}\bs c_{\ell,n}^* |\langle\bs c_{\ell,n}, \hat{\bm X}_{\ell,n}[\ell,\sim n]\rangle |^2\right\|\\
&\lesssim \left\|\sum_{(\ell,n)}\|\hat{\bm X}_{\ell,n}[\ell,\sim n]\|^2 (\bs e_n\bs e_n^*)\otimes \bm I_K\right\|\nonumber + \left\|\sum_{(\ell,n)}\hat{\bm X}_{\ell,n}^T[\ell,\sim n]\overline{\hat{\bm X}}_{\ell,n}[\ell,\sim n]\right\|\nonumber \\
&\lesssim \left( \frac{\overline{\rho}^2_{\bm X}}{N}\right). \label{boundlemma2variance2}
\end{align}
The last line comes from orthogonality of the $\bm X[\ell,\sim n_1]$ for $\bm X[\ell,\sim n_2]$ for distinct $n_1,n_2$. Combining~\eqref{boundlemma2variance1b} with~\eqref{boundlemma2variance2}, we get the following bound on $\sigma$,
\begin{align}
\sigma \lesssim \left( \frac{\overline{\mu}^2_{\bm X}K}{L} + \frac{\overline{\rho}^2_{\bm X}}{N}\right)
\end{align}
We now compute the bound on the $\Psi_1$ norm. Following the same argument as in previous lemmas, we get 
\begin{align}
\|\|\bm A_{\ell,n}\langle \bm A_{\ell,n}, \bm X_{\ell,n} \rangle\|\|_{\psi_1}& \leq \|\bs f_\ell\|\|\|\bs c_{\ell,n}\|\|_{\psi_2} \||\langle \bs c_{\ell,n},\hat{\bm X}_{\ell,n}[\ell,\sim n]\rangle |\|_{\psi_2} \nonumber \\
&\leq  \sqrt{K}\|\bm X_{\ell,n}[\ell,\sim n]\|\\
& \leq  \sqrt{K}\frac{\overline{\nu}_{\bm X}}{\sqrt{LN}} \label{psi1tmpLemma2}
\end{align}
The last line follows from the discussion at the end of section~\ref{subgaussianAndSubexponential}. Substituting those bounds into proposition~\ref{bernstein}, we have
\begin{multline*}
\left\|\sum_{(\ell,n)} \mathcal{A}_{\ell,n}^*\mathcal{A}_{\ell,n}\bm X_{\ell,n} - \bm X_{\ell,n}\right\|\lesssim \\
\max\left\{\sqrt{\left(\frac{\overline{\mu}^2_{\bm X}K}{L}  + \frac{\overline{\rho}^2_{\bm X}}{N} \right)}\sqrt{t+\log(LKN)},  \frac{\overline{\nu}_{\bm X}}{\sqrt{LN}} \sqrt{K} \log(LKN) (t+ \log(LKN))\right\}\label{eq:BernsteinSum}
\end{multline*}
with probability at least $1-e^{-t}$. Taking $t = \beta\log(LN)$ concludes the proof of lemma~\ref{lemmaSecondTermIndependent}.

\end{proof}

\subsection{\label{proofconvergenceCoherence}Proof of lemma~\ref{convergenceCoherence}}

Again, we recall the statement of lemma~\ref{convergenceCoherence} for clarity before proving this lemma.

\convergenceCoherence* 

\begin{proof}
The proof is similar for any of the three relations on either $\nu^2_{\bm X_{\ell,n}}$, $\mu^2_{\bm X_{\ell,n}}$ or $\rho^2_{\bm X_{\ell,n}}$.  We thus only show the relation $\nu^2_{\bm X_{\ell,n}}\leq \mu_h^2\mu_m^2$. From the definition of $\nu_{\bm X_{\ell,n}}$, we have,
\begin{align}
&\nu^2_{\bm X_{\ell_1,n_1}}/LN = \sup_{\ell_2,n_2}\|\hat{\bm X}_{\ell_1,n_1}[\ell_2,\sim n_2]\|^2_F \\
& = \sup_{\ell_2,n_2} \sum_{k_2\sim n_2} \left|\sum_{(\ell,n)\neq (\ell_1,n_1)} \bs F \mathcal{P}_T (\bm A_{\ell,n})\langle \bm A_{\ell,n}, \bs h\bs m^*\rangle  - \mathbb{E}\sum_{(\ell,n)\neq (\ell_1,n_1)} \bs F \mathcal{P}_T (\bm A_{\ell,n})\langle \bm A_{\ell,n}, \bs h\bs m^*\rangle \right|^2_{\ell_2,k_2}\label{modulusToBound}
\end{align}
The notation above really means the modulus squared of the entry $(\ell_2,k_2)$ of the underlying matrix.
We start by bounding the sum of terms inside the modulus. Let us denote each of the terms of this sum by $z_{\ell,n}$. We have
\begin{align}
z_{\ell,n} = \left(\bs F\mathcal{P}_T(\bm A_{\ell,n})\langle \bm A_{\ell,n}, \bs h\bs m^*\rangle - \mathbb{E}\bs F\mathcal{P}_T(\bm A_{\ell,n})\langle \bm A_{\ell,n}, \bs h \bs m^*\rangle\right)_{(\ell_2,k_2)}. 
\end{align}
for fixed $(\ell,n)\neq (\ell_1,n_1)$, $\ell_2,n_2$ and $k_2$. We will use proposition~\ref{bernstein} to bound the sum $\sum_{(\ell,n)\neq (\ell_1,n_1)}z_{\ell,n}$. We start by deriving a bound on the variance of the $z_{\ell,n}$. 
\begin{align*}
\mathbb{E}|z_{\ell,n}|^2 & = \mathbb{E}\left|\bs F\mathcal{P}_T(\bm A_{\ell,n})\langle \bm A_{\ell,n}, \bs h\bs m^*\rangle\right|^2_{\ell_2,k_2} -  \left|\mathbb{E}\bs F\mathcal{P}_T(\bm A_{\ell,n})\langle \bm A_{\ell,n}, \bm X\rangle\right|^2_{\ell_2,k_2}\\
&\leq \mathbb{E}\left|\bs F\mathcal{P}_T(\bm A_{\ell,n})\langle \bm A_{\ell,n}, \bm X\rangle\right|^2_{\ell_2,k_2}.
\end{align*}
Let us use $\hat{\mathcal{P}}_T$ to denote the Fourier transform of the projector onto the tangent space, $\hat{\mathcal{P}}_T = \bs F\mathcal{P}_T$. From definition~\eqref{projectorTBD}, we have
\begin{align*}
&\left|\hat{\mathcal{P}}_T(\bm A_{\ell,n})\langle \bm A_{\ell,n}, \bs h\bs m^*\rangle\right|^2_{\ell_2,k_2}\leq
\left|\hat{\mathcal{P}}_T(\bm A_{\ell,n})\right|^2_{(k_2,\ell_2)} |\langle \bs c_{\ell,n},\hat{h}[\ell]\bs m_n\rangle |^2\\
 &\leq \left||\hat{h}[\ell]|h[\ell_2]\bs c_{\ell,n}[k_2] + \bs e_\ell[\ell_2]\langle \bs c_{\ell,n}, \bs m_n \rangle \bs m^*[k_2] - 
\hat{h}[\ell_2]\hat{h}[\ell]\langle \bs c_{\ell,n}, \bs m\rangle \bs m^*[k_2]\right|^2|\langle \bs c_{\ell,n},\hat{h}[\ell]\bs m^*_n\rangle |^2\\
&\lesssim  \left(|\hat{h}[\ell]|^2|\hat{h}(\ell_2)|^2|\bs c_{\ell,n}[k_2]|^2\right)|\langle \bs c_{\ell,n}, \hat{h}[\ell]\bs m^*_n\rangle |^2\\
&+ \left(|\bs e_\ell[\ell_2]|^2|\langle \bs c_{\ell,n},\bs m_n\rangle |^2|\bs m_{n_2}[k_2]|^2\right)|\langle \bs c_{\ell,n}, \hat{h}[\ell]\bs m^*_n\rangle |^2\\
&+  \left(|\hat{h}[\ell_2]|^2|\hat{h}[\ell]|^2|\langle \bs c_{\ell,n}, \bs m\rangle |^2|\bs m_{n_2}[k_2]|^2\right)|\langle \bs c_{\ell,n}, \hat{h}[\ell]\bs m^*_n\rangle |^2.
\end{align*}
Taking the expectation gives,
\begin{align*}
\mathbb{E}|z_{\ell,n}|^2&\lesssim |\hat{h}[\ell]|^2 |\hat{h}[\ell_2]|^2\delta(n,n_2)|\hat{h}[\ell]\bs m_n^*[k_2]|^2+ \delta(\ell_2,\ell)\|\bs m_n\|^2 |\bs m_n[k_2] |^2\|\hat{h}[\ell]\bs m_n^*\|^2\\
&+ |\hat{h}[\ell_2]|^2|\hat{h}[\ell]|^2\|\bs m_n\|^2|\bs m_n[k_2]|^2\|\hat{h}[\ell]\bs m_n^*\|^2,\\
&\lesssim \left(|\hat{h}[\ell]|^2 |\hat{h}[\ell_2]|^2\delta(n,n_2)+ \delta(\ell_2,\ell)\|\bs m_n\|^2 |\bs m_n[k_2] |^2 + |\hat{h}[\ell_2]|^2|\hat{h}[\ell]|^2\|\bs m_n\|^2|\bs m_n[k_2]|^2\right)\frac{\mu_h^2\mu_m^2}{LN}
\end{align*}
Summing over the measurements, we get 
\begin{align}
\sigma^2 &= \sum_{(\ell,n)\neq (\ell_1,n_1)}\mathbb{E}|z_{\ell,n}|^2\lesssim  (|\hat{h}[\ell_2]|^2 + |\bs m[k_2]|^2)\frac{\mu_h^2\mu_m^2}{LN}\label{boundVarianceCoherenceLemma}
\end{align}
In order to apply proposition~\ref{bernstein}, we are left with computing the bound on the $\bs \Psi_1$ norm of the $z_{\ell,n}$. Again, we can write
\begin{align*}
&\left|\bs F\mathcal{P}_T(\bm A_{\ell,n})\langle \bm A_{\ell,n}, \bs h\bs m^*\rangle\right|_{\ell_2,k_2}\\
&=\left| \left(\hat{h}[\ell]\hat{h}[\ell_2]\bs c_{\ell,n}[k_2]+ \bs e_{\ell}[\ell_2]\langle \bs c_{\ell,n}, \bs m\rangle m[k_2] - \hat{h}[\ell_2]\hat{h}[\ell]\langle \bs c_{\ell,n},\bs m\rangle m[k_2]\right) \langle \bs c_{\ell,n},\hat{h}[\ell]\bs m_n \rangle \right| \\
&\lesssim   \left|\hat{h}[\ell]\hat{h}[\ell_2]\bs c_{\ell,n}[k_2]\langle \bs c_{\ell,n},\hat{h}[\ell]\bs m_n \rangle \right|\\
&+\left| \bs e_{\ell}[\ell_2]\langle \bs c_{\ell,n}, \bs m\rangle m[k_2] \langle \bs c_{\ell,n},\hat{h}[\ell]\bs m_n \rangle \right|\\
& + \left|\hat{h}[\ell_2]\hat{h}[\ell]\langle \bs c_{\ell,n},\bs m\rangle m[k_2] \langle \bs c_{\ell,n},\hat{h}[\ell]\bs m_n \rangle \right|
\end{align*}
Each of the terms above can be bounded by again using the fact that the product of two subgaussian random variables follows a subexponential distribution and that the Orlicz one norm of this product can be upper bounded by the product of the Orlicz-$2$ norms of each of the subgaussian random variables. In other words, $\|XY\|_{\bs \Psi_1}\leq \|X\|_{\bs \Psi_2}\|Y\|_{\bs \Psi_2}$. For each of the three terms in the expression of $z_{\ell,n}-\mathbb{E}z_{\ell,n}$, we can use the fact that $\|\bs c_{\ell,n}[k']\|_{\bs \Psi_2}\leq c$, $\langle \bs c_{\ell,n}, \bs m_n\rangle$is a mixture of centered gaussians. Note that the corresponding variable $|\langle \bs c_{\ell,n}, \bs m_n\rangle|^2$ is a chi-squared that is sub-exponential with parameters $(2, 4\|\bs m_n\|^2)$ as we saw in the proof of lemma~\ref{lemmaT} and using proposition~\ref{EqNormPsi}, we get,
\begin{align}
\|\left|\mathcal{P}_T(\bm A_{\ell,n})\langle \bm A_{\ell,n},\bs h\bs m^*\rangle\right|_{\ell_2,n_2}\|_{\bs \Psi_1}\lesssim \left(|\hat{h}[\ell]\hat{h}[\ell_2]| + e_\ell[\ell_2] m[k_2]\|\bs m_n\|\right)\frac{\mu_h\mu_m}{\sqrt{LN}}\label{boundPsiNormtt}
\end{align}
Finally for the expectation, simply recall that $\|X - \mathbb{E}X\|_{\psi_q}\leq 2\|X\|_{\psi_q}$
%%
%\begin{align*}
%\|\left|\mathbb{E}\bs F\mathcal{P}_T(\bm A_{\ell,n})\langle \bm A_{\ell,n}, \bs h\bs m^*\rangle\right|_{\ell_2,k_2}\|_{\bs \Psi_1}&\lesssim \|\left|\mathbb{E}\bs F\mathcal{P}_T(\bm A_{\ell,n})\langle \bm A_{\ell,n}, \bs h\bs m^*\rangle\right|_{\ell_2,k_2}\|\\
%&\lesssim \frac{\mu_m\mu_h}{\sqrt{LN}}\label{boundPsiExp}
%\end{align*}
%Where we used the following relations,
%
%\begin{align}
%&|\mathbb{E}\left(\bs h\hat{h}[\ell]\bs c_{\ell,n}^* + \bs f_\ell\langle \bs m_n,\bs c_{\ell,n}\rangle \bs m^* + \bs h\hat{h}[\ell]\langle \bs c_{\ell,n}^*,\bs m_n\rangle\bs m^*\right)\hat{h}[\ell]\langle \bs m_n,\bs c_{\ell,n}\rangle|_{\ell_2,k_2}\\
%& \lesssim  ||\hat{h}[\ell]|^2\hat{h}[\ell_2]\bs m_n^*[k_2]\delta_{n,n_2} + \delta_{\ell,\ell_2}\hat{h}[\ell]m[k_2]\|\bs m_n\|^2 + \hat{h}[\ell_2]\frac{\mu_h^2}{L}\frac{\mu^2_m}{N}m[k_2]|\\
%&\lesssim \frac{\mu_h\mu_m}{\sqrt{LN}}\left(|\hat{h}[\ell]\hat{h}[\ell_2]\delta_{n,n_2}| + |\delta_{\ell,\ell_2}m[k_2]|
%\|\bs m_n\| + |m[k_2]||\hat{h}[\ell_2]|\|\bs m_n\||\hat{h}[\ell]|\right)
%\end{align}
 
Using the bounds~\eqref{boundPsiNormtt} and~\eqref{boundVarianceCoherenceLemma}, one can now apply proposition~\ref{bernstein} to derive the tail bound guaranteeing that each each term within the modulus of~\eqref{modulusToBound} remains sufficiently small. Applying proposition~\ref{bernstein} gives, 
\begin{align*}
&\left|\sum_{(\ell,n)\neq (\ell_1,n_1)}\bs F\mathcal{P}_T(\bm A_{\ell,n})\langle \bm A_{\ell,n}, \bs h\bs m^*\rangle -  \mathbb{E}\bs F\mathcal{P}_T(\bm A_{\ell,n})\langle \bm A_{\ell,n}, \bs h\bs m^*\rangle\right|_{\ell_2,k_2}\\
&\leq \beta\max\left\{\left(\sqrt{(|\hat{h}[\ell_2]|^2 + |\bs m[k_2]|^2)}\right), \log(N\vee L)\left(|\hat{h}[\ell]\hat{h}[\ell_2]| + e_\ell[\ell_2] m[k_2]\|\bs m\|\right)\right\}\frac{\mu_h\mu_m}{\sqrt{LN}}
\end{align*}
with probability at least $1 - c_3(LN)^{-\beta}$ for some constant $c_3$. We then square the modulus and take the sum over the indices $k_2\sim n_2$, which gives 
\begin{align}
\sup_{\ell_2,n_2}\sum_{k_2\sim n_2}|\sum_{(\ell,n)\neq (\ell_1,n_1)}&z_{\ell,n}|^2_{\ell_2,k_2}\lesssim \\
&\max\left\{K \|\hat{\bs h}\|^2_{\infty} + \|\bs m_{n_2}\|^2,K\log^2(N\vee L)\left(\|\hat{\bs h}\|^2_{\infty} + |m[k_2]|\|\bs m_n\|\right)^2\right\}\frac{\mu_h^2\mu_m^2}{LN}.\label{beforelastBoundCoherenceNu}\\
\lesssim &\max\left\{K\|\hat{\bs h}\|^2_{\infty} + \|\bs m_{n_2}\|^2, K\|\bs h\|^4_{\infty} + \|\bs m_n\|^2\|\bs m_{n_2}\|^2\right\}\frac{\mu_h^2\mu_m^2}{LN}\label{lastBoundCoherenceNu}
\end{align}
With probability at least $1 - 2c_3(LN)^{-\beta}$. In~\eqref{beforelastBoundCoherenceNu}, we used $\sum_{k_2=1}^K|m_{n_2}[k_2]|^2 = \|\bs m_{n_2}\|^2$

The bound~\eqref{lastBoundCoherenceNu} can be made smaller than any constant from the coherences $\mu^2_h/L = \|\hat{\bs h}\|_\infty^2$ and $\mu^2_m/N =\max_{n}\|\bs m_n\|^2$ as soon as $L\gtrsim K\mu^2_h$ and $N\gtrsim \mu^2_m $ up to log factors. This concludes the proof of lemma~\ref{convergenceCoherence}.

\end{proof}

\subsection{\label{sectionExpectationSwap}Proof of lemma~\ref{lemmaLemmaExpectationSwap}}

\lemmaLemmaExpectationSwap*

\begin{proof}
We use proposition~\ref{bernstein}. The terms within the norm~\eqref{normLemmaLemmaExpectationSwap} expand as
\begin{align}
&\sum_{\ell,n} \bs f_\ell \left(|\hat{h}[\ell]|^2\tilde{\bs c}^*_{\ell,n} + \|\bs f_\ell\|^2\langle \bs m_n , \tilde{\bs c}^*_{\ell,n}\rangle\bs m_n^* - |\hat{h}[\ell]|^2\langle \bs m_n , \tilde{\bs c}^*_{\ell,n}\rangle\bs m_n^*\right)|\hat{h}[\ell]|\langle \tilde{\bs c}_{\ell,n}, \bs m_n\rangle\\
&-\mathbb{E}\bs f_\ell \left(|\hat{h}[\ell]|^2\tilde{\bs c}^*_{\ell,n} + \|\bs f_\ell\|^2\langle \bs m_n , \tilde{\bs c}^*_{\ell,n}\rangle\bs m_n^* - |\hat{h}[\ell]|^2\langle \bs m_n , \tilde{\bs c}^*_{\ell,n}\rangle\bs m_n^*\right)|\hat{h}[\ell]|\langle \tilde{\bs c}_{\ell,n}, \bs m_n\rangle
\end{align}
We will bound each of the first two term through proposition~\ref{bernstein}. The bound on the last term follows from the first two. For each of those two terms, letting aside the expectation, the first and second terms in the variance bound can be expressed as 
\begin{align}
&\left\|\mathbb{E}\sum_{\ell,n} \bs f_\ell\bs f_\ell^* |\hat{h}[\ell]|^4 \|\tilde{\bs c}_{\ell,n}\|^2 |\hat{h}[\ell]|^2|\langle \tilde{\bs c}_{\ell,n}, \bs m_n\rangle|^2\right\|\lesssim \frac{K\mu_h^6}{L^3}\frac{\mu_m^2}{N}\\
&\left\|\mathbb{E}\sum_{\ell,n} \|\bs f_\ell\|^2 \bs c_{\ell,n}\bs c_{\ell,n}^* |\hat{h}[\ell]|^4  |\hat{h}[\ell]|^2|\langle \tilde{\bs c}_{\ell,n}, \bs m_n\rangle|^2\right\|\lesssim \frac{\mu_h^4}{L^2}(\|\bs m_n\|^2 + 1)
\end{align}
for the first term, and 
\begin{align}
&\left\|\mathbb{E}\sum_{\ell,n} \bs f_\ell\bs f_\ell^* \|\bs f_\ell\|^2|\langle \bs m_n , \tilde{\bs c}^*_{\ell,n}\rangle|^4 \|\bs m_n \|^2  |\hat{h}[\ell]|^2 \right\|\lesssim \frac{\mu_h^2}{L} \|\bs m_n\|^4\\
&\left\|\mathbb{E}\sum_{\ell,n} \|\bs f_\ell\|^4|\langle \bs m_n , \tilde{\bs c}^*_{\ell,n}\rangle|^4 \bs m_n \bs m_n^*  |\hat{h}[\ell]|^2\right\|\lesssim \frac{\mu_h^4}{L^2}(\|\bs m_n\|^2 + 1)
\end{align}
For the Orlicz norm, simply note that 
\begin{align}
\left\|\left\|\bs f_\ell \left(|\hat{h}[\ell]|^2\tilde{\bs c}^*_{\ell,n} \right)|\hat{h}[\ell]|\langle \tilde{\bs c}_{\ell,n}, \bs m_n\rangle\right\|\right\|_{\psi_1}&\lesssim \left\|\left\|\bs f_\ell \left(|\hat{h}[\ell]|^2\tilde{\bs c}^*_{\ell,n} \right)|\hat{h}[\ell]|\right\|\right\|_{\psi_2} \left\|\langle \tilde{\bs c}_{\ell,n}, \bs m_n\rangle\right\|_{\psi_2}\\
&\lesssim \|\bs m_n\|\sqrt{K}\frac{\mu_h}{\sqrt{L}}\frac{\mu_h^2}{L}
\end{align}
as well as 
\begin{align}
&\left\|\left\|\bs f_\ell \left(\|\bs f_\ell\|^2\langle \bs m_n , \tilde{\bs c}^*_{\ell,n}\rangle\bs m_n^*\right)|\hat{h}[\ell]|\langle \tilde{\bs c}_{\ell,n}, \bs m_n\rangle\right\|\right\|_{\psi_1}\nonumber \\
&\lesssim \left\|\left\|\bs f_\ell \left(\|\bs f_\ell\|^2\langle \bs m_n , \tilde{\bs c}^*_{\ell,n}\rangle\bs m_n^*\right)\right\|\right\|_{\psi_2}\left\||\hat{h}[\ell]|\langle \tilde{\bs c}_{\ell,n}, \bs m_n\rangle\right\|_{\psi_2}\\
&\lesssim \left\|\bs m_n\right\|^3\frac{\mu_h}{\sqrt{L}}
\end{align}
All of these terms can be made less than $\delta$ whenever $L\gtrsim (1/\delta)\beta\mu_h^2K$ and $N\gtrsim \mu_m^2$. The conclusion follows from applying proposition~\ref{bernstein} with $t=\beta\log(LN)$. 

\end{proof}

\subsection{\label{finNeumann}Proof of lemma~\ref{lemmafinNeumann}}

The proof of lemma~\ref{lemmafinNeumann} follows the approach in~\cite{candes2009exact}. simply note that we have,
\begin{align}
\|\mathcal{P}_T^\perp\left(\mathcal{A}^*\mathcal{A}\mathcal{P}_T\right)\left(\mathcal{P}_T - \mathcal{P}_T\mathcal{A}^*\mathcal{A}\mathcal{P}_T\right)^k\bs h\bs m^*\| &\leq \|\mathcal{P}_T^\perp\left(\mathcal{A}^*\mathcal{A}\mathcal{P}_T\right)\| \left\|\left(\mathcal{P}_T - \mathcal{P}_T\mathcal{A}^*\mathcal{A}\mathcal{P}_T\right)^k\bs h\bs m^*\right\| \\
&\leq \|\mathcal{P}_T^\perp\left(\mathcal{A}^*\mathcal{A}\mathcal{P}_T\right)\| \left\|\mathcal{P}_T - \mathcal{P}_T\mathcal{A}^*\mathcal{A}\mathcal{P}_T\right\|^k \|\bs h\bs m^*\|\\
&\leq \|\mathcal{A}^*\|\|\mathcal{P}_T(\mathcal{A})\| \left\|\mathcal{P}_T - \mathcal{P}_T\mathcal{A}^*\mathcal{A}\mathcal{P}_T\right\|^k \|\bs h\bs m^*\|
\end{align}
Now taking the sum and using the fact that $\|\mathcal{A}^*\mathcal{A}\|\leq \|\mathcal{A}\|^2\lesssim K\log(L\vee N)$ as well as 
$$|\|\mathcal{A}\mathcal{P}_T\|^2 - \|\mathcal{P}_T\||= |\|\mathcal{P}_T(\mathcal{A}^*\mathcal{A})\mathcal{P}_T\| - \|\mathcal{P}_T\||\lesssim \delta$$ 
gives,
\begin{align}
\sum_{k=k_0}^\infty \|\mathcal{P}_T^\perp\left(\mathcal{A}^*\mathcal{A}\mathcal{P}_T\right)\left(\mathcal{P}_T - \mathcal{P}_T\mathcal{A}^*\mathcal{A}\mathcal{P}_T\right)^k\bs h\bs m^*\| &\leq \sqrt{K\log(L\vee N)} \frac{\|\mathcal{P}_T - \mathcal{P}_T\mathcal{A}^*\mathcal{A}\mathcal{P}_T\|^{k_0}}{1 - \|\mathcal{P}_T - \mathcal{P}_T\mathcal{A}^*\mathcal{A}\mathcal{P}_T\|} \|\bs h\bs m^*\|
\end{align}
Now we can use the fact that lemma~\ref{lemmaT} gives $\|\mathcal{P}_T - \mathcal{P}_T\mathcal{A}^*\mathcal{A}\mathcal{P}_T\|\lesssim \sqrt{\mu_h^2K/L + \mu_m^2/N}\lesssim \delta$ as soon as $L\gtrsim K\mu_h^2$ and $N\gtrsim \mu_m^2$, the denominator is always larger than $1-\delta$. To conclude, it suffices to again apply lemma~\ref{lemmaT} on the numerator to get 
\begin{align}
\sum_{k=k_0}^\infty \|\mathcal{P}_T^\perp\left(\mathcal{A}^*\mathcal{A}\mathcal{P}_T\right)\left(\mathcal{P}_T - \mathcal{P}_T\mathcal{A}^*\mathcal{A}\mathcal{P}_T\right)^k\bs h\bs m^*\| &\leq  \sqrt{K\log(L\vee N)} \left(\frac{\mu_h^2K}{L} + \frac{\mu_m^2}{N}\right)^{k_0/2}\label{almostFinalBound}
\end{align}
In particular, as soon as $L\gtrsim K^{1+1/k_0}$ and $N\gtrsim K^{1/k_0}$, the bound~\eqref{almostFinalBound} reduces to 
\begin{align}
\sum_{k=k_0}^\infty \|\mathcal{P}_T^\perp\left(\mathcal{A}^*\mathcal{A}\mathcal{P}_T\right)\left(\mathcal{P}_T - \mathcal{P}_T\mathcal{A}^*\mathcal{A}\mathcal{P}_T\right)^k\bs h\bs m^*\| &\leq  C\label{almostFinalBound}
\end{align}
 for a constant $C$ which can be taken arbitrarly small so as to satisfy the bound on $\|\mathcal{P}_T^{\perp}(\bs Y)\|$. In our case, as $k_0 = 2$, we get the sample complexities $L\gtrsim K^{3/2}\mu_h^2$ and $N\gtrsim K^{1/2}\mu_m^2$.

\section{\label{sec:num}Numerical simulations}

In this section we provide additional numerical experiments and study the resulting phase diagrams in order to 
quantify the probability of success for the formulation~\eqref{nuclearNormDeconvolution}. To conduct the numerical experiments, we first write problem~\eqref{nuclearNormDeconvolution} as a trace minimization problem (see~\cite{fazel2001rank}). Introducing $\bm V$ as a proxy for the rank one matrix $\bm V_0$, 
$$\bm V = \left(\begin{array}{cc}
\bm V_{11} & \bm V_{12}\\
\bm V_{21} & \bm V_{22}
\end{array}\right)  \approx \bm V_0 = \left(\begin{array}{cc}
\bm h\bm h^\H & \bm h\bs m^\H\\
\bs m\bm h^\H & \bs m \bs m^\H\\
\end{array}\right),$$ 
one can write problem~\eqref{nuclearNormDeconvolution} as
\begin{align}\begin{split}
\min \quad & \text{Tr}(\bm V_{11}) + \text{Tr}(\bm V_{22})\\
s.t \quad & \mathcal{A}(\bm V_{12}) = \{ \vm{y}_n \},\quad \bm V\succeq 0.
\end{split}\label{nucnorm2bb}
\end{align}
To handle reasonably large instances of~\eqref{nuclearNormDeconvolutionLinearMap}, we follow the approach in~\cite{burer2003nonlinear} and replace the matrix $\bm V$ by a low rank factorization,
\begin{equation}
\bm V= \left[\begin{array}{c}
\bm R_{1}\\
\bm R_{2}\end{array}\right]
\left[\begin{array}{c}
\bm R_{1}\\
\bm R_{2}\end{array}\right]^\H = \bm R\bm R^\H\label{eq:factorization}
\end{equation}
where $\bm R_{1}\in \mathbb{C}^{L\times r}$ and $\bm R_2\in \mathbb{C}^{KN\times r}$ for some rank $r\ll L+KN$. Introducing multipliers $\lambda_{\ell,n}$ for each one of the constraints $\{\langle \bm A_{\ell,n}, \bm X\rangle  = \bs y_{n}[\ell]\}_{\ell,n}$ and $\sigma>0$, we then minimize the augmented Lagrangian $\mathcal{L}(\bm R)$ associated to~\eqref{nucnorm2bb} with respect to $\displaystyle \bm R_1$ and $\bm R_2$,
\begin{align}
\mathcal{L}(\bm R)=\frac{1}{2}\|\bm R_1\|_F^2+ \frac{1}{2}\|\bm R_2\|_F^2 + \sum_{\ell,n} \lambda_{\ell,n} ( \langle \bm A_{\ell,n},\bm R_1\bm R_2^\H\rangle  - \vm{y}_n[\l] )
+ \sum_{\ell,n}\frac{\sigma}{2}(\langle \bm A_{\ell,n},\bm R_1\bm R_2^\H\rangle -\vm{y}_n[\l] )^2.\label{augmentedLagrangian}
\end{align}
In addition to dealing with fewer variables, the factorization introduced in~\eqref{eq:factorization} is also a very convenient way to circumvent the constraint $\bm V\succeq 0$ since $\bm V\succeq 0 \Leftrightarrow \bm V = \bm Y\bm Y^\H$ for some $\bm Y$. In particular, taking $r=L+KN$ is equivalent to solving the original semidefinite program~\eqref{nuclearNormDeconvolutionLinearMap}. As explained in~\cite{burer2003nonlinear}, one can thus minimize~\eqref{augmentedLagrangian} for a very small rank $r\ll L+KN$, check optimality with respect to the original problem and then increase the rank in case the rank-$r$ factorization doesn't lead to the solution of the original problem. In practice we don't even need to do so and setting $r=4$ is enough to achieve convergence to the minimizer of~\eqref{nuclearNormDeconvolutionLinearMap}. 

\subsection{\label{sectionPhaseTrans}Phase transition}

The success rates for different values of the parameters $L$, $N$ and $K$ are shown in Fig.~\ref{figure1}. 
The first set of numerical experiments, corresponding to the first diagram of Fig.~\ref{figure1}, shows the rate of success as a function of both the size of the input subspace 
$K$ and the size of the ambient space $L$, the number of input signals $N$ being set to $40$. For each of the values of the pairs $(K,L)$, $100$ 
experiments were run by taking gaussian i.i.d sensing matrices $\bm C_{n}$ with $\bm C_{n}[\ell,k]\sim \mathcal{N}(0,1)$, as well as gaussian i.i.d vectors $\bs h$ and $\bs m$. For each of the $100$ experiments, we ran $40$ iterations of 
the limited memory BFGS on the augmented Lagrangian~\eqref{augmentedLagrangian}. An experiment was classified as a success (white/1)  as opposed to failure (black/0) whenever 
the relative difference between the obtained matrix $\bm X$ and the optimal solution $\bm X_0 = \bs h\bs m^{\H}$ was less than $2\%$. In other words, 
\begin{equation}
\frac{\|\bm X - \bs h\bs m^{\H}\|_F}{\|\bs h\bs m^{\H}\|_F}<.02
\end{equation}
The second diagram of Figure~\ref{figure1} shows the rate of success as a function of $K$ and $N$ for the same experimental framework. Here $L$ is fixed to $800$.

The third diagram shows the rate of success for various values of $L$ and $N$ for a fixed $K = 40$. It is worth noting that as long as there is a sufficient number of columns ($>50$), the recovery mainly depends on the parameters $K$ and $L$.

Finally, we provide a phase diagram highlighting the independence of the recovery vis a vis the sparsity of the impulse response $\bs h$ (Sparsity was studied as an additional condition on the recovery for example in~\cite{ahmed2016leveraging,ling2015self, li2016identifiability}). Figure~\eqref{figureSparse} shows that for fixed $L$ and $N$, the sparsity has no influence on the recovery which is essentially driven by the subspace dimension. This observation further motivates the extension to blind super-resolution discussed below.

\begin{figure}[h!]
  \begin{minipage}{0.5\textwidth}
    \hspace{0.4cm}\includegraphics[trim = 0cm 7cm 0cm 6cm, clip, width=\textwidth]{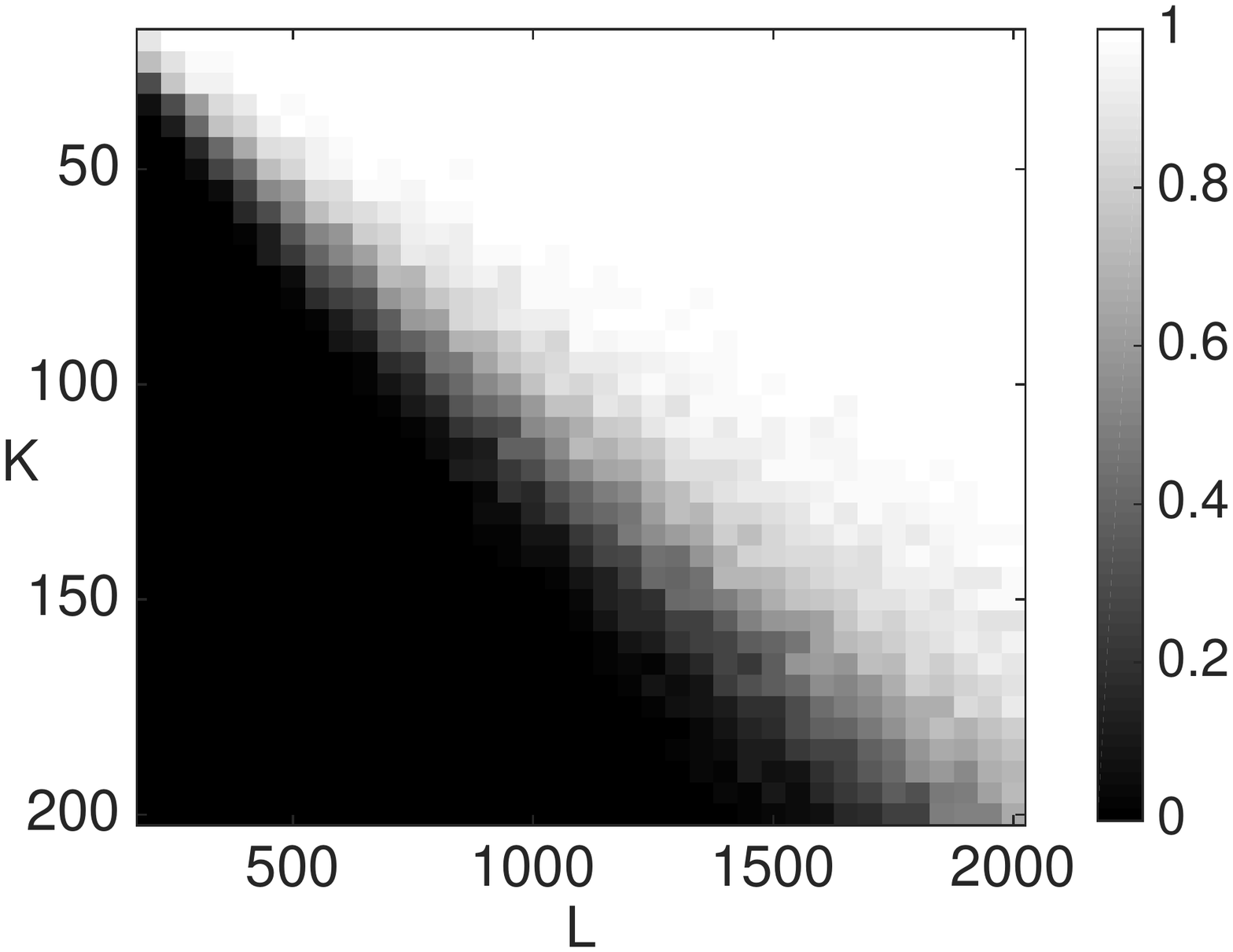}
  \end{minipage}
  \begin{minipage}{0.5\textwidth}
    \includegraphics[trim = 0cm 7cm 0cm 6cm, clip, width=\textwidth]{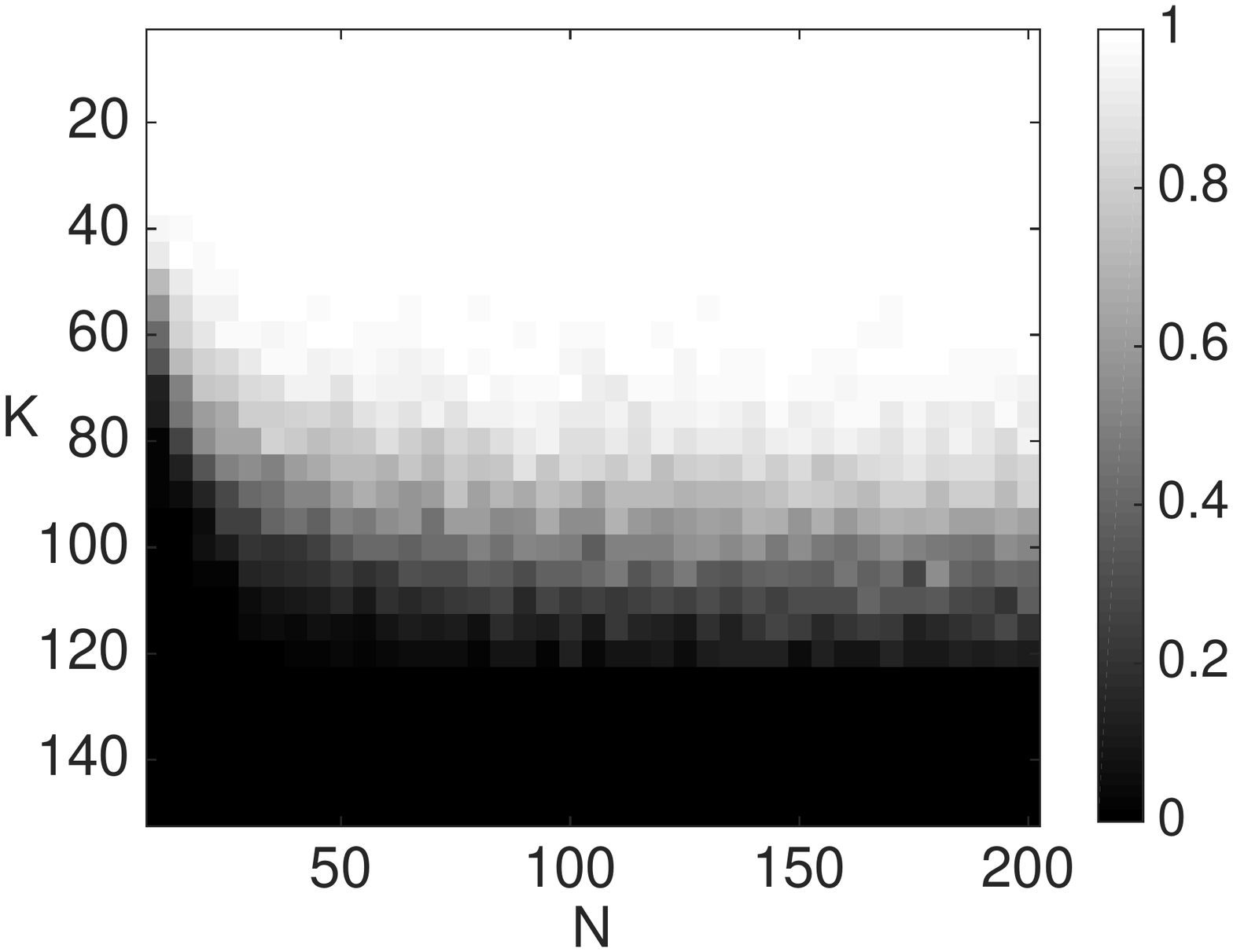}
  \end{minipage}\\
  \begin{minipage}{0.5\textwidth}
   \hspace{0.35cm}\includegraphics[trim = 0cm 7cm 0cm 6cm, clip, width=\textwidth]{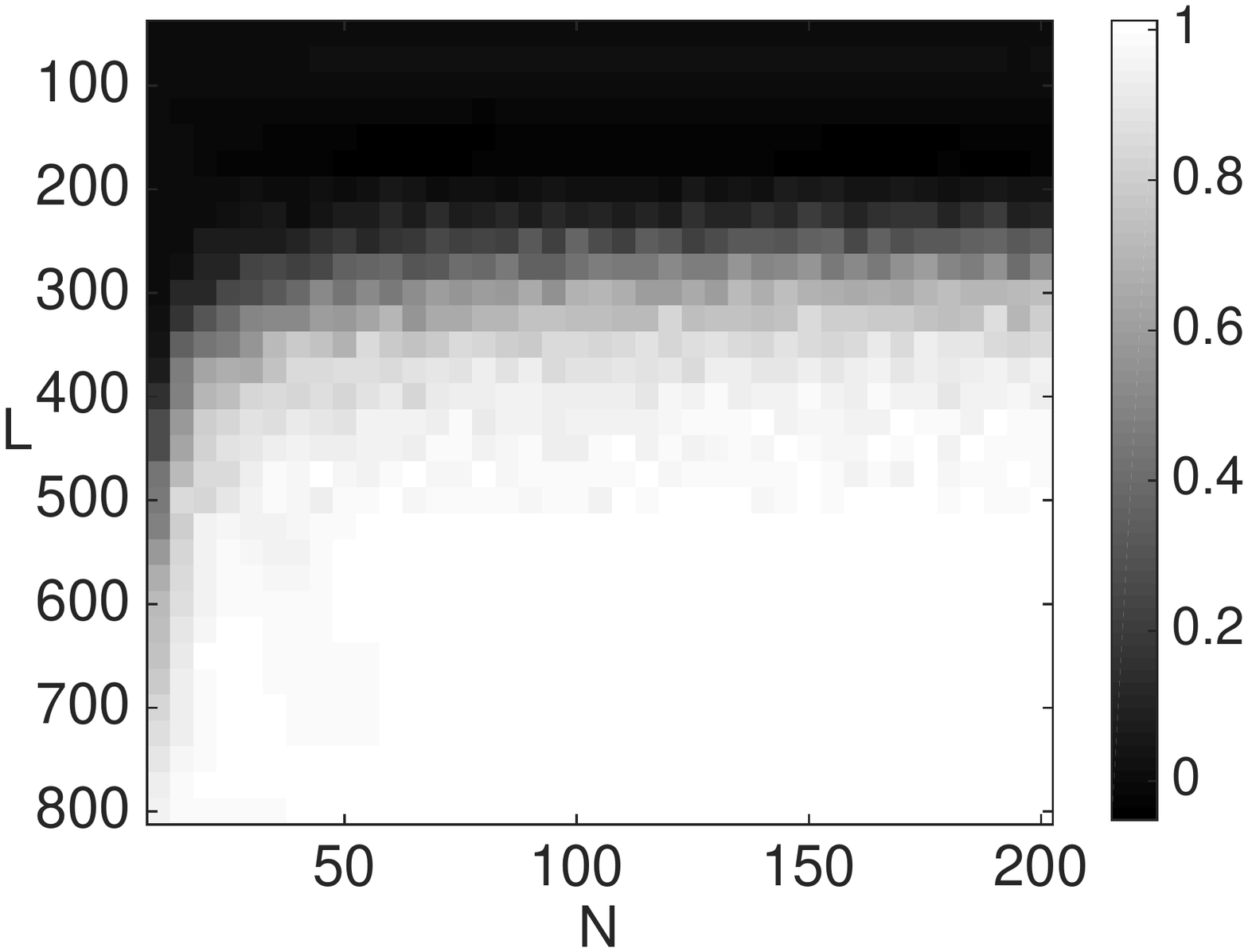}
  \end{minipage}
\caption{\label{figure1}Recovery of $\hh\mm^\H$ from nuclear norm minimization for (top left) various values of $K$ and $L$, (top right) various values of $K$ and $N$, and (bottom) various values of $L$ and $N$. White (1) is used to indicate success while black (0) is used to indicate failure. From the first figure (top left), the empirical recovery rate is seen to increase when $L$ increases and when $K$ decreases, with an expected phase transition when $L$ is a multiple of $K$. From the remaining two figures, we see that a minimum (small) value of $N$ is needed for recovery, but past this threshold, $K$ and $L$ are the only parameters affecting recovery. All these phase transitions can be explained by noting that the number of unknowns in the original problem is given by $KN + L$ whether the number of measurements is given by $LN$.}
\end{figure}

\begin{figure}[h!]\centering
  \begin{minipage}{0.8\textwidth}
    \hspace{0.6cm}\includegraphics[trim = 0cm 4.5cm 0cm 11cm, clip, width=.9\textwidth]{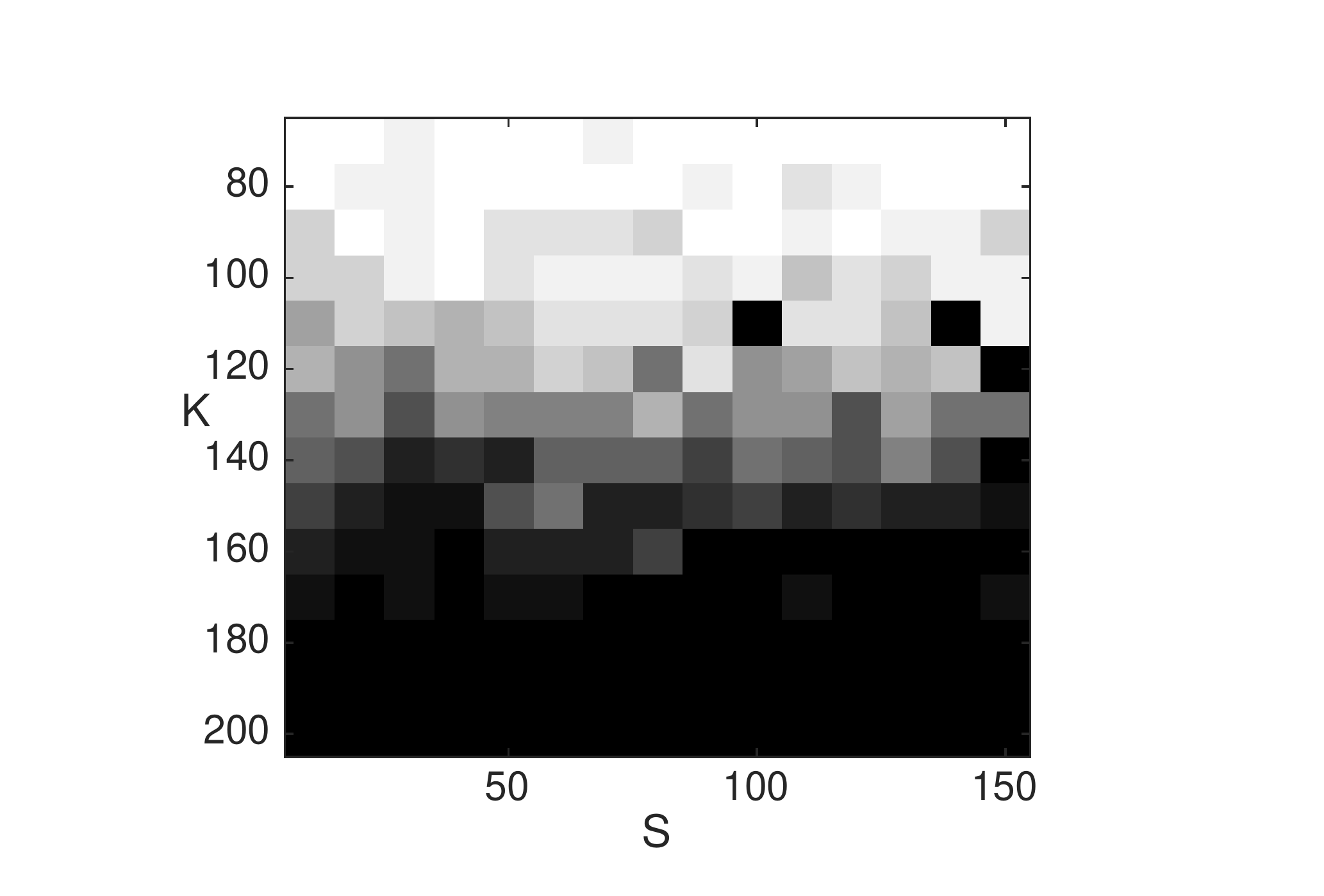}
  \end{minipage}
\caption{\label{figureSparse}Recovery of $\hh\mm^\H$ from nuclear norm minimization for various values of $K$ and $S$-sparse vectors $h$ with randomized supports. White (1) is used to indicate success while black (0) is used to indicate failure. The phase transition here illustrates the fact that recovery through the nuclear norm minimization~\eqref{nuclearNormDeconvolutionLinearMap} is essentially driven by the dimension of the inputs subspace as well as the coherences~\eqref{coherenceh} and~\eqref{coherencem} more than the sparsity $S$ of the filter.}
\end{figure}

\subsection{\label{certificateConvergence}Convergence of the certificate}

In this section, we provide numerial evidence supporting the choice of the ansatz~\eqref{candidateCert0}. As explained in section~\ref{sectionCandidateDualCert}, the most natural certificate construction as soon as the normal map $\mathcal{A}^*\mathcal{A}$ concentrates to the identity sufficiently fast, would be to use $\bs Y_1 = \mathcal{A}^*\mathcal{A}\bs h\bs m^*$. In practice, however, such a certificate performs poorly and does not achieve the sample complexity observed empirically in Fig.~\ref{figure1}. In particular, it requires poor ratios $L/K$ and $N/K$ in order for the condition $\mathcal{P}_T(\bs Y_1)$ to be met in~\eqref{dualCertConditions}. 

As soon as the map $\mathcal{P}_T\mathcal{A}^*\mathcal{A}\mathcal{P}_T$ can be shown to be injective, the second ansatz $\bs Y_2 = \mathcal{A}^*\mathcal{A}(\mathcal{P}_T\mathcal{A}^*\mathcal{A}\mathcal{P}_T)^{-1}\bs h\bs m^*$ provides a much better candidate to satisfy the first condition in~\eqref{dualCertConditions}.  In this section, we compare the two ansatz and provide numerical evidence that the conditions~\eqref{dualCertConditions} are met for both, albeit in different sample complexity regimes.

Fig.~\ref{certificateConvergence1a} first compares the projection of the two ansatz on the tangent space $T$ for the particular choice of dimensions given by $K=8$, $L = 80$ and $N=20$, and in a framework where the minimization program~\eqref{nuclearNormDeconvolutionLinearMap} is known to recover the exact solution. In this regime, the direct ansatz fails to meet the condition in~\eqref{dualCertConditions}. The use of the inverse however enables the second ansatz to satisfy the first condition exactly, in passing highlighting the fact that in practice, the constant in the first condition of~\eqref{dualCertConditions} can be taken arbitrarily large. The norms corresponding to the first condition are respectively given by $\|\mathcal{P}_T(\bs Y_1)-\bs h\bs m^*\|\approx 0.48$ and $\|\mathcal{P}_T(\bs Y_2)-\bs h\bs m^*\|\approx 1e^{-13}$.

Fig.~\ref{certificateConvergence1c} then compares the projection of the two ansatz onto the orthogonal complement $T^\perp$. The two ansatz exhibit comparable magnitude on the orthogonal complement $T^\perp$ and both of them satisfy the second condition in~\eqref{dualCertConditions} although the direct certificate seems to perform slightly better for the particular dimensions considered when considering this second condition. The norms corresponding to the second condition in~\eqref{dualCertConditions} are given respectively by $\|\mathcal{P}_T^\perp(\bs Y_1)\|\approx 0.67$ and $\|\mathcal{P}_T^\perp(\bs Y_2)\|\approx 0.72$, thus showing that for the particular choice $K=8$, $L = 80$ and $N=20$, the first candidate certificate fails to satisfy the optimality conditions corresponding to a zero subgradient whether the second candidate certificate meets both of these conditions.

\begin{figure}\centering
\input{CertificateGolfingTFinal.tex}  
\caption{\label{certificateConvergence1a}Convergence of the dual certificates on $T$ for (Top) the direct construction $\bs Y_1 = \mathcal{A}^*\mathcal{A}(\bs h\bs m^*)$ and (bottom) the construction involving the inverse, $\bs Y_2 = \mathcal{A}^*\mathcal{A}\mathcal{P}_T(\mathcal{P}_T\mathcal{A}^*\mathcal{A}\mathcal{P}_T)^{-1}\bs h\bs m^*$. Here we take $K = 8$, $N = 20$ and $L= 80$. In this framework, the nuclear norm minimization program~\eqref{nuclearNormDeconvolutionLinearMap} is known to recover the exact solution. The certificates have been downsampled by a factor $50$ for clarity.}
\end{figure}

%\begin{figure}\centering
%\input{CertificateGolfingTperpFinal.tex}  
%\caption{\label{certificateConvergence1b}Convergence of the dual certificates on $T^\perp$ for (Top) the direct construction $\bs Y_1 = \mathcal{A}^*\mathcal{A}(\bs h\bs m^*)$ and (bottom) the construction involving the inverse, $\bs Y_2 = \mathcal{A}^*\mathcal{A}\mathcal{P}_T(\mathcal{P}_T\mathcal{A}^*\mathcal{A}\mathcal{P}_T)^{-1}\bs h\bs m^*$. Here we take $K = 8$, $N = 20$ and $L= 80$. The certificates have been downsampled by a factor $50$ for clarity.}
%\end{figure}

\begin{figure}\centering
\input{ComparisonTperpFinal.tex}  
\caption{\label{certificateConvergence1c}Convergence of the dual certificates on $T^\perp$ for the direct construction $\bs Y_1 = \mathcal{A}^*\mathcal{A}(\bs h\bs m^*)$ and the construction involving the inverse, $\bs Y_2 = \mathcal{A}^*\mathcal{A}\mathcal{P}_T(\mathcal{P}_T\mathcal{A}^*\mathcal{A}\mathcal{P}_T)^{-1}\bs h\bs m^*$. Here we take $K = 8$, $N = 20$ and $L= 80$. In this framework, the nuclear norm minimization program~\eqref{nuclearNormDeconvolutionLinearMap} is known to recover the exact solution. The certificates have been downsampled by a factor $50$ for clarity.}
\end{figure}
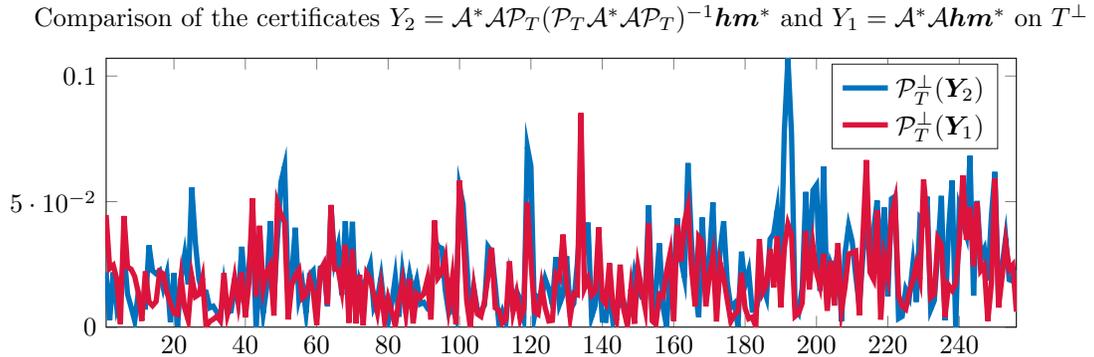

\subsection{Applications}

In this section, we discuss some of the application of the result of this paper. Among the many applications of blind deconvolution, one should mention astronomical imaging, movie deblurring, seismic data processing, super-resolution and medical imaging. In this section we discuss two of these applications in greater details: Communication and blind super-resolution for medical imaging. 

\subsubsection{Communication and Rayleigh fading}

When transmitting signals within a densely built environment such as in wireless communication, the 
multipath nature of the transmission arising from the many reflections that the signal will face will result in a channel impulse response that can be considered completely arbitrary. A common assumption in this case is to view each of the entries in the channel transfer matrix are identically and i.i.d., following a gaussian distribution~\cite{tse2005fundamentals,li2016rapid}. This assumption is known as i.i.d. Rayleigh fading. In such a regime, it is clear that 
\begin{itemize}
\item No sparsity or non vanishing assumption holds on the channel impulse response, and
\item when some of the entries from the channel transfer matrix are very small, or vanish, linearized approaches such as~\cite{gribonval2012blind,ling2016self} do not hold.
\end{itemize}
 
This paper is precisely interested in this regime as it removes the need for sparsity or non vanishing assumptions on the filter $\bs h$. For more details on Rayleigh fading, see chapters 2, 7 in~\cite{tse2005fundamentals}.

\subsubsection{\label{secBlindSuperResMI}Blind super-resolution and medical imaging}

To evaluate the interest of the nuclear norm minimization~\eqref{nuclearNormDeconvolutionLinearMap} for the super-resolution of signals, we consider $3$ datasets. Before discussing each of these datasets and provide corresponding results of the nuclear norm reconstruction for each of these datasets, we briefly recall the framework of blind super-resolution.  
In super-resolution, one is interested in recovering a signal $\bs x$ from the result of its convolution with a known (ideal) low pass filter $h(x,y)$. In Fourier space, the convolution with the low pass filter thus reads as 
\begin{align}
\hat{y}(\omega_1,\omega_2) = \hat{x}(\omega_1,\omega_2)\hat{h}(\omega_1,\omega_2)
\end{align}

with $\hat{h}(\omega_1,\omega_2) = 0$ for $|\omega|:=\sqrt{\omega_1^2 +\omega_2^2}>\Omega_c$ (see for example~\cite{candes2014towards}). The super-resolution framework naturally extends to the problem of recovering a signal from the result of its convolution with an unknown low pass filter. This extension, known as blind super-resolution, arises in many engineering applications such as medical imaging, microscopy, astronomy or even seismic imaging. In these applications, acquisition at lower resolution is often either performed by adding a blurring filter in order to reduce the influence of the side lobes arising from the convolution with the sinc, or the ideal low pass filter is corrupted by spurious reflections, noise or phase ambiguities, and is therefore only partially known. An illustration of the convolution with an ideal low pass filter with and without gaussian blurr is given in Fig.~\ref{superResolutionKernel}. As explained in section~\ref{connectionsExistingWork}, recent developments on the subject include~\cite{yang2016super} in which Yang et al. discuss the problem in the case where the output is given by a sum of spikes modulated by unknown (random) waveforms. The spikes have to satisfy a minimum separation condition and the waveforms are assumed to be generated as random combinations of random vectors satisfying an incoherence condition.

Although this paper primarily focuses on blind deconvolution, it also applies in the case of blind super-resolution as is shown below. Using a gaussian ideal low pass filter, we now provide three examples of blind super-resolution through the minimization program~\eqref{nuclearNormDeconvolutionLinearMap}.
\begin{itemize}
\item In the first one, we consider the recovery of one dimensional wavelet trains from their convolution with a one dimensional gaussian ideal low pass filter. The results are shown in Figs.~\ref{SRwavelet1} and~\ref{SRwavelet2}. In these examples, nuclear norm minimization is shown to recover the exact wavelet train and to remove the spurious oscillations arising from the convolution with the sinc.
\item In the second example, we consider the 3D Shepp Logan phantom of Schabel\footnote{https://www.mathworks.com/matlabcentral/fileexchange/9416-3d-shepp-logan-phantom}. For each of the $x$-$y$ slices from this phantom, we use as subspaces $\bs C_n$, the wavelets corresponding to the $K$ largest coefficients from the discrete wavelet transform of each of the frames. We consider data generated from the convolution of the slices with a gaussian (ideal) low pass filter similar to the one shown in Fig.~\ref{superResolutionKernel}. The low resolution images together with the recovered (super-resolved) ones are shown in Fig.~\ref{superResolution1},~\ref{superResolution2}. Figs~\ref{superresolution1a} and~\ref{superresolution2a} illustrate the evolution of the recovery with the cutoff frequency $\omega_c$.
\item Super-resolution is particularly interesting in medical imaging and Magnetic Resonance Imaging (MRI) where a reduction in the number of samples acquired in $k$-space, the equivalent of the two-dimensional Fourier space, leads to a reduction in the acquisition time. For an expensive imaging modality such as MRI, a reduction in the acquisition time means a larger number of faster examinations and thus a reduction of the waiting time for a fixed number of magnets. As a third example, we consider the MRI dataset distributed by the  Stanford data archive\footnote{https://graphics.stanford.edu/data/voldata/}, originally generated at the University of North Carolina. This MRI dataset consists of $109$ slices of the skull, each being of size $256$ by $256$. We start by applying a low pass gaussian filter to each slice in the volume. We then study reconstruction of these slices through the nuclear norm minimization program~\eqref{nuclearNormDeconvolutionLinearMap}. The results are shown in Figs~\ref{MRsequenceNoisy}, ~\ref{MRComparisonImages2},~\ref{MRComparisonImages2Arrows} and ~ \ref{MRComparisonImages2Arrows2}. The subspaces $\bs C_n$ are defined from the wavelets corresponding to the largest $K$ coefficients of the discrete wavelet transform obtained for each slice.  Figs~\ref{MRsequenceNoisy} and~\ref{MRComparisonImages2} compare the original image to the low pass image and recovered image as well as the correponding original low pass filter and recovered filter. Figs~\ref{MRComparisonImages2Arrows} and~\ref{MRComparisonImages2Arrows2} provide further highlights on the details that can be recovered through nuclear norm minimization. 
\end{itemize}

\begin{figure}[h!]\centering
  \begin{minipage}{0.8\textwidth}
    \hspace{0.6cm}\includegraphics[trim = 0cm 1cm 0cm 8cm, clip, width=.9\textwidth]{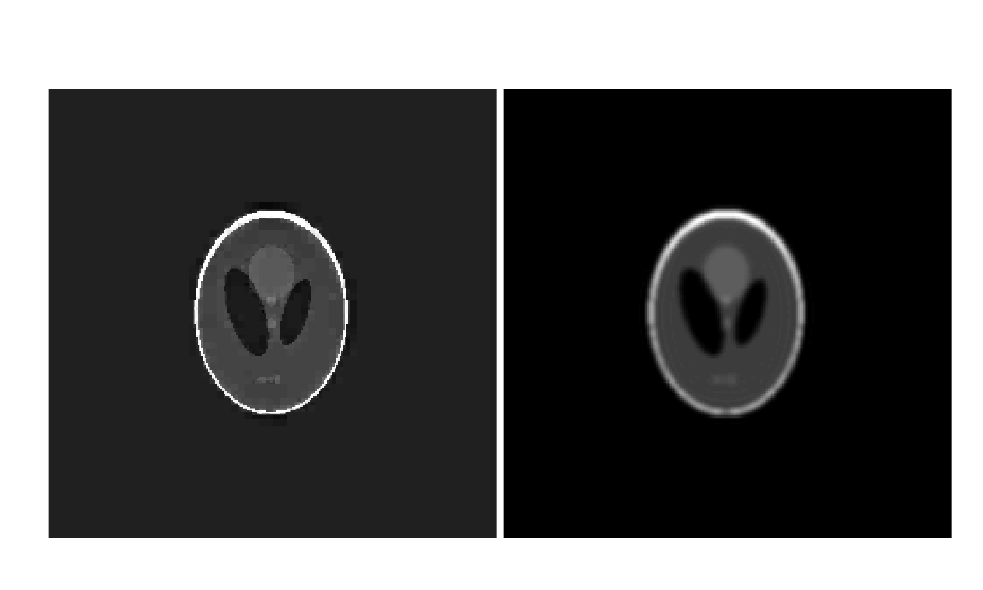}
  \end{minipage}
\caption{\label{superresolution1a}Evolution of the recovery in the blind super-resolution framework, for different sizes of the filter support (Part I). The original (low resolution) slice from the 3D Shepp Logan volume is shown on the right and the result obtained through the nuclear norm relaxation~\eqref{nuclearNormDeconvolutionLinearMap} is shown on the left. Here $K \approx L/20$, $L= 200^2$, $|\mbox{supp}(\bs h)| = L/10$ and the subspaces $\bs C_n$ are defined from the discrete wavelet transform with Daubechies wavelets by taking the wavelets corresponding to the $K$ largest coefficients of the DWT.}
\end{figure}

\begin{figure}[h!]\centering
  \begin{minipage}{0.8\textwidth}
    \hspace{0.6cm}\includegraphics[trim = 0cm 1cm 0cm 8cm, clip, width=.9\textwidth]{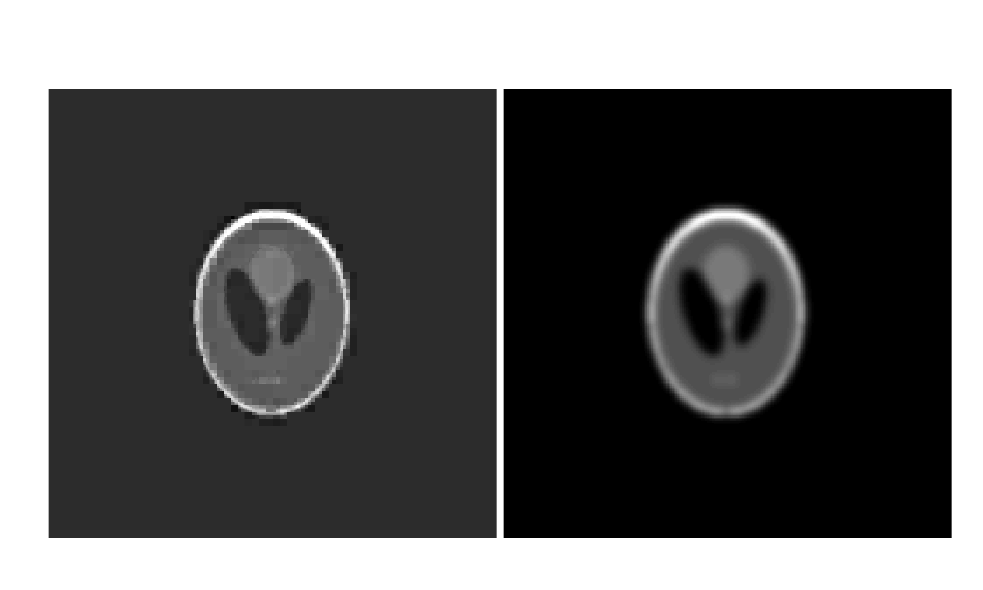}
  \end{minipage}
\caption{\label{superresolution2a}Evolution of the recovery in the blind super-resolution framework, for different sizes of the filter support (Part I). The original (low resolution) slice from the 3D Shepp Logan volume is shown on the right and the result obtained through the nuclear norm relaxation~\eqref{nuclearNormDeconvolutionLinearMap} is shown on the left. Here $K \approx L/20$, $L= 200^2$, $\mbox{supp}(\bs h) = L/16$ and the subspaces $\bs C_n$ are defined from the discrete wavelet transform with Daubechies wavelets by taking the wavelets corresponding to the $K$ largest coefficients of the DWT.}
\end{figure}

\begin{figure}[h!]\centering
  \begin{minipage}{0.8\textwidth}
    \hspace{0.6cm}\includegraphics[trim = 0cm 0cm 0cm 6cm, clip, width=.9\textwidth]{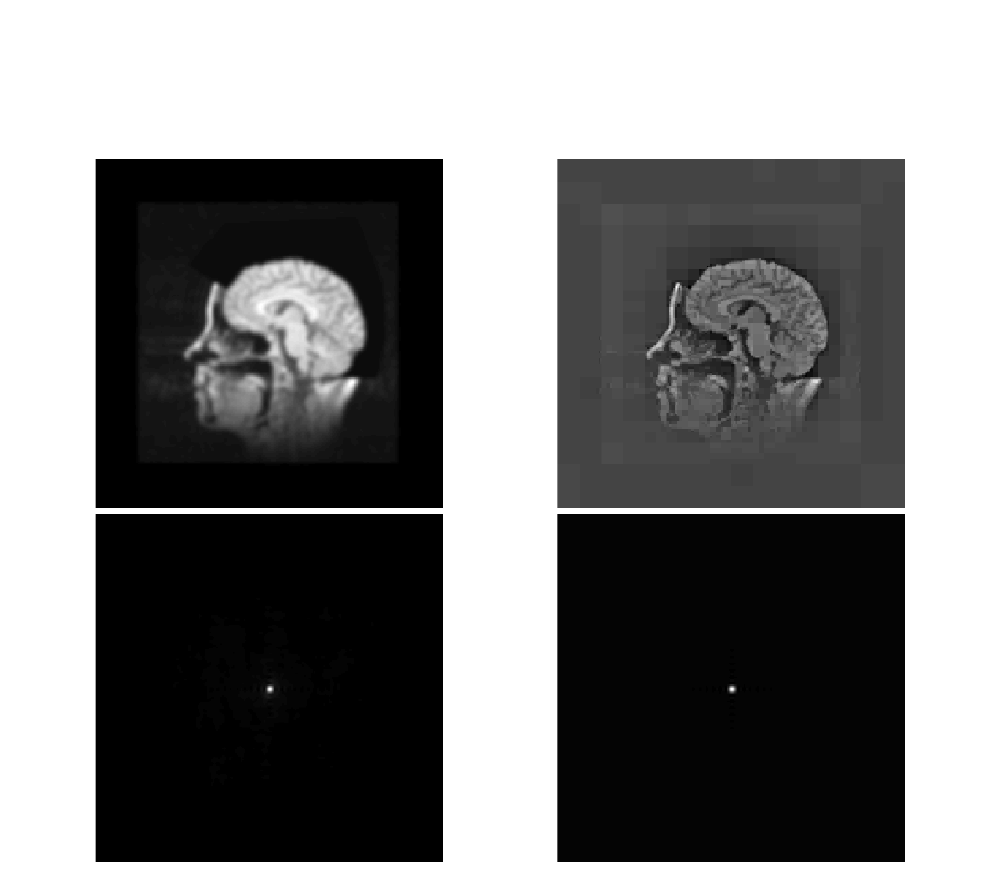}
  \end{minipage}
\caption{\label{MRsequenceNoisy}Ideal low pass filtering and blind super-resolution of MR images through nuclear norm minimization. From top to botton, left to right. (Top left) Original ideal low pass filtering of the Magentic resonance image shown in Fig.~\ref{MRComparisonImages2}. (Top right) recovery through semidefinite programming. (Bottom left) Original low pass (sinc) filter and (bottom right) recovery of the low pass filter. Additional comparison of the quality of the super-resolved image vs low pass image is given in Fig.~\ref{MRComparisonImages2Arrows}. The change in contrast appearing in the low pass image comes from the blurring of the forehead which has the highest intensity in the original image. Such blurring homogeneize the forehead with the background then resulting in a reduced contrast and apparent higher intensity of the global image.}
\end{figure}

\begin{figure}[h!]\centering
\includegraphics[width=\linewidth]{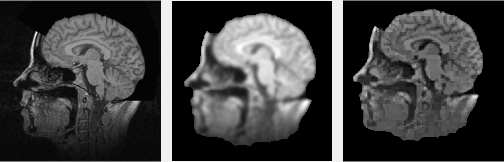}
\caption{\label{MRComparisonImages2} From left to right, Slice from the original MR volume, approximation obtained throug ideal low pass filtering through a truncated gaussian filter similar to the one depicted in Fig~\ref{superResolutionKernel} and recovered (super-resolved) image obtained through the nuclear norm minimization program~\eqref{nuclearNormDeconvolutionLinearMap}. An additional background subtraction step has been added to facilitate the comparison. The original (unprocessed) images are shown in Fig.~\ref{MRsequenceNoisy}.}
\end{figure}

\begin{figure}[h!]\centering
    \begin{minipage}{\textwidth}
    \hspace{0.5cm}\includegraphics[trim = 0cm 0cm 0cm 0cm, clip, width=.9\textwidth]{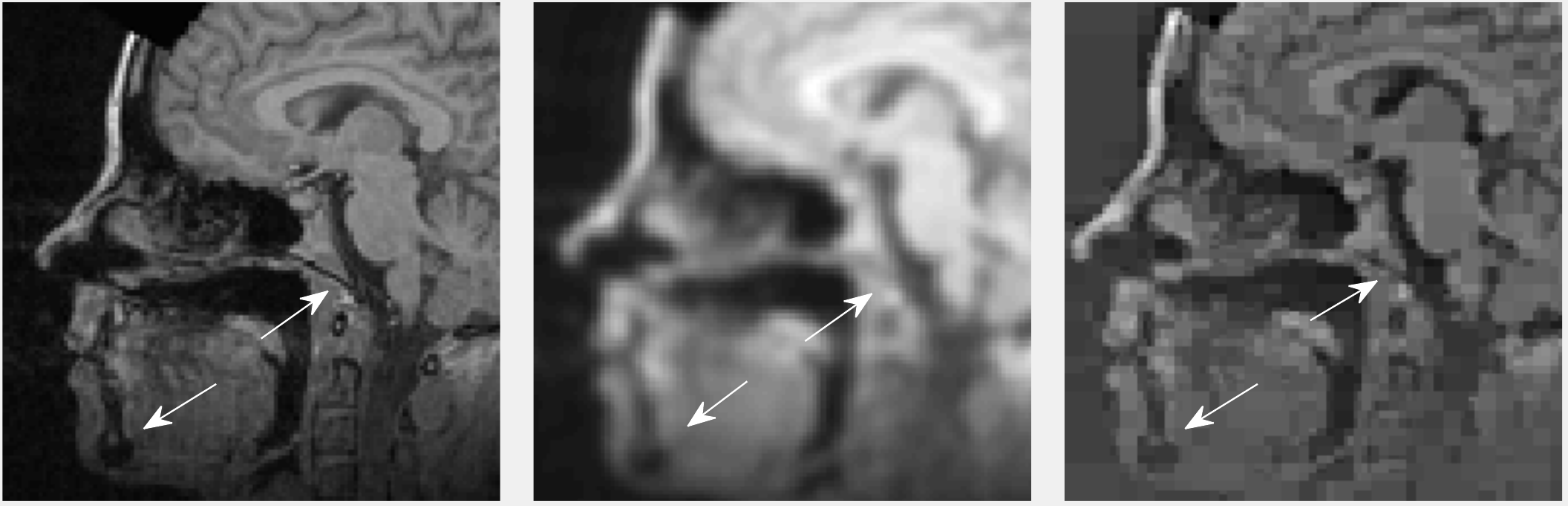}
  \end{minipage}
\caption{\label{MRComparisonImages2Arrows}. Additional illustration of the use of the nuclear norm minimization program~\eqref{nuclearNormDeconvolutionLinearMap} in the framework of the blind super-resolution of MR images from a same volume. (Part I)}
\end{figure}

\begin{figure}[h!]\centering
    \begin{minipage}{\textwidth}
    \hspace{0.5cm}\includegraphics[trim = 0cm 0cm 0cm 0cm, clip, width=.9\textwidth]{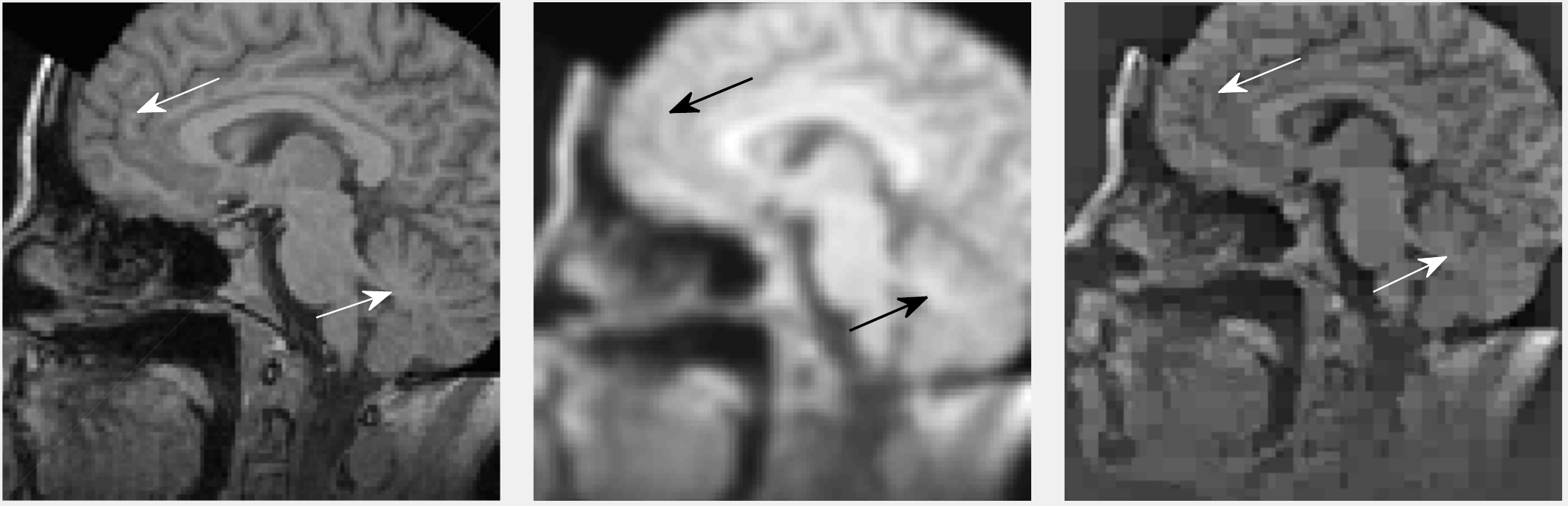}
  \end{minipage}
\caption{\label{MRComparisonImages2Arrows2}. Additional illustration of the use of the nuclear norm minimization program~\eqref{nuclearNormDeconvolutionLinearMap} in the framework of the blind super-resolution of MR images from a same volume (Part 2).}
\end{figure}

%\begin{figure}[h]\centering
%   \begin{minipage}[c]{.46\linewidth}\flushright
%     \includegraphics[trim = 2cm 4cm 10cm 13cm, clip, width=.8\linewidth]{Basis}
%   \end{minipage} \hfill
%   \begin{minipage}[c]{.46\linewidth}\flushleft
%      \includegraphics[trim = 3cm 4cm 8cm 13cm, clip, width=.8\linewidth]{MRbasis}
%   \end{minipage}
%\caption{\label{subspace}. Illustration of the subspace obtained by the first singular values of the sequence of Fig~\ref{MRsequence}.}
%\end{figure}

\begin{figure}
\input{plotArxivFinal1}
\caption{\label{SRwavelet1}Blind super-resolution of a wavelet train from its convolution with an ideal gaussian low pass filter. The side lobes arising from the convolution with $\mbox{sinc}$ induced by the windowing in Fourier space are clearly visible on the right figure. The subspaces $\bs C_n$ are defined from the wavelets corresponding to the $K$ largest coefficients of the discrete wavelet transform of each train. Elimination of the spurious high frequency oscillations following from the low pass nature of the filter, through the minimization program~\eqref{nuclearNormDeconvolutionLinearMap} are further highlighted in Fig.~\ref{SRwavelet2} where the recovered train is superimposed on its low pass approximation.}
\end{figure}
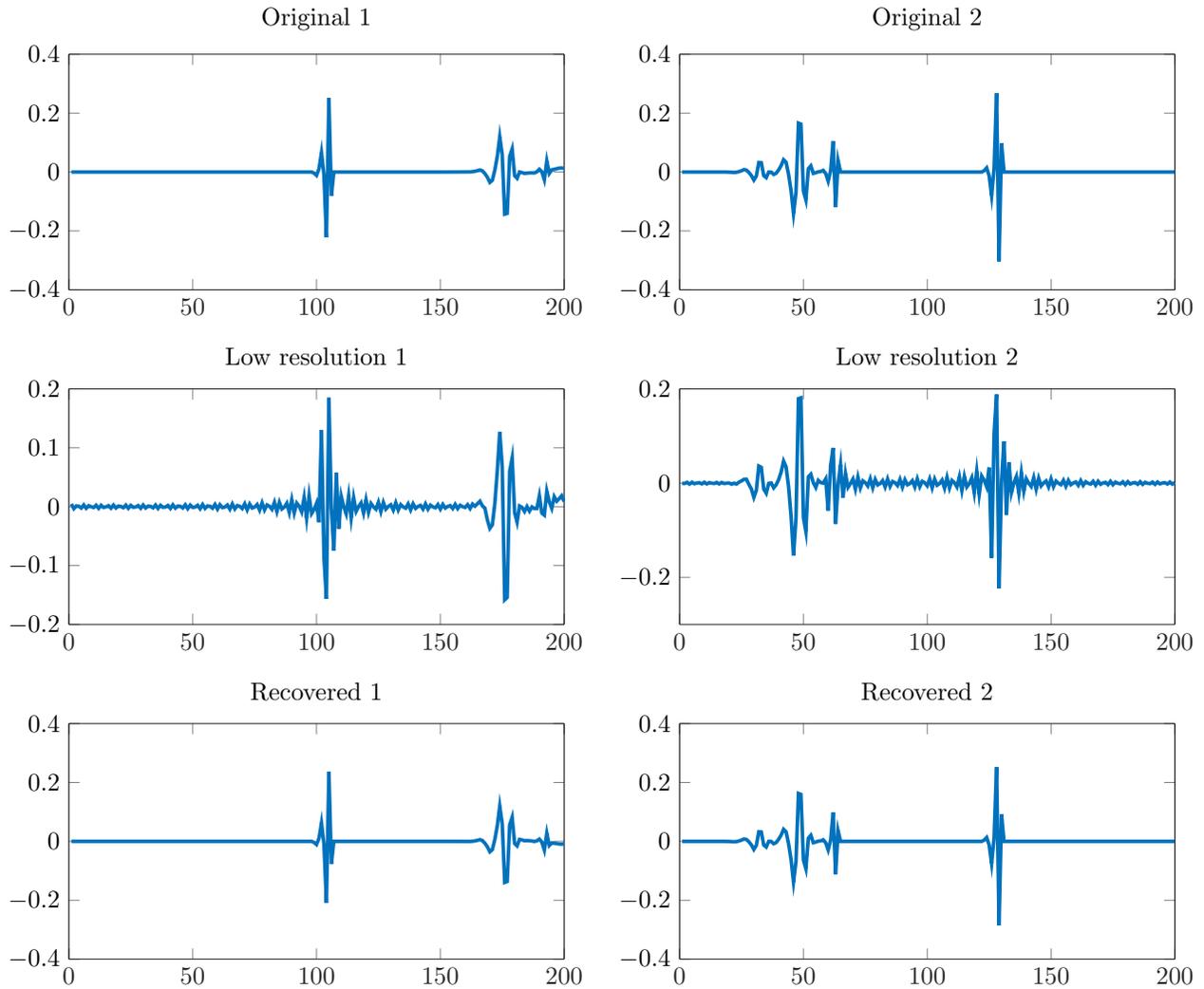

 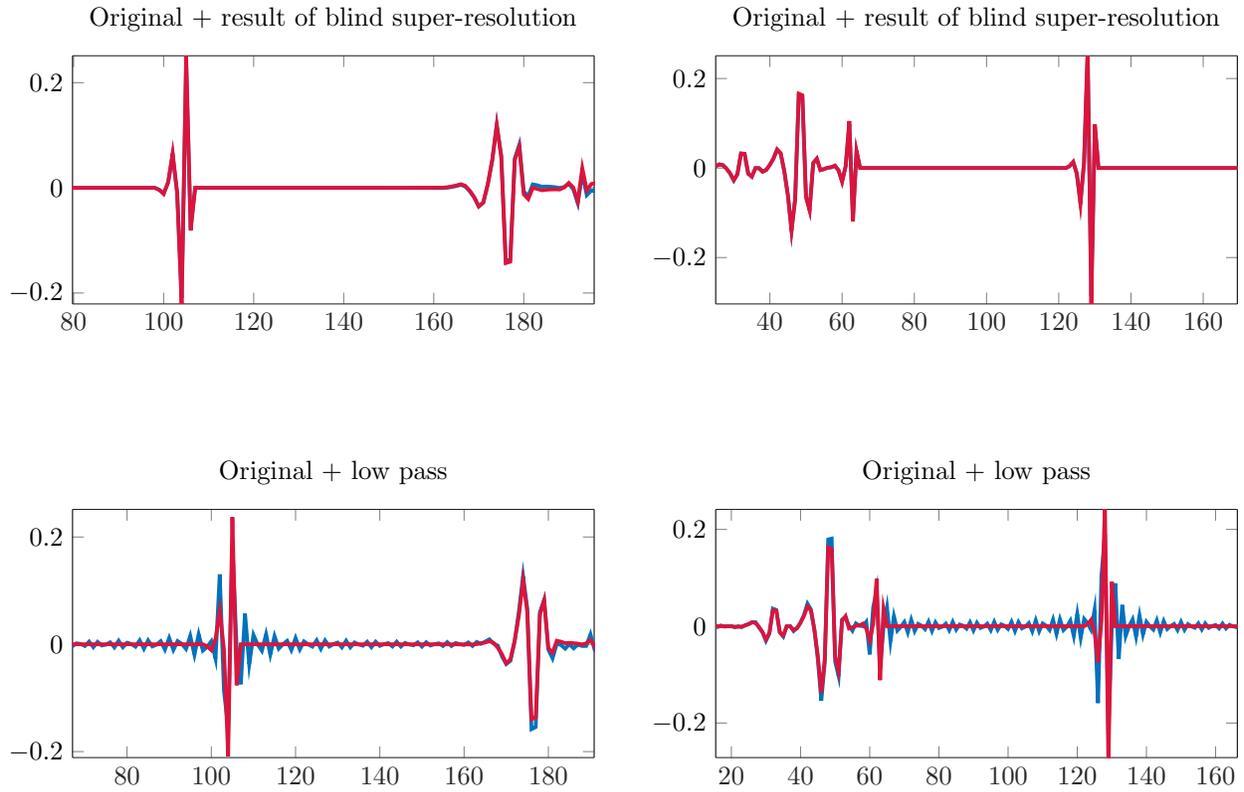
\begin{figure}
\input{plotArxivFinal2b}
\caption{\label{SRwavelet2}Further illustration of the blind super-resolution of a wavelet train from its convolution with an ideal gaussian low pass filter. The recovered (super-resolved) wavelet train (red) is superimposed on its low pass approximation. Note the side lobes arising from the $\mbox{sinc}$ and resulting from the multiplication in Fourier space with the ideal low pass filter of Fig.~\ref{superResolutionKernel}. The subspaces $\bs C_n$ are defined from the wavelets corresponding to the $K$ largest coefficients of the discrete wavelet transform of each train.}
\end{figure}

\begin{figure}\centering
     \includegraphics[trim = 0cm 4cm 0cm 5cm, clip, width=.8\linewidth]{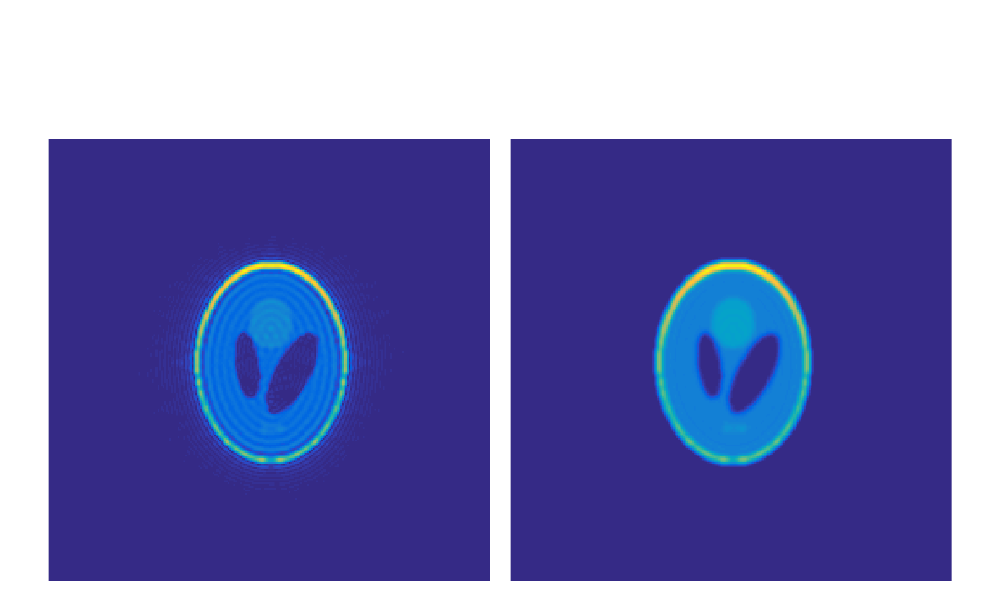}
\caption{\label{superResolutionFramework} Low resolution images obtained through ideal low pass filtering with (right) and without (left) the addition of a blurring kernel. In the left image, which is obtained without the blurring kernel, the effect of the side lobes arising from the sinc is clearly visible whether on the right image, the blurring step helped mitigrate those effects. Both images result from a windowed acquisition in frequency.}
\end{figure}

\begin{figure}\centering
    \includegraphics[trim = 0cm 0cm 0cm 0cm, clip, width=.5\linewidth]{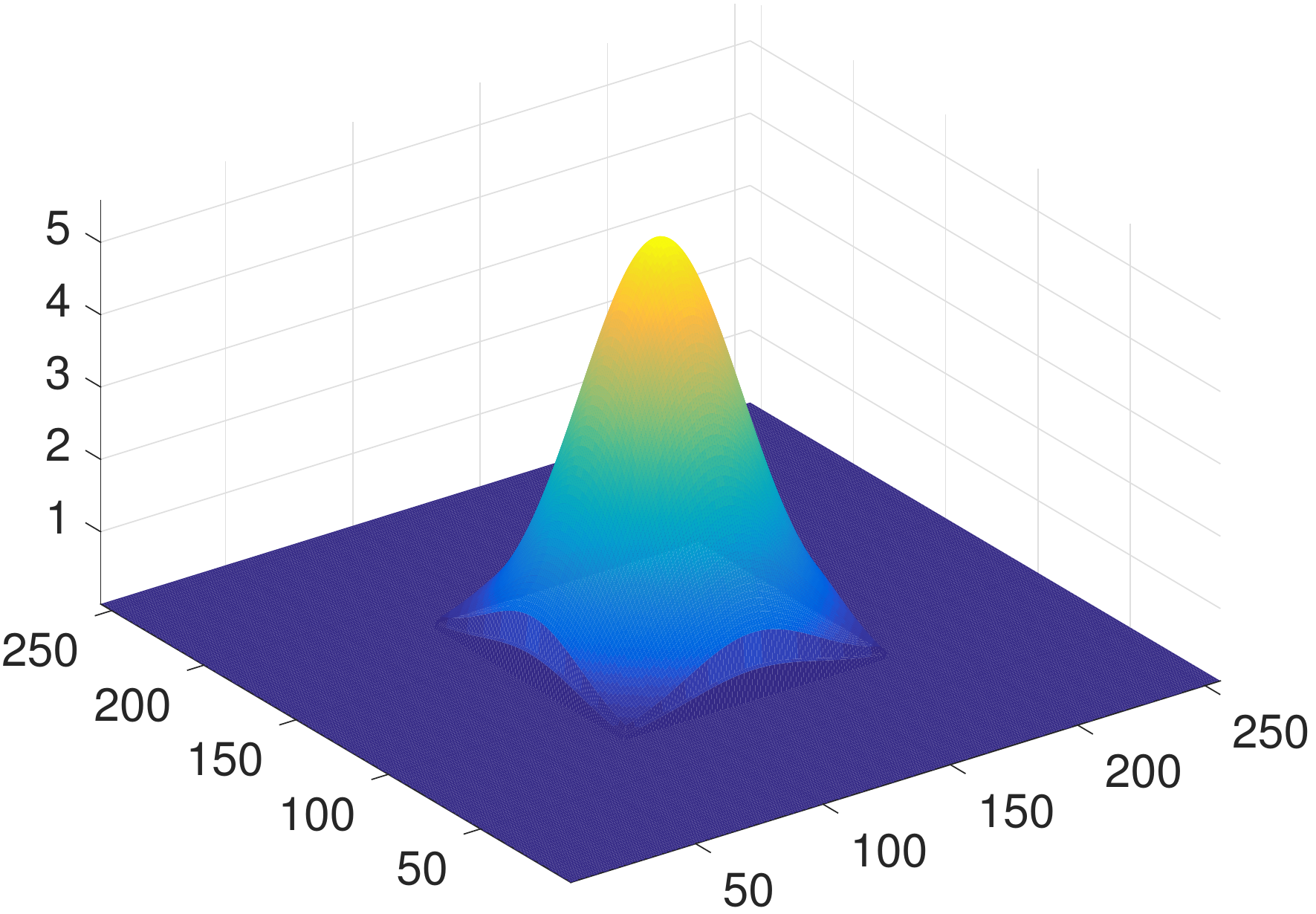}
\caption{\label{superResolutionKernel} Representation of a gaussian ideal low pass filter such as the one used in the super-resolution examples of section~\ref{secBlindSuperResMI}. All the coefficients above and below the cutoff frequencies $\pm \omega_1$ and $\pm \omega_2$ have been set to zero.}
\end{figure}

%\begin{figure}\centering
%    \includegraphics[trim = 0cm 2cm 0cm 2cm, clip, width=\linewidth]{FrequencyExtension}
%\caption{\label{frequencyExtension} Recovery of the windowed spectrum through the semidefinite program~\eqref{nuclearNormDeconvolutionLinearMap}. From left to right, spectrum of the recovered image, spectrum of the recovered filter and spectrum of the original image. The central peak arises from the sligth offset in the background values (see Figs.~\ref{superResolution1} and~\ref{superResolution2})}
%\end{figure}

\begin{figure}\centering
     \includegraphics[trim = 0cm 2cm 0cm 2cm, clip, width=.8\linewidth]{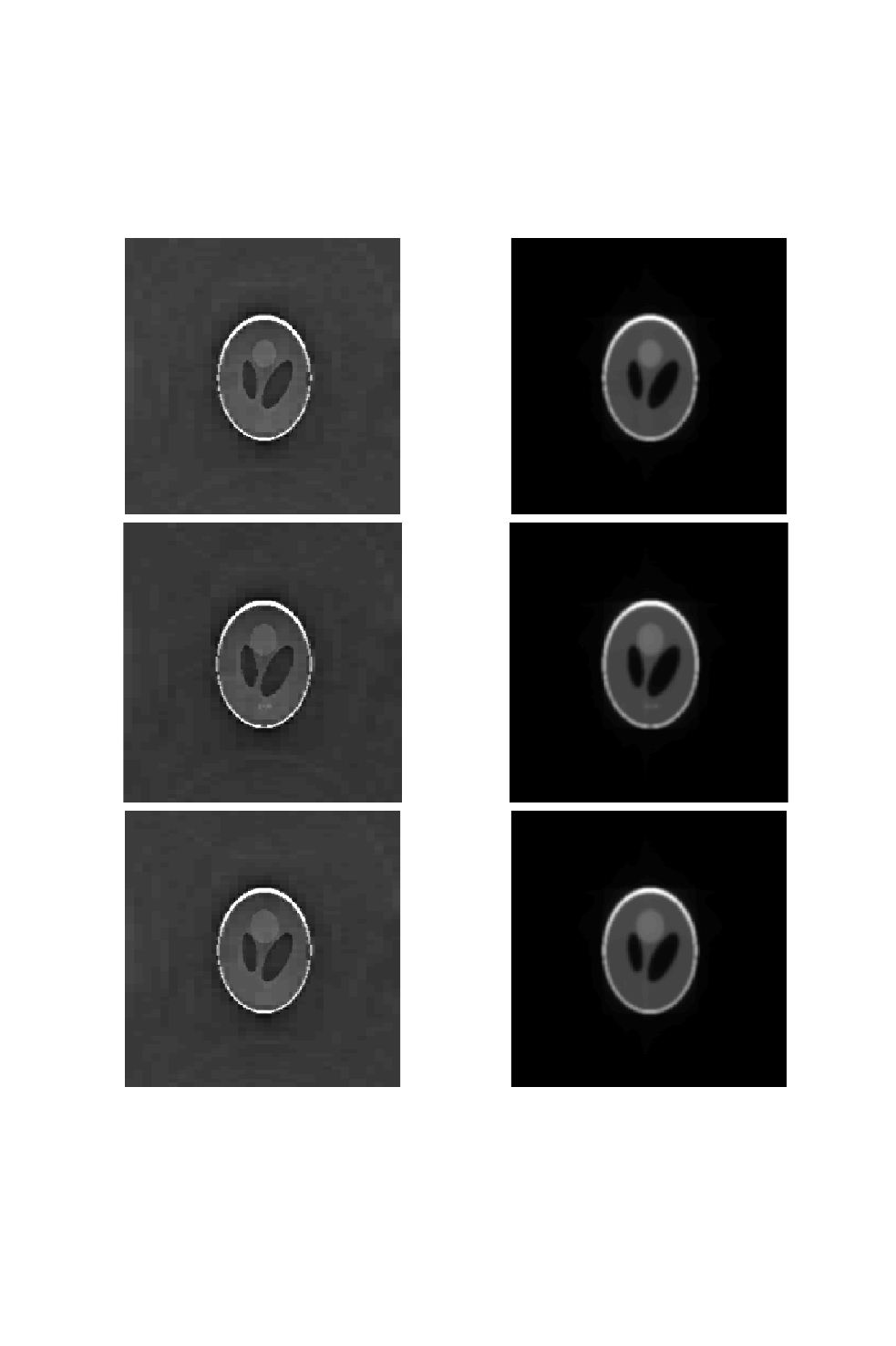}
\caption{\label{superResolution1} Reconstruction of the 3D Shepp-Logan phantom from blurred, low pass sample images. The original image before undersampling together with the filter and its corresponding transfer function are shown in Fig.~\ref{SingleFrameFail} (top row). The filter used here is gaussian and multiplied by a rectangular window zeroing out all frequencies above $\pm (k^c_x,k_y^c)$ such as shown in Fig.~\ref{superResolutionKernel}. The subspaces $\bs C_n$ are defined by considering the wavelets corresponding to the $K$ largest coefficients obtained from the discrete wavelet transform of each of the volume slices.}
\end{figure}

\begin{figure}\centering
     \includegraphics[trim = 0cm 2cm 0cm 4cm, clip, width=.8\linewidth]{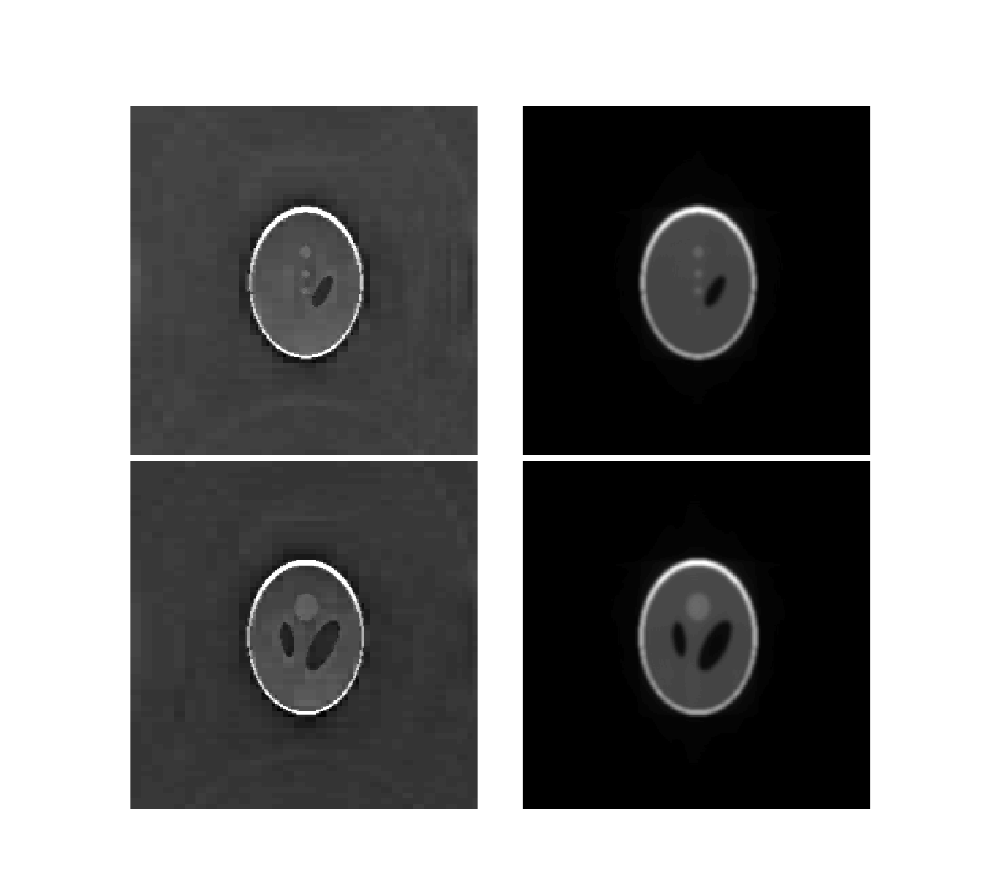}
\caption{\label{superResolution2} Reconstruction of the 3D Shepp-Logan phantom from blurred, low pass samples (continued). (Right) Original, blurred and low pass images. (Left) Recovery through the nuclear norm minimization program~\eqref{nuclearNormDeconvolutionLinearMap}. The gaussian ideal low pass filter has the same structure as shown in Fig.~\ref{superResolutionKernel} and the subspaces $\bs C_n$ are defined by considering the wavelets corresponding to the $K$ largest coeffcients obtained from the discrete wavelet transform of each of the volume slices.}
\end{figure}

\section{\label{sec:conclusion}Conclusions and perspectives}

In this paper we have considered a version of the blind deconvolution problem where the paradigm is shifted towards an arbitrary transfer function at the expense of requiring a small number of sufficiently distinct input signals to guarantee the recovery. Theory shows that whenever the number of inputs $N$, the ambient dimension $L$ and the dimension of the inputs subspace $K$ satisfy $L\gtrsim K^{3/2}\mu_h^2$ and $N\gtrsim K^{1/2}\mu_m^2$ up to log factors and for coherences defined as in~\eqref{coherenceh} and~\eqref{coherencem}, both the unknown filter as well as the unknown input signals can be recovered from the outputs to the filter by means of the nuclear norm minimization program~\eqref{nuclearNormDeconvolutionLinearMap}. When trying to recover an unknown filter of length $L$ whose Fourier transform is sufficiently "spread out" from its outputs, having a collection of a few outputs is therefore helpful. Such a framework finds applications in medical imaging, astronomy or microscopy where multiple compressible slices of a same volume are usually acquired. 

The importance of the coherences~\eqref{coherencem} and~\eqref{coherenceh} in the formulation of Theorem~\ref{theorem:BlindDeconvTh} illustrates the fact that the measured outputs have to carry a sufficient amount of information on both the impulse response h and the input signals. An intuition for this is that if the energy of the filter is concentrated at a single position in Fourier space, most of the measurements will be zero and won’t therefore carry information on the input signals. An equivalent statement holds whenever a number of input signals are zero. In other words, what really matters is the number of "meaningful" input signals and the fact that the probing of the filter through those signals should provide a sufficient amount of information on this filter. The result of this paper also shows that the notion of coherence alone is sufficient to describe the recovery. This is in contrast to previous work where both sparsity and sufficient incoherence were required (see for example~\cite{ahmed2016leveraging} or~\cite{ling2015self}). A direct consequence of this conclusion is that as soon as the support of the Fourier transform of the filter is on the order of the ambient dimension, it remains possible to recover both the filter and input signals including when the support is stricly smaller than this dimension. This observation implies that the nuclear norm minimization program~\eqref{nuclearNormDeconvolutionLinearMap} can be used in the framework of blind super-resolution as well.

To certify recovery through nuclear norm minimization, we construct a certificate of optimality. Exact recovery can then be shown by proving that the certificate satisfies the conditions derived from requiring the subgradient of the Lagrangian to vanish at the solution $\bs X_0$. In this case, the candidate certificate expands as a Neumann series, and certifying these conditions can be done by controlling each of the terms in the series. We applied \textit{ad-hoc} concentration results to the first two terms only as it is assumed that deriving bounds on a constant number of higher order terms should follow the same approach and is likely to become gradually heavier in terms of the derivations needed. It is likely that the sample complexity will benefit from further applications of the decoupling inequality of de la Pen\~a and Gin\'e~\cite{de1995decoupling}. It is not clear, however whether such an inequality can be used to achieve the sample complexity observed empirically ($L\gtrsim K$, $N\gtrsim 1$) as the constant appearing in this inequality scales badly with the order of the $U$-statistics involved.  

This work naturally raises an important open question: How far can we go in the complexity of the filter while still being able to certify the recovery? So far we have assumed that the filter remained constant spatially. In many applications, however, the point spread function varies with the position and it is not clear whether the proof techniques used in this paper can translate easily to that framework. Similar open problems include the more general field of blind linear system identification or even non linear system identification. 

Certainly equally interesting is the question of the efficiency of nuclear norm minimization for blind deconvolution. Why does the nuclear norm perform so well on the blind deconvolution problem? The probabilistic method reveals an elegant tool to derive recovery guarantees yet it does not make use of the particular structure of the problem, and as a consequence, is unable to explain the efficiency of nuclear norm minimization for that particular problem structure. Blind deconvolution however seems a natural candidate for a better understanding of the propagation of information in semidefinite relaxations such as discussed in~\cite{cosserank}, in the framework of matrix completion.

\bibliography{biblio.bib}
\bibliographystyle{abbrv}

\end{document}

%% file: FrequencyExtensionSingleFrame2b.tex
% This file was created by matlab2tikz v0.3.3.
% Copyright (c) 2008--2013, Nico Schlömer <nico.schloemer@gmail.com>
% All rights reserved.
% 
% The latest updates can be retrieved from
%   http://www.mathworks.com/matlabcentral/fileexchange/22022-matlab2tikz
% where you can also make suggestions and rate matlab2tikz.
% 
% 
% 
\begin{tikzpicture}
\hspace{-.4cm}
\begin{axis}[%
width=5cm,
height=5cm,
axis on top,
scale only axis,
separate axis lines,
every outer x axis line/.append style={darkgray!60!black},
every x tick label/.append style={font=\color{darkgray!60!black}},
xmin=0.5,
xmax=256.5,
every outer y axis line/.append style={darkgray!60!black},
every y tick label/.append style={font=\color{darkgray!60!black}},
y dir=reverse,
ymin=0.5,
ymax=256.5,
hide axis,
name=plot1,
title={\bf Recovered Image}
]
\addplot graphics [xmin=0.5,xmax=256.5,ymin=0.5,ymax=256.5] {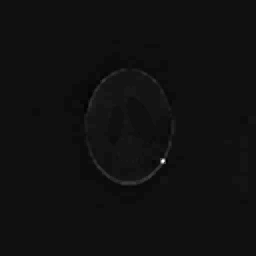};
\end{axis}
\hspace{1cm}

\begin{axis}[%
width=5cm,
height=5cm,
axis on top,
scale only axis,
separate axis lines,
every outer x axis line/.append style={darkgray!60!black},
every x tick label/.append style={font=\color{darkgray!60!black}},
xmin=0.5,
xmax=256.5,
every outer y axis line/.append style={darkgray!60!black},
every y tick label/.append style={font=\color{darkgray!60!black}},
y dir=reverse,
ymin=0.5,
ymax=256.5,
hide axis,
at=(plot1.right of south east),
anchor=left of south west,
title={\bf Recovered filter}
]
\addplot graphics [xmin=0.5,xmax=256.5,ymin=0.5,ymax=256.5] {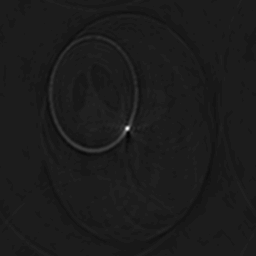};
\end{axis}
\end{tikzpicture}%

%% file: CertificateGolfingTFinal.tex
\definecolor{mycolor1}{rgb}{0.00000,0.44700,0.74100}%
\definecolor{mycolor2}{rgb}{0.85000,0.32500,0.09800}%

\definecolor{mycolor1}{rgb}{0,0.447,0.741}
\definecolor{mycolor2}{rgb}{0.8627,0.0784,0.2353}
\begin{tikzpicture}

\begin{axis}[%
scale=0.7,
width=6.806in,
height=2.012in,
at={(0.389in,3.237in)},
scale only axis,
xmin=1,
xmax=256,
ymin=2.74954213762184e-06,
ymax=0.0321723036780314,
axis background/.style={fill=white},
title style={},
title={Projection of the certificate on $T$},
legend style={legend cell align=left, align=left, draw=white!15!black}
]
\addplot [color=mycolor1, line width=2.0pt]
  table[row sep=crcr]{%
1	0.0208993644097785\\
2	0.00766694223361853\\
3	0.00222142645405987\\
4	0.000564425279023338\\
5	0.00367920786265728\\
6	0.000169299645851649\\
7	0.00159917924638979\\
8	0.000234983362147749\\
9	0.00112445637007171\\
10	0.000412507378910074\\
11	0.00988353485811572\\
12	0.00251123187528157\\
13	0.00193137683472639\\
14	0.000547755179030353\\
15	0.00517401386164405\\
16	0.00170581188611752\\
17	0.00791012398494975\\
18	0.00290183292009563\\
19	0.0103308107962335\\
20	0.00262487680181569\\
21	0.00462635318341471\\
22	0.00273433710286409\\
23	0.0258281411371968\\
24	0.00422236539933731\\
25	0.000134850196159089\\
26	4.94698615647911e-05\\
27	0.00833341533207118\\
28	0.00211737384572214\\
29	0.000691090642101898\\
30	0.000628574210631807\\
31	0.00593741840038489\\
32	0.00406431272050303\\
33	0.0164722491202107\\
34	0.00604285278663015\\
35	0.00013153555895262\\
36	3.34208654207927e-05\\
37	0.00027811622305829\\
38	0.000633273035257404\\
39	0.00598180279815418\\
40	0.00778929088058262\\
41	0.0123122618434181\\
42	0.00451675938405594\\
43	0.0128018548348712\\
44	0.00325272550616421\\
45	0.00149146511148717\\
46	0.00308682358404119\\
47	0.0291576759539756\\
48	0.00351434064737007\\
49	0.00163132237439718\\
50	0.00059845142482231\\
51	0.0129624121992118\\
52	0.00329352030042876\\
53	0.00168895399550369\\
54	0.00155477601805867\\
55	0.0146861827640365\\
56	0.00637286865977405\\
57	0.00580491851540952\\
58	0.00212953724600753\\
59	0.00468912390694538\\
60	0.00119142367496148\\
61	0.00168497603862869\\
62	0.000666963641212612\\
63	0.00630003924555393\\
64	0.0107998406102502\\
65	0.00836598538946033\\
66	0.00306906590318499\\
67	0.00553549588032311\\
68	0.00140647186454173\\
69	0.0013078211694626\\
70	0.00111080325117992\\
71	0.0104924821146158\\
72	0.00843377804629928\\
73	0.0147312056338219\\
74	0.00540415011727458\\
75	0.0155250896830068\\
76	0.00394465144690192\\
77	0.000658440697531626\\
78	0.000375447395086955\\
79	0.0035464201907444\\
80	0.00088119134713512\\
81	0.00989753050344076\\
82	0.00363091399037249\\
83	0.000446893954826438\\
84	0.000113547871317444\\
85	0.000828613405034562\\
86	0.000964927086799104\\
87	0.00911455758649745\\
88	0.00393937954763641\\
89	0.00849276803898551\\
90	0.00311557618125257\\
91	0.00962567394336838\\
92	0.00244571396516149\\
93	0.000353504868620651\\
94	0.000394294381884355\\
95	0.00372444602175998\\
96	0.00350103818118946\\
97	0.00567156625337909\\
98	0.00208061690232323\\
99	0.00405086193558416\\
100	0.00102925256611511\\
101	0.00544799891281129\\
102	0.00220957429549657\\
103	0.020871309794769\\
104	0.00231309807226853\\
105	0.00514197547465559\\
106	0.00188633626161483\\
107	0.00606241048780362\\
108	0.00154035157224267\\
109	0.00232445567343937\\
110	0.000402584442849401\\
111	0.00380275270326494\\
112	0.00125008959648731\\
113	0.00574599493527837\\
114	0.002107921101315\\
115	0.0137924338000938\\
116	0.00350441414875627\\
117	0.00114608688066786\\
118	0.00053661443360817\\
119	0.00506878003921728\\
120	0.00650309475988816\\
121	0.00204807061731629\\
122	0.00075133572511844\\
123	0.0104248042674231\\
124	0.00264875888492746\\
125	0.00182026329451793\\
126	0.00255965163461648\\
127	0.0241780882143912\\
128	0.00576502405246526\\
129	0.0136523866755083\\
130	0.00500838484557793\\
131	0.00549514569707429\\
132	0.00139621959470078\\
133	3.01553317499116e-05\\
134	0.00209664774269842\\
135	0.0198046223915396\\
136	0.00947447420710476\\
137	0.00536859923877221\\
138	0.00196947330225298\\
139	0.0128549146968066\\
140	0.00326620708117791\\
141	5.12390823414868e-05\\
142	3.91305069907111e-05\\
143	0.000369620942592456\\
144	0.00041320232969526\\
145	0.0251169523774743\\
146	0.00921416647831369\\
147	0.00350271105443222\\
148	0.000889977095853363\\
149	0.00161167711781857\\
150	0.000575518920722189\\
151	0.00543626603170834\\
152	0.00820636591044666\\
153	0.00436966906680484\\
154	0.00160301527158153\\
155	0.0139879390610951\\
156	0.00355408858712874\\
157	0.00362774917716451\\
158	0.000562923407493308\\
159	0.00531729068919106\\
160	0.000223753166846305\\
161	0.00940310241806247\\
162	0.00344953280121535\\
163	0.00496051986682699\\
164	0.00126038060131033\\
165	0.00354233068284933\\
166	0.000763069258518772\\
167	0.00720783859672424\\
168	0.00261090481568977\\
169	0.00653676255910868\\
170	0.00239801459761703\\
171	0.00141971061924231\\
172	0.000360723426577435\\
173	0.000993896179494429\\
174	0.00131081496169328\\
175	0.0123817629508449\\
176	0.00699358589409415\\
177	0.000467654082872048\\
178	0.000171559133014507\\
179	0.0137146003270142\\
180	0.00348463803612377\\
181	0.00400892929445278\\
182	0.000699535069926613\\
183	0.00660770411137614\\
184	0.0110802508473021\\
185	0.0112147766576068\\
186	0.00411414639751313\\
187	0.00120038499618114\\
188	0.000304996654364461\\
189	0.00152536615706363\\
190	0.000289393836163205\\
191	0.00273357108632613\\
192	0.00157808238480221\\
193	0.00873543974970219\\
194	0.0032046004190867\\
195	0.0205938972408153\\
196	0.00523254604044248\\
197	0.00232629581853464\\
198	0.0001993751036482\\
199	0.00188326754257004\\
200	2.74954213762184e-06\\
201	0.0212418929654884\\
202	0.00779259900472894\\
203	0.0203494027396575\\
204	0.00517042430025006\\
205	0.00290242193759935\\
206	0.00230356440275867\\
207	0.0217591263530576\\
208	0.00884746592989398\\
209	0.00707044951945841\\
210	0.00259379792459334\\
211	0.000487578962142162\\
212	0.000123885214047956\\
213	0.00195270769420965\\
214	0.000205379333704923\\
215	0.0019399826056699\\
216	0.00424040362877211\\
217	0.00632751861212068\\
218	0.00232125335154101\\
219	0.0146404861204323\\
220	0.00371988928486021\\
221	0.00447729116017169\\
222	0.00102725678037767\\
223	0.00970331459129366\\
224	0.00963322523624546\\
225	0.00129445138675058\\
226	0.000474870135371225\\
227	0.0216271779063142\\
228	0.0054950844318746\\
229	0.00160319618111635\\
230	2.28860487864473e-05\\
231	0.000216178209157158\\
232	0.00108922414934127\\
233	0.00814878772424372\\
234	0.00298938683161889\\
235	0.00932017643709453\\
236	0.00236809243737945\\
237	0.00184010567262849\\
238	2.40101866520383e-05\\
239	0.000226796648054008\\
240	0.0133289145716826\\
241	0.0071977380794045\\
242	0.00264049379615059\\
243	0.0194918523439561\\
244	0.0049525358707296\\
245	0.00373698789470482\\
246	0.000937543228152134\\
247	0.00885589373511126\\
248	0.00180245353070454\\
249	0.00448613330782302\\
250	0.00164574023635371\\
251	0.0212241932690687\\
252	0.0053926931436536\\
253	0.00282200742826718\\
254	0.000240098638352413\\
255	0.0022679359877461\\
256	0.00607671644269401\\
};
\addlegendentry{$\bs h\bs m^*$}

\addplot [color=mycolor2, line width=2.0pt]
  table[row sep=crcr]{%
1	0.0287518613175069\\
2	0.00884327507378738\\
3	0.00671163980132639\\
4	0.00196513457227748\\
5	0.00792799808981979\\
6	0.000470174949629772\\
7	0.00676224959155537\\
8	0.00195310054529091\\
9	0.00788501428157714\\
10	0.00298432171175411\\
11	0.00183516137325432\\
12	0.00162229260620314\\
13	0.00515825447200653\\
14	0.000732520129995389\\
15	0.000590282034755755\\
16	0.000160482945968548\\
17	0.0104984245271291\\
18	0.00320627544748003\\
19	0.00430921014693955\\
20	0.00230321961318919\\
21	0.00565811026550526\\
22	0.00661145200953746\\
23	0.0249638551934862\\
24	0.00826129338042096\\
25	0.00316531016903792\\
26	0.00117219274987879\\
27	0.0133265418107153\\
28	0.00436074264767539\\
29	0.00142931420776672\\
30	0.00150277255816807\\
31	0.00557740319267138\\
32	0.00461493059117968\\
33	0.014406613668841\\
34	0.00394174855926208\\
35	0.00212943290863118\\
36	0.000556436169734653\\
37	0.000414866370901627\\
38	0.00159659023579769\\
39	0.00639917153949941\\
40	0.00728992179934014\\
41	0.0210426969914937\\
42	0.00671544925330703\\
43	0.0155762806780901\\
44	0.00545500486907948\\
45	0.00354093544184281\\
46	0.00653691075183056\\
47	0.019427266347074\\
48	0.00496307221200499\\
49	0.00450432548409223\\
50	0.00151937842627197\\
51	0.00997494743058496\\
52	0.00405058420379147\\
53	0.00285875409753668\\
54	0.00412567171246698\\
55	0.0176549709434766\\
56	0.00484384175324602\\
57	0.00647468498328835\\
58	0.0019018457505897\\
59	0.00705964514413926\\
60	0.00234218653374492\\
61	0.00214884549044752\\
62	0.00152907007056902\\
63	0.00529949807743575\\
64	0.0139794579646151\\
65	0.00993774487539547\\
66	0.00296341398318975\\
67	0.00422116554693623\\
68	0.00171997336138173\\
69	0.00244141390990201\\
70	0.0030289011262391\\
71	0.0133817718707717\\
72	0.00674001649472445\\
73	0.00818204995471888\\
74	0.0018002477250147\\
75	0.018211289310131\\
76	0.00644303251385431\\
77	0.00171441333033891\\
78	0.00117842880413701\\
79	0.005983992596674\\
80	0.00110418105269378\\
81	0.0134677131682101\\
82	0.00413348590282169\\
83	0.00307289171297066\\
84	0.000728497408169048\\
85	0.000912426691352475\\
86	0.00277086539702585\\
87	0.0129443229605144\\
88	0.00595584511770264\\
89	0.00502171902118976\\
90	0.00114962822777522\\
91	0.012396515573933\\
92	0.00427558616704239\\
93	8.36949581229412e-05\\
94	0.000398072258235126\\
95	0.00164552564085834\\
96	0.00309250581297192\\
97	0.00423391340833893\\
98	0.00109069256482244\\
99	0.000990395417513356\\
100	0.000725444029763848\\
101	0.00800811061978442\\
102	0.00454555904084897\\
103	0.0126440774570244\\
104	0.00441083384705127\\
105	0.00775927359674958\\
106	0.00242716124755929\\
107	0.00668706816008584\\
108	0.00240814537116395\\
109	0.00452203295811146\\
110	0.00105110391733878\\
111	0.00440925994345235\\
112	0.0018599598922609\\
113	0.00126186267853194\\
114	5.67544173165095e-06\\
115	0.0167686441846499\\
116	0.00587382475085142\\
117	0.000561819016003729\\
118	0.00171038585228569\\
119	0.00879922823189584\\
120	0.00732133867013033\\
121	0.00151388527687197\\
122	0.000388347802112161\\
123	0.0118676203664649\\
124	0.00423467211279296\\
125	0.00212609506949992\\
126	0.00526800333697428\\
127	0.0146687240929833\\
128	0.00525306212294328\\
129	0.0261901462340955\\
130	0.0084945053383651\\
131	0.0039005116735372\\
132	0.0016337813790207\\
133	0.000481027271578171\\
134	0.00620734145234923\\
135	0.029889215621673\\
136	0.00823340678525118\\
137	0.00594417880885607\\
138	0.00174281147570377\\
139	0.0107920727419772\\
140	0.00424562807396062\\
141	0.00106862886325111\\
142	0.000326494769642996\\
143	0.00362048841528452\\
144	0.00412569705848263\\
145	0.0192111444281032\\
146	0.0049993107288665\\
147	0.00646001732821012\\
148	0.0020510649657604\\
149	0.00282506417553846\\
150	0.00136500101456899\\
151	0.00500340571053812\\
152	0.00809346186338986\\
153	0.000490134640644529\\
154	0.000176543802496674\\
155	0.0175007499061435\\
156	0.00608270683427447\\
157	0.00893599185675007\\
158	0.000486082284534187\\
159	0.0031260490612028\\
160	0.000161169188373957\\
161	0.0112318450770478\\
162	0.00335357757909002\\
163	0.00978998764858826\\
164	0.00306766176117697\\
165	0.00618002424439608\\
166	0.00139009385605637\\
167	0.0026691704538035\\
168	0.00383680748977396\\
169	0.00270754346228995\\
170	0.000460185553302233\\
171	0.000631474285665805\\
172	0.000326500461886556\\
173	0.000622309467381805\\
174	0.00322439940530192\\
175	0.0124863193471118\\
176	0.00378272495639863\\
177	0.00105432075140224\\
178	0.000424915649517014\\
179	0.00953855759108742\\
180	0.00402768361950375\\
181	0.00364686912038256\\
182	0.00120892778141203\\
183	0.00182893722964014\\
184	0.010578786404722\\
185	0.00511002944205643\\
186	0.000960043012027517\\
187	0.0111031734551992\\
188	0.0029615214740991\\
189	0.00145323269144233\\
190	0.0012992005491329\\
191	0.00830455282660814\\
192	0.00260711663305443\\
193	0.00160477809760754\\
194	0.000123668318695374\\
195	0.0191894539985316\\
196	0.00728443117097227\\
197	0.0020478683601904\\
198	0.000814219204509291\\
199	0.00495762089989619\\
200	0.00283169523347225\\
201	0.0152922869772542\\
202	0.00387769005699813\\
203	0.017999447771179\\
204	0.00695347470445354\\
205	0.00300793028821426\\
206	0.0060111037446409\\
207	0.0251988185221859\\
208	0.013615561042036\\
209	0.00396168154138784\\
210	0.000876744141539227\\
211	0.00145601118955572\\
212	0.000426975586796158\\
213	0.0036356141152812\\
214	0.000998505325876227\\
215	0.00661604398142616\\
216	0.00916419675357971\\
217	0.00315168037811522\\
218	0.00064018140985979\\
219	0.0171289829900985\\
220	0.00606457145063963\\
221	0.00663294038206199\\
222	0.00113289213529213\\
223	0.00338224985817413\\
224	0.0124750965712338\\
225	0.000879718768705374\\
226	0.000217161742298359\\
227	0.0266615821893978\\
228	0.00930382302261565\\
229	1.55166346921488e-05\\
230	0.000767125310740141\\
231	0.0075599122805869\\
232	0.00269628807848736\\
233	0.00479123736129424\\
234	0.00109312778058504\\
235	0.00439690591846994\\
236	2.70580056510421e-05\\
237	0.00185337778085625\\
238	0.000153563776930067\\
239	0.0017797129808159\\
240	0.0173034945561411\\
241	0.000126510594104073\\
242	0.000540570710810164\\
243	0.0321723036780314\\
244	0.0104542455093232\\
245	0.00504713828139721\\
246	0.00193259697603995\\
247	0.00540158870571615\\
248	0.000137047049123451\\
249	0.0111485856499581\\
250	0.00372401650391041\\
251	0.0182597375594722\\
252	0.00712192810184843\\
253	0.00554460302355959\\
254	0.00084569990641731\\
255	0.00469667555822039\\
256	0.00548799401393419\\
};
\addlegendentry{$\mathcal{P}_T(\bs Y_1)$}

\end{axis}

\begin{axis}[%
scale=0.7,
width=6.806in,
height=2.012in,
at={(0.389in,0.583in)},
scale only axis,
xmin=1,
xmax=256,
ymin=2.74954213762184e-06,
ymax=0.029157675954045,
title style={},
title={Projection of the certificate on $T$},
axis background/.style={fill=white},
legend style={legend cell align=left, align=left, draw=white!15!black}
]
\addplot [color=mycolor1, line width=2.0pt]
  table[row sep=crcr]{%
1	0.0208993644097785\\
2	0.00766694223361853\\
3	0.00222142645405987\\
4	0.000564425279023338\\
5	0.00367920786265728\\
6	0.000169299645851649\\
7	0.00159917924638979\\
8	0.000234983362147749\\
9	0.00112445637007171\\
10	0.000412507378910074\\
11	0.00988353485811572\\
12	0.00251123187528157\\
13	0.00193137683472639\\
14	0.000547755179030353\\
15	0.00517401386164405\\
16	0.00170581188611752\\
17	0.00791012398494975\\
18	0.00290183292009563\\
19	0.0103308107962335\\
20	0.00262487680181569\\
21	0.00462635318341471\\
22	0.00273433710286409\\
23	0.0258281411371968\\
24	0.00422236539933731\\
25	0.000134850196159089\\
26	4.94698615647911e-05\\
27	0.00833341533207118\\
28	0.00211737384572214\\
29	0.000691090642101898\\
30	0.000628574210631807\\
31	0.00593741840038489\\
32	0.00406431272050303\\
33	0.0164722491202107\\
34	0.00604285278663015\\
35	0.00013153555895262\\
36	3.34208654207927e-05\\
37	0.00027811622305829\\
38	0.000633273035257404\\
39	0.00598180279815418\\
40	0.00778929088058262\\
41	0.0123122618434181\\
42	0.00451675938405594\\
43	0.0128018548348712\\
44	0.00325272550616421\\
45	0.00149146511148717\\
46	0.00308682358404119\\
47	0.0291576759539756\\
48	0.00351434064737007\\
49	0.00163132237439718\\
50	0.00059845142482231\\
51	0.0129624121992118\\
52	0.00329352030042876\\
53	0.00168895399550369\\
54	0.00155477601805867\\
55	0.0146861827640365\\
56	0.00637286865977405\\
57	0.00580491851540952\\
58	0.00212953724600753\\
59	0.00468912390694538\\
60	0.00119142367496148\\
61	0.00168497603862869\\
62	0.000666963641212612\\
63	0.00630003924555393\\
64	0.0107998406102502\\
65	0.00836598538946033\\
66	0.00306906590318499\\
67	0.00553549588032311\\
68	0.00140647186454173\\
69	0.0013078211694626\\
70	0.00111080325117992\\
71	0.0104924821146158\\
72	0.00843377804629928\\
73	0.0147312056338219\\
74	0.00540415011727458\\
75	0.0155250896830068\\
76	0.00394465144690192\\
77	0.000658440697531626\\
78	0.000375447395086955\\
79	0.0035464201907444\\
80	0.00088119134713512\\
81	0.00989753050344076\\
82	0.00363091399037249\\
83	0.000446893954826438\\
84	0.000113547871317444\\
85	0.000828613405034562\\
86	0.000964927086799104\\
87	0.00911455758649745\\
88	0.00393937954763641\\
89	0.00849276803898551\\
90	0.00311557618125257\\
91	0.00962567394336838\\
92	0.00244571396516149\\
93	0.000353504868620651\\
94	0.000394294381884355\\
95	0.00372444602175998\\
96	0.00350103818118946\\
97	0.00567156625337909\\
98	0.00208061690232323\\
99	0.00405086193558416\\
100	0.00102925256611511\\
101	0.00544799891281129\\
102	0.00220957429549657\\
103	0.020871309794769\\
104	0.00231309807226853\\
105	0.00514197547465559\\
106	0.00188633626161483\\
107	0.00606241048780362\\
108	0.00154035157224267\\
109	0.00232445567343937\\
110	0.000402584442849401\\
111	0.00380275270326494\\
112	0.00125008959648731\\
113	0.00574599493527837\\
114	0.002107921101315\\
115	0.0137924338000938\\
116	0.00350441414875627\\
117	0.00114608688066786\\
118	0.00053661443360817\\
119	0.00506878003921728\\
120	0.00650309475988816\\
121	0.00204807061731629\\
122	0.00075133572511844\\
123	0.0104248042674231\\
124	0.00264875888492746\\
125	0.00182026329451793\\
126	0.00255965163461648\\
127	0.0241780882143912\\
128	0.00576502405246526\\
129	0.0136523866755083\\
130	0.00500838484557793\\
131	0.00549514569707429\\
132	0.00139621959470078\\
133	3.01553317499116e-05\\
134	0.00209664774269842\\
135	0.0198046223915396\\
136	0.00947447420710476\\
137	0.00536859923877221\\
138	0.00196947330225298\\
139	0.0128549146968066\\
140	0.00326620708117791\\
141	5.12390823414868e-05\\
142	3.91305069907111e-05\\
143	0.000369620942592456\\
144	0.00041320232969526\\
145	0.0251169523774743\\
146	0.00921416647831369\\
147	0.00350271105443222\\
148	0.000889977095853363\\
149	0.00161167711781857\\
150	0.000575518920722189\\
151	0.00543626603170834\\
152	0.00820636591044666\\
153	0.00436966906680484\\
154	0.00160301527158153\\
155	0.0139879390610951\\
156	0.00355408858712874\\
157	0.00362774917716451\\
158	0.000562923407493308\\
159	0.00531729068919106\\
160	0.000223753166846305\\
161	0.00940310241806247\\
162	0.00344953280121535\\
163	0.00496051986682699\\
164	0.00126038060131033\\
165	0.00354233068284933\\
166	0.000763069258518772\\
167	0.00720783859672424\\
168	0.00261090481568977\\
169	0.00653676255910868\\
170	0.00239801459761703\\
171	0.00141971061924231\\
172	0.000360723426577435\\
173	0.000993896179494429\\
174	0.00131081496169328\\
175	0.0123817629508449\\
176	0.00699358589409415\\
177	0.000467654082872048\\
178	0.000171559133014507\\
179	0.0137146003270142\\
180	0.00348463803612377\\
181	0.00400892929445278\\
182	0.000699535069926613\\
183	0.00660770411137614\\
184	0.0110802508473021\\
185	0.0112147766576068\\
186	0.00411414639751313\\
187	0.00120038499618114\\
188	0.000304996654364461\\
189	0.00152536615706363\\
190	0.000289393836163205\\
191	0.00273357108632613\\
192	0.00157808238480221\\
193	0.00873543974970219\\
194	0.0032046004190867\\
195	0.0205938972408153\\
196	0.00523254604044248\\
197	0.00232629581853464\\
198	0.0001993751036482\\
199	0.00188326754257004\\
200	2.74954213762184e-06\\
201	0.0212418929654884\\
202	0.00779259900472894\\
203	0.0203494027396575\\
204	0.00517042430025006\\
205	0.00290242193759935\\
206	0.00230356440275867\\
207	0.0217591263530576\\
208	0.00884746592989398\\
209	0.00707044951945841\\
210	0.00259379792459334\\
211	0.000487578962142162\\
212	0.000123885214047956\\
213	0.00195270769420965\\
214	0.000205379333704923\\
215	0.0019399826056699\\
216	0.00424040362877211\\
217	0.00632751861212068\\
218	0.00232125335154101\\
219	0.0146404861204323\\
220	0.00371988928486021\\
221	0.00447729116017169\\
222	0.00102725678037767\\
223	0.00970331459129366\\
224	0.00963322523624546\\
225	0.00129445138675058\\
226	0.000474870135371225\\
227	0.0216271779063142\\
228	0.0054950844318746\\
229	0.00160319618111635\\
230	2.28860487864473e-05\\
231	0.000216178209157158\\
232	0.00108922414934127\\
233	0.00814878772424372\\
234	0.00298938683161889\\
235	0.00932017643709453\\
236	0.00236809243737945\\
237	0.00184010567262849\\
238	2.40101866520383e-05\\
239	0.000226796648054008\\
240	0.0133289145716826\\
241	0.0071977380794045\\
242	0.00264049379615059\\
243	0.0194918523439561\\
244	0.0049525358707296\\
245	0.00373698789470482\\
246	0.000937543228152134\\
247	0.00885589373511126\\
248	0.00180245353070454\\
249	0.00448613330782302\\
250	0.00164574023635371\\
251	0.0212241932690687\\
252	0.0053926931436536\\
253	0.00282200742826718\\
254	0.000240098638352413\\
255	0.0022679359877461\\
256	0.00607671644269401\\
};
\addlegendentry{$\bs h\bs m^*$}

\addplot [color=mycolor2, line width=2.0pt]
  table[row sep=crcr]{%
1	0.0208993644098182\\
2	0.00766694223354651\\
3	0.00222142645403367\\
4	0.000564425279025636\\
5	0.00367920786260671\\
6	0.000169299645855414\\
7	0.00159917924634704\\
8	0.000234983362133089\\
9	0.00112445637009981\\
10	0.000412507378915722\\
11	0.00988353485805733\\
12	0.00251123187530658\\
13	0.00193137683470166\\
14	0.000547755179049522\\
15	0.00517401386157177\\
16	0.00170581188615309\\
17	0.00791012398496809\\
18	0.00290183292006959\\
19	0.0103308107961889\\
20	0.00262487680184598\\
21	0.00462635318336391\\
22	0.00273433710301813\\
23	0.0258281411373871\\
24	0.00422236539933234\\
25	0.000134850196238172\\
26	4.94698615932442e-05\\
27	0.00833341533209895\\
28	0.00211737384576279\\
29	0.000691090642072165\\
30	0.000628574210662766\\
31	0.0059374184003866\\
32	0.00406431272050172\\
33	0.0164722491202326\\
34	0.00604285278656995\\
35	0.00013153555897677\\
36	3.34208654274587e-05\\
37	0.000278116223048446\\
38	0.000633273035286153\\
39	0.00598180279813284\\
40	0.00778929088054963\\
41	0.0123122618434265\\
42	0.00451675938400805\\
43	0.0128018548349801\\
44	0.00325272550624348\\
45	0.00149146511147588\\
46	0.00308682358419969\\
47	0.029157675954045\\
48	0.00351434064734947\\
49	0.00163132237440624\\
50	0.000598451424818876\\
51	0.0129624121991647\\
52	0.00329352030046903\\
53	0.00168895399546154\\
54	0.00155477601813346\\
55	0.0146861827640239\\
56	0.00637286865975499\\
57	0.00580491851542978\\
58	0.00212953724599091\\
59	0.0046891239069097\\
60	0.00119142367497131\\
61	0.00168497603859877\\
62	0.000666963641246283\\
63	0.0063000392455635\\
64	0.0107998406101919\\
65	0.00836598538947506\\
66	0.00306906590315573\\
67	0.00553549588028459\\
68	0.00140647186455426\\
69	0.00130782116946765\\
70	0.00111080325123664\\
71	0.0104924821146378\\
72	0.00843377804624248\\
73	0.0147312056338537\\
74	0.00540415011722521\\
75	0.0155250896829649\\
76	0.00394465144695386\\
77	0.000658440697531822\\
78	0.000375447395104734\\
79	0.00354642019073869\\
80	0.000881191347121688\\
81	0.00989753050345489\\
82	0.00363091399033668\\
83	0.000446893954869589\\
84	0.000113547871330209\\
85	0.000828613405016411\\
86	0.000964927086844912\\
87	0.00911455758648386\\
88	0.00393937954761023\\
89	0.00849276803900048\\
90	0.00311557618122289\\
91	0.00962567394335876\\
92	0.00244571396519784\\
93	0.000353504868624717\\
94	0.000394294381907528\\
95	0.0037244460217965\\
96	0.0035010381811589\\
97	0.00567156625345277\\
98	0.00208061690232677\\
99	0.00405086193563203\\
100	0.0010292525661436\\
101	0.00544799891274306\\
102	0.002209574295614\\
103	0.0208713097948562\\
104	0.00231309807226573\\
105	0.00514197547465905\\
106	0.0018863362615948\\
107	0.00606241048776571\\
108	0.00154035157225747\\
109	0.00232445567342014\\
110	0.00040258444286939\\
111	0.00380275270326756\\
112	0.00125008959647541\\
113	0.00574599493525071\\
114	0.00210792110128106\\
115	0.0137924338000473\\
116	0.00350441414880005\\
117	0.00114608688067402\\
118	0.000536614433640332\\
119	0.00506878003927289\\
120	0.00650309475985993\\
121	0.00204807061730563\\
122	0.000751335725106047\\
123	0.0104248042673647\\
124	0.00264875888495463\\
125	0.00182026329451732\\
126	0.0025596516347402\\
127	0.024178088214376\\
128	0.00576502405243991\\
129	0.0136523866755724\\
130	0.00500838484554491\\
131	0.00549514569708291\\
132	0.00139621959472512\\
133	3.01553317541521e-05\\
134	0.00209664774279659\\
135	0.0198046223914972\\
136	0.00947447420702831\\
137	0.00536859923882482\\
138	0.00196947330225005\\
139	0.0128549146967386\\
140	0.00326620708121246\\
141	5.12390823342746e-05\\
142	3.91305069892227e-05\\
143	0.000369620942560299\\
144	0.000413202329687015\\
145	0.0251169523775166\\
146	0.00921416647822518\\
147	0.00350271105446441\\
148	0.000889977095875661\\
149	0.00161167711780689\\
150	0.000575518920756525\\
151	0.00543626603176649\\
152	0.00820636591037412\\
153	0.00436966906677423\\
154	0.0016030152715522\\
155	0.0139879390610564\\
156	0.00355408858717529\\
157	0.00362774917711781\\
158	0.000562923407526105\\
159	0.0053172906892405\\
160	0.000223753166841522\\
161	0.00940310241799159\\
162	0.0034495328011504\\
163	0.00496051986687045\\
164	0.00126038060134136\\
165	0.00354233068281543\\
166	0.000763069258553163\\
167	0.00720783859669616\\
168	0.00261090481569096\\
169	0.0065367625591433\\
170	0.00239801459760266\\
171	0.00141971061928766\\
172	0.00036072342659468\\
173	0.000993896179476598\\
174	0.00131081496175738\\
175	0.0123817629508441\\
176	0.00699358589403943\\
177	0.000467654082829584\\
178	0.000171559132996992\\
179	0.0137146003269112\\
180	0.00348463803615287\\
181	0.00400892929440022\\
182	0.000699535069954124\\
183	0.00660770411131246\\
184	0.0110802508472323\\
185	0.0112147766576494\\
186	0.00411414639748232\\
187	0.00120038499616152\\
188	0.000304996654364315\\
189	0.0015253661570691\\
190	0.000289393836181741\\
191	0.00273357108636738\\
192	0.00157808238483052\\
193	0.00873543974971407\\
194	0.00320460041905488\\
195	0.0205938972408108\\
196	0.00523254604052433\\
197	0.00232629581852506\\
198	0.000199375103659612\\
199	0.00188326754258562\\
200	2.74954214391743e-06\\
201	0.0212418929655618\\
202	0.00779259900466788\\
203	0.0203494027396143\\
204	0.00517042430032111\\
205	0.00290242193758672\\
206	0.00230356440286261\\
207	0.021759126352974\\
208	0.00884746592986658\\
209	0.00707044951946892\\
210	0.00259379792456792\\
211	0.000487578962125064\\
212	0.000123885214045577\\
213	0.00195270769419028\\
214	0.000205379333719871\\
215	0.00193998260571612\\
216	0.00424040362869689\\
217	0.00632751861210301\\
218	0.00232125335150832\\
219	0.0146404861204945\\
220	0.00371988928493504\\
221	0.00447729116013707\\
222	0.00102725678043095\\
223	0.00970331459132188\\
224	0.00963322523617071\\
225	0.00129445138669571\\
226	0.000474870135345735\\
227	0.0216271779062702\\
228	0.0054950844319506\\
229	0.00160319618110686\\
230	2.28860487936302e-05\\
231	0.000216178209214422\\
232	0.00108922414930533\\
233	0.00814878772426028\\
234	0.00298938683159122\\
235	0.00932017643709939\\
236	0.00236809243741825\\
237	0.00184010567261611\\
238	2.40101866615973e-05\\
239	0.000226796648133196\\
240	0.0133289145715986\\
241	0.00719773807937556\\
242	0.00264049379611016\\
243	0.0194918523439255\\
244	0.00495253587080039\\
245	0.0037369878946615\\
246	0.000937543228192191\\
247	0.00885589373505601\\
248	0.00180245353068702\\
249	0.00448613330777542\\
250	0.00164574023631766\\
251	0.0212241932690235\\
252	0.00539269314372765\\
253	0.00282200742824623\\
254	0.000240098638359993\\
255	0.00226793598770665\\
256	0.00607671644265806\\
};
\addlegendentry{$\mathcal{P}_T(\bs Y_2)$}

\end{axis}
\end{tikzpicture}%

%% file: ComparisonTperpFinal.tex
\definecolor{mycolor1}{rgb}{0.00000,0.44700,0.74100}%
\definecolor{mycolor2}{rgb}{0.85000,0.32500,0.09800}%

\definecolor{mycolor1}{rgb}{0,0.447,0.741}
\definecolor{mycolor2}{rgb}{0.8627,0.0784,0.2353}
\begin{tikzpicture}

\begin{axis}[%
scale=0.7,
width=6.806in,
height=2.012in,
at={(0.389in,3.237in)},
scale only axis,
xmin=1,
xmax=256,
ymin=4.97689229801025e-05,
ymax=0.107090388724753,
axis background/.style={fill=white},
title style={},
title={Comparison of the certificates $Y_2 = \mathcal{A}^*\mathcal{A}\mathcal{P}_T(\mathcal{P}_T\mathcal{A}^*\mathcal{A}\mathcal{P}_T)^{-1}\bs h\bs m^*$ and $Y_1 = \mathcal{A}^*\mathcal{A}\bs h\bs m^*$ on $T^\perp$},
legend style={legend cell align=left, align=left, draw=white!15!black}
]
\addplot [color=mycolor1, line width=2.0pt]
  table[row sep=crcr]{%
1	0.021960171029648\\
2	0.00290594291142637\\
3	0.021879605394768\\
4	0.00777983966628963\\
5	0.0129717995611379\\
6	0.0300616156407456\\
7	0.0131752512589718\\
8	0.00799402669116684\\
9	0.00252696860176789\\
10	0.013571430665576\\
11	0.00849240195422307\\
12	0.00817774819507911\\
13	0.0326203747986508\\
14	0.0224783536032311\\
15	0.0212731412145007\\
16	0.0225475122139778\\
17	0.0157465822803306\\
18	0.0213581908726841\\
19	0.00207271066618506\\
20	0.0215628075625092\\
21	0.00026524828918184\\
22	0.019920930221249\\
23	0.0289487960935292\\
24	0.017109117678643\\
25	0.0556345973828727\\
26	0.0337000823007916\\
27	0.0190127177020899\\
28	0.0143857040294857\\
29	0.0140596624740556\\
30	0.00772509477220368\\
31	0.00836132023875734\\
32	0.00606711390049784\\
33	0.00315922764311164\\
34	0.016636955559989\\
35	0.00618691241794265\\
36	0.00648886649276205\\
37	0.0153949846282708\\
38	0.0105963969538129\\
39	0.0320365376382137\\
40	0.0149475760113902\\
41	0.0236640801747478\\
42	0.0417270912338904\\
43	4.97689229801025e-05\\
44	0.0207684909242792\\
45	0.00556614134198313\\
46	0.016083470572149\\
47	0.0422639624788865\\
48	0.0115253195794046\\
49	0.0355975033838759\\
50	0.0579669436639747\\
51	0.0640610650239318\\
52	0.0142794320361543\\
53	0.0207090450393642\\
54	0.0396198128166691\\
55	0.0113884454694031\\
56	0.0173048949614414\\
57	0.00566246582136476\\
58	0.0209900151053349\\
59	0.0225397366027388\\
60	0.00588698460911247\\
61	0.0236305721892106\\
62	0.0240357622829028\\
63	0.00838272049138581\\
64	0.0484310166457202\\
65	0.026502605358953\\
66	0.0325099583849662\\
67	0.0128602959035165\\
68	0.0421813815365354\\
69	0.0141354560187936\\
70	0.0419912310733032\\
71	0.00299126361938533\\
72	0.018746967646812\\
73	0.0147691626332329\\
74	0.0203744324775202\\
75	0.0250655568766548\\
76	0.0080726379194783\\
77	0.011425963438218\\
78	0.0198794023466748\\
79	0.00683148069823963\\
80	0.00201586448402902\\
81	0.00932680785108905\\
82	0.0198555077567701\\
83	0.00417663701005187\\
84	0.0216825918580849\\
85	0.0115639182083799\\
86	0.0209730270872867\\
87	0.0132422464279937\\
88	0.0187278774061815\\
89	0.00924733624972785\\
90	0.00990074827078055\\
91	0.0076371994395874\\
92	0.00887041894095193\\
93	0.0353298393362471\\
94	0.0320786590407224\\
95	0.0314853885650545\\
96	0.0147724187563172\\
97	0.00553605703739311\\
98	0.00790745093604426\\
99	0.00356182168618269\\
100	0.055745722676911\\
101	0.0489274287736132\\
102	0.0262648391060016\\
103	0.00465927421528985\\
104	0.01554463655931\\
105	0.00703549946990163\\
106	0.00636097986095247\\
107	0.00884539371073417\\
108	0.0320589265277912\\
109	0.0300024866342896\\
110	0.0165393083243465\\
111	0.00186676386710919\\
112	0.00627689816366828\\
113	0.00205167307294765\\
114	0.0187881778433424\\
115	0.00868658149198343\\
116	0.014201484119515\\
117	0.0110626831128251\\
118	0.00377860139443538\\
119	0.0717779525634328\\
120	0.0641209716663681\\
121	0.00451568966651447\\
122	0.0122057370862071\\
123	0.0116444459033846\\
124	0.0193097830363341\\
125	0.0152820686651252\\
126	0.00917489054795182\\
127	0.0278958850411755\\
128	0.0090774068891785\\
129	0.0129987683371629\\
130	0.028250007163068\\
131	0.00695443336401192\\
132	0.00943790883912839\\
133	0.0254955949935434\\
134	0.04672029376384\\
135	0.0309534931671798\\
136	0.0416761174021871\\
137	0.00411396886027403\\
138	0.00863776888138701\\
139	0.0261312791580489\\
140	0.00281992766606542\\
141	0.00270898639370471\\
142	0.0154720939179034\\
143	0.00192507765986215\\
144	0.00737069013245258\\
145	0.00900785342554599\\
146	0.0100699084537701\\
147	0.00112928085249039\\
148	0.0117737411370209\\
149	0.0219253176010465\\
150	0.0157438318973856\\
151	0.0120774822768844\\
152	0.0121960487272295\\
153	0.0485501976841381\\
154	0.00300650155979881\\
155	0.0154112037276275\\
156	0.033408960886222\\
157	0.00803212049706556\\
158	0.000276641794752129\\
159	0.0224403678400728\\
160	0.0202492959684816\\
161	0.0433572525621508\\
162	0.0271165692671649\\
163	0.0198064278074302\\
164	0.0653678629725063\\
165	0.0332942221051328\\
166	0.0147087610645732\\
167	0.00405333299319137\\
168	0.0436844584980555\\
169	0.0119067621625134\\
170	0.0278630233877598\\
171	0.0496453147409837\\
172	0.00694822188387536\\
173	0.0306335305265414\\
174	0.0421867051113687\\
175	0.0180192249361274\\
176	0.00794440677614689\\
177	0.00948344589018028\\
178	0.000244571588370778\\
179	0.0301166405578607\\
180	0.0168690328781254\\
181	0.0223165365752064\\
182	0.00676964147195307\\
183	0.00656920213580889\\
184	0.0129553506410261\\
185	0.0165076744580838\\
186	0.0118869955393032\\
187	0.0351228496340302\\
188	0.0372348512290965\\
189	0.0466159541145785\\
190	0.0196570707531681\\
191	0.0783597998294707\\
192	0.107090388724753\\
193	0.0773516391945649\\
194	0.00828202634945415\\
195	0.00582128472939345\\
196	0.015690795508808\\
197	0.0538785532386345\\
198	0.0236624528534686\\
199	0.054525169485212\\
200	0.05662556536007\\
201	0.0144882433674809\\
202	0.0639757150182861\\
203	0.00796476954389366\\
204	0.0250838401201711\\
205	0.0210306771359301\\
206	0.0309558180396562\\
207	0.00242135191093462\\
208	0.0295560167569696\\
209	0.0402487526545116\\
210	0.0341897389479358\\
211	0.020499572813309\\
212	0.00855121927871823\\
213	0.0282808902320282\\
214	0.0461538922552231\\
215	0.0241496604669644\\
216	0.038621829257323\\
217	0.0504377907278036\\
218	0.0240567363192314\\
219	0.0476669279065287\\
220	0.0123662716038171\\
221	0.0511312576504865\\
222	0.0521119759922016\\
223	0.00347030340491048\\
224	0.00429850791382419\\
225	0.0150427288119176\\
226	0.0118974921193854\\
227	0.0251280794992246\\
228	0.0100715021303514\\
229	0.0435470239399456\\
230	0.0418467343502207\\
231	0.0520065803663701\\
232	0.00705805990144255\\
233	0.0105487344447877\\
234	0.0303291262190194\\
235	0.0523042819518294\\
236	0.00276656103143895\\
237	0.0384357252694202\\
238	0.0584239496879415\\
239	9.90532085724568e-05\\
240	0.0472345593534911\\
241	0.0418984038570913\\
242	0.0368374791191882\\
243	0.0682493217914651\\
244	0.0125777199577027\\
245	0.0502903047518362\\
246	0.0282516269250233\\
247	0.0312413388369324\\
248	0.0185792138703458\\
249	0.0444173255980522\\
250	0.061850324067537\\
251	0.0111375127684287\\
252	0.0242613483729819\\
253	0.0355434483885422\\
254	0.0189365786016022\\
255	0.0183680273470513\\
256	0.0225753142837684\\
};
\addlegendentry{$\mathcal{P}_T^\perp(\bs Y_2)$}

\addplot [color=mycolor2, line width=2.0pt]
  table[row sep=crcr]{%
1	0.044659605930333\\
2	0.023387096135982\\
3	0.024672638924581\\
4	0.0191545830361675\\
5	0.00131888670653409\\
6	0.0442929496069279\\
7	0.0239157330707851\\
8	0.0231286231682984\\
9	0.0202385304457408\\
10	0.0138111790906227\\
11	0.00253376494065552\\
12	0.0222625624076691\\
13	0.0100965607146528\\
14	0.0085903179058704\\
15	0.0102357366270356\\
16	0.0221137777854064\\
17	0.0212451085243698\\
18	0.013932997336875\\
19	0.00946157413659331\\
20	0.00563482708238546\\
21	0.00508290614924571\\
22	0.0194902945658796\\
23	0.00757224336734197\\
24	0.00274153833289391\\
25	0.0200443405023855\\
26	0.0121713967467228\\
27	0.0131207353217278\\
28	0.0195156942106163\\
29	0.000109818701667493\\
30	0.00230241316550876\\
31	0.0031619446069566\\
32	0.00433006380318616\\
33	0.00256067430224198\\
34	0.0215271962609017\\
35	0.00518541567756932\\
36	0.0106290130494572\\
37	0.0195620098331771\\
38	0.00970882355961443\\
39	0.0217943815770207\\
40	0.00394776193137061\\
41	0.0179094070706079\\
42	0.0513639308934445\\
43	0.0121602072980081\\
44	0.0404089206200658\\
45	0.0154440210387713\\
46	0.021978398645835\\
47	0.0247178960064811\\
48	0.00463602675063726\\
49	0.0493514175072003\\
50	0.0448732047424337\\
51	0.0418629053361557\\
52	0.00311856077327911\\
53	0.0167944459528115\\
54	0.0206171631533692\\
55	0.0229253929625968\\
56	0.0116269355791969\\
57	0.0129224304527403\\
58	0.0191904604530169\\
59	0.0144009256077004\\
60	0.000900754282200006\\
61	0.0226274762860862\\
62	0.0252823303945559\\
63	0.00903518083901602\\
64	0.048690333012091\\
65	0.0232355254787891\\
66	0.0229342797533218\\
67	0.0150008249678974\\
68	0.0328390746210703\\
69	0.00175430283197283\\
70	0.0309841840008175\\
71	0.00149784690981712\\
72	0.0208787305048443\\
73	0.000718678972934478\\
74	0.0215042877834622\\
75	0.0192049200255721\\
76	0.0127414641847675\\
77	0.00590704172600493\\
78	0.0158744646041809\\
79	0.00261328003774397\\
80	0.0115480765151033\\
81	0.00631705950473273\\
82	0.000470554145133314\\
83	0.0103667868203999\\
84	0.0180129243771165\\
85	0.00550429311740924\\
86	0.0192545294642229\\
87	0.00154247747657722\\
88	0.0132884941271074\\
89	0.011748233662531\\
90	0.0133518478069638\\
91	0.0191991094113515\\
92	0.00304747968638172\\
93	0.0424904488519062\\
94	0.0197888205837108\\
95	0.0218270535084103\\
96	0.0275053215138348\\
97	0.00877849310170863\\
98	0.0197128002638882\\
99	0.00123497347097659\\
100	0.0583842625539848\\
101	0.0376842822968821\\
102	0.0106406143660671\\
103	0.000156280472220448\\
104	0.0155292255333264\\
105	0.00528598055116773\\
106	0.00429225205927276\\
107	0.0102290111747397\\
108	0.0204407898092192\\
109	0.0316691254095201\\
110	0.00774849989443289\\
111	0.00683055038932701\\
112	0.00163508282876148\\
113	0.0031261194114254\\
114	0.0262159753805351\\
115	0.00630823252037028\\
116	0.0131071075489564\\
117	0.00311844825716684\\
118	0.00715617533897408\\
119	0.0496022464786051\\
120	0.0268943621779654\\
121	0.00604429061045869\\
122	0.00517904750649348\\
123	0.0164749690852245\\
124	0.00778346643959604\\
125	0.00251270218816361\\
126	0.00274156444483838\\
127	0.0210152395993282\\
128	0.0192282730906277\\
129	0.0369175018599206\\
130	0.0178741749519386\\
131	0.00396934326853742\\
132	0.0134732405003613\\
133	0.0129675788700249\\
134	0.0852553880056426\\
135	0.0155789868922997\\
136	0.0192562781277412\\
137	0.0126811412761229\\
138	0.0142629659093061\\
139	0.0397620788010622\\
140	0.0111086641317731\\
141	0.00916696937576829\\
142	0.0255052617504826\\
143	0.00823708536872262\\
144	0.00275903823916219\\
145	0.0249357957918976\\
146	0.00697671790029552\\
147	0.000921647981931456\\
148	0.0151878294737142\\
149	0.00140935142593119\\
150	0.0236502464699816\\
151	0.0152341748591426\\
152	0.00754573810722324\\
153	0.0413981508583985\\
154	0.00248768338559987\\
155	0.00118352240443551\\
156	0.0130568832816049\\
157	0.0121829941681405\\
158	0.0222957952289362\\
159	0.0293527976727858\\
160	0.00251322553525698\\
161	0.0406841567594215\\
162	0.0260989123445072\\
163	0.0418823417148746\\
164	0.0480786757357357\\
165	0.0232687252726041\\
166	0.00842004693362888\\
167	0.0359626792184811\\
168	0.0340466568614051\\
169	0.00833325069066388\\
170	0.0206599831155156\\
171	0.0286896749665115\\
172	0.00229187099963576\\
173	0.0255311228296311\\
174	0.0198077804168557\\
175	0.0109574406541726\\
176	0.00234918978507531\\
177	0.00500747697014853\\
178	0.00624892067173089\\
179	0.0215053203453362\\
180	0.00760958562722362\\
181	0.00340316664956754\\
182	0.00402601124761665\\
183	0.000252662568282136\\
184	0.0350186837057204\\
185	0.0139908101533385\\
186	0.0169069176843511\\
187	0.0311133670783371\\
188	0.0158525812912037\\
189	0.0362272325510256\\
190	0.00790200165607498\\
191	0.0317345179999494\\
192	0.0411175772251642\\
193	0.0368536849651043\\
194	0.00437521938735088\\
195	0.00561053339380973\\
196	0.00965793231018451\\
197	0.0382976983451577\\
198	0.0150986503379558\\
199	0.0326774114256638\\
200	0.0218135042272518\\
201	0.016421206520312\\
202	0.029284337949079\\
203	0.0065642340438659\\
204	0.0265163112140121\\
205	0.00864314793096604\\
206	0.0334395406599134\\
207	0.00683386216939533\\
208	0.0160431190738572\\
209	0.029132673292406\\
210	0.0297384313599811\\
211	0.0303142444179222\\
212	0.004745392311565\\
213	0.0360384382770586\\
214	0.0665418213322958\\
215	0.02186924752053\\
216	0.0190382074494493\\
217	0.0465881975301323\\
218	0.00302536209757769\\
219	0.0275510082171746\\
220	0.0202867142822969\\
221	0.0441918530638823\\
222	0.0504477039074906\\
223	0.0127473621906241\\
224	0.00846072769517708\\
225	0.00722918093542354\\
226	0.0140682190095045\\
227	0.018558374896015\\
228	0.0251867568536357\\
229	0.0190381299628275\\
230	0.0587979506129839\\
231	0.0433889855323281\\
232	0.0118474561866905\\
233	0.0132020311998755\\
234	0.0353868314714879\\
235	0.032843495535806\\
236	0.00407096518109023\\
237	0.0167468206273407\\
238	0.0147041649092139\\
239	0.0173148457383388\\
240	0.0333353635477265\\
241	0.0604494359479178\\
242	0.0347791410233023\\
243	0.0478696902626791\\
244	0.0243142386706202\\
245	0.0498983721525378\\
246	0.0226171471088764\\
247	0.0245012880769072\\
248	0.00247548114323833\\
249	0.0173601100764908\\
250	0.0592570340817077\\
251	0.0078405144966304\\
252	0.0200490343324435\\
253	0.0337093025551677\\
254	0.0230927163159695\\
255	0.0248047616717256\\
256	0.00620004318853432\\
};
\addlegendentry{$\mathcal{P}_T^\perp(\bs Y_1)$}

\end{axis}
\end{tikzpicture}%

%% file: plotArxivFinal1.tex
\definecolor{mycolor1}{rgb}{0.00000,0.44700,0.74100}%
\begin{tikzpicture}

\begin{axis}[%
width=\fwidth,
height=\fheight,
scale only axis,
separate axis lines,
every outer x axis line/.append style={darkgray!60!black},
every x tick label/.append style={font=\color{darkgray!60!black}},
at={(0.259in,2.139in)},
scale only axis,
xmin=0,
xmax=200,
ymin=-0.2,
ymax=0.2,
axis background/.style={fill=white},
title={Low resolution 1}
]
\addplot [color=mycolor1, forget plot,line width=0.5mm]
  table[row sep=crcr]{%
1	0.00465930083831772\\
2	-0.00331691404848387\\
3	0.00131773920932422\\
4	0.000736880895812353\\
5	-0.00229942550791691\\
6	0.00300309598761647\\
7	-0.00273757528560362\\
8	0.00165742622915965\\
9	-0.000124072902809996\\
10	-0.00140124969416619\\
11	0.0024845516037784\\
12	-0.00283308061967144\\
13	0.00236757815301265\\
14	-0.00123565968978043\\
15	-0.0002352195974041\\
16	0.00163113935148692\\
17	-0.00256570735553869\\
18	0.00278417757534228\\
19	-0.00223040445967589\\
20	0.00105957381408155\\
21	0.000405843765611382\\
22	-0.00176450176155021\\
23	0.00264499159155799\\
24	-0.00280595480131827\\
25	0.00220091478966192\\
26	-0.000990737522840935\\
27	-0.000499441295364696\\
28	0.00186728050552069\\
29	-0.00274111775412525\\
30	0.00287998880570749\\
31	-0.00223914060983927\\
32	0.000983465555052134\\
33	0.000554989555415935\\
34	-0.00196449276759062\\
35	0.00286314701166922\\
36	-0.00300152712879914\\
37	0.00233108246042704\\
38	-0.00101976323418022\\
39	-0.000589889737851745\\
40	0.00206977528242044\\
41	-0.00301936582946288\\
42	0.00317324629739066\\
43	-0.00247441981492121\\
44	0.00109336748704237\\
45	0.000613122000870988\\
46	-0.0021935102897222\\
47	0.00322011498335401\\
48	-0.00340383776711398\\
49	0.00267450956699823\\
50	-0.00120480805461891\\
51	-0.000629861758760674\\
52	0.00234638991601899\\
53	-0.00348004418704784\\
54	0.00370899569608021\\
55	-0.00294426321926368\\
56	0.00136031591269413\\
57	0.000643442805809877\\
58	-0.00254213675752554\\
59	0.00382134981941333\\
60	-0.00411471233239537\\
61	0.00330680711784407\\
62	-0.00157302482782849\\
63	-0.000656351772360952\\
64	0.00280106236153397\\
65	-0.00427963211091809\\
66	0.00466439200950187\\
67	-0.00380210076029708\\
68	0.00186695004142402\\
69	0.000670877249514711\\
70	-0.00315650553417293\\
71	0.00491583475745658\\
72	-0.00543418898113866\\
73	0.00450183687167388\\
74	-0.00228652740679735\\
75	-0.000689750209851512\\
76	0.00366881301989427\\
77	-0.00584350403073891\\
78	0.00656915381339336\\
79	-0.0055455994254968\\
80	0.00292060725181579\\
81	0.000717129118160909\\
82	-0.00446086021655575\\
83	0.00730070042694993\\
84	-0.00838182128331266\\
85	0.00724296253195468\\
86	-0.00397300123457756\\
87	-0.000760473763343713\\
88	0.00582805774631723\\
89	-0.00988373689748858\\
90	0.0116886368469882\\
91	-0.0104409375767264\\
92	0.00603292451719721\\
93	0.000832142858781931\\
94	-0.00870241848727529\\
95	0.0156192151720988\\
96	-0.019497071791841\\
97	0.0185672043100425\\
98	-0.0117908202369515\\
99	-0.00493213713458223\\
100	0.00558342656177428\\
101	-0.0264326525355788\\
102	0.13032544473625\\
103	-0.0859586041211799\\
104	-0.156609397766655\\
105	0.185131294398993\\
106	-0.00270372309489277\\
107	-0.0744791459515443\\
108	0.0573758521583257\\
109	-0.0373686235980185\\
110	0.017411139249784\\
111	-0.000243728761747558\\
112	-0.0120810351320317\\
113	0.0185206223640755\\
114	-0.0191545266283438\\
115	0.0150729960399814\\
116	-0.008076220310715\\
117	0.000255633564660609\\
118	0.00644972626061425\\
119	-0.010628339706409\\
120	0.0116289752987957\\
121	-0.00961198625260342\\
122	0.00542156524755536\\
123	-0.000320286158779367\\
124	-0.00434503826912241\\
125	0.00746634497800928\\
126	-0.00841404213796629\\
127	0.00714496624313993\\
128	-0.00416341919409036\\
129	0.000355061025625814\\
130	0.00325819584001476\\
131	-0.00577860462095346\\
132	0.00664491480490685\\
133	-0.00574847933892765\\
134	0.00343089812709265\\
135	-0.000370810522906455\\
136	-0.00260606839811818\\
137	0.00474266924043256\\
138	-0.00553736694642979\\
139	0.00485701775055242\\
140	-0.00295177274920598\\
141	0.00037268865679759\\
142	0.00218272153546237\\
143	-0.00405497881428097\\
144	0.00478940558985573\\
145	-0.00424379536322577\\
146	0.00261209772168508\\
147	-0.00036135030352947\\
148	-0.00189888264526889\\
149	0.00357888344330038\\
150	-0.0042604121350617\\
151	0.00380757436652693\\
152	-0.0023271534115163\\
153	0.000382192568194598\\
154	0.00179358776228596\\
155	-0.00313317524605754\\
156	0.00402815157237609\\
157	-0.00326783152914075\\
158	0.00238373114980378\\
159	1.94334087380965e-05\\
160	-0.00126425390911275\\
161	0.00343471469342407\\
162	-0.00314297516457721\\
163	0.00444590781159389\\
164	0.000895141056517406\\
165	0.00505914687185495\\
166	0.0092121366891496\\
167	0.00127296148396519\\
168	-0.00256939266075346\\
169	-0.0226089412503548\\
170	-0.036851936821143\\
171	-0.0307479102778015\\
172	0.00525903078387847\\
173	0.0607470007534062\\
174	0.127190584685893\\
175	0.0672719182519593\\
176	-0.158147108240976\\
177	-0.15419461803309\\
178	0.0600617210635463\\
179	0.0825745818850831\\
180	-0.0100187787753938\\
181	-0.0260997105566183\\
182	0.00111242065664737\\
183	-0.00120625285574739\\
184	-0.00813695561132968\\
185	0.000251753939067187\\
186	-0.00758248779161592\\
187	-0.000437438676825044\\
188	-0.00323632601267295\\
189	-0.00308662552385153\\
190	0.0183163924274325\\
191	-0.010960715164633\\
192	-0.0143892370993463\\
193	0.0263150276568802\\
194	0.00824376704270032\\
195	-0.00272275860443068\\
196	0.0170200840789255\\
197	0.00933505093074593\\
198	0.0135241584481246\\
199	0.0185199514532914\\
200	0.00866789549203257\\
};
\end{axis}

\begin{axis}[%
xshift=-1cm,
width=\fwidth,
height=\fheight,
scale only axis,
separate axis lines,
every outer x axis line/.append style={darkgray!60!black},
every x tick label/.append style={font=\color{darkgray!60!black}},
at={(4.019in,2.139in)},
scale only axis,
xmin=0,
xmax=200,
ymin=-0.3,
ymax=0.2,
axis background/.style={fill=white},
title={Low resolution 2}
]
\addplot [color=mycolor1, forget plot,line width=0.5mm]
  table[row sep=crcr]{%
1	-0.000286061624359592\\
2	-0.000846667906693175\\
3	0.001715523393525\\
4	-0.00209882415112596\\
5	0.00191105088661467\\
6	-0.00122139511355449\\
7	0.000230597735944419\\
8	0.000786623368187795\\
9	-0.0015572758338678\\
10	0.00188320632204591\\
11	-0.00169209709248616\\
12	0.00105320638070935\\
13	-0.000154490436880303\\
14	-0.000751772740162952\\
15	0.00141937921240519\\
16	-0.0016751158210419\\
17	0.0014640347599792\\
18	-0.000861172924188925\\
19	-0.000116994879267343\\
20	0.000235053245278522\\
21	-0.00225363574801269\\
22	-0.000103084962589772\\
23	-0.00198776556690191\\
24	0.00214318530327393\\
25	0.00470830283920787\\
26	0.00809875082833615\\
27	0.0082189801018279\\
28	-0.00281986341450236\\
29	-0.012371115960746\\
30	-0.0302013811405909\\
31	-0.0152875186140071\\
32	0.0362304759729051\\
33	0.0335674417961683\\
34	-0.0128079798451848\\
35	-0.0214302141303982\\
36	-0.000459498297174085\\
37	0.000321203679293801\\
38	-0.0100815910749637\\
39	-0.00356632594041721\\
40	0.00750649493915219\\
41	0.0233722725415592\\
42	0.0466219765973374\\
43	0.034704093325862\\
44	-0.00679037728879114\\
45	-0.0689936225229944\\
46	-0.153929980059759\\
47	-0.0752973707854892\\
48	0.179297762108189\\
49	0.180949135325234\\
50	-0.0728577421085516\\
51	-0.102805485267936\\
52	0.0136529916849763\\
53	0.0189749228348323\\
54	0.000971673722761715\\
55	-0.00923756317092622\\
56	0.00585929288042115\\
57	0.00111854829619137\\
58	5.72723408084225e-05\\
59	0.00901435735812265\\
60	-0.0582867269764356\\
61	0.0383179991056855\\
62	0.074772284835243\\
63	-0.0863856718954189\\
64	-0.00151843908453392\\
65	0.0389231497242348\\
66	-0.0307611637773025\\
67	0.0202613831119902\\
68	-0.00907367851531699\\
69	-0.00105886540171644\\
70	0.00863782851668232\\
71	-0.0126992451091541\\
72	0.0129922800744833\\
73	-0.010007386755693\\
74	0.00484757054658649\\
75	0.00102861787653968\\
76	-0.0061348269263989\\
77	0.009281827996352\\
78	-0.00982680966181874\\
79	0.00779176860215084\\
80	-0.00382611027715007\\
81	-0.000975203711885713\\
82	0.00535303510310783\\
83	-0.00819702840288825\\
84	0.00881525287960075\\
85	-0.00709647373737152\\
86	0.00352712972523722\\
87	0.000939675258225918\\
88	-0.00512995059844402\\
89	0.007948220209058\\
90	-0.0086562512526944\\
91	0.00706174317896482\\
92	-0.00356955779395233\\
93	-0.000918163808886303\\
94	0.00522920758743511\\
95	-0.00822003954562146\\
96	0.00907172589758147\\
97	-0.00750819216525019\\
98	0.00388010181328256\\
99	0.000909831481397385\\
100	-0.00562506609499053\\
101	0.00900558278589486\\
102	-0.0100944527387752\\
103	0.00849840029679282\\
104	-0.00451403541845614\\
105	-0.000917325574179366\\
106	0.00642190834194957\\
107	-0.010526680353108\\
108	0.0120354411525458\\
109	-0.0103557195406378\\
110	0.00569229143089311\\
111	0.000948204356287293\\
112	-0.00795380713035601\\
113	0.0134691446368323\\
114	-0.0158453675774437\\
115	0.0140757531909401\\
116	-0.00811262694710512\\
117	-0.00101646245971983\\
118	0.0113331615323808\\
119	-0.0202601038960024\\
120	0.025134757355032\\
121	-0.0237807603071422\\
122	0.0150150102246644\\
123	0.00595025031216452\\
124	-0.00749779066151394\\
125	0.033212288112627\\
126	-0.15882139147255\\
127	0.105055872929902\\
128	0.188310993708331\\
129	-0.223470860367916\\
130	0.00396685113099463\\
131	0.088685005975975\\
132	-0.0678479542765364\\
133	0.0438821905172866\\
134	-0.0202983307333077\\
135	0.000265595857128709\\
136	0.0139232131653085\\
137	-0.0211915315523397\\
138	0.021766433186146\\
139	-0.0170104739046682\\
140	0.00904735366515859\\
141	-0.000270576050362674\\
142	-0.00715667071755788\\
143	0.0117070977466333\\
144	-0.0127237459204841\\
145	0.0104482555617067\\
146	-0.00585464080675999\\
147	0.000341448467303136\\
148	0.00463786244897374\\
149	-0.00791896032111643\\
150	0.00887019200795681\\
151	-0.00748917289887775\\
152	0.00434295327482321\\
153	-0.000378922435516281\\
154	-0.00333852781608218\\
155	0.0058973534144195\\
156	-0.00674860766043704\\
157	0.0058125274803426\\
158	-0.00346210946973506\\
159	0.000396952549025753\\
160	0.00255424812218322\\
161	-0.00464940089972881\\
162	0.00541273735628557\\
163	-0.00473656824662533\\
164	0.00288479361074342\\
165	-0.000403294040014477\\
166	-0.00203495660596304\\
167	0.00380761894491273\\
168	-0.0044977985041203\\
169	0.00398837876105718\\
170	-0.00247542483126585\\
171	0.000401542262006285\\
172	0.00166986940384651\\
173	-0.00320469290106706\\
174	0.00383327563961715\\
175	-0.00343711985284578\\
176	0.00216731918836357\\
177	-0.000393189251685381\\
178	-0.00140265678826864\\
179	0.00275361694133753\\
180	-0.00332865473789143\\
181	0.00301188121109785\\
182	-0.00192340817748685\\
183	0.000378510356802575\\
184	0.00120192910507113\\
185	-0.00240470730769337\\
186	0.00293118842299077\\
187	-0.00267042588792283\\
188	0.00172086651684067\\
189	-0.000356895267494826\\
190	-0.00104912599947197\\
191	0.00212751873447952\\
192	-0.00260769804543471\\
193	0.00238538999040356\\
194	-0.00154412239071061\\
195	0.00032689367032099\\
196	0.000932958305967389\\
197	-0.00190229689443983\\
198	0.00233574514017623\\
199	-0.00213738737536715\\
200	0.00138119712475934\\
};
\end{axis}

\begin{axis}[%
width=\fwidth,
height=\fheight,
scale only axis,
separate axis lines,
every outer x axis line/.append style={darkgray!60!black},
every x tick label/.append style={font=\color{darkgray!60!black}},
at={(0.259in,3.986in)},
scale only axis,
xmin=0,
xmax=200,
ymin=-0.4,
ymax=0.4,
axis background/.style={fill=white},
title={Original 1}
]
\addplot [color=mycolor1, forget plot,line width=0.5mm]
  table[row sep=crcr]{%
1	0\\
2	0\\
3	0\\
4	0\\
5	0\\
6	0\\
7	0\\
8	0\\
9	0\\
10	0\\
11	0\\
12	0\\
13	0\\
14	0\\
15	0\\
16	0\\
17	0\\
18	0\\
19	0\\
20	0\\
21	0\\
22	0\\
23	0\\
24	0\\
25	0\\
26	0\\
27	0\\
28	0\\
29	0\\
30	0\\
31	0\\
32	0\\
33	0\\
34	0\\
35	0\\
36	0\\
37	0\\
38	0\\
39	0\\
40	0\\
41	0\\
42	0\\
43	0\\
44	0\\
45	0\\
46	0\\
47	0\\
48	0\\
49	0\\
50	0\\
51	0\\
52	0\\
53	0\\
54	0\\
55	0\\
56	0\\
57	0\\
58	0\\
59	0\\
60	0\\
61	0\\
62	0\\
63	0\\
64	0\\
65	0\\
66	0\\
67	0\\
68	0\\
69	0\\
70	0\\
71	0\\
72	0\\
73	0\\
74	0\\
75	0\\
76	0\\
77	0\\
78	0\\
79	0\\
80	0\\
81	0\\
82	0\\
83	0\\
84	0\\
85	0\\
86	0\\
87	0\\
88	0\\
89	0\\
90	0\\
91	0\\
92	0\\
93	0\\
94	0\\
95	0\\
96	0\\
97	0\\
98	0\\
99	-0.00373051717316239\\
100	-0.0115755391951491\\
101	0.0108568408481115\\
102	0.0658403438175068\\
103	-0.00985089878612181\\
104	-0.222083826463113\\
105	0.251641625157386\\
106	-0.0810980282054584\\
107	0\\
108	0\\
109	0\\
110	0\\
111	0\\
112	0\\
113	0\\
114	0\\
115	0\\
116	0\\
117	0\\
118	0\\
119	0\\
120	0\\
121	0\\
122	0\\
123	0\\
124	0\\
125	0\\
126	0\\
127	0\\
128	0\\
129	0\\
130	0\\
131	0\\
132	0\\
133	0\\
134	0\\
135	0\\
136	0\\
137	0\\
138	0\\
139	0\\
140	0\\
141	0\\
142	0\\
143	0\\
144	0\\
145	0\\
146	0\\
147	0\\
148	0\\
149	0\\
150	0\\
151	7.69392110682931e-06\\
152	2.38737100518669e-05\\
153	4.49432166443214e-05\\
154	7.31439134730651e-05\\
155	0.000104074055893115\\
156	0.000137585814336866\\
157	0.000177935234807332\\
158	0.000224343468598561\\
159	0.00027129580859171\\
160	0.000318968437465827\\
161	0.000369620877819756\\
162	0.000422002187612073\\
163	0.00113928223220231\\
164	0.00258965105018521\\
165	0.00446206597595127\\
166	0.00694951969458869\\
167	0.00383851944078271\\
168	-0.00548461268707073\\
169	-0.017926571687872\\
170	-0.035254429775193\\
171	-0.0281012829077763\\
172	0.00638799952782684\\
173	0.0527922115404258\\
174	0.119315445887919\\
175	0.058881424233174\\
176	-0.143167511424887\\
177	-0.140591922726901\\
178	0.0529752496611853\\
179	0.078448106666805\\
180	-0.0123824847268994\\
181	-0.0213488003204193\\
182	4.96426418115869e-05\\
183	-0.0020811979253784\\
184	-0.00464550042390388\\
185	-0.00366873153795451\\
186	-0.00306842845005148\\
187	-0.00286943654691587\\
188	-0.00313041144142476\\
189	0.00084964842599228\\
190	0.00933398882897311\\
191	0.000649934723316456\\
192	-0.0259661372074788\\
193	0.0374740807770349\\
194	-0.00483792472715695\\
195	0.00736543639864456\\
196	0.00918571160681934\\
197	0.0111074282158192\\
198	0.0131816285816044\\
199	0.0136573338444448\\
200	0.0123414675684024\\
};
\end{axis}

\begin{axis}[%
xshift=-1cm,
width=\fwidth,
height=\fheight,
scale only axis,
separate axis lines,
every outer x axis line/.append style={darkgray!60!black},
every x tick label/.append style={font=\color{darkgray!60!black}},
at={(4.019in,3.986in)},
scale only axis,
xmin=0,
xmax=200,
ymin=-0.4,
ymax=0.4,
axis background/.style={fill=white},
title={Original 2}
]
\addplot [color=mycolor1, forget plot,line width=0.5mm]
  table[row sep=crcr]{%
1	0\\
2	0\\
3	0\\
4	0\\
5	0\\
6	0\\
7	0\\
8	0\\
9	0\\
10	0\\
11	0\\
12	0\\
13	0\\
14	0\\
15	0\\
16	0\\
17	0\\
18	0\\
19	-0.000150041144346376\\
20	-0.000465567390962221\\
21	-0.000876449285393847\\
22	-0.00142639836400737\\
23	-0.00071646438840749\\
24	0.00139139593718514\\
25	0.00420043447166738\\
26	0.00810838871455793\\
27	0.00643453994969401\\
28	-0.00147117720668598\\
29	-0.012104739418711\\
30	-0.0273334446345161\\
31	-0.0136398289446236\\
32	0.0323149288823354\\
33	0.0316493875093995\\
34	-0.0125279530495936\\
35	-0.0191952639819224\\
36	-0.000231198121055781\\
37	-0.000412601393343244\\
38	-0.00823451835538304\\
39	-0.00414849671568392\\
40	0.00730240347004614\\
41	0.0215920058869101\\
42	0.0416805399629512\\
43	0.0330762508998743\\
44	-0.00756247172089371\\
45	-0.0622234691556969\\
46	-0.14050544091067\\
47	-0.0701144771702277\\
48	0.16611237227217\\
49	0.162691208737934\\
50	-0.0643989658265945\\
51	-0.0947070798230155\\
52	0.011113668252435\\
53	0.0210382968381128\\
54	-0.00463783671487898\\
55	-0.00239321287222001\\
56	0.000771274803560894\\
57	0.00176484819523519\\
58	0.00547620303275946\\
59	-0.00513619830370667\\
60	-0.0311480169012301\\
61	0.00466030315292556\\
62	0.105064317393852\\
63	-0.119047640686395\\
64	0.0383661841165596\\
65	0\\
66	0\\
67	0\\
68	0\\
69	0\\
70	0\\
71	0\\
72	0\\
73	0\\
74	0\\
75	0\\
76	0\\
77	0\\
78	0\\
79	0\\
80	0\\
81	0\\
82	0\\
83	0\\
84	0\\
85	0\\
86	0\\
87	0\\
88	0\\
89	0\\
90	0\\
91	0\\
92	0\\
93	0\\
94	0\\
95	0\\
96	0\\
97	0\\
98	0\\
99	0\\
100	0\\
101	0\\
102	0\\
103	0\\
104	0\\
105	0\\
106	0\\
107	0\\
108	0\\
109	0\\
110	0\\
111	0\\
112	0\\
113	0\\
114	0\\
115	0\\
116	0\\
117	0\\
118	0\\
119	0\\
120	0\\
121	0\\
122	0\\
123	0.00450343146481491\\
124	0.0139738392865893\\
125	-0.0131062360563873\\
126	-0.0794816006035538\\
127	0.0118918759761455\\
128	0.268096686195563\\
129	-0.303778473597071\\
130	0.0979004773338996\\
131	0\\
132	0\\
133	0\\
134	0\\
135	0\\
136	0\\
137	0\\
138	0\\
139	0\\
140	0\\
141	0\\
142	0\\
143	0\\
144	0\\
145	0\\
146	0\\
147	0\\
148	0\\
149	0\\
150	0\\
151	0\\
152	0\\
153	0\\
154	0\\
155	0\\
156	0\\
157	0\\
158	0\\
159	0\\
160	0\\
161	0\\
162	0\\
163	0\\
164	0\\
165	0\\
166	0\\
167	0\\
168	0\\
169	0\\
170	0\\
171	0\\
172	0\\
173	0\\
174	0\\
175	0\\
176	0\\
177	0\\
178	0\\
179	0\\
180	0\\
181	0\\
182	0\\
183	0\\
184	0\\
185	0\\
186	0\\
187	0\\
188	0\\
189	0\\
190	0\\
191	0\\
192	0\\
193	0\\
194	0\\
195	0\\
196	0\\
197	0\\
198	0\\
199	0\\
200	0\\
};
\end{axis}

\begin{axis}[%
width=\fwidth,
height=\fheight,
scale only axis,
separate axis lines,
every outer x axis line/.append style={darkgray!60!black},
every x tick label/.append style={font=\color{darkgray!60!black}},
at={(0.259in,0.292in)},
scale only axis,
xmin=0,
xmax=200,
ymin=-0.4,
ymax=0.4,
axis background/.style={fill=white},
title={Recovered 1}
]
\addplot [color=mycolor1, forget plot,line width=0.5mm]
  table[row sep=crcr]{%
1	0\\
2	0\\
3	0\\
4	0\\
5	0\\
6	0\\
7	0\\
8	0\\
9	0\\
10	0\\
11	0\\
12	0\\
13	0\\
14	0\\
15	0\\
16	0\\
17	0\\
18	0\\
19	0\\
20	0\\
21	0\\
22	0\\
23	0\\
24	0\\
25	0\\
26	0\\
27	0\\
28	0\\
29	0\\
30	0\\
31	0\\
32	0\\
33	0\\
34	0\\
35	0\\
36	0\\
37	0\\
38	0\\
39	0\\
40	0\\
41	0\\
42	0\\
43	0\\
44	0\\
45	0\\
46	0\\
47	0\\
48	0\\
49	0\\
50	0\\
51	0\\
52	0\\
53	0\\
54	0\\
55	0\\
56	0\\
57	0\\
58	0\\
59	0\\
60	0\\
61	0\\
62	0\\
63	0\\
64	0\\
65	0\\
66	0\\
67	0\\
68	0\\
69	0\\
70	0\\
71	0\\
72	0\\
73	0\\
74	0\\
75	0\\
76	0\\
77	0\\
78	0\\
79	0\\
80	0\\
81	0\\
82	0\\
83	0\\
84	0\\
85	0\\
86	0\\
87	0\\
88	0\\
89	0\\
90	0\\
91	0\\
92	0\\
93	0\\
94	0\\
95	0\\
96	0\\
97	0\\
98	0\\
99	-0.00351313804668061\\
100	-0.0109010266592199\\
101	0.0102242072291343\\
102	0.0620037936122769\\
103	-0.00927688192095555\\
104	-0.209142889332576\\
105	0.236978340115651\\
106	-0.0763724049976303\\
107	0\\
108	0\\
109	0\\
110	0\\
111	0\\
112	0\\
113	0\\
114	0\\
115	0\\
116	0\\
117	0\\
118	0\\
119	0\\
120	0\\
121	0\\
122	0\\
123	0\\
124	0\\
125	0\\
126	0\\
127	0\\
128	0\\
129	0\\
130	0\\
131	0\\
132	0\\
133	0\\
134	0\\
135	0\\
136	0\\
137	0\\
138	0\\
139	0\\
140	0\\
141	0\\
142	0\\
143	0\\
144	0\\
145	0\\
146	0\\
147	0\\
148	0\\
149	0\\
150	0\\
151	-5.30346831698016e-06\\
152	-1.64562988248556e-05\\
153	-3.09796425290564e-05\\
154	-5.04185606140431e-05\\
155	-7.17389027499684e-05\\
156	-9.48387690840501e-05\\
157	-0.000122651878953826\\
158	-0.000154641367036869\\
159	-0.000187005910954627\\
160	-0.000219866954538323\\
161	-0.000254782000958066\\
162	-0.000290888767979487\\
163	0.000320946063654758\\
164	0.00164758770622277\\
165	0.0033863720293048\\
166	0.00572655519349118\\
167	0.00263661033801139\\
168	-0.00647825027131745\\
169	-0.0186131041141609\\
170	-0.0354824346607803\\
171	-0.0280717773959356\\
172	0.00644377111903955\\
173	0.0528588148027425\\
174	0.119285897314719\\
175	0.0599797701375527\\
176	-0.139573238285297\\
177	-0.136441539638581\\
178	0.055853019685268\\
179	0.0817661733605682\\
180	-0.00740531273156058\\
181	-0.0154599403545823\\
182	0.00661795671376851\\
183	0.00486723334745235\\
184	0.0020959122370631\\
185	0.00252887977468045\\
186	0.00211508707222433\\
187	0.00118501863682445\\
188	-0.000302508435239946\\
189	0.00172189715506959\\
190	0.00756004600407074\\
191	-0.00254176150837254\\
192	-0.0293041425557245\\
193	0.0276540012091323\\
194	-0.0139021529530547\\
195	-0.0050770417370502\\
196	-0.00633176891201051\\
197	-0.00765642028398766\\
198	-0.0090861796706866\\
199	-0.00941408630693424\\
200	-0.00850705138839587\\
};
\end{axis}

\begin{axis}[%
xshift=-1cm,
width=\fwidth,
height=\fheight,
scale only axis,
separate axis lines,
every outer x axis line/.append style={darkgray!60!black},
every x tick label/.append style={font=\color{darkgray!60!black}},
at={(4.019in,0.292in)},
scale only axis,
xmin=0,
xmax=200,
ymin=-0.4,
ymax=0.4,
axis background/.style={fill=white},
title={Recovered 2}
]
\addplot [color=mycolor1, forget plot,line width=0.5mm]
  table[row sep=crcr]{%
1	0\\
2	0\\
3	0\\
4	0\\
5	0\\
6	0\\
7	0\\
8	0\\
9	0\\
10	0\\
11	0\\
12	0\\
13	0\\
14	0\\
15	0\\
16	0\\
17	0\\
18	0\\
19	-0.000151215002073358\\
20	-0.00046920979106316\\
21	-0.00088330624966488\\
22	-0.00143755789460572\\
23	-0.000722069699284742\\
24	0.00140228162376979\\
25	0.00423329687406194\\
26	0.00817182527915777\\
27	0.00648488103761647\\
28	-0.00148268706779343\\
29	-0.0121994417216138\\
30	-0.0275472898123108\\
31	-0.0137465411312783\\
32	0.0325677470618465\\
33	0.031896998777922\\
34	-0.012625966394897\\
35	-0.0193241596410194\\
36	-0.000166976922968838\\
37	-0.000291525299605699\\
38	-0.00809664000490295\\
39	-0.00407933883538131\\
40	0.0071621969381565\\
41	0.0211651977826944\\
42	0.0408566428068719\\
43	0.0324224342968649\\
44	-0.00741298471930224\\
45	-0.0609935009421292\\
46	-0.137728076863014\\
47	-0.0687285277945265\\
48	0.162828837288547\\
49	0.159475299723381\\
50	-0.0631259946787605\\
51	-0.092835009696378\\
52	0.0108939849259971\\
53	0.0206224338730677\\
54	-0.00454616082768688\\
55	-0.00234590635265316\\
56	0.000756029052959431\\
57	0.0016528019032629\\
58	0.0051285310655247\\
59	-0.00481011248152748\\
60	-0.0291704984917172\\
61	0.00436443085684826\\
62	0.0983940172431853\\
63	-0.111489570398552\\
64	0.0359304003029756\\
65	0\\
66	0\\
67	0\\
68	0\\
69	0\\
70	0\\
71	0\\
72	0\\
73	0\\
74	0\\
75	0\\
76	0\\
77	0\\
78	0\\
79	0\\
80	0\\
81	0\\
82	0\\
83	0\\
84	0\\
85	0\\
86	0\\
87	0\\
88	0\\
89	0\\
90	0\\
91	0\\
92	0\\
93	0\\
94	0\\
95	0\\
96	0\\
97	0\\
98	0\\
99	0\\
100	0\\
101	0\\
102	0\\
103	0\\
104	0\\
105	0\\
106	0\\
107	0\\
108	0\\
109	0\\
110	0\\
111	0\\
112	0\\
113	0\\
114	0\\
115	0\\
116	0\\
117	0\\
118	0\\
119	0\\
120	0\\
121	0\\
122	0\\
123	0.00423263436496211\\
124	0.0131335744391763\\
125	-0.0123181413020239\\
126	-0.0747022701966714\\
127	0.0111768022481873\\
128	0.251975689202687\\
129	-0.285511885043318\\
130	0.0920135962870013\\
131	0\\
132	0\\
133	0\\
134	0\\
135	0\\
136	0\\
137	0\\
138	0\\
139	0\\
140	0\\
141	0\\
142	0\\
143	0\\
144	0\\
145	0\\
146	0\\
147	0\\
148	0\\
149	0\\
150	0\\
151	0\\
152	0\\
153	0\\
154	0\\
155	0\\
156	0\\
157	0\\
158	0\\
159	0\\
160	0\\
161	0\\
162	0\\
163	0\\
164	0\\
165	0\\
166	0\\
167	0\\
168	0\\
169	0\\
170	0\\
171	0\\
172	0\\
173	0\\
174	0\\
175	0\\
176	0\\
177	0\\
178	0\\
179	0\\
180	0\\
181	0\\
182	0\\
183	0\\
184	0\\
185	0\\
186	0\\
187	0\\
188	0\\
189	0\\
190	0\\
191	0\\
192	0\\
193	0\\
194	0\\
195	0\\
196	0\\
197	0\\
198	0\\
199	0\\
200	0\\
};
\end{axis}
\end{tikzpicture}%

%% file: plotArxivFinal2b.tex
\definecolor{mycolor1}{rgb}{0.00000,0.44700,0.74100}%
\definecolor{mycolor2}{rgb}{0.85000,0.32500,0.09800}%
\definecolor{mycolor1}{rgb}{0,0.447,0.741}
\definecolor{mycolor2}{rgb}{0.8627,0.0784,0.2353}
\begin{tikzpicture}

\begin{axis}[%
width=\fwidth,
height=\fheight,
scale only axis,
separate axis lines,
every outer x axis line/.append style={darkgray!60!black},
every x tick label/.append style={font=\color{darkgray!60!black}},
at={(0.259in,0.292in)},
scale only axis,
xmin=67.0634920634921,
xmax=190.873015873016,
ymin=-0.211173184357542,
ymax=0.251396648044693,
axis background/.style={fill=white},
title={Original + low pass}
]
\addplot [color=mycolor1, forget plot,line width=0.5mm]
  table[row sep=crcr]{%
1	0.00465930083831772\\
2	-0.00331691404848387\\
3	0.00131773920932422\\
4	0.000736880895812353\\
5	-0.00229942550791691\\
6	0.00300309598761647\\
7	-0.00273757528560362\\
8	0.00165742622915965\\
9	-0.000124072902809996\\
10	-0.00140124969416619\\
11	0.0024845516037784\\
12	-0.00283308061967144\\
13	0.00236757815301265\\
14	-0.00123565968978043\\
15	-0.0002352195974041\\
16	0.00163113935148692\\
17	-0.00256570735553869\\
18	0.00278417757534228\\
19	-0.00223040445967589\\
20	0.00105957381408155\\
21	0.000405843765611382\\
22	-0.00176450176155021\\
23	0.00264499159155799\\
24	-0.00280595480131827\\
25	0.00220091478966192\\
26	-0.000990737522840935\\
27	-0.000499441295364696\\
28	0.00186728050552069\\
29	-0.00274111775412525\\
30	0.00287998880570749\\
31	-0.00223914060983927\\
32	0.000983465555052134\\
33	0.000554989555415935\\
34	-0.00196449276759062\\
35	0.00286314701166922\\
36	-0.00300152712879914\\
37	0.00233108246042704\\
38	-0.00101976323418022\\
39	-0.000589889737851745\\
40	0.00206977528242044\\
41	-0.00301936582946288\\
42	0.00317324629739066\\
43	-0.00247441981492121\\
44	0.00109336748704237\\
45	0.000613122000870988\\
46	-0.0021935102897222\\
47	0.00322011498335401\\
48	-0.00340383776711398\\
49	0.00267450956699823\\
50	-0.00120480805461891\\
51	-0.000629861758760674\\
52	0.00234638991601899\\
53	-0.00348004418704784\\
54	0.00370899569608021\\
55	-0.00294426321926368\\
56	0.00136031591269413\\
57	0.000643442805809877\\
58	-0.00254213675752554\\
59	0.00382134981941333\\
60	-0.00411471233239537\\
61	0.00330680711784407\\
62	-0.00157302482782849\\
63	-0.000656351772360952\\
64	0.00280106236153397\\
65	-0.00427963211091809\\
66	0.00466439200950187\\
67	-0.00380210076029708\\
68	0.00186695004142402\\
69	0.000670877249514711\\
70	-0.00315650553417293\\
71	0.00491583475745658\\
72	-0.00543418898113866\\
73	0.00450183687167388\\
74	-0.00228652740679735\\
75	-0.000689750209851512\\
76	0.00366881301989427\\
77	-0.00584350403073891\\
78	0.00656915381339336\\
79	-0.0055455994254968\\
80	0.00292060725181579\\
81	0.000717129118160909\\
82	-0.00446086021655575\\
83	0.00730070042694993\\
84	-0.00838182128331266\\
85	0.00724296253195468\\
86	-0.00397300123457756\\
87	-0.000760473763343713\\
88	0.00582805774631723\\
89	-0.00988373689748858\\
90	0.0116886368469882\\
91	-0.0104409375767264\\
92	0.00603292451719721\\
93	0.000832142858781931\\
94	-0.00870241848727529\\
95	0.0156192151720988\\
96	-0.019497071791841\\
97	0.0185672043100425\\
98	-0.0117908202369515\\
99	-0.00493213713458223\\
100	0.00558342656177428\\
101	-0.0264326525355788\\
102	0.13032544473625\\
103	-0.0859586041211799\\
104	-0.156609397766655\\
105	0.185131294398993\\
106	-0.00270372309489277\\
107	-0.0744791459515443\\
108	0.0573758521583257\\
109	-0.0373686235980185\\
110	0.017411139249784\\
111	-0.000243728761747558\\
112	-0.0120810351320317\\
113	0.0185206223640755\\
114	-0.0191545266283438\\
115	0.0150729960399814\\
116	-0.008076220310715\\
117	0.000255633564660609\\
118	0.00644972626061425\\
119	-0.010628339706409\\
120	0.0116289752987957\\
121	-0.00961198625260342\\
122	0.00542156524755536\\
123	-0.000320286158779367\\
124	-0.00434503826912241\\
125	0.00746634497800928\\
126	-0.00841404213796629\\
127	0.00714496624313993\\
128	-0.00416341919409036\\
129	0.000355061025625814\\
130	0.00325819584001476\\
131	-0.00577860462095346\\
132	0.00664491480490685\\
133	-0.00574847933892765\\
134	0.00343089812709265\\
135	-0.000370810522906455\\
136	-0.00260606839811818\\
137	0.00474266924043256\\
138	-0.00553736694642979\\
139	0.00485701775055242\\
140	-0.00295177274920598\\
141	0.00037268865679759\\
142	0.00218272153546237\\
143	-0.00405497881428097\\
144	0.00478940558985573\\
145	-0.00424379536322577\\
146	0.00261209772168508\\
147	-0.00036135030352947\\
148	-0.00189888264526889\\
149	0.00357888344330038\\
150	-0.0042604121350617\\
151	0.00380757436652693\\
152	-0.0023271534115163\\
153	0.000382192568194598\\
154	0.00179358776228596\\
155	-0.00313317524605754\\
156	0.00402815157237609\\
157	-0.00326783152914075\\
158	0.00238373114980378\\
159	1.94334087380965e-05\\
160	-0.00126425390911275\\
161	0.00343471469342407\\
162	-0.00314297516457721\\
163	0.00444590781159389\\
164	0.000895141056517406\\
165	0.00505914687185495\\
166	0.0092121366891496\\
167	0.00127296148396519\\
168	-0.00256939266075346\\
169	-0.0226089412503548\\
170	-0.036851936821143\\
171	-0.0307479102778015\\
172	0.00525903078387847\\
173	0.0607470007534062\\
174	0.127190584685893\\
175	0.0672719182519593\\
176	-0.158147108240976\\
177	-0.15419461803309\\
178	0.0600617210635463\\
179	0.0825745818850831\\
180	-0.0100187787753938\\
181	-0.0260997105566183\\
182	0.00111242065664737\\
183	-0.00120625285574739\\
184	-0.00813695561132968\\
185	0.000251753939067187\\
186	-0.00758248779161592\\
187	-0.000437438676825044\\
188	-0.00323632601267295\\
189	-0.00308662552385153\\
190	0.0183163924274325\\
191	-0.010960715164633\\
192	-0.0143892370993463\\
193	0.0263150276568802\\
194	0.00824376704270032\\
195	-0.00272275860443068\\
196	0.0170200840789255\\
197	0.00933505093074593\\
198	0.0135241584481246\\
199	0.0185199514532914\\
200	0.00866789549203257\\
};
\addplot [color=mycolor2, forget plot,line width=0.5mm]
  table[row sep=crcr]{%
1	0\\
2	0\\
3	0\\
4	0\\
5	0\\
6	0\\
7	0\\
8	0\\
9	0\\
10	0\\
11	0\\
12	0\\
13	0\\
14	0\\
15	0\\
16	0\\
17	0\\
18	0\\
19	0\\
20	0\\
21	0\\
22	0\\
23	0\\
24	0\\
25	0\\
26	0\\
27	0\\
28	0\\
29	0\\
30	0\\
31	0\\
32	0\\
33	0\\
34	0\\
35	0\\
36	0\\
37	0\\
38	0\\
39	0\\
40	0\\
41	0\\
42	0\\
43	0\\
44	0\\
45	0\\
46	0\\
47	0\\
48	0\\
49	0\\
50	0\\
51	0\\
52	0\\
53	0\\
54	0\\
55	0\\
56	0\\
57	0\\
58	0\\
59	0\\
60	0\\
61	0\\
62	0\\
63	0\\
64	0\\
65	0\\
66	0\\
67	0\\
68	0\\
69	0\\
70	0\\
71	0\\
72	0\\
73	0\\
74	0\\
75	0\\
76	0\\
77	0\\
78	0\\
79	0\\
80	0\\
81	0\\
82	0\\
83	0\\
84	0\\
85	0\\
86	0\\
87	0\\
88	0\\
89	0\\
90	0\\
91	0\\
92	0\\
93	0\\
94	0\\
95	0\\
96	0\\
97	0\\
98	0\\
99	-0.00351313804668061\\
100	-0.0109010266592199\\
101	0.0102242072291343\\
102	0.0620037936122769\\
103	-0.00927688192095555\\
104	-0.209142889332576\\
105	0.236978340115651\\
106	-0.0763724049976303\\
107	0\\
108	0\\
109	0\\
110	0\\
111	0\\
112	0\\
113	0\\
114	0\\
115	0\\
116	0\\
117	0\\
118	0\\
119	0\\
120	0\\
121	0\\
122	0\\
123	0\\
124	0\\
125	0\\
126	0\\
127	0\\
128	0\\
129	0\\
130	0\\
131	0\\
132	0\\
133	0\\
134	0\\
135	0\\
136	0\\
137	0\\
138	0\\
139	0\\
140	0\\
141	0\\
142	0\\
143	0\\
144	0\\
145	0\\
146	0\\
147	0\\
148	0\\
149	0\\
150	0\\
151	-5.30346831698016e-06\\
152	-1.64562988248556e-05\\
153	-3.09796425290564e-05\\
154	-5.04185606140431e-05\\
155	-7.17389027499684e-05\\
156	-9.48387690840501e-05\\
157	-0.000122651878953826\\
158	-0.000154641367036869\\
159	-0.000187005910954627\\
160	-0.000219866954538323\\
161	-0.000254782000958066\\
162	-0.000290888767979487\\
163	0.000320946063654758\\
164	0.00164758770622277\\
165	0.0033863720293048\\
166	0.00572655519349118\\
167	0.00263661033801139\\
168	-0.00647825027131745\\
169	-0.0186131041141609\\
170	-0.0354824346607803\\
171	-0.0280717773959356\\
172	0.00644377111903955\\
173	0.0528588148027425\\
174	0.119285897314719\\
175	0.0599797701375527\\
176	-0.139573238285297\\
177	-0.136441539638581\\
178	0.055853019685268\\
179	0.0817661733605682\\
180	-0.00740531273156058\\
181	-0.0154599403545823\\
182	0.00661795671376851\\
183	0.00486723334745235\\
184	0.0020959122370631\\
185	0.00252887977468045\\
186	0.00211508707222433\\
187	0.00118501863682445\\
188	-0.000302508435239946\\
189	0.00172189715506959\\
190	0.00756004600407074\\
191	-0.00254176150837254\\
192	-0.0293041425557245\\
193	0.0276540012091323\\
194	-0.0139021529530547\\
195	-0.0050770417370502\\
196	-0.00633176891201051\\
197	-0.00765642028398766\\
198	-0.0090861796706866\\
199	-0.00941408630693424\\
200	-0.00850705138839587\\
};
\end{axis}

\begin{axis}[%
xshift=-1cm,
width=\fwidth,
height=\fheight,
scale only axis,
separate axis lines,
every outer x axis line/.append style={darkgray!60!black},
every x tick label/.append style={font=\color{darkgray!60!black}},
at={(4.019in,0.292in)},
scale only axis,
xmin=15.4761904761905,
xmax=166.269841269841,
ymin=-0.271508379888268,
ymax=0.241340782122905,
axis background/.style={fill=white},
title={Original + low pass}
]
\addplot [color=mycolor1, forget plot,line width=0.5mm]
  table[row sep=crcr]{%
1	-0.000286061624359592\\
2	-0.000846667906693175\\
3	0.001715523393525\\
4	-0.00209882415112596\\
5	0.00191105088661467\\
6	-0.00122139511355449\\
7	0.000230597735944419\\
8	0.000786623368187795\\
9	-0.0015572758338678\\
10	0.00188320632204591\\
11	-0.00169209709248616\\
12	0.00105320638070935\\
13	-0.000154490436880303\\
14	-0.000751772740162952\\
15	0.00141937921240519\\
16	-0.0016751158210419\\
17	0.0014640347599792\\
18	-0.000861172924188925\\
19	-0.000116994879267343\\
20	0.000235053245278522\\
21	-0.00225363574801269\\
22	-0.000103084962589772\\
23	-0.00198776556690191\\
24	0.00214318530327393\\
25	0.00470830283920787\\
26	0.00809875082833615\\
27	0.0082189801018279\\
28	-0.00281986341450236\\
29	-0.012371115960746\\
30	-0.0302013811405909\\
31	-0.0152875186140071\\
32	0.0362304759729051\\
33	0.0335674417961683\\
34	-0.0128079798451848\\
35	-0.0214302141303982\\
36	-0.000459498297174085\\
37	0.000321203679293801\\
38	-0.0100815910749637\\
39	-0.00356632594041721\\
40	0.00750649493915219\\
41	0.0233722725415592\\
42	0.0466219765973374\\
43	0.034704093325862\\
44	-0.00679037728879114\\
45	-0.0689936225229944\\
46	-0.153929980059759\\
47	-0.0752973707854892\\
48	0.179297762108189\\
49	0.180949135325234\\
50	-0.0728577421085516\\
51	-0.102805485267936\\
52	0.0136529916849763\\
53	0.0189749228348323\\
54	0.000971673722761715\\
55	-0.00923756317092622\\
56	0.00585929288042115\\
57	0.00111854829619137\\
58	5.72723408084225e-05\\
59	0.00901435735812265\\
60	-0.0582867269764356\\
61	0.0383179991056855\\
62	0.074772284835243\\
63	-0.0863856718954189\\
64	-0.00151843908453392\\
65	0.0389231497242348\\
66	-0.0307611637773025\\
67	0.0202613831119902\\
68	-0.00907367851531699\\
69	-0.00105886540171644\\
70	0.00863782851668232\\
71	-0.0126992451091541\\
72	0.0129922800744833\\
73	-0.010007386755693\\
74	0.00484757054658649\\
75	0.00102861787653968\\
76	-0.0061348269263989\\
77	0.009281827996352\\
78	-0.00982680966181874\\
79	0.00779176860215084\\
80	-0.00382611027715007\\
81	-0.000975203711885713\\
82	0.00535303510310783\\
83	-0.00819702840288825\\
84	0.00881525287960075\\
85	-0.00709647373737152\\
86	0.00352712972523722\\
87	0.000939675258225918\\
88	-0.00512995059844402\\
89	0.007948220209058\\
90	-0.0086562512526944\\
91	0.00706174317896482\\
92	-0.00356955779395233\\
93	-0.000918163808886303\\
94	0.00522920758743511\\
95	-0.00822003954562146\\
96	0.00907172589758147\\
97	-0.00750819216525019\\
98	0.00388010181328256\\
99	0.000909831481397385\\
100	-0.00562506609499053\\
101	0.00900558278589486\\
102	-0.0100944527387752\\
103	0.00849840029679282\\
104	-0.00451403541845614\\
105	-0.000917325574179366\\
106	0.00642190834194957\\
107	-0.010526680353108\\
108	0.0120354411525458\\
109	-0.0103557195406378\\
110	0.00569229143089311\\
111	0.000948204356287293\\
112	-0.00795380713035601\\
113	0.0134691446368323\\
114	-0.0158453675774437\\
115	0.0140757531909401\\
116	-0.00811262694710512\\
117	-0.00101646245971983\\
118	0.0113331615323808\\
119	-0.0202601038960024\\
120	0.025134757355032\\
121	-0.0237807603071422\\
122	0.0150150102246644\\
123	0.00595025031216452\\
124	-0.00749779066151394\\
125	0.033212288112627\\
126	-0.15882139147255\\
127	0.105055872929902\\
128	0.188310993708331\\
129	-0.223470860367916\\
130	0.00396685113099463\\
131	0.088685005975975\\
132	-0.0678479542765364\\
133	0.0438821905172866\\
134	-0.0202983307333077\\
135	0.000265595857128709\\
136	0.0139232131653085\\
137	-0.0211915315523397\\
138	0.021766433186146\\
139	-0.0170104739046682\\
140	0.00904735366515859\\
141	-0.000270576050362674\\
142	-0.00715667071755788\\
143	0.0117070977466333\\
144	-0.0127237459204841\\
145	0.0104482555617067\\
146	-0.00585464080675999\\
147	0.000341448467303136\\
148	0.00463786244897374\\
149	-0.00791896032111643\\
150	0.00887019200795681\\
151	-0.00748917289887775\\
152	0.00434295327482321\\
153	-0.000378922435516281\\
154	-0.00333852781608218\\
155	0.0058973534144195\\
156	-0.00674860766043704\\
157	0.0058125274803426\\
158	-0.00346210946973506\\
159	0.000396952549025753\\
160	0.00255424812218322\\
161	-0.00464940089972881\\
162	0.00541273735628557\\
163	-0.00473656824662533\\
164	0.00288479361074342\\
165	-0.000403294040014477\\
166	-0.00203495660596304\\
167	0.00380761894491273\\
168	-0.0044977985041203\\
169	0.00398837876105718\\
170	-0.00247542483126585\\
171	0.000401542262006285\\
172	0.00166986940384651\\
173	-0.00320469290106706\\
174	0.00383327563961715\\
175	-0.00343711985284578\\
176	0.00216731918836357\\
177	-0.000393189251685381\\
178	-0.00140265678826864\\
179	0.00275361694133753\\
180	-0.00332865473789143\\
181	0.00301188121109785\\
182	-0.00192340817748685\\
183	0.000378510356802575\\
184	0.00120192910507113\\
185	-0.00240470730769337\\
186	0.00293118842299077\\
187	-0.00267042588792283\\
188	0.00172086651684067\\
189	-0.000356895267494826\\
190	-0.00104912599947197\\
191	0.00212751873447952\\
192	-0.00260769804543471\\
193	0.00238538999040356\\
194	-0.00154412239071061\\
195	0.00032689367032099\\
196	0.000932958305967389\\
197	-0.00190229689443983\\
198	0.00233574514017623\\
199	-0.00213738737536715\\
200	0.00138119712475934\\
};
\addplot [color=mycolor2, forget plot,line width=0.5mm]
  table[row sep=crcr]{%
1	0\\
2	0\\
3	0\\
4	0\\
5	0\\
6	0\\
7	0\\
8	0\\
9	0\\
10	0\\
11	0\\
12	0\\
13	0\\
14	0\\
15	0\\
16	0\\
17	0\\
18	0\\
19	-0.000151215002073358\\
20	-0.00046920979106316\\
21	-0.00088330624966488\\
22	-0.00143755789460572\\
23	-0.000722069699284742\\
24	0.00140228162376979\\
25	0.00423329687406194\\
26	0.00817182527915777\\
27	0.00648488103761647\\
28	-0.00148268706779343\\
29	-0.0121994417216138\\
30	-0.0275472898123108\\
31	-0.0137465411312783\\
32	0.0325677470618465\\
33	0.031896998777922\\
34	-0.012625966394897\\
35	-0.0193241596410194\\
36	-0.000166976922968838\\
37	-0.000291525299605699\\
38	-0.00809664000490295\\
39	-0.00407933883538131\\
40	0.0071621969381565\\
41	0.0211651977826944\\
42	0.0408566428068719\\
43	0.0324224342968649\\
44	-0.00741298471930224\\
45	-0.0609935009421292\\
46	-0.137728076863014\\
47	-0.0687285277945265\\
48	0.162828837288547\\
49	0.159475299723381\\
50	-0.0631259946787605\\
51	-0.092835009696378\\
52	0.0108939849259971\\
53	0.0206224338730677\\
54	-0.00454616082768688\\
55	-0.00234590635265316\\
56	0.000756029052959431\\
57	0.0016528019032629\\
58	0.0051285310655247\\
59	-0.00481011248152748\\
60	-0.0291704984917172\\
61	0.00436443085684826\\
62	0.0983940172431853\\
63	-0.111489570398552\\
64	0.0359304003029756\\
65	0\\
66	0\\
67	0\\
68	0\\
69	0\\
70	0\\
71	0\\
72	0\\
73	0\\
74	0\\
75	0\\
76	0\\
77	0\\
78	0\\
79	0\\
80	0\\
81	0\\
82	0\\
83	0\\
84	0\\
85	0\\
86	0\\
87	0\\
88	0\\
89	0\\
90	0\\
91	0\\
92	0\\
93	0\\
94	0\\
95	0\\
96	0\\
97	0\\
98	0\\
99	0\\
100	0\\
101	0\\
102	0\\
103	0\\
104	0\\
105	0\\
106	0\\
107	0\\
108	0\\
109	0\\
110	0\\
111	0\\
112	0\\
113	0\\
114	0\\
115	0\\
116	0\\
117	0\\
118	0\\
119	0\\
120	0\\
121	0\\
122	0\\
123	0.00423263436496211\\
124	0.0131335744391763\\
125	-0.0123181413020239\\
126	-0.0747022701966714\\
127	0.0111768022481873\\
128	0.251975689202687\\
129	-0.285511885043318\\
130	0.0920135962870013\\
131	0\\
132	0\\
133	0\\
134	0\\
135	0\\
136	0\\
137	0\\
138	0\\
139	0\\
140	0\\
141	0\\
142	0\\
143	0\\
144	0\\
145	0\\
146	0\\
147	0\\
148	0\\
149	0\\
150	0\\
151	0\\
152	0\\
153	0\\
154	0\\
155	0\\
156	0\\
157	0\\
158	0\\
159	0\\
160	0\\
161	0\\
162	0\\
163	0\\
164	0\\
165	0\\
166	0\\
167	0\\
168	0\\
169	0\\
170	0\\
171	0\\
172	0\\
173	0\\
174	0\\
175	0\\
176	0\\
177	0\\
178	0\\
179	0\\
180	0\\
181	0\\
182	0\\
183	0\\
184	0\\
185	0\\
186	0\\
187	0\\
188	0\\
189	0\\
190	0\\
191	0\\
192	0\\
193	0\\
194	0\\
195	0\\
196	0\\
197	0\\
198	0\\
199	0\\
200	0\\
};
\end{axis}

\begin{axis}[%
yshift=-1cm,
width=\fwidth,
height=\fheight,
scale only axis,
separate axis lines,
every outer x axis line/.append style={darkgray!60!black},
every x tick label/.append style={font=\color{darkgray!60!black}},
at={(0.259in,3.062in)},
scale only axis,
xmin=79.7619047619048,
xmax=195.634920634921,
ymin=-0.220786516853933,
ymax=0.251123595505618,
axis background/.style={fill=white},
title={Original + result of blind super-resolution}
]
\addplot [color=mycolor1, forget plot,line width=0.5mm]
  table[row sep=crcr]{%
1	0\\
2	0\\
3	0\\
4	0\\
5	0\\
6	0\\
7	0\\
8	0\\
9	0\\
10	0\\
11	0\\
12	0\\
13	0\\
14	0\\
15	0\\
16	0\\
17	0\\
18	0\\
19	0\\
20	0\\
21	0\\
22	0\\
23	0\\
24	0\\
25	0\\
26	0\\
27	0\\
28	0\\
29	0\\
30	0\\
31	0\\
32	0\\
33	0\\
34	0\\
35	0\\
36	0\\
37	0\\
38	0\\
39	0\\
40	0\\
41	0\\
42	0\\
43	0\\
44	0\\
45	0\\
46	0\\
47	0\\
48	0\\
49	0\\
50	0\\
51	0\\
52	0\\
53	0\\
54	0\\
55	0\\
56	0\\
57	0\\
58	0\\
59	0\\
60	0\\
61	0\\
62	0\\
63	0\\
64	0\\
65	0\\
66	0\\
67	0\\
68	0\\
69	0\\
70	0\\
71	0\\
72	0\\
73	0\\
74	0\\
75	0\\
76	0\\
77	0\\
78	0\\
79	0\\
80	0\\
81	0\\
82	0\\
83	0\\
84	0\\
85	0\\
86	0\\
87	0\\
88	0\\
89	0\\
90	0\\
91	0\\
92	0\\
93	0\\
94	0\\
95	0\\
96	0\\
97	0\\
98	0\\
99	-0.00351313804668061\\
100	-0.0109010266592199\\
101	0.0102242072291343\\
102	0.0620037936122769\\
103	-0.00927688192095555\\
104	-0.209142889332576\\
105	0.236978340115651\\
106	-0.0763724049976303\\
107	0\\
108	0\\
109	0\\
110	0\\
111	0\\
112	0\\
113	0\\
114	0\\
115	0\\
116	0\\
117	0\\
118	0\\
119	0\\
120	0\\
121	0\\
122	0\\
123	0\\
124	0\\
125	0\\
126	0\\
127	0\\
128	0\\
129	0\\
130	0\\
131	0\\
132	0\\
133	0\\
134	0\\
135	0\\
136	0\\
137	0\\
138	0\\
139	0\\
140	0\\
141	0\\
142	0\\
143	0\\
144	0\\
145	0\\
146	0\\
147	0\\
148	0\\
149	0\\
150	0\\
151	-5.30346831698016e-06\\
152	-1.64562988248556e-05\\
153	-3.09796425290564e-05\\
154	-5.04185606140431e-05\\
155	-7.17389027499684e-05\\
156	-9.48387690840501e-05\\
157	-0.000122651878953826\\
158	-0.000154641367036869\\
159	-0.000187005910954627\\
160	-0.000219866954538323\\
161	-0.000254782000958066\\
162	-0.000290888767979487\\
163	0.000320946063654758\\
164	0.00164758770622277\\
165	0.0033863720293048\\
166	0.00572655519349118\\
167	0.00263661033801139\\
168	-0.00647825027131745\\
169	-0.0186131041141609\\
170	-0.0354824346607803\\
171	-0.0280717773959356\\
172	0.00644377111903955\\
173	0.0528588148027425\\
174	0.119285897314719\\
175	0.0599797701375527\\
176	-0.139573238285297\\
177	-0.136441539638581\\
178	0.055853019685268\\
179	0.0817661733605682\\
180	-0.00740531273156058\\
181	-0.0154599403545823\\
182	0.00661795671376851\\
183	0.00486723334745235\\
184	0.0020959122370631\\
185	0.00252887977468045\\
186	0.00211508707222433\\
187	0.00118501863682445\\
188	-0.000302508435239946\\
189	0.00172189715506959\\
190	0.00756004600407074\\
191	-0.00254176150837254\\
192	-0.0293041425557245\\
193	0.0276540012091323\\
194	-0.0139021529530547\\
195	-0.0050770417370502\\
196	-0.00633176891201051\\
197	-0.00765642028398766\\
198	-0.0090861796706866\\
199	-0.00941408630693424\\
200	-0.00850705138839587\\
};
\addplot [color=mycolor2, forget plot,line width=0.5mm]
  table[row sep=crcr]{%
1	0\\
2	0\\
3	0\\
4	0\\
5	0\\
6	0\\
7	0\\
8	0\\
9	0\\
10	0\\
11	0\\
12	0\\
13	0\\
14	0\\
15	0\\
16	0\\
17	0\\
18	0\\
19	0\\
20	0\\
21	0\\
22	0\\
23	0\\
24	0\\
25	0\\
26	0\\
27	0\\
28	0\\
29	0\\
30	0\\
31	0\\
32	0\\
33	0\\
34	0\\
35	0\\
36	0\\
37	0\\
38	0\\
39	0\\
40	0\\
41	0\\
42	0\\
43	0\\
44	0\\
45	0\\
46	0\\
47	0\\
48	0\\
49	0\\
50	0\\
51	0\\
52	0\\
53	0\\
54	0\\
55	0\\
56	0\\
57	0\\
58	0\\
59	0\\
60	0\\
61	0\\
62	0\\
63	0\\
64	0\\
65	0\\
66	0\\
67	0\\
68	0\\
69	0\\
70	0\\
71	0\\
72	0\\
73	0\\
74	0\\
75	0\\
76	0\\
77	0\\
78	0\\
79	0\\
80	0\\
81	0\\
82	0\\
83	0\\
84	0\\
85	0\\
86	0\\
87	0\\
88	0\\
89	0\\
90	0\\
91	0\\
92	0\\
93	0\\
94	0\\
95	0\\
96	0\\
97	0\\
98	0\\
99	-0.00373051717316239\\
100	-0.0115755391951491\\
101	0.0108568408481115\\
102	0.0658403438175068\\
103	-0.00985089878612181\\
104	-0.222083826463113\\
105	0.251641625157386\\
106	-0.0810980282054584\\
107	0\\
108	0\\
109	0\\
110	0\\
111	0\\
112	0\\
113	0\\
114	0\\
115	0\\
116	0\\
117	0\\
118	0\\
119	0\\
120	0\\
121	0\\
122	0\\
123	0\\
124	0\\
125	0\\
126	0\\
127	0\\
128	0\\
129	0\\
130	0\\
131	0\\
132	0\\
133	0\\
134	0\\
135	0\\
136	0\\
137	0\\
138	0\\
139	0\\
140	0\\
141	0\\
142	0\\
143	0\\
144	0\\
145	0\\
146	0\\
147	0\\
148	0\\
149	0\\
150	0\\
151	7.69392110682931e-06\\
152	2.38737100518669e-05\\
153	4.49432166443214e-05\\
154	7.31439134730651e-05\\
155	0.000104074055893115\\
156	0.000137585814336866\\
157	0.000177935234807332\\
158	0.000224343468598561\\
159	0.00027129580859171\\
160	0.000318968437465827\\
161	0.000369620877819756\\
162	0.000422002187612073\\
163	0.00113928223220231\\
164	0.00258965105018521\\
165	0.00446206597595127\\
166	0.00694951969458869\\
167	0.00383851944078271\\
168	-0.00548461268707073\\
169	-0.017926571687872\\
170	-0.035254429775193\\
171	-0.0281012829077763\\
172	0.00638799952782684\\
173	0.0527922115404258\\
174	0.119315445887919\\
175	0.058881424233174\\
176	-0.143167511424887\\
177	-0.140591922726901\\
178	0.0529752496611853\\
179	0.078448106666805\\
180	-0.0123824847268994\\
181	-0.0213488003204193\\
182	4.96426418115869e-05\\
183	-0.0020811979253784\\
184	-0.00464550042390388\\
185	-0.00366873153795451\\
186	-0.00306842845005148\\
187	-0.00286943654691587\\
188	-0.00313041144142476\\
189	0.00084964842599228\\
190	0.00933398882897311\\
191	0.000649934723316456\\
192	-0.0259661372074788\\
193	0.0374740807770349\\
194	-0.00483792472715695\\
195	0.00736543639864456\\
196	0.00918571160681934\\
197	0.0111074282158192\\
198	0.0131816285816044\\
199	0.0136573338444448\\
200	0.0123414675684024\\
};
\end{axis}

\begin{axis}[%
xshift=-1cm,
yshift=-1cm,
width=\fwidth,
height=\fheight,
scale only axis,
separate axis lines,
every outer x axis line/.append style={darkgray!60!black},
every x tick label/.append style={font=\color{darkgray!60!black}},
at={(4.019in,3.062in)},
scale only axis,
xmin=25,
xmax=169.444444444444,
ymin=-0.303651685393258,
ymax=0.250842696629214,
axis background/.style={fill=white},
title={Original + result of blind super-resolution}
]
\addplot [color=mycolor1, forget plot,line width=0.5mm]
  table[row sep=crcr]{%
1	0\\
2	0\\
3	0\\
4	0\\
5	0\\
6	0\\
7	0\\
8	0\\
9	0\\
10	0\\
11	0\\
12	0\\
13	0\\
14	0\\
15	0\\
16	0\\
17	0\\
18	0\\
19	-0.000151215002073358\\
20	-0.00046920979106316\\
21	-0.00088330624966488\\
22	-0.00143755789460572\\
23	-0.000722069699284742\\
24	0.00140228162376979\\
25	0.00423329687406194\\
26	0.00817182527915777\\
27	0.00648488103761647\\
28	-0.00148268706779343\\
29	-0.0121994417216138\\
30	-0.0275472898123108\\
31	-0.0137465411312783\\
32	0.0325677470618465\\
33	0.031896998777922\\
34	-0.012625966394897\\
35	-0.0193241596410194\\
36	-0.000166976922968838\\
37	-0.000291525299605699\\
38	-0.00809664000490295\\
39	-0.00407933883538131\\
40	0.0071621969381565\\
41	0.0211651977826944\\
42	0.0408566428068719\\
43	0.0324224342968649\\
44	-0.00741298471930224\\
45	-0.0609935009421292\\
46	-0.137728076863014\\
47	-0.0687285277945265\\
48	0.162828837288547\\
49	0.159475299723381\\
50	-0.0631259946787605\\
51	-0.092835009696378\\
52	0.0108939849259971\\
53	0.0206224338730677\\
54	-0.00454616082768688\\
55	-0.00234590635265316\\
56	0.000756029052959431\\
57	0.0016528019032629\\
58	0.0051285310655247\\
59	-0.00481011248152748\\
60	-0.0291704984917172\\
61	0.00436443085684826\\
62	0.0983940172431853\\
63	-0.111489570398552\\
64	0.0359304003029756\\
65	0\\
66	0\\
67	0\\
68	0\\
69	0\\
70	0\\
71	0\\
72	0\\
73	0\\
74	0\\
75	0\\
76	0\\
77	0\\
78	0\\
79	0\\
80	0\\
81	0\\
82	0\\
83	0\\
84	0\\
85	0\\
86	0\\
87	0\\
88	0\\
89	0\\
90	0\\
91	0\\
92	0\\
93	0\\
94	0\\
95	0\\
96	0\\
97	0\\
98	0\\
99	0\\
100	0\\
101	0\\
102	0\\
103	0\\
104	0\\
105	0\\
106	0\\
107	0\\
108	0\\
109	0\\
110	0\\
111	0\\
112	0\\
113	0\\
114	0\\
115	0\\
116	0\\
117	0\\
118	0\\
119	0\\
120	0\\
121	0\\
122	0\\
123	0.00423263436496211\\
124	0.0131335744391763\\
125	-0.0123181413020239\\
126	-0.0747022701966714\\
127	0.0111768022481873\\
128	0.251975689202687\\
129	-0.285511885043318\\
130	0.0920135962870013\\
131	0\\
132	0\\
133	0\\
134	0\\
135	0\\
136	0\\
137	0\\
138	0\\
139	0\\
140	0\\
141	0\\
142	0\\
143	0\\
144	0\\
145	0\\
146	0\\
147	0\\
148	0\\
149	0\\
150	0\\
151	0\\
152	0\\
153	0\\
154	0\\
155	0\\
156	0\\
157	0\\
158	0\\
159	0\\
160	0\\
161	0\\
162	0\\
163	0\\
164	0\\
165	0\\
166	0\\
167	0\\
168	0\\
169	0\\
170	0\\
171	0\\
172	0\\
173	0\\
174	0\\
175	0\\
176	0\\
177	0\\
178	0\\
179	0\\
180	0\\
181	0\\
182	0\\
183	0\\
184	0\\
185	0\\
186	0\\
187	0\\
188	0\\
189	0\\
190	0\\
191	0\\
192	0\\
193	0\\
194	0\\
195	0\\
196	0\\
197	0\\
198	0\\
199	0\\
200	0\\
};
\addplot [color=mycolor2, forget plot,line width=0.5mm]
  table[row sep=crcr]{%
1	0\\
2	0\\
3	0\\
4	0\\
5	0\\
6	0\\
7	0\\
8	0\\
9	0\\
10	0\\
11	0\\
12	0\\
13	0\\
14	0\\
15	0\\
16	0\\
17	0\\
18	0\\
19	-0.000150041144346376\\
20	-0.000465567390962221\\
21	-0.000876449285393847\\
22	-0.00142639836400737\\
23	-0.00071646438840749\\
24	0.00139139593718514\\
25	0.00420043447166738\\
26	0.00810838871455793\\
27	0.00643453994969401\\
28	-0.00147117720668598\\
29	-0.012104739418711\\
30	-0.0273334446345161\\
31	-0.0136398289446236\\
32	0.0323149288823354\\
33	0.0316493875093995\\
34	-0.0125279530495936\\
35	-0.0191952639819224\\
36	-0.000231198121055781\\
37	-0.000412601393343244\\
38	-0.00823451835538304\\
39	-0.00414849671568392\\
40	0.00730240347004614\\
41	0.0215920058869101\\
42	0.0416805399629512\\
43	0.0330762508998743\\
44	-0.00756247172089371\\
45	-0.0622234691556969\\
46	-0.14050544091067\\
47	-0.0701144771702277\\
48	0.16611237227217\\
49	0.162691208737934\\
50	-0.0643989658265945\\
51	-0.0947070798230155\\
52	0.011113668252435\\
53	0.0210382968381128\\
54	-0.00463783671487898\\
55	-0.00239321287222001\\
56	0.000771274803560894\\
57	0.00176484819523519\\
58	0.00547620303275946\\
59	-0.00513619830370667\\
60	-0.0311480169012301\\
61	0.00466030315292556\\
62	0.105064317393852\\
63	-0.119047640686395\\
64	0.0383661841165596\\
65	0\\
66	0\\
67	0\\
68	0\\
69	0\\
70	0\\
71	0\\
72	0\\
73	0\\
74	0\\
75	0\\
76	0\\
77	0\\
78	0\\
79	0\\
80	0\\
81	0\\
82	0\\
83	0\\
84	0\\
85	0\\
86	0\\
87	0\\
88	0\\
89	0\\
90	0\\
91	0\\
92	0\\
93	0\\
94	0\\
95	0\\
96	0\\
97	0\\
98	0\\
99	0\\
100	0\\
101	0\\
102	0\\
103	0\\
104	0\\
105	0\\
106	0\\
107	0\\
108	0\\
109	0\\
110	0\\
111	0\\
112	0\\
113	0\\
114	0\\
115	0\\
116	0\\
117	0\\
118	0\\
119	0\\
120	0\\
121	0\\
122	0\\
123	0.00450343146481491\\
124	0.0139738392865893\\
125	-0.0131062360563873\\
126	-0.0794816006035538\\
127	0.0118918759761455\\
128	0.268096686195563\\
129	-0.303778473597071\\
130	0.0979004773338996\\
131	0\\
132	0\\
133	0\\
134	0\\
135	0\\
136	0\\
137	0\\
138	0\\
139	0\\
140	0\\
141	0\\
142	0\\
143	0\\
144	0\\
145	0\\
146	0\\
147	0\\
148	0\\
149	0\\
150	0\\
151	0\\
152	0\\
153	0\\
154	0\\
155	0\\
156	0\\
157	0\\
158	0\\
159	0\\
160	0\\
161	0\\
162	0\\
163	0\\
164	0\\
165	0\\
166	0\\
167	0\\
168	0\\
169	0\\
170	0\\
171	0\\
172	0\\
173	0\\
174	0\\
175	0\\
176	0\\
177	0\\
178	0\\
179	0\\
180	0\\
181	0\\
182	0\\
183	0\\
184	0\\
185	0\\
186	0\\
187	0\\
188	0\\
189	0\\
190	0\\
191	0\\
192	0\\
193	0\\
194	0\\
195	0\\
196	0\\
197	0\\
198	0\\
199	0\\
200	0\\
};
\end{axis}
\end{tikzpicture}%